\documentclass[10pt,a4paper,twocolumn]{article}
\usepackage[margin=1.8cm]{geometry}
\usepackage{iftex}
\usepackage[T1]{fontenc}
\ifPDFTeX
  \usepackage[utf8]{inputenc}
\fi
\usepackage{lmodern,microtype}
\usepackage{amsmath,amssymb,amsthm,mathtools,mathrsfs}
\usepackage{booktabs,enumitem}
\usepackage{float}\usepackage[noend]{algpseudocode}
\usepackage{cuted}
\floatstyle{ruled}
\newfloat{algorithm}{htbp}{loa}
\floatname{algorithm}{Algorithm}
\usepackage{graphicx}
\usepackage[numbers,sort&compress]{natbib}
\usepackage{xcolor}
\usepackage{tikz-cd}
\usepackage{titlesec}
\usepackage[colorlinks=true,allcolors=blue!55!black]{hyperref}
\usepackage{orcidlink}
\newcommand{\Eprint}[2]{\unskip\ifnum\spacefactor=1250 \else,\fi\ \href{#1}{#2}\ignorespaces}
\makeatletter
\providecommand{\bibinfo}{}
\renewcommand{\bibinfo}[2]{\ifcsname qdm@bibinfo@#1\endcsname
  \csname qdm@bibinfo@#1\endcsname{#2}\else#2\fi}
\newcommand{\qdm@bibinfo@journal}[1]{\emph{#1}}
\makeatother

\renewcommand{\thesection}{\Roman{section}}
\renewcommand{\thesubsection}{\Alph{subsection}}
\renewcommand{\thesubsubsection}{\thesubsection.\arabic{subsubsection}}
\titleformat{\section}
  {\normalfont\normalsize\centering}{\thesection.}{0.6em}{}
\titleformat{\subsection}
  {\normalfont\normalsize\raggedright}{\thesubsection.}{0.6em}{}
\titleformat{\subsubsection}
  {\normalfont\normalsize\itshape\raggedright}{\thesubsubsection.}{0.6em}{}
\titlespacing*{\section}{0pt}{2.3ex plus .6ex minus .3ex}{.6ex}
\titlespacing*{\subsection}{0pt}{1.7ex plus .4ex minus .2ex}{.6ex}
\titlespacing*{\subsubsection}{0pt}{1.3ex plus .3ex}{.6ex}

\newtheoremstyle{qdmplain}
  {6pt}{6pt}{\itshape}{}{\normalfont}{.}{0.5em}{}
\theoremstyle{qdmplain}

\makeatletter
\newcommand{\qdmfrontmatter}{  {\Large\normalfont\@title\par}  \vspace{.8em}  {\normalsize\begin{tabular}[t]{c}\@author\end{tabular}\par}  \vspace{1em}}
\makeatother
\setlist{leftmargin=*,itemsep=3pt,topsep=4pt}
\newtheorem{theorem}{Theorem}

\newtheorem{alem}{Lemma}[section]

\newtheorem{acor}[alem]{Corollary}
\makeatletter
\providecommand{\theHALG@line}{\thealgorithm.\arabic{ALG@line}}
\makeatother
\algrenewcommand\algorithmicrequire{\textbf{Input:}}
\algrenewcommand\algorithmicensure{\textbf{Output:}}
\newcommand{\AlgPart}[1]{\Statex\hspace{-\algorithmicindent}\textbf{#1}}
\newcommand{\AlgGuard}{\emph{Guard:}}
\makeatletter
\newcommand{\AlgState}[1]{\State\parbox[t]{\dimexpr\linewidth-\ALG@thistlm}{#1}}
\makeatother
\DeclareMathOperator{\Tr}{Tr}
\DeclareMathOperator{\Var}{Var}

\DeclareMathOperator{\TV}{TV}
\DeclareMathOperator{\sgn}{sgn}
\DeclareMathOperator{\diag}{diag}
\DeclareMathOperator{\spn}{span}
\newcommand{\E}{\mathbb E}
\newcommand{\Prob}{\mathbb P}
\newcommand{\R}{\mathbb R}
\newcommand{\cO}{\mathcal O}

\newcommand{\cH}{\mathcal H}
\newcommand{\cS}{\mathcal S}
\newcommand{\op}{\mathrm{op}}
\newcommand{\trd}{D_{\mathrm{tr}}}
\newcommand{\obs}{\mathcal B}
\newcommand{\ket}[1]{|#1\rangle}
\newcommand{\bra}[1]{\langle#1|}
\newcommand{\braket}[2]{\langle#1|#2\rangle}
\newcommand{\ketbra}[2]{|#1\rangle\langle#2|}

\DeclareMathOperator{\polar}{pol}
\DeclareMathOperator{\dist}{dist}
\DeclareMathOperator{\infid}{infid}
\DeclareMathOperator{\Log}{Log}
\DeclareMathOperator{\skw}{skw}
\newcommand{\T}{\mathsf T}
\newcommand{\SO}{\mathrm{SO}}
\newcommand{\cA}{\mathcal A}
\newcommand{\cX}{\mathcal X}
\newcommand{\cK}{\mathcal K}

\newcommand{\dd}{\mathfrak d}
\newcommand{\sM}{\sigma_{\mathsf M}}
\newcommand{\piu}{\pi_{\mathrm u}}
\newcommand{\nP}{\|_{\Xi}}
\newcommand{\od}{\mathrm o}
\newcommand{\gref}{\widetilde\gamma}
\allowdisplaybreaks
\title{Copy-Optimal and Fast Learning of Structured Fermionic States\\ Under General Quadratic Evolution}
\author{Erfan Amidi$^1$\,\orcidlink{0009-0007-2612-6313}\quad
Ali Asadian$^1$\,\orcidlink{0000-0002-5081-0162}\quad
Ali Hamed Moosavian$^2$\,\orcidlink{0000-0002-1818-5037}\\[.4em]
\small $^1$Institute for Advanced Studies in Basic Sciences\\
\small $^2$Ferdowsi University of Mashhad}
\date{}
\begin{document}
\makeatletter
\renewcommand{\p@subsection}{\thesection.}
\makeatother
\raggedbottom
\twocolumn[{
\begin{center}
\qdmfrontmatter
\begin{minipage}{.94\textwidth}
We prove copy-optimal learning of structured non-Gaussian fermionic states
under general quadratic evolution in two access models. For the full
Gaussian orbit of a fixed paired-magic input, one fixed-basis quartic sign
per copy permits nonadaptive reconstruction within trace distance $\epsilon$
using $\Theta(n^6/\epsilon^2)$ copies at fixed confidence, optimal in the
worst case for infidelity $\epsilon^2$ even with adaptive quartet choices.
For unknown disjoint-branch inputs with at least three particles per
branch, the learner also measures input copies and sends chosen Fock
states through the same unknown circuit. With these resources the optimal
worst-case expected number of preparations is
$\Theta((n^2+\log(1/\delta))/f)$ for $n\ge6$, infidelity $0<f\le1/n$
and failure probability $0<\delta\le1/2$. The learner requires neither a minimum branch weight nor a fixed upper bound on particles per branch. Both learners return preparation circuits using
polynomial classical processing; the second also returns a structured input.
\end{minipage}
\end{center}
\vspace{1em}
}]

\section{Introduction}
A pure state of $n$ fermionic modes carries $2^n$ amplitudes, and no
experiment of reasonable size can reveal them all. Learning becomes possible
where the state has structure: one that was prepared by a short recipe can be
learned by recovering the recipe rather than the
amplitudes~\cite{gebhart2023learning,anshu2024survey,eisert2020certification}.
The question is then how much of its preparation a measurement record
discloses, and at what cost in copies.

Free fermions give the classic example. A Gaussian state is fixed by the
covariance matrix of its $2n$ Majorana operators, and polynomially many
copies suffice to learn it~\cite{aaronson2023,bittel2025}. The same
structure, however, makes Gaussian circuits efficiently simulable on a
classical computer~\cite{terhal2002}. Interest therefore turns to the next
layer, where the evolution is still Gaussian, with $O(n^2)$ angles, but the
input is not. Such states underlie fermion-sampling proposals for quantum
advantage~\cite{oszmaniec2022} and recent quantum machine-learning
architectures~\cite{kerenidis2026scalable}. They can hold correlations that
no covariance records: on the whole orbit of the paired
input~\cite{oszmaniec2022} studied below, the Majorana covariance vanishes.

Existing tools reach part of the way. Gaussian
compression~\cite{mele2025} treats states with few non-Gaussian gates, and
fermionic partial tomography~\cite{zhao2021,wan2023} estimates low-degree
Majorana correlations. Neither directly gives the reconstruction studied
here, because two unknowns meet in the data: the rotation that mixes the
modes, and the amplitudes inside each block. A kindred aim is pursued in Hamiltonian
learning~\cite{anshu2026syk}. 

Here the object to be learned is a preparation, not a state vector: a
classical recipe, an input together with a Gaussian circuit, that a
laboratory can run to produce the state again. Such a recipe has only
polynomially many parameters and is judged solely by the state it prepares,
so it need not coincide with the original. We determine, up to constants,
how many preparations a learner must consume to return one, in two settings
(Fig.~\ref{fig:learning-setting}): a known paired input under a hidden
circuit, and unknown blocks under the same kind of hidden evolution; both learners need only polynomial classical time.

Section~\ref{sec:setting-access} fixes the evolution, the states and the two
forms of access. Section~\ref{sec:results} states the two theorems. Numerical
examples and a discussion of scope follow; the procedures and proofs are
given in the appendices.

\begin{figure*}[!t]
\centering
\includegraphics[width=\textwidth]{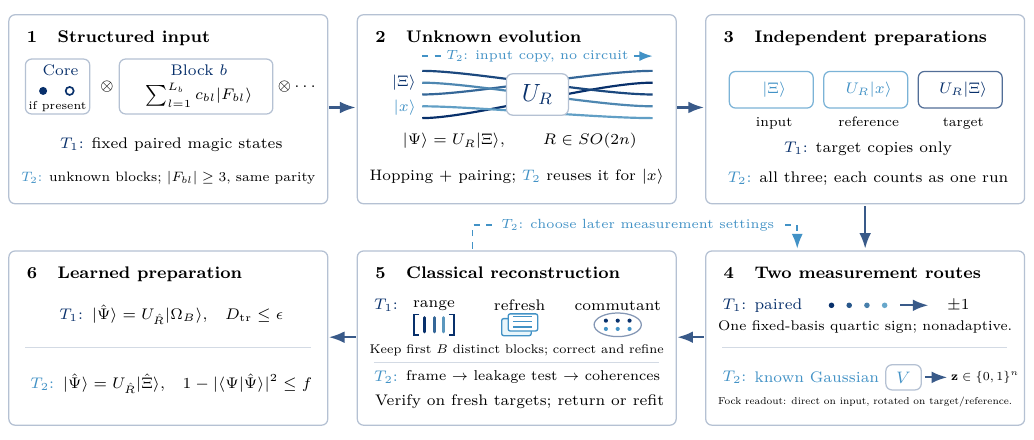}
\caption{Learning routes under unknown quadratic evolution, including pair
creation and annihilation. Inputs comprise disjoint blocks of branches and optional
core modes (filled/open dots: occupied/vacant).
Theorem~\ref{thm:reconstruction} ($T_1$) learns the fixed paired-input orbit
from target copies, with one nonadaptive quartic sign per copy.
Theorem~\ref{thm:active-main} ($T_2$) also sends chosen Fock states $\ket x$
through the same circuit, measures unevolved input copies, calibrates the
frame and may verify its proposal on fresh targets. Both return preparations,
not necessarily the originals, to the displayed accuracy with probability at
least $1-\delta$. Here $\ket\Xi$ includes the core and $\mathbf z$ is a
readout string.}
\label{fig:learning-setting}
\end{figure*}

\section{Setting and access}
\label{sec:setting-access}

\subsection{Setting}

\noindent\emph{Modes and quadratic evolution.}
There are $n$ fermionic modes, with creation and annihilation operators
defined in Appendix~\ref{app:fermions}. The unknown evolution is generated
by the quadratic Hamiltonian
\[
\begin{gathered}
H(t)=c(t)I+\sum_{i,j=1}^{n}h_{ij}(t)a_i^\dagger a_j\\
{}+\tfrac12\sum_{i,j=1}^{n}
 \bigl[\Delta_{ij}(t)a_i^\dagger a_j^\dagger
 +\overline{\Delta_{ij}(t)}\,a_j a_i\bigr],\\[2pt]
h_{ij}(t)=\overline{h_{ji}(t)},\quad
\Delta_{ij}(t)=-\Delta_{ji}(t),\quad c(t)\in\mathbb R .
\end{gathered}
\]
The hopping terms $a_i^\dagger a_j$ conserve the particle number; the pairing
terms create and annihilate pairs. The resulting evolution is a Gaussian
circuit. It preserves fermion parity and acts on the $2n$ Majorana operators
as a rotation $R\in SO(2n)$; we write $U_R$ for the circuit with rotation
$R$. A circuit without pairing terms is passive. Neither the Hamiltonian nor
$R$ is known to the learner.

\noindent\emph{Target states.}
A Fock basis state $\ket F$ has one particle in each mode of $F$ and
all other modes of the specified mode space empty. The input is a product
of blocks and core modes, written in a fixed mode ordering. Each block is
a superposition of Fock basis states, called its branches. Different
branches occupy disjoint sets of modes, and different blocks use disjoint
modes. Each core mode is either occupied or empty. The unknown Gaussian
circuit maps this pure input to the target state.

In Theorem~\ref{thm:reconstruction}, the input is the known paired state:
each block has two branches of two particles each. Only the circuit is
unknown. In Theorem~\ref{thm:active-main}, the core occupations, block
structure, branch amplitudes and circuit are all unknown. Each branch
contains at least three particles. Within each block the branch sizes
are either all even or all odd, and they may differ.

\noindent\emph{Accuracy and computation.}
The learner returns a preparation, which need not coincide with the original
one; only the prepared state is compared with the target.
Theorem~\ref{thm:reconstruction} specifies a trace-distance tolerance
$\epsilon$, Theorem~\ref{thm:active-main} an infidelity tolerance $f$. For a
pure target and a pure reconstruction, infidelity equals squared trace
distance, so the two correspond through $f=\epsilon^2$. Each guarantee holds
with probability at least $1-\delta$.

\subsection{Access models}
\label{sec:access}

\noindent\emph{Fixed-basis readout.}
In Theorem~\ref{thm:reconstruction}, the laboratory Majorana basis is fixed and known. Each fresh copy
supplies one sign, $+1$ or $-1$, of a product $B_I$ of four Majorana
operators, Eq.~\eqref{eq:main-quartic-coordinates}. Quartet choices may be randomized.
Nonadaptive acquisition fixes all quartets and batch sizes beforehand.
Adaptive acquisition may choose the next quartet from earlier outcomes.
No extra outcomes, state-dependent Gaussian prerotations or collective
measurements are included in this model.

\noindent\emph{Reference-assisted runs.}
In Theorem~\ref{thm:active-main}, the unknown input $\ket\Xi$ and target
$\ket\Psi=U_R\ket\Xi$ are defined in Section~\ref{gen:sec:setting}.
Write $\ket x$ for the Fock basis state with occupation string
$x\in\{0,1\}^n$. In each run the learner requests one fresh state
and measures it once:
\begin{enumerate}[label=(\roman*)]
\item a target copy $\ket\Psi$;
\item an evolved reference $U_R\ket x$ for a chosen $x$;
\item an unevolved input copy $\ket\Xi$.
\end{enumerate}
For requests (i) and (ii), the learner applies a known even Gaussian
circuit and then measures every mode in the Fock basis. For request
(iii), it measures the unevolved input directly in the same basis.
Each measurement returns the occupation bit of every mode.
There are no joint measurements of several states. Later runs may depend on
earlier outcomes. $N$ denotes the total number of runs.

\noindent\emph{Comparison.}
The first model uses independent copies of the target alone. The second adds
copies of the unknown input and the ability to send chosen Fock states
$\ket x$ through the same repeatable unknown circuit. This capability is an
extra experimental resource; it does not follow from access to target copies.
All three kinds of preparation are counted, although their laboratory
durations may differ. Since the input promises differ as well, neither
theorem follows from the other by a change of measurement model alone.

Both results assume ideal preparations and measurements. The paired
algorithm includes finite-precision circuit compilation. For the
reference-assisted algorithm we prove a polynomial count of real-arithmetic
operations, without a finite-bit runtime or a device-noise guarantee.
Appendix~\ref{app:learning-algorithms} gives the detailed procedures.

\section{Results}
\label{sec:results}
\subsection{Quartic learning of the paired-magic orbit}
On $n=4B$ modes, take the fixed paired input~\cite{oszmaniec2022}
and its full even Gaussian orbit,
\[
\begin{gathered}
 \ket{\Omega_B}=\bigotimes_{b=1}^{B}
 \frac{\ket{1100}_b+\ket{0011}_b}{\sqrt2},\\
 P_0=\ketbra{\Omega_B}{\Omega_B},\qquad
 \cO=\{U_R P_0 U_R^\dagger : R\in\SO(2n)\}.
\end{gathered}
\]
For \mbox{$1\le i_1<i_2<i_3<i_4\le 2n$}, write
$I=(i_1,i_2,i_3,i_4)$. The quartic coordinate map is
\[
 x:\cO\longrightarrow\mathbb R^{\binom{2n}{4}},\qquad
 \rho\longmapsto\bigl(x_I(\rho)\bigr)_I ,
\]
with one coordinate for each quartet:
\begin{equation}
 B_I=\gamma_{i_1}\gamma_{i_2}\gamma_{i_3}\gamma_{i_4},
 \qquad x_I(\rho)=\Tr(\rho B_I).
 \label{eq:main-quartic-coordinates}
\end{equation}
Here $\gamma_1,\dots,\gamma_{2n}$ are the Majorana operators and
$\trd$ denotes trace distance. Although the Majorana covariance
vanishes throughout the orbit, the quartic coordinates distinguish
its states and bound their trace distance
(Appendix~\ref{app:fermions}).

\begin{theorem}[Copy-optimal learning of the full paired-magic orbit]
\label{thm:reconstruction}
Fix $0<\epsilon_{\max}<1/\sqrt2$. For $n=4B$, $B\ge1$,
$0<\epsilon\le \epsilon_{\max}$ and $0<\delta\le1/8$, independent copies of any
unknown $\rho\in\mathcal O$ admit a nonadaptive learner in the
fixed-basis measurement model of Section~\ref{sec:access}. It returns a finite-precision
Gaussian circuit $U_{\hat R}$ such that
\begin{equation}
 \Prob\!\left\{\trd\bigl(\rho,U_{\hat R}P_0U_{\hat R}^\dagger\bigr)
                    \le\epsilon\right\}\ge1-\delta.
 \label{eq:reconstruction-guarantee}
\end{equation}
Thus the pure output has infidelity at most $\epsilon^2$. Its sufficient
worst-case copy bound is
\[
 N_{\rm copies}
 =O\!\left(\frac{n^4[n^2+\log(1/\delta)]}{\epsilon^2}\right).
\]
Every learner in the same measurement model that attains infidelity
at most $\epsilon^2$ with probability at least $1-\delta$ on every
$\rho\in\mathcal O$ requires
$\Omega_{\epsilon_{\max}}(n^4[n^2+\log(1/\delta)]/\epsilon^2)$ copies in the worst case,
even with adaptive quartet choices and unlimited classical computation.
Thus the optimal adaptive and nonadaptive copy budgets have the same order.

The constructed learner uses polynomial classical processing in
$n$, $\epsilon^{-2}$ and $\log(1/\delta)$; its output circuit consists of
$O(n^2)$ Gaussian plane rotations. No starting approximation, hidden frame or conditioning
promise is supplied. The guarantee assumes ideal measurements and gates;
optimal classical runtime and resistance to additional device noise
are not asserted.
\end{theorem}

For looser requests, run at $\min(\epsilon,1/2)$ and
$\min(\delta,1/8)$. The sufficient guarantee remains valid; matching
optimality is asserted only in the theorem's range.

\subsection{Unknown blocks under general quadratic evolution}
\label{gen:sec:setting}
\noindent\emph{States.}
Fix $n\ge6$ fermionic modes with annihilation operators $a_1,\dots,a_n$.
An occupation state is $\ket x$, $x\in\{0,1\}^n$. The input and the target are
\begin{equation}
\begin{gathered}
 \ket\Xi=\ket{1}_{C_1}\otimes\bigotimes_{b=1}^{B}\ket{\omega_b}
          \otimes\ket{0}_{C_0},\qquad
 \ket\Psi=U_R\ket\Xi,\\
 \ket{\omega_b}=\sum_{l=1}^{L_b}c_{bl}\ket{F_{bl}},\qquad
 w_{bl}=|c_{bl}|^2>0 .
\end{gathered}
\label{gen:eq:family}
\end{equation}
Each block is normalized, $\sum_{l}w_{bl}=1$. The core sets $C_1,C_0$ and
the branch sets $F_{bl}$ are pairwise disjoint.
Together they contain all $n$ modes. The number of blocks is $B\ge0$, and
$C_1$ or $C_0$ may be empty. Block $b$ consists of the modes of
$F_{b1},\dots,F_{bL_b}$. The state $\ket{F_{bl}}$ has the modes of $F_{bl}$
occupied and the other modes of block $b$ empty. Every block has $L_b\ge2$
branches. The sizes $|F_{bl}|$ of the branches of one block have the same
parity. The fermionic product in~\eqref{gen:eq:family} follows the ordering
in Appendix~\ref{app:fermions}. Since each factor has definite parity,
reordering the factors changes at most the global sign. Each branch
occupies at least three modes, so disjointness gives
\begin{equation}
 |F_{bl}|\ge3\quad\text{for all }b,l,\qquad\text{hence}\qquad
 \sum_{b=1}^{B}L_b\le\frac n3 .
 \label{gen:eq:size-three}
\end{equation}
A \emph{pattern of $\Xi$} is an occupation pattern
$C_1\cup\bigcup_bF_{b,l_b}$, one branch for each block. Two different patterns
differ in at least six modes. $U_R$ is the even Gaussian circuit of a rotation
$R\in\SO(2n)$, with conventions in~\eqref{gen:eq:conventions}. Every $R$ is allowed. The
sets, the coefficients and $R$ are unknown.

The learner returns a normalized structured input $\ket{\hat\Xi}$ and an even
Gaussian circuit $U_{\hat R}$. Together they prepare the reconstructed state
\[
 \ket{\hat\Psi}=U_{\hat R}\ket{\hat\Xi}.
\]
We measure its error by its infidelity with the target $\ket\Psi$,
\[
 X=1-|\langle\Psi\mid\hat\Psi\rangle|^2 .
\]
Different input--circuit pairs can prepare the same state, so any pair
achieving $X\le f$ meets the requested accuracy. The total cost counts all
preparations,
\[
 N=N_{\rm input}+N_{\rm reference}+N_{\rm target},
\]
including unevolved input copies, chosen Fock reference preparations and
evolved target copies. Each fresh preparation of the same reference pattern
counts as one run.
\begin{theorem}[Copy-optimal learning with reference preparations]
\label{thm:active-main}\label{thm:general-lower}
For the family of Section~\ref{gen:sec:setting} and the access of
Section~\ref{sec:access}, let $n\ge6$,
$0<f\le1/2$ and $0<\delta<1$. Algorithm~\ref{alg:general-reference}
returns a preparation satisfying, for every target,
\begin{equation}
 \Prob\{X\le f\}\ge1-\delta .
 \label{gen:eq:guarantee}
\end{equation}
It uses polynomially many real-arithmetic operations in $n$, $1/f$
and $\log(1/\delta)$. For $f\le1/n$ and $0<\delta\le1/2$,
\begin{equation}
 \sup_{\Xi,R}\E_{\Xi,R}N
 \le C_D\,\frac{n^2+\log(1/\delta)}{f},
 \label{gen:eq:bound}
\end{equation}
where $C_D$ depends only on the fixed calibration $D$ in Table~\ref{gen:tab:design}.
Over this same range the minimax expected total number of runs is
\begin{equation}
 \inf_{\mathcal A}\sup_{\Xi,R}\E_{\Xi,R} N_{\mathcal A}
 =\Theta\!\left(\frac{n^2+\log(1/\delta)}{f}\right).
 \label{gen:eq:complexity}
\end{equation}
Here $N_{\mathcal A}$ counts all preparations requested by learner $\mathcal A$.
The infimum ranges over learners satisfying~\eqref{gen:eq:guarantee}
on every target. The lower bound also allows arbitrary finite-outcome
measurements on each supplied state, adaptive stopping, mixed reports
and unlimited classical computation.
\end{theorem}

\noindent\emph{Sufficient runs.}
Algorithm~\ref{alg:general-reference} can either fit the target directly
at the requested accuracy and confidence or test a sequence of proposals
on fresh target copies. The verifier determines how accurately each
proposal should be fitted. If its infidelity is $X$, the rejection
probability $p$ on one fresh target satisfies, by~\eqref{gen:eq:verifier-gap},
\[
 X/5\le p\le X.
\]
Thus a proposal outside the requested tolerance, $X>f$, has rejection
probability greater than $f/5$. The learner aims for $X\le f/10$ so that
a successful fit has rejection probability at most half that threshold.
Counting rejections on fresh targets distinguishes the two cases.

Each proposal reaches $X\le f/10$ with probability at least $3/4$, and
the test accepts such a proposal with probability at least $2/3$. Each
fresh attempt therefore has probability at least $1/2$ of returning an
accurate preparation. Before taking data, the learner chooses the attempt
limit that minimizes its expected cost bound, including the cost of a
final fit if all proposals are rejected. A limit of zero gives a direct
fit. Counting all preparations gives the upper bound in~\eqref{gen:eq:bound}.
Equation~\eqref{gen:eq:count} bounds the cost of one fit;
Section~\ref{gen:app:verifier} gives the verification test size and
confidence allocation.

\noindent\emph{Necessary runs.}
Let $M_{\mathcal A}=\sup_{\Xi,R}\E_{\Xi,R}N_{\mathcal A}$ be the
worst-case expected run count of any learner satisfying~\eqref{gen:eq:guarantee}.
Two subfamilies explain the dependence on system size and confidence.

First, fix the input to the vacuum and vary only the Gaussian circuit.
With $d=n(n-1)$, whenever $(1+2/d)(f+\delta)\le1$,
\begin{equation}
 M_{\mathcal A}\ge\frac{d}{2f}
       \Bigl[1-\bigl(1+\tfrac2d\bigr)(f+\delta)\Bigr]^2 .
 \label{gen:eq:lower}
\end{equation}
Second, fix the circuit to the identity and vary only the amplitudes
of two disjoint three-mode occupation patterns. For $0<f<1/2$,
\begin{equation}
 M_{\mathcal A}\ge\frac{\log(1/\delta)}{-\log(1-f)}-1 .
 \label{gen:eq:confidence-lower}
\end{equation}
Every learner must handle both subfamilies, even when it chooses when to
stop. In the optimality range of Theorem~\ref{thm:active-main}, these two
necessary costs together match Algorithm~\ref{alg:general-reference}'s
upper bound up to constants, giving the minimax rate in~\eqref{gen:eq:complexity}.

\subsection{Learning procedures}
Each learner returns a classical preparation description using its stated
access model and accuracy parameters. Appendix~\ref{app:learning-algorithms}
states the procedures with their tests, retry limits and fallback outputs.
The budgets of Algorithm~\ref{alg:paired-quartic} are fixed up to one
universal constant that is not computed; those of
Algorithm~\ref{alg:general-reference} are given by~\eqref{gen:eq:validation}
and the schedule in Appendix~\ref{gen:app:schedule}.

\noindent\emph{Nonadaptive paired-orbit learner
(Theorem~\ref{thm:reconstruction}; Algorithm~\ref{alg:paired-quartic}).}
\begin{enumerate}[itemsep=2pt,topsep=3pt]
 \item Acquire four independent banks of fixed-basis quartic measurements
 at predetermined budgets, before processing any outcomes.
 \item Estimate the quartic range with the first bank and refresh it with
 the second. Form the symmetric commutant, test the top-eight eigenspaces
 of sampled pencils, and retain the first representative of each block.
 \item Assemble the hidden mode blocks, applying three frame corrections
 after each insertion with the third bank. Return the paired reference
 if a required guard fails.
 \item Refine the state in the resulting local orbit chart with the fourth
 bank, then compile and round the Gaussian circuit within the reserved error.
\end{enumerate}
For one block, the full Majorana space replaces block discovery. All
corrections are classical; measurement acquisition remains nonadaptive.
Appendix~\ref{app:reconstruction} proves the guarantee.

\noindent\emph{Reference-assisted structured-state learner
(Theorem~\ref{thm:active-main}; Algorithm~\ref{alg:general-reference}).}
Use the same fit and the same probe construction for every input.
No unknown branch count, weight or frame enters a batch size.
\begin{enumerate}[itemsep=2pt,topsep=3pt]
 \item \emph{Fit a preparation.} Read input occupations to locate the core.
 Group frequent branches by disjoint records. Assign each remaining
 eligible branch to the unique group absent from all its occupied rows. Estimate the Gaussian frame with chosen Fock references.
 Refine it and test its leakage on target copies. In the fitted frame,
 combine measured branch coherences with independently estimated weights.
 \item \emph{Certify and return.} Fix in advance the round count $j$
 minimizing the run bound~\eqref{gen:eq:finite-verified}; $j=0$ is a
 direct fit. Each round fits at $(f/10,1/4)$ and tests it on fresh
 targets. Return the first accepted proposal, else a fresh fit at
 $(f,z_j)$ with $z_j$ from~\eqref{gen:eq:validation}.
\end{enumerate}
Lemma~\ref{gen:lem:whole-test} bounds the rejection probability between
$X/5$ and $X$ for every proposed state.

\section{Numerical examples}
\label{sec:paired-numerics}
Two examples illustrate the two learners: one dimer for
Theorem~\ref{thm:reconstruction} and one block, of six to twelve modes, for
Theorem~\ref{thm:active-main}. In both examples the state evolves for time
$\tau=0.85$ under the quench Hamiltonian $H_q$ of~\eqref{eq:example-quench}.
It combines hopping and site energies with nonzero pairing, as in a
Kitaev-type chain~\cite{kitaev2001wire}. The learners receive measurement
outcomes only, not the Hamiltonian, the phase or the hidden frame.

\begin{figure}[!t]
 \centering
 \includegraphics[width=\columnwidth]{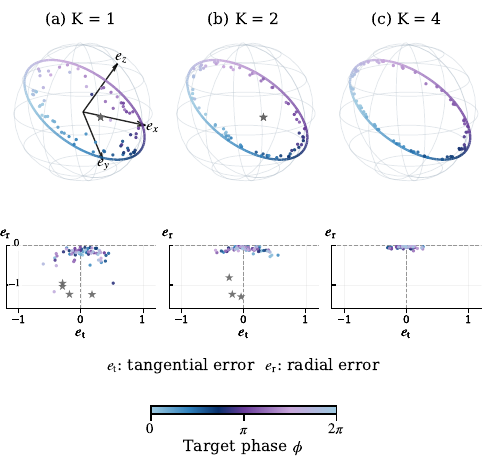}
 \caption{Four-mode logical reconstruction for $K=1,2,4$ in (a)--(c),
 with three trials at each of 24 phases; total copies follow
 Eq.~\eqref{eq:example-copies-main}.
 Dots and target circles share the target-phase colours, used only for display;
 grey stars mark fallbacks. Below each sphere, the same trials show signed
 tangential and radial errors $e_{\rm t},e_{\rm r}$ with the same axis limits in (a)--(c) ($e_{\rm r}\le0$);
 dashed lines mark zero. Black axes appear in (a).
 These are three-observable projections, not the full paired orbit or
 full-state infidelity.}
 \label{fig:paired-logical}
\end{figure}
\begin{figure}[!t]
 \centering
 \includegraphics[width=\columnwidth]{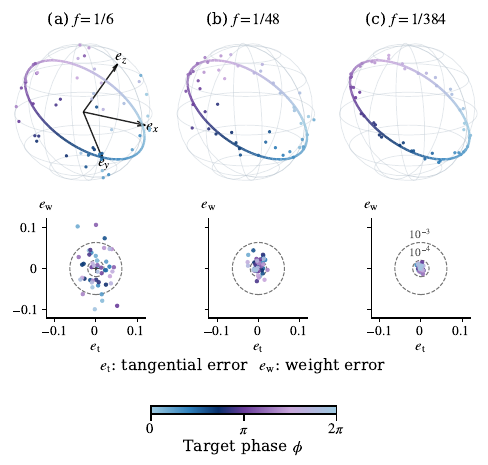}
 \caption{Two-branch reconstruction on six modes at (a) $f=1/6$,
 (b) $f=1/48$ and (c) $f=1/384$, with $\delta=1/2$ and two trials per
 phase at 24 phases. The coloured circle traces the targets. Dots show
 reconstructions projected onto the two branches after undoing the quench
 for display. Each sphere magnifies angular displacements from the corresponding
 target tenfold. Below, tangential and weight errors $(e_{\rm t},e_{\rm w})$
 appear at true scale on common axes. Distance $r$ from the origin gives
 the full-state infidelity bound $X\ge r^2/4$.}
 \label{fig:block-sphere}
\end{figure}
\subsection{One dimer: logical-state reconstruction}
We test Algorithm~\ref{alg:paired-quartic} on one dimer ($B=1$, $n=4$
modes). A dimer has two sites, each with spin-up and spin-down
modes; at zero phase its two electrons form a spin singlet. The inputs and
evolved targets are
\begin{equation}
 \begin{gathered}
 \ket{s_\phi}=\frac{e^{i\phi}\ket{1001}-\ket{0110}}{\sqrt2},\\
 \ket{\Psi_\phi^{(B)}}=e^{-i\tau H_q}\ket{s_\phi}^{\otimes B}.
 \end{gathered}
 \label{eq:example-target}
\end{equation}
The targets lie exactly on the paired
orbit of Theorem~\ref{thm:reconstruction}. The learner receives quartic
outcome counts. For one dimer, Algorithm~\ref{alg:paired-quartic} skips
block discovery, so this example exercises frame assembly and refinement only.
At the budgets shown, the refinement iterate ends on the boundary of its
chart in every trial without fallback and changes the infidelity by less
than $0.003$; the plotted accuracy is therefore set by the frame assembly.
Appendices~\ref{app:example-geometry}--\ref{app:example-records} give the
parameters, orbit membership and implementation details.
Each copy supplies one raw fixed-basis quartic sign. $K$ counts copies of
each quartic in each of the three initial banks. With
$\epsilon^2=\delta=1/8$, the four independent banks use $(K,K,K,8K)$ copies
per observable, so the total is
\begin{equation}
 N_{\rm copies}=11\binom{2n}{4}K.
 \label{eq:example-copies-main}
\end{equation}

Figure~\ref{fig:paired-logical} displays reconstructed four-mode states for
24 input phases. Each dot is a projection onto three logical expectation
values~\eqref{eq:example-logical}, evaluated after undoing the known quench
for visualization only.
Colour links each dot to its target phase on the circle; it does not encode error.
Real symmetry already confines the dots to the target circle's plane:
being in that plane does not establish recovery of the correct state.
The dots are not normalized to the sphere.
The lower panels resolve each displacement from its own target into a
signed tangential component $e_{\rm t}$ and a radial component $e_{\rm r}$,
as defined in Appendix~\ref{app:example-geometry}.
Negative $e_{\rm r}$ means an inward radial displacement; the origin marks
agreement of the logical coordinates, not a separate full-state error test.
Every fallback returns the same state with logical coordinates
$b=(0.296,0.035,0)$. Its stars coincide on the sphere but separate in the
error plot because each displacement is measured from the trial's own
target. A dot without a star passed all guards; it need not be accurate.

The implementation evaluates the real-arithmetic construction
numerically and rounds the output circuits, without certifying every
finite-bit guard. The budgets $K=1,2,4$ are not chosen from the theorem's
unspecified sufficient constant, so the plot illustrates reconstruction
on a small example with pairing rather than certifies the worst-case
confidence guarantee. Appendix~\ref{app:example-records} reports the
fallback counts and median infidelities. One system size cannot establish
polynomial scaling; the scaling guarantee follows from
Theorem~\ref{thm:reconstruction}.

\begin{figure*}[!t]
 \centering
 \includegraphics[width=\textwidth]{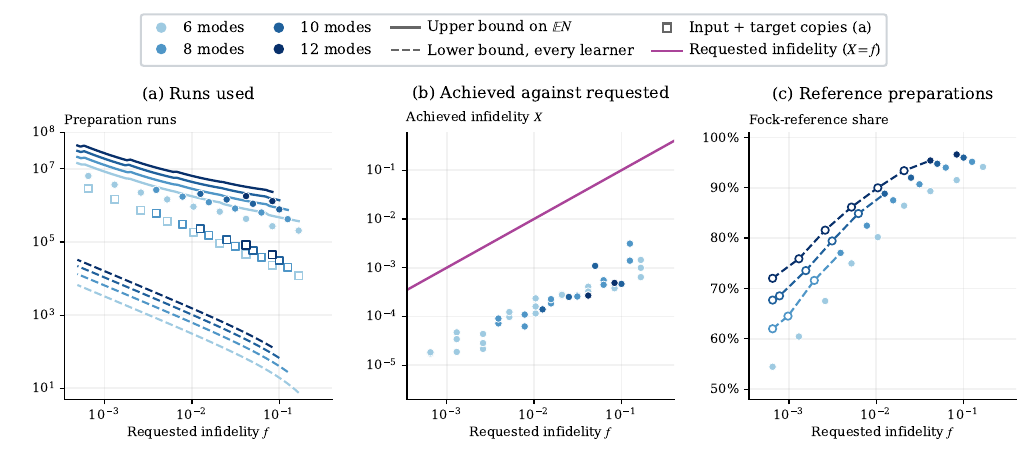}
 \caption{Cost and accuracy of Algorithm~\ref{alg:general-reference} with full test batches.
 (a) At $\delta=1/2$, filled circles count all preparations, including
 repeated Fock references; open squares count only unevolved input and evolved
 target copies. Both bounds concern total runs. Solid curves give
 the selected upper bound on expected runs; dashed curves give the
 worst-case lower bound. (b) Achieved full-state infidelity $X$ against
 the requested infidelity $f$ for the same 45 executions. The diagonal is $X=f$.
 (c) Fock-reference share in the two-branch example at $\delta=1/2$.
Filled dots show recorded runs. Open dots joined by dashed lines show
scheduled shares if every stage accepts on its first attempt.}
 \label{fig:reference-runs}
\end{figure*}
\subsection{One block with reference preparations}
\label{sec:reference-numerics}
For even $n\ge6$, let $\ket u$ occupy the first $n/2$ modes and
$\ket v$ the last $n/2$. These complementary patterns form a single
block in~\eqref{gen:eq:family}. Its input and evolved target are
\begin{equation}
 \ket{\Xi_\phi}=\frac{\ket u+e^{i\phi}\ket v}{\sqrt2},\qquad
 \ket{\Psi_\phi}=e^{-i\tau H_q}\ket{\Xi_\phi}.
 \label{eq:block-example}
\end{equation}
The two input patterns are equally likely at every phase. Recovering
$\phi$ therefore tests whether Algorithm~\ref{alg:general-reference}
also learns the coherence between them after a quench that creates and
removes pairs. Figure~\ref{fig:block-sphere} uses six modes, the same 24
phases and the same confidence at three requests, with two independent
trials at each phase. Each reconstruction used $205\,278$, $674\,692$ and
$2\,248\,632$ preparations at $f=1/6$, $1/48$ and $1/384$, respectively.
 Most runs at each accuracy were spent repeatedly preparing Fock references
 from the $2^6=64$ possible occupation patterns: $193\,220$, $583\,091$ and
 $1\,517\,947$ runs, or approximately $94\%$, $86\%$ and $68\%$ of the totals.
The deviations are too small to see on a sphere, so each dot is drawn at
ten times its angle from its own target. Below each sphere the same trials
appear at true scale. A horizontal shift measures displacement along the
circle, to first order the error in the relative phase; a vertical shift
is the difference of the two reconstructed branch weights. The distance
$r$ of a dot from the centre bounds the full-state infidelity from below,
$X\ge r^2/4$, even when weight leaves the two
branches~\eqref{eq:block-board}. These panels do not show that weight.
Appendix~\ref{app:block-example} reports the full-state infidelities of
all 144 trials; the largest ratio $X/f$ was $0.0384$ rounded upward.
Theorem~\ref{thm:active-main} guarantees only $X\le f$; the margin seen
here belongs to this example.

\subsection{Runs and achieved accuracy}
\label{sec:reference-runs}
At $\delta=1/2$ all 21 distinct requests in
Fig.~\ref{fig:reference-runs}(a,b) select direct reconstruction.
Each block has two branches of equal weight. At each size the requested
infidelity starts at $1/n$ and is halved successively. Every size uses the
same selection rule.
Across all 45 executions the largest observed ratio $X/f$ was $0.0363$ rounded upward. In panel (a), the difference between each filled circle and its open square counts the Fock reference preparations. The solid curves
bound the expected cost of Algorithm~\ref{alg:general-reference}, while the
dashed curves give a lower bound on the worst-case expected cost of any
learner with the same access. Theorem~\ref{thm:active-main} establishes copy
optimality by proving that the upper and lower bounds have the same
asymptotic order.

For this two-branch example at fixed confidence and $0<f\le1/n$,
the Fock-reference share of the expected preparation cost tends to one as
$n$ grows, with
$\E N_{\rm reference}/\E N=1-O(n^{-1/2})$
(Appendix~\ref{app:example-reference-share}).
Theorem~\ref{thm:active-main} determines the optimal order of the total
preparation count. The fraction spent on Fock references depends on the
learner and the target.

Figure~\ref{fig:size-scaling} follows this example to $n=1024$ at
$f=1/n$; its curves are exact, since the schedule depends only on $n$, $f$
and $\delta$. Full simulations up to $n=16$ used exactly the first-attempt
count and reached $X\le0.025f$. Both bounds stay within fixed multiples of
the rate, as Theorem~\ref{thm:active-main} states. Across the 510 plotted even sizes at $f=1/n$ and $\delta=1/2$,
the upper-to-lower bound ratio is below $5\times10^4$. At these settings
it bounds the possible improvement in worst-case expected run count.
\begin{figure}[!t]
 \centering
 \includegraphics[width=\columnwidth]{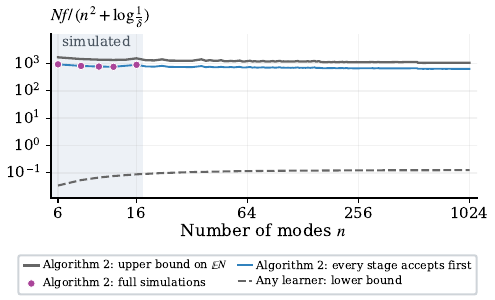}
 \caption{Runs of Algorithm~\ref{alg:general-reference} for the two-branch
 example at $f=1/n$, $\delta=1/2$, divided by $(n^2+\log(1/\delta))/f$.
 Dots are full simulations, seeded as in Appendix~\ref{app:block-example},
 all in the shaded range $n\le16$; elsewhere the curves are the exact
 schedule. The worst-case cost of any better learner lies between the two
 bounds, so only the constant can improve, not the rate.}
 \label{fig:size-scaling}
\end{figure}

\section{Scope and interpretation}
The two optimality statements concern different statistical experiments.
For the paired orbit, the worst-case upper and lower copy bounds match
in their dependence on $n$, $\epsilon$ and $\delta$ throughout the stated
range. For reference-assisted learning, the cost is the expected total
number of preparations, including tests and retries. Its optimal order is
$\Theta((n^2+\log(1/\delta))/f)$ for $n\ge6$, uniformly over $f\le1/n$ and
$0<\delta\le1/2$. A finite sufficient bound also covers larger $f$.

Neither result proves optimal classical runtime.
Lemma~\ref{lem:necessary-copies} gives $\Omega((n^2+\log(1/\delta))/f)$
copies for the paired orbit under arbitrary collective measurements of
target copies, for $n\ge4$, $f\le1/4$ and $\delta\le1/8$. It is not the
lower bound for reference access: that bound is proved in
Appendix~\ref{gen:app:lower}, for every finite-outcome measurement of single
states, using the vacuum orbit for dimension and a two-branch
input for confidence.

\section{Discussion}
The paired-magic orbit illustrates how a state can be easy to describe by
its preparation even when its Gaussian expansion is long.
Wei et al.~\cite{wei2026gaussianranks} showed that approximating $B$
copies of a fixed pure non-Gaussian block of definite parity can require
exponentially many Gaussian terms in $B$. This bound persists after
Gaussian evolution: applying the inverse circuit to each term would give
an equally short expansion of the input. Our learner represents the state
by its input blocks and the circuit acting on them. It recovers this
preparation in polynomial time, without a Gaussian expansion.

In Fig.~\ref{fig:reference-runs}(b) every reconstruction is far more
accurate than requested: the achieved infidelity never exceeds $0.037f$.
Such a margin is expected. The schedule is fixed before any data are taken
and is sized for the least favourable target on $n$ modes, whereas the
example has only two branches of equal weight. The margin therefore does
not indicate that fewer preparations would suffice in the worst case.
Theorem~\ref{thm:active-main} concerns that worst case and settles its
rate; the constants of the upper and lower bounds differ, as the two
curves of Fig.~\ref{fig:reference-runs}(a) show. Our interest is in this
worst-case rate; a sharp constant, or a budget tuned to a particular
target such as this one, is not an objective of this work.

Measurement access is part of the complexity statement. A single quartic
sign, a complete occupation record, and a reference preparation expose
different information. For reference access,
the remaining questions include the optimal rate for $f>1/n$,
two-particle branches and finite-bit implementation. Practical
constants, robustness to device noise and optimal runtime remain separate
objectives for both routes.

\section*{Code availability}
The code used for the numerical simulations and figure generation is
available from the authors upon reasonable request. A separate Lean
project formalizes selected components of the proof of
both Theorems; it does not formalize either theorem as a
whole.

\section*{Note added}
Version~1 of this article established polynomial-time reconstruction from
evolved target copies for disjoint-branch inputs with fixed particle
number per block and a fixed particle cap, under general quadratic
evolution. Version~2 proved copy-optimal learning of the paired-magic orbit
with one fixed-basis quartic measurement per copy. For unknown blocks
with fixed particle cap $r_\star$, its learner uses single-copy Gaussian
rotations followed by occupation readout and polynomial classical
processing. Its sufficient bound of $\widetilde O_{r_\star}(n^7/f)$ copies
remains our reconstruction guarantee for that family when only copies of
the evolved target are available.

Version~3 retains the paired-orbit result. For unknown blocks,
Theorem~\ref{thm:active-main} also allows the learner to measure unevolved
input copies and send chosen Fock states through the unknown circuit.
Counting all three kinds of preparation gives the optimal worst-case
expected order in~\eqref{gen:eq:complexity}, with polynomial classical
processing. Each branch contains at least three particles. Branch sizes
within a block may differ provided they have the same parity, and no
fixed upper particle cap is required.

After preparing the first version, we became aware of the work of
Heo, Shin and Oh~\cite{heo2026learning} on disjoint-branch states under
particle-number-conserving evolution. Both of our theorems allow
quadratic evolution that also creates and annihilates particle pairs.

\section*{Acknowledgments}
E.~Amidi conceived the research ideas and made the principal contribution to this work. Generative AI systems were used to assist in generating numerical examples and figures, in refining the rigor of the proofs, in developing parts of some intermediate lemmas and in writing the Lean formalization. All AI-assisted content was reviewed and verified by the authors. The research ideas, theorem statements, learning algorithms, and overall mathematical arguments were developed by the authors, who take full responsibility for the accuracy, originality, and integrity of the manuscript.
\par
\setlength{\stripsep}{4pt plus .5pt minus .5pt}
\begin{strip}
 \centering\rule{.45\textwidth}{.3pt}\par
\end{strip}
\begingroup
\renewcommand{\bibsection}{}
\interlinepenalty=10000
\bibliographystyle{unsrtnat}
\bibliography{references}
\endgroup
\clearpage
\appendix
\makeatletter
\renewcommand{\p@subsection}{}
\makeatother
\renewcommand{\thesection}{\Alph{section}}
\renewcommand{\thesubsection}{\thesection.\arabic{subsection}}
\renewcommand{\thesubsubsection}{\thesubsection.\arabic{subsubsection}}
\numberwithin{equation}{section}
\titleformat{\section}
 {\normalfont\normalsize\centering}{Appendix \thesection.}{0.6em}{}
\titleformat{\subsection}
 {\normalfont\normalsize\centering}{\thesubsection.}{0.6em}{}
\section{Observable algebra and state certificates}
\label{app:fermions}
\noindent\emph{Index conventions.}
For a finite set $S$, let $\mathcal I_k(S)=\{I\subseteq S:|I|=k\}$;
$\mathcal I_k(d)$ abbreviates $\mathcal I_k(\{1,\ldots,d\})$.
A sum over subsets includes each subset once. Empty sums are zero and
empty products are the identity.
For mode indices $1\le i,j\le n$ and Majorana indices $1\le p,q\le2n$,
the conventions are
\[
\begin{gathered}
 \{a_i,a_j\}=\{a_i^\dagger,a_j^\dagger\}=0,\quad \{a_i,a_j^\dagger\}=\delta_{ij}I,\\
 a_i\ket0=0,\quad I=\{i_1<\cdots<i_p\},\\
 a_I=a_{i_p}\cdots a_{i_1},\quad \ket I=a_I^\dagger\ket0,\\
 \gamma_{2j-1}=a_j+a_j^\dagger,\quad \gamma_{2j}=-i(a_j-a_j^\dagger),\\
 U_R^\dagger\gamma_pU_R=\sum_{q=1}^{2n}R_{pq}\gamma_q.
\end{gathered}
\]
In the fermionic product~\eqref{gen:eq:family}, write
$a_F^\dagger=\prod_{j\in F}a_j^\dagger$ in increasing order of $j$.
For each branch use the block vacuum in $\ket{F_{bl}}=a_{F_{bl}}^\dagger\ket0$.
On the full vacuum the input is
\[
 \ket\Xi=a_{C_1}^\dagger\,O_1\cdots O_B\ket0,\qquad
 O_b=\sum_{l=1}^{L_b}c_{bl}a_{F_{bl}}^\dagger .
\]
For density operators $\rho,\sigma$ and unit vectors $\psi,\phi$,
\[
\begin{gathered}
 \trd(\rho,\sigma)=\tfrac12\|\rho-\sigma\|_1,\quad \trd(\psi,\phi)^2=1-|\braket\psi\phi|^2,\\
 d_2(\psi,\phi):=\min_{\theta\in\R}\|\psi-e^{i\theta}\phi\|
 =\sqrt{2-2|\braket\psi\phi|},\\
 \trd(\psi,\phi)\le d_2(\psi,\phi)\le\sqrt2\,\trd(\psi,\phi).
\end{gathered}
\]
In trace distances, vectors denote their rank-one projectors. The paired orbit is
\[
 \cO=\{U_RP_0U_R^\dagger:R\in SO(2n)\}.
\]
\noindent\emph{Trace and Majorana basis.}
In an orthonormal basis $\{\ket\ell\}_{\ell=1}^{2^n}$,
\begin{equation}
\begin{gathered}
 \sigma=\sum_{k=1}^{2^n}s_k\ketbra{v_k}{v_k},\quad I=\sum_{\ell=1}^{2^n}\ketbra\ell\ell,\\
 \Tr(A\sigma)=\sum_{\ell,k=1}^{2^n}s_k
 \bra\ell A\ket{v_k}\braket{v_k}\ell
 =\sum_{k=1}^{2^n}s_k\bra{v_k}A\ket{v_k}.
\end{gathered}
\label{eq:spectral-trace}
\end{equation}
For a sorted $k$-element set $I$, put $s=k(k-1)/2$. Reversing the
anticommuting factors gives
\[
\begin{gathered}
 C_I=i^s\gamma_{i_1}\cdots\gamma_{i_k},\quad C_I^2=i^{2s}(-1)^sI=I,\\
 C_I^\dagger=(-i)^s\gamma_{i_k}\cdots\gamma_{i_1}
 =(-i)^s(-1)^s\gamma_{i_1}\cdots\gamma_{i_k}=C_I.
\end{gathered}
\]
If $I\ne J$, their product is a phase times $C_{I\triangle J}$.
For a nonempty even set $K$, choose $j\in K$; for an odd set, choose
$j\notin K$, which exists because $2n$ is even. Then
\begin{equation}
\begin{gathered}
\Tr C_K=\Tr(\gamma_jC_K\gamma_j)=-\Tr(C_K\gamma_j^2)
       =-\Tr C_K,\\
  2^{-n}\Tr(C_IC_J)=\delta_{IJ}.
\end{gathered}
\label{eq:majorana-orthogonality}
\end{equation}
These $2^{2n}$ orthonormal operators form an operator basis. For quartets
$C_I=-B_I$; the raw convention remains~\eqref{eq:main-quartic-coordinates}.
\subsection{Gaussian covariance of the observable representation}
\begin{alem}[Exterior action and quartic flattening]
\label{lem:closure}
For every $R\in SO(2n)$ and fixed degree $k$,
\begin{equation}
 U_R^\dagger C_IU_R=\sum_{J\in\mathcal I_k(2n)}\det(R_{I,J})C_J.
 \label{eq:wedge-action}
\end{equation}
The same formula holds for the raw quartics $B_I$.
Here $R_{I,J}$ is the $k\times k$ submatrix of $R$ with rows $I$ and
columns $J$, and $\wedge^kR$ is the matrix with entries $\det(R_{I,J})$.
It is orthogonal. Hence $x(U_R\rho U_R^\dagger)=\wedge^4(R)\,x(\rho)$ and
$\|x(U_R\rho U_R^\dagger)\|_2=\|x(\rho)\|_2$ for every density
operator $\rho$.
Let $Y(z)$ be the alternating four-tensor whose entries on increasing
index tuples are given by $z\in\R^{\binom{2n}{4}}$.
In the orthonormal basis $\ket{p\wedge q}$ with $p<q$, define $Q(z)$
by its matrix entries
\[
 Q(z)_{(pq),(rs)}=Y(z)_{pqrs}.
\]
Then
\[
 Q((\wedge^4 R)z)
 =(\wedge^2 R)\,Q(z)\,(\wedge^2 R)^{\mathsf T}.
\]
A plane generator is $H_{pq}=i\gamma_p\gamma_q/2$; its Gaussian lift
$e^{\theta\gamma_p\gamma_q/2}$ induces a real plane rotation.
Every $R\in SO(2n)$ is a product of $O(n^2)$ such rotations.
\end{alem}
\begin{proof}
\noindent\emph{Antisymmetrization.} The transformed Majoranas
$\gamma'_i=\sum_{j=1}^{2n}R_{ij}\gamma_j$ obey
$\{\gamma'_i,\gamma'_l\}=2\delta_{il}I$.
For distinct $i_1,\ldots,i_k$, their product equals its antisymmetrization.
Repeated output indices cancel. The coefficient of an ordered product
with index set $J$ is
\[
 \sum_{\pi\in S_k}\sgn(\pi)\prod_{a=1}^kR_{i_a j_{\pi(a)}}
 =\det(R_{I,J}).
\]
The common Hermitian phase gives~\eqref{eq:wedge-action}.
Cauchy--Binet yields
$(\wedge^kR)(\wedge^kR)^{\mathsf T}=\wedge^k(RR^{\mathsf T})=I$.
\noindent\emph{States and measured moments.} Let $\mathcal S_n$ be the density
operators on the $n$-mode Fock space, and define
$m_k(\rho)=(\Tr(\rho C_I))_{|I|=k}$, with sorted index sets.
For fixed $R$, the action on states and the action on their real moments
fit into the square
\[
 \begin{tikzcd}[column sep=large,row sep=large]
 \mathcal S_n\arrow[r,"\rho\,\mapsto\,U_R\rho U_R^\dagger"]
 \arrow[d,"m_k"'] & \mathcal S_n\arrow[d,"m_k"]\\
 \R^{\binom{2n}{k}}\arrow[r,"\wedge^kR"'] &\R^{\binom{2n}{k}}.
 \end{tikzcd}
\]
Conjugation preserves positivity and trace, and each $C_I$ is Hermitian,
so all maps have the displayed domains and codomains. Cyclicity and
\eqref{eq:wedge-action} verify the two compositions coordinate by coordinate:
\[
 \begin{aligned}
 [m_k(U_R\rho U_R^\dagger)]_I
 &=\Tr(\rho\,U_R^\dagger C_I U_R)\\
 &=\sum_{J\in\mathcal I_k(2n)}
   \det(R_{I,J})\,\Tr(\rho C_J)\\
 &=[(\wedge^k R)m_k(\rho)]_I.
 \end{aligned}
\]
This equivariance does not assert that $m_k$ is injective.
\noindent\emph{Pair flattening.} Put $m=\binom{2n}{2}$ and $S_R=\wedge^2R$.
Both arrows in the following square are defined on all alternating tensors:
\[
 \begin{tikzcd}[column sep=large,row sep=large]
 \R^{\binom{2n}{4}}\arrow[r,"\wedge^4R"]\arrow[d,"Q"']
 &\R^{\binom{2n}{4}}\arrow[d,"Q"]\\
 \operatorname{Sym}_m(\R)\arrow[r,"A\,\mapsto\,S_RA S_R^{\mathsf T}"']
 &\operatorname{Sym}_m(\R).
 \end{tikzcd}
\]
Antisymmetrizing both index pairs verifies the two compositions:
\begin{equation}
 \begin{array}{@{}r@{}l@{}}
 \displaystyle Q((\wedge^4 R)z)_{(pq),(rs)}
 & \displaystyle {}=\sum_{a,b,c,d=1}^{2n}
   R_{pa}R_{qb}R_{rc}R_{sd}\,Y(z)_{abcd}\\[3pt]
 \multicolumn{2}{@{}c@{}}{\displaystyle
   \medmuskip=1mu\relax
   {}=\smashoperator{\sum_{\substack{1\le a<b\le 2n\\1\le c<d\le 2n}}}
   (R_{pa}R_{qb}-R_{pb}R_{qa})
   \cdot(R_{rc}R_{sd}-R_{rd}R_{sc})\,Y(z)_{abcd}
 }\\[3pt]
 & \displaystyle {}=(S_RQ(z)S_R^{\mathsf T})_{(pq),(rs)}.
 \end{array}
 \label{eq:rec-q-equivariance}
\end{equation}
\noindent\emph{Plane rotations.} Anticommutation gives
$i[H_{pq},\gamma_j]=\delta_{pj}\gamma_q-\delta_{qj}\gamma_p$.
For $U(\theta)=e^{-i\theta H_{pq}}$, the product rule gives
\[
 \begin{gathered}
 \begin{aligned}
 \frac{d}{d\theta}(U^\dagger\gamma_j U)
 &= iU^\dagger H_{pq}\gamma_j U
    -iU^\dagger\gamma_j H_{pq}U\\
 &= \delta_{pj}U^\dagger\gamma_q U
    -\delta_{qj}U^\dagger\gamma_p U,
 \end{aligned}\\[3pt]
 U^\dagger\gamma_p U=\cos\theta\,\gamma_p+\sin\theta\,\gamma_q,\\
 U^\dagger\gamma_q U=-\sin\theta\,\gamma_p+\cos\theta\,\gamma_q.
 \end{gathered}
\]
All other Majoranas are fixed. Orthogonal elimination uses at most
$n(2n-1)$ such rotations for any $R\in SO(2n)$; lift phases do not
affect states.
\end{proof}
\subsection{The quartic parent of the paired reference}
Write $\obs=(B_I)_{|I|=4}$, $x_0=x(P_0)$, and $n=4B$.
We use the Jordan--Wigner representation with $\ket1$ the occupied state,
$a_j=Z_1\cdots Z_{j-1}(X_j+iY_j)/2$, so that
\begin{equation}
\begin{gathered}
\gamma_{2j-1}=a_j+a_j^\dagger=Z_1\cdots Z_{j-1}X_j,\\
 \gamma_{2j}=-i(a_j-a_j^\dagger)=Z_1\cdots Z_{j-1}Y_j .
\end{gathered}
\label{eq:jw}
\end{equation}
In the occupation basis $X_j$ and $Z_j$ are real, and $Y_j$ is imaginary. One
block carries modes $1,\dots,4$ and Majoranas $\gamma_1,\dots,\gamma_8$; block
$b$ is shifted by $4(b-1)$ modes. Put $\ket{c_1}=\ket{1100}$,
$\ket{c_2}=\ket{0011}$, $\ket\omega=(\ket{c_1}+\ket{c_2})/\sqrt2$,
$P_\omega=\ketbra\omega\omega$, and $n_j=a_j^\dagger a_j$. 
Using the identity
\[
 \begin{aligned}
 \gamma_{2j-1}\gamma_{2j}
 &=-i(a_j+a_j^\dagger)(a_j-a_j^\dagger)\\
 &=-i(a_j^\dagger a_j-a_ja_j^\dagger)
   =-i(2n_j-1),
 \end{aligned}
\]
we obtain
\[
 (-1)^{n_j}=1-2n_j=-i\gamma_{2j-1}\gamma_{2j},
\]
\begin{equation}
 \begin{gathered}
 \Pi_b=\prod_{j=4b-3}^{4b}(-1)^{n_j}
 =\prod_{j=4b-3}^{4b}(-i\gamma_{2j-1}\gamma_{2j})\\
 =\gamma_{8b-7}\gamma_{8b-6}\cdots\gamma_{8b}.
 \end{gathered}
 \label{eq:parity}
\end{equation}
Commuting through $\Pi_b$ gives
$\gamma_p\Pi_b=(-1)^7\Pi_b\gamma_p$ for $p\in b$ and
$\gamma_q\Pi_b=(-1)^8\Pi_b\gamma_q$ for $q\notin b$.
For a monomial $C$ with $k_b$ factors in block $b$, use
$\Pi_b\ket{\Omega_B}=\ket{\Omega_B}$ to obtain
\begin{equation}
\begin{gathered}
\Pi_bC\Pi_b=(-1)^{k_b}C,\\
 \bra{\Omega_B}C\ket{\Omega_B}
 =\bra{\Omega_B}\Pi_bC\Pi_b\ket{\Omega_B}
 =(-1)^{k_b}\bra{\Omega_B}C\ket{\Omega_B},
\end{gathered}
\label{eq:selection}
\end{equation}
Thus any odd $k_b$ forces zero expectation, also for one block
$(\omega,\Pi_1)$. A block-supported even monomial acts only on that
block's qubits, identically in every block.
\begin{alem}[Quadratic blindness]
\label{lem:quadratic}
$\bra\omega a_i^\dagger a_j\ket\omega=\delta_{ij}/2$,
$\bra\omega a_ia_j\ket\omega=\bra\omega a_i^\dagger a_j^\dagger\ket\omega=0$,
and $\bra{\Omega_B}\gamma_p\gamma_q\ket{\Omega_B}=\delta_{pq}$. The Majorana
covariance $[\Gamma_{\rm A}(\rho)]_{pq}=\frac i2\Tr(\rho[\gamma_p,\gamma_q])$
vanishes on $\cO$.
\end{alem}
\begin{proof}
The occupations are
\[
\begin{gathered}
 n_i\ket{c_1}=[i\le2]\ket{c_1},\quad n_i\ket{c_2}=[i\ge3]\ket{c_2},\\
 \bra\omega n_i\ket\omega
 =\tfrac12\bigl(\bra{c_1}n_i\ket{c_1}+\bra{c_2}n_i\ket{c_2}\bigr)=\tfrac12 .
\end{gathered}
\]
For $i\ne j$, $a_i^\dagger a_j\ket{c_l}$ is zero or a signed occupation ket
with one particle in each of $\{1,2\}$ and $\{3,4\}$, hence orthogonal
to both $\ket{c_1},\ket{c_2}$:
\[
 \bra\omega a_i^\dagger a_j\ket\omega
 =\frac12\sum_{k,l=1}^2\bra{c_k}a_i^\dagger a_j\ket{c_l}=0 .
\]
For $N_{\rm loc}=\sum_{i=1}^4n_i$, the number-changing terms vanish because
\[
 \begin{gathered}
 N_{\rm loc}\ket{c_l}=2\ket{c_l}\quad(l=1,2),\qquad
 N_{\rm loc}\ket\omega=2\ket\omega,\\[3pt]
 \begin{aligned}
 -2\bra\omega a_i a_j\ket\omega
 &=\bra\omega[N_{\rm loc},a_i a_j]\ket\omega\\
 &=(2-2)\bra\omega a_i a_j\ket\omega=0,\\[3pt]
 2\bra\omega a_i^\dagger a_j^\dagger\ket\omega
 &=\bra\omega[N_{\rm loc},a_i^\dagger a_j^\dagger]\ket\omega\\
 &=(2-2)\bra\omega a_i^\dagger a_j^\dagger\ket\omega=0.
 \end{aligned}
 \end{gathered}
\]
Together with~\eqref{eq:parity}, this gives
\[
 \bra\omega\gamma_{2i-1}\gamma_{2i}\ket\omega
 =i\bra\omega(1-2n_i)\ket\omega=i\bigl(1-2\cdot\tfrac12\bigr)=0 .
\]
For two modes $i\ne j$ of one block, the four products are
\[
\begin{gathered}
 \gamma_{2i-1}\gamma_{2j-1}=a_ia_j+a_ia_j^\dagger+a_i^\dagger a_j+a_i^\dagger a_j^\dagger,\\
  \gamma_{2i-1}\gamma_{2j}=-i(a_ia_j-a_ia_j^\dagger+a_i^\dagger a_j-a_i^\dagger a_j^\dagger),\\
 \gamma_{2i}\gamma_{2j-1}=-i(a_ia_j+a_ia_j^\dagger-a_i^\dagger a_j-a_i^\dagger a_j^\dagger),\\
  \gamma_{2i}\gamma_{2j}=-a_ia_j+a_ia_j^\dagger+a_i^\dagger a_j-a_i^\dagger a_j^\dagger.
\end{gathered}
\]
Every term has zero expectation, using also $a_ia_j^\dagger=-a_j^\dagger a_i$.
Cross-block products vanish by~\eqref{eq:selection}; $\gamma_p^2=I$ gives
$\bra{\Omega_B}\gamma_p\gamma_q\ket{\Omega_B}=\delta_{pq}$ and
$\Gamma_{\rm A}(P_0)=0$, where $P_0=\ketbra{\Omega_B}{\Omega_B}$.
For $\rho=U_RP_0U_R^\dagger$,
\begin{align*}
 [\Gamma_{\rm A}(\rho)]_{pq}
 &=\frac{i}{2}
   \bra{\Omega_B}U_R^\dagger[\gamma_p,\gamma_q]U_R\ket{\Omega_B}\\
 &=\sum_{p',q'=1}^{2n}R_{pp'}R_{qq'}\,[\Gamma_{\rm A}(P_0)]_{p'q'}\\
 &=\bigl(R\,\Gamma_{\rm A}(P_0)R^{\mathsf T}\bigr)_{pq}=0. \qedhere
\end{align*}
\end{proof}
\begin{alem}[Block expansion and block parent]
\label{lem:block}
On one block let $\Pi=\Pi_1$ be its parity operator from~\eqref{eq:parity}. For $x^\omega_I=\bra\omega B_I\ket\omega$, $|I|=4$,
\begin{equation}
 \begin{gathered}
 16P_\omega=I+\Pi+\sum_{I\in\mathcal I_4(8)}x^\omega_I B_I,
 \quad \sum_{I\in\mathcal I_4(8)}(x^\omega_I)^2=14,\\
 K_b:=I-\frac{1}{14}\sum_{I\in\mathcal I_4(8)}x^\omega_I B_I
 =\frac{15I+\Pi-16P_\omega}{14}.
 \end{gathered}
 \label{eq:block-expansion}
\end{equation}
$K_b$ has eigenvalue $0$ on $\ket\omega$, $8/7$ on the seven even states
orthogonal to $\ket\omega$, and $1$ on the eight odd states.
\end{alem}
\begin{proof}
\noindent\emph{Expansion.} The four-mode basis
\eqref{eq:majorana-orthogonality} has $2^8=16^2$ members,
indexed by $J\subseteq\{1,\ldots,8\}$. By~\eqref{eq:spectral-trace},
\[
 \begin{aligned}
 16P_\omega
 &=\sum_{J\subseteq\{1,\ldots,8\}}\Tr(C_JP_\omega)\,C_J\\
 &=\sum_{J\subseteq\{1,\ldots,8\}}\bra\omega C_J\ket\omega\,C_J.
 \end{aligned}
\]
For odd $|J|$, the parity rule~\eqref{eq:selection} gives
\[
 \bra\omega C_J\ket\omega
 =(-1)^{|J|}\bra\omega C_J\ket\omega=0.
\]
For $|J|=2$, Lemma~\ref{lem:quadratic} gives $\bra\omega C_J\ket\omega=0$.
For $|J|=6$, the complement has $|J^c|=2$, and
\[
 C_J=\zeta'\Pi C_{J^c},\qquad |\zeta'|=1.
\]
Since $\bra\omega\Pi=\bra\omega$, the $|J^c|=2$ case gives
\[
 \bra\omega C_J\ket\omega
 =\zeta'\bra\omega C_{J^c}\ket\omega=0.
\]
The remaining coefficients are
\[
\begin{gathered}
 \bra\omega C_\emptyset\ket\omega=1,\quad C_{\{1,\dots,8\}}=i^{28}\gamma_1\cdots\gamma_8=\Pi,\\
 \bra\omega\Pi\ket\omega=1,\quad C_J=i^6B_J=-B_J,\\
 \bra\omega C_J\ket\omega\,C_J=x^\omega_JB_J\quad (|J|=4),
\end{gathered}
\]
which is the expansion in~\eqref{eq:block-expansion}.
\noindent\emph{Normalization and spectrum.} Squaring the expansion, taking
$\frac1{16}\Tr$, and using orthonormality and $P_\omega^2=P_\omega$,
\[
 16=16\Tr P_\omega=\frac1{16}\Tr\bigl[(16P_\omega)^2\bigr]
 =1+1+\sum_{I\in\mathcal I_4(8)}(x^\omega_I)^2 .
\]
Rearranging, $K_b=I-(16P_\omega-I-\Pi)/14=(15I+\Pi-16P_\omega)/14$. For an
even $\ket e$ with $\braket\omega e=0$ and an odd $\ket o$
($\Pi\ket o=-\ket o$, $\braket\omega o=0$),
\[
 \begin{gathered}
 K_b\ket\omega=\tfrac1{14}\bigl(15\ket\omega+\ket\omega-16\ket\omega\braket\omega\omega\bigr)=0,\\
 K_b\ket e=\tfrac1{14}\bigl(15\ket e+\ket e-16\ket\omega\braket\omega e\bigr)=\tfrac87\ket e,\\
 K_b\ket o=\tfrac1{14}\bigl(15\ket o-\ket o-16\ket\omega\braket\omega o\bigr)=\ket o .
 \end{gathered}
\]
The even and odd sectors both have dimension $8$, so the multiplicities are
$1$, $7$, and $8$.
\end{proof}
\begin{acor}[Global parent]
\label{cor:global-parent}
$x_0=x(P_0)$ is supported on quartets inside single blocks,
$\|x_0\|^2=14B$, and
\begin{equation}
\begin{gathered}
 K_0=BI-\frac1{14}x_0^{\mathsf T}\obs=\sum_{b=1}^{B}K_b,\quad K_0\ket{\Omega_B}=0,\\
 K_0\ge I-P_0 .
\end{gathered}
\label{eq:global-parent}
\end{equation}
\end{acor}
\begin{proof}
A quartet meeting several blocks has block multiplicities $(3,1)$,
$(2,1,1)$, $(1,1,1,1)$, or $(2,2)$. In the first three some multiplicity is
odd, and~\eqref{eq:selection} gives $\bra{\Omega_B}B_I\ket{\Omega_B}=0$. In
the last, $B_I=\pm(\gamma_p\gamma_q)(\gamma_r\gamma_s)$ with $p\ne q$ in block
$b$ and $r\ne s$ in block $b'$, and the two even factors act on different
tensor factors:
\[
 \bra{\Omega_B}B_I\ket{\Omega_B}
 =\pm\bra\omega\gamma_p\gamma_q\ket\omega\,\bra\omega\gamma_r\gamma_s\ket\omega=0
\]
by Lemma~\ref{lem:quadratic}. Put $\mathcal J_b=\{8b-7,\ldots,8b\}$ for the Majorana indices in block $b$. Hence $x_0=\bigoplus_{b=1}^{B}x^\omega_b$,
$\|x_0\|^2=\sum_{b=1}^{B}14=14B$, and
\[
 \begin{aligned}
 K_0
 &=BI-\frac{1}{14}x_0^{\mathsf T}\obs\\
 &=\sum_{b=1}^{B}\Bigl(
   I-\frac{1}{14}\sum_{I\in\mathcal I_4(\mathcal J_b)}x^\omega_I B_I
   \Bigr)=\sum_{b=1}^{B}K_b.
 \end{aligned}
\]
Let $\{\ket{\kappa_k}\}_{k=0}^{15}$ be the eigenbasis of Lemma~\ref{lem:block}, with
eigenvalues $\kappa_k\in\{0,\frac87,1\}$ and $\ket{\kappa_0}=\ket\omega$. The
product vectors $\ket{\mathbf k}=\bigotimes_{b=1}^{B}\ket{\kappa_{k_b}}$ form an
orthonormal basis, $\ket{\Omega_B}=\ket{\mathbf0}$, and each $K_b$ acts on its
own factor, so
\[
\begin{gathered}
 K_0\ket{\mathbf k}=\Bigl(\sum_{b=1}^{B}\kappa_{k_b}\Bigr)\ket{\mathbf k},\quad \sum_{b=1}^{B}\kappa_{k_b}=0\ \text{if}\ \mathbf k=\mathbf0,\\
 \sum_{b=1}^{B}\kappa_{k_b}\ge1\ \text{otherwise}.
\end{gathered}
\]
Consequently, in this product basis,
\[
 \begin{aligned}
 K_0
 &=\sum_{\mathbf k\in\{0,\ldots,15\}^B}
   \Bigl(\sum_{b=1}^{B}\kappa_{k_b}\Bigr)\ketbra{\mathbf k}{\mathbf k}\\
 &\ge\sum_{\substack{\mathbf k\in\{0,\ldots,15\}^B\\
                     \mathbf k\ne\mathbf0}}
   \ketbra{\mathbf k}{\mathbf k}=I-P_0,
 \end{aligned}
\]
and $K_0\ket{\Omega_B}=0$.
\end{proof}
For $\rho=U_RP_0U_R^\dagger$, conjugating~\eqref{eq:global-parent}
by $U_R$ gives
\[
 K_\rho=BI-\tfrac{1}{14}x(\rho)^{\mathsf T}\obs\ge I-\rho.
\]
For any $\sigma\in\cO$, Lemma~\ref{lem:closure} and
Corollary~\ref{cor:global-parent} give
$\|x(\rho)\|^2=\|x(\sigma)\|^2=14B$.
Since $\sigma$ is positive and has unit trace, tracing the operator
inequality against $\sigma$ yields
\begin{equation}
 \begin{aligned}
 \trd(\rho,\sigma)^2
 &=1-\Tr(\rho\sigma)\le\Tr(K_\rho\sigma)\\
 &=B-\tfrac{1}{14}x(\rho)^{\mathsf T}x(\sigma)\\
 &=\frac{\|x(\rho)\|^2+\|x(\sigma)\|^2
          -2x(\rho)^{\mathsf T}x(\sigma)}{28}\\
 &=\frac{\|x(\rho)-x(\sigma)\|^2}{28}.
 \end{aligned}
 \label{eq:paired-certificate-proof}
\end{equation}
This is the bound quoted in the main text. In particular,
$x(\rho)=x(\sigma)$ forces $\trd(\rho,\sigma)=0$, hence $\rho=\sigma$.

\clearpage
\section{Copy-optimal learning from fixed-basis quartic measurements}
\label{app:reconstruction}
We prove Theorem~\ref{thm:reconstruction}. Throughout this appendix,
$d=2n=8B$, $m=\binom d2$, $r=\binom d4$, and $x_0=x(P_0)$.
Write the unknown state as $\rho=U_{R_*}P_0U_{R_*}^\dagger$ with
$R_*\in SO(d)$. Put $x_*=x(\rho)=\wedge^4(R_*)x_0$, by
Lemma~\ref{lem:closure}. Each copy of $\rho$ yields one projective sign
of a fixed-basis $B_I$. A bank uses $K$ independent copies per quartet
and has empirical mean $y=x_*+\xi$. Banks use disjoint copies and predetermined
measurement schedules. The accuracy parameter $\epsilon$ is a
trace-distance tolerance; for the pure output, the corresponding
infidelity tolerance is $\epsilon^2$.
Identify $\ket{p\wedge q}$ with
$(\ketbra pq-\ketbra qp)/\sqrt2$ when a pair vector is regarded as a
unit-Frobenius skew matrix. Use the flattening $Q$ of Lemma~\ref{lem:closure}.
In execution order, Part I uses Lemmas~\ref{lem:opt-range-refresh}--\ref{lem:opt-candidates}
to locate all blocks. Part II uses Lemmas~\ref{lem:rec-decoder}, \ref{lem:opt-joint-noise}
and~\ref{lem:opt-initialization} to assemble a frame. Part III uses
Lemmas~\ref{lem:opt-geometry}--\ref{lem:opt-refinement} and the parent certificate to bound
state error. The proofs start with the local geometry used in all three parts.
Lemma~\ref{lem:opt-lower} treats fixed-basis quartic readout and
Lemma~\ref{lem:necessary-copies} permits collective measurements of target
copies. The reference-assisted lower bound is proved in Section~\ref{gen:app:lower}.
\subsection{Local algebra and the quartic matrix}
\begin{alem}[Local real representation and quartic spectrum]
\label{lem:rec-local}
On one four-mode block let $\mathsf K_{\rm occ}$ be occupation-basis conjugation,
$J=\gamma_2\gamma_4\gamma_6\gamma_8\mathsf K_{\rm occ}$, and
$\cH_+^J=\{v:\Pi v=v,\ Jv=v\}$. This is an eight-dimensional real
inner-product space. Choose Spin representatives of the Gaussian unitaries,
modulo their physically irrelevant global phase. Their restrictions to
$\cH_+^J$ realize all of $SO(8)$, and the block orbit consists of its
unit-vector projectors.
Extend $x^\omega$ alternately to $Y^\omega_{pqrs}$ and define
$Q_\omega$ on the orthonormal pair basis $\ket{p\wedge q}$, $p<q$, by
$(Q_\omega)_{(pq),(rs)}=Y^\omega_{pqrs}$. Then
\begin{equation}
\begin{gathered}
(Q_\omega+3I)(Q_\omega-I)=0,\\
\operatorname{mult}(-3)=7,\quad
\operatorname{mult}(1)=21,\quad
\star x^\omega=x^\omega.
\end{gathered}
\label{eq:rec-local-spectrum}
\end{equation}
Here $\star$ is the Hodge star for the standard orientation on $\R^8$:
$(\star x)_{I^c}=\epsilon_Ix_I$, where $I^c$ is the sorted complement of
$I$ and $\epsilon_I$ is the sign of $(I,I^c)$ as a permutation of
$(1,\ldots,8)$.
\end{alem}
\begin{proof}
\noindent\emph{Real structure.}
By~\eqref{eq:jw}, complex conjugation fixes odd-indexed Majoranas and
negates even-indexed ones. Set $C=\gamma_2\gamma_4\gamma_6\gamma_8$.
Anticommuting a Majorana through the four factors gives
\begin{equation}
\begin{gathered}
C\gamma_{2j-1}C^{-1}=\gamma_{2j-1},\qquad
C\gamma_{2j}C^{-1}=-\gamma_{2j},\\
J\gamma_pJ^{-1}=\gamma_p,\quad
J^2=C\overline C=C^2=(-1)^6I=I,\\
J\Pi J^{-1}=\Pi.
\end{gathered}
 \label{eq:rec-real}
\end{equation}
Jordan--Wigner gives $C=-X\otimes Y\otimes X\otimes Y$, hence
\[
\begin{gathered}
C\ket{1100}=-\ket{0011},\qquad
C\ket{0011}=-\ket{1100},\\
C\ket\omega=-\ket\omega,\qquad
J(i\ket\omega)=i\ket\omega,\\
v=\frac{v+Jv}{2}+i\frac{v-Jv}{2i},\\
J\frac{v+Jv}{2}=\frac{v+Jv}{2},\quad
J\frac{v-Jv}{2i}=\frac{v-Jv}{2i}\quad
(v\in\cH_+).
\end{gathered}
\]
Here $\cH_+=\{v:\Pi v=v\}$ is the even sector. Its complex dimension is
eight. $J$ maps it to itself, since $J\Pi J^{-1}=\Pi$. The last two
lines show $\cH_+=\cH_+^J\oplus i\cH_+^J$. Hence $\dim_{\R}\cH_+^J=8$.
For $J$-fixed $u,v$, $\braket uv=\overline{\braket{Ju}{Jv}}
=\overline{\braket uv}$. A real orthonormal basis
$\{\ket{e_a}\}_{a=1}^8$ of $\cH_+^J$ also spans $\cH_+$ over $\mathbb C$. Fix
$E=\sum_{a=1}^{8}\ketbra{e_a}{a}$ and $v_0=E^\dagger i\ket\omega\in\R^8$.

\noindent\emph{Generator basis.}
The matrices $G_{pq}=E^\dagger\gamma_p\gamma_qE/2$, $p<q$, are real
by~\eqref{eq:rec-real} and skew by adjunction. Write
$P_+=(I+\Pi)/2=EE^\dagger$. The monomial trace rule
\eqref{eq:majorana-orthogonality} gives
\begin{equation}
\begin{aligned}
 \Tr(G_{pq}^{\mathsf T}G_{rs})
 &=-\tfrac14\Tr(P_+\gamma_p\gamma_q\gamma_r\gamma_s)\\
 &=-\tfrac18\Tr(\gamma_p\gamma_q\gamma_r\gamma_s)\\
 &\quad
   -\tfrac18\Tr(\Pi\gamma_p\gamma_q\gamma_r\gamma_s) =2\delta_{(pq),(rs)}.
\end{aligned}
 \label{eq:rec-generator-gram}
\end{equation}
Equal pairs give $-I$ in the first trace; all other traces vanish.
The 28 generators span $\mathfrak{so}(8)$, so every $A\in\mathfrak{so}(8)$ has
\begin{equation}
\begin{gathered}
A=\sum_{1\le p<q\le8}c_{pq}G_{pq},\quad
c_{pq}=\tfrac12\Tr(G_{pq}^{\mathsf T}A),\\
e^{\sum_{1\le p<q\le8} c_{pq}\gamma_p\gamma_q/2}E=Ee^A.
\end{gathered}
\label{eq:rec-spin-map}
\end{equation}
With $\mathcal U_A=\exp(\sum_{1\le p<q\le8}c_{pq}\gamma_p\gamma_q/2)$,
the isometry $E:\R^8\to\cH_+^J$ intertwines the two actions:
\[
 \begin{tikzcd}[column sep=large,row sep=large]
 \R^8\arrow[r,"e^A"]\arrow[d,"E"'] &\R^8\arrow[d,"E"]\\
 \cH_+^J\arrow[r,"\mathcal U_A"'] &\cH_+^J.
 \end{tikzcd}
\]
\eqref{eq:rec-real} makes the lower map a real isometry of
$\cH_+^J$. For every $v\in\R^8$, the series gives both compositions:
\[
\begin{aligned}
\mathcal U_AEv
 &=\sum_{j=0}^{\infty}\frac1{j!}
     \left(\sum_{1\le p<q\le8}c_{pq}\gamma_p\gamma_q/2\right)^jEv\\
 &=E\sum_{j=0}^{\infty}\frac{A^jv}{j!} =Ee^Av.
\end{aligned}
\]
Exponentials exhaust $SO(8)$, whose transitive sphere action proves the
orbit assertion and gives a fixed-size inverse for local compilation.

\noindent\emph{Spectral identity without a quartic table.}
Put $V=[G_{pq}v_0]_{p<q}$, an $8\times28$ matrix.
The entries of $V^{\mathsf T}V$ are $-v_0^{\mathsf T}G_{pq}G_{rs}v_0$,
because $G_{pq}$ is skew. Use $EE^\dagger=P_+$ and
$P_+\ket\omega=\ket\omega$. Equal pairs give
$\gamma_p\gamma_q\gamma_p\gamma_q=-I$, hence the entry $1/4$.
Overlapping distinct pairs reduce to quadratics, whose expectations
vanish by Lemma~\ref{lem:quadratic}. Disjoint pairs give the quartic
$Y^\omega_{pqrs}$. $Q_\omega$ vanishes in the first two cases, so
\begin{equation}
\begin{aligned}
(V^{\mathsf T}V)_{(pq),(rs)}
 &=-\tfrac14\bra\omega\gamma_p\gamma_q\gamma_r\gamma_s\ket\omega\\
 &=\tfrac14(I-Q_\omega)_{(pq),(rs)}.
\end{aligned}
\label{eq:rec-tangent-gram}
\end{equation}
The generators in~\eqref{eq:rec-generator-gram} have the same Gram matrix
as $F_{ab}=\ketbra ab-\ketbra ba$, $a<b$, so the change between these
bases is orthogonal. For a real unit vector $v$,
\[
\begin{gathered}
F_{ab}v=v_b\ket a-v_a\ket b,\\
 \sum_{1\le a<b\le8}(F_{ab}v)(F_{ab}v)^{\mathsf T}=I-vv^{\mathsf T}.
\end{gathered}
\]
The diagonal entries are $\sum_{\substack{1\le b\le8\\b\ne a}}v_b^2=1-v_a^2$; the off-diagonal
entries are $-v_av_b$. Therefore
\begin{equation}
\begin{gathered}
VV^{\mathsf T}=I-v_0v_0^{\mathsf T},\qquad
(V^{\mathsf T}V)^2=V^{\mathsf T}V,\\
\operatorname{rank}V=7.
\end{gathered}
\label{eq:rec-tangent-projector}
\end{equation}
Here $V^{\mathsf T}v_0=0$ because every $G_{pq}$ is skew.
Set $P=V^{\mathsf T}V$. Equations~\eqref{eq:rec-tangent-gram}
and~\eqref{eq:rec-tangent-projector} give
\[
\begin{gathered}
Q_\omega=I-4P,\quad
(Q_\omega+3I)(Q_\omega-I)=16(P^2-P)=0,\\
Q_\omega|_{\operatorname{ran}P}=-3I,\qquad
Q_\omega|_{\ker P}=I,\\
\dim\operatorname{ran}P=7,\qquad
\dim\ker P=28-7=21.
\end{gathered}
\]

\noindent\emph{Orientation.}
For a sorted quartet $I$, write $I^c$ in increasing order and let
$\epsilon_I$ be the sign of $(I,I^c)$ relative to $(1,\ldots,8)$.
Anticommutation and $B_I^2=I$ give
\[
\begin{gathered}
\Pi=\epsilon_I B_IB_{I^c},\qquad
B_{I^c}=\epsilon_I B_I\Pi,\\
x^\omega_{I^c}=\epsilon_I\bra\omega B_I\Pi\ket\omega
 =\epsilon_I x^\omega_I.
\end{gathered}
\]
This is the last identity in~\eqref{eq:rec-local-spectrum}.
\end{proof}
\newcommand{\proofdisplays}{%
  \setlength{\abovedisplayskip}{6.5pt plus 2pt minus 3pt}%
  \setlength{\belowdisplayskip}{6.5pt plus 2pt minus 3pt}}
\subsection{Orbit geometry and local refinement}
Let $G_0$ be the stabilizer of $x_0$ in $SO(d)$, including permutations
of the identical blocks. Its Lie algebra is denoted by $\mathfrak g_0$.
Use the Frobenius-orthogonal complement
$\mathfrak m=\mathfrak g_0^\perp\subset\mathfrak{so}(d)$ and write
\[
\begin{gathered}
F(X)=\wedge^4(e^X)x_0,\qquad
L=DF(0),\\
G_X=\frac12\sum_{1\le p<q\le d}X_{pq}\gamma_p\gamma_q .
\end{gathered}
\]
The sign convention is that
\[
 U_{e^X}=e^{G_X} \quad\text{and}\quad x(U_{e^X}P_0U_{e^X}^\dagger)=F(X).
\]
The matrix $X$ has a crossing part $X_{\rm c}$, connecting distinct
eight-dimensional blocks, and an intrablock part $X_{\rm i}$.
\begin{alem}[Dimension-independent orbit geometry]
\label{lem:opt-geometry}
The kernel of $L$ is the direct sum of the local stabilizer algebras.
Its dimension is $21B$, and $\dim\mathfrak m=D=m-21B$.
The first line below holds for $X\in\mathfrak m$. The other two lines
hold for all skew $X,Z,Z_1,\ldots,Z_k$. At $X=0$ and $k=1$ they give
$\|LZ\|_2\le4\|Z\|_F$ for every skew $Z$.
\begin{equation}
\begin{gathered}
 \|LX\|_2^2=7\|X_{\rm c}\|_F^2+16\|X_{\rm i}\|_F^2,\\
 \|D^kF(X)[Z_1,\ldots,Z_k]\|_2\le4^k\prod_{j=1}^k\|Z_j\|_F
 \quad(k\ge1),\\
 \|DF(X)[Z]-LZ\|_2\le16\|X\|_F\|Z\|_F .
\end{gathered}
 \label{eq:opt-derivatives}
\end{equation}
If $V\in SO(d)$ and $r_V=\|V-I\|_F\le1/16$, then
$V=e^Xe^Z$ with $X\in\mathfrak m$, $Z\in\mathfrak g_0$ and
\[
 \|X\|_F\le\frac{r_V}{\sqrt{1-r_V^2/4}} .
\]
\end{alem}
\begin{proof}
\proofdisplays
\noindent\emph{The tangent metric from the parent.}
Let \(\rho_4=\Lambda^4,\ A_X=d\rho_4(X)\) and \(G_X=\frac12\sum_{p<q}X_{pq}\gamma_p\gamma_q\).
Take \(F\) on all of \(\mathfrak{so}(d)\) before restricting the chart to
\(\mathfrak m\). Put \(\ket\Omega=\ket{\Omega_B}\) and \(K_0=\sum_b K_b\)
as in~\eqref{eq:global-parent}. For \(\rho,\sigma\in\cO\) the last three lines
of~\eqref{eq:paired-certificate-proof} are equalities:
\[
 \Tr(K_\rho\sigma)=\|x(\rho)-x(\sigma)\|^2/28.
\]
Take \(\rho=P_0\), so
\(K_\rho=K_0\), and \(\sigma=e^{tG_X}P_0e^{-tG_X}\), so \(x(\sigma)=F(tX)\):
\[
 \|F(tX)-x_0\|^2/28
 =\bra\Omega e^{-tG_X}K_0e^{tG_X}\ket\Omega .
\]
Because \(F(0)=x_0\), differentiating the left hand side twice at zero gives \(\|LX\|^2/14\).
The product rule on the right hand side gives
\begin{equation}
\begin{aligned}
\|LX\|^2/14
 &=\bra\Omega(G_X^2K_0-G_XK_0G_X\\
 &\qquad{}-G_XK_0G_X+K_0G_X^2)\ket\Omega\\
 &=2\langle G_X\Omega\,|\,K_0\,|\,G_X\Omega\rangle .
\end{aligned}
 \label{eq:opt-parent-tangent}
\end{equation}
The first and last terms vanish since \(K_0\ket\Omega=0\), and \(G_X^\dagger=-G_X\) changes
the sign of the mixed terms.

\noindent\emph{Three kinds of tangent kets.}
Split \(X=X_{\rm c}+\sum_bX_b\), with \(X_b\) supported in block \(b\). Then
\[
 G_X\ket\Omega=G_{X_{\rm c}}\ket\Omega+\sum_bG_{X_b}\ket\Omega.
\]

For a crossing pair, each touched block is odd. The block parent is the
identity on the odd sector by Lemma~\ref{lem:block}. So
\(K_0\gamma_p\gamma_q\ket\Omega=2\gamma_p\gamma_q\ket\Omega\).
Distinct crossing pairs have zero inner product on the seed, because
their product contains an odd block or a vanishing nontrivial quadratic
expectation. Each ket \(\gamma_p\gamma_q\ket\Omega\) has norm one, and
\(\|X_{\rm c}\|_F^2=2\sum X_{pq}^2\) over crossing pairs \(p<q\). Consequently
\[
 \|G_{X_{\rm c}}\Omega\|^2=\tfrac14\sum X_{pq}^2
 =\|X_{\rm c}\|_F^2/8 .
\]
On block \(b\) use Lemma~\ref{lem:rec-local}, with real seed spinor \(v_0\).
Put \(A_b=\sum_{p<q}X_{pq}G_{pq}\) over the pairs of block \(b\).
Equation~\eqref{eq:rec-generator-gram} gives \(\|A_b\|_F=\|X_b\|_F\).
Since \(\gamma_p\gamma_qE/2=EG_{pq}\), the ket \(G_{X_b}Ev_0\) equals \(EA_bv_0\).
Let \(\mathfrak k_b=\{X_b:A_bv_0=0\}\). Its dimension is \(21\)
by~\eqref{eq:rec-tangent-projector}. The matrices
\(\ketbra u{v_0}-\ketbra{v_0}u\), \(u\perp v_0\), span seven dimensions.
They are Frobenius orthogonal to every skew \(A\) with \(Av_0=0\). So the
part \(X_b^{\rm h}\) of \(X_b\) orthogonal to \(\mathfrak k_b\) has
\[
\begin{gathered}
A_b^{\rm h}=\ketbra{u_b}{v_0}-\ketbra{v_0}{u_b},\qquad u_b=A_bv_0,\\
\|X_b^{\rm h}\|_F^2=2\|u_b\|^2,\qquad
\|G_{X_b}\Omega\|^2=\|u_b\|^2 .
\end{gathered}
\]
The ket \(G_{X_b}\ket\Omega\) is even and orthogonal to \(\ket\omega\) in
block \(b\). It equals the seed in the other blocks. Its parent energy is
\(8/7\) by Lemma~\ref{lem:block}. Both kinds of kets are eigenvectors of
\(K_0\), with eigenvalues \(2\) and \(8/7\). Kets of different blocks are
orthogonal, because each is orthogonal to \(\ket\omega\) in its own block.
A crossing ket is odd in two blocks, so it is orthogonal to every block ket.
So the cross terms in~\eqref{eq:opt-parent-tangent} vanish. For every
skew \(X\),
\begin{equation}
\begin{aligned}
\|LX\|^2&=28\Bigl(\tfrac28\|X_{\rm c}\|_F^2
   +\tfrac87\cdot\tfrac12\sum_b\|X_b^{\rm h}\|_F^2\Bigr)\\
 &=7\|X_{\rm c}\|_F^2+16\sum_b\|X_b^{\rm h}\|_F^2 .
\end{aligned}
\label{eq:opt-parent-gram}
\end{equation}
Thus \(\ker L=\bigoplus_b\mathfrak k_b\), of dimension \(21B\). Now \(LZ=0\)
is equivalent to \(e^{t\,d\rho_4(Z)}x_0=x_0\) for all real \(t\). So
\(\mathfrak g_0=\ker L\). Hence \(X\in\mathfrak m\) means
\(X_b=X_b^{\rm h}\) for all \(b\), and \(D=m-21B\). On \(\mathfrak m\),
\eqref{eq:opt-parent-gram} is the first line
of~\eqref{eq:opt-derivatives}. For every skew \(Z\) it gives
\(\|LZ\|\le4\|Z\|_F\), since projections contract.

\noindent\emph{All derivative orders with one estimate.} Put \(A=A_X\), \(S_j=d\rho_4(Z_j)\).
For the simplex \(t_i\ge0,\ \sum_{i=1}^kt_i\le1,\ t_0=1-\sum_{i=1}^kt_i\), differentiation of
the exponential yields
\[
\begin{aligned}
&D^kF(X)[Z_1,\ldots,Z_k]\\
&\quad=\sum_{\pi\in S_k}\int_{\Delta_k} e^{t_0A}S_{\pi(1)}e^{t_1A}
   \cdots S_{\pi(k)}e^{t_kA}x_0\,dt_1\cdots dt_k .
\end{aligned}
\]
To check the coefficients, expand the \(k+1\) exponentials. For \(j_i\) powers of \(A\) in gap
\(i\), the integral of \(\prod_i t_i^{j_i}\) is \(\prod_i j_i!/(k+\sum_i j_i)!\). The
factorials of the gaps cancel, leaving exactly the coefficient of that ordered word in the
differentiated exponential series. Absolute convergence permits the expansion.
For \(g=e^{t_kX}\), the rightmost factor \(S_{\pi(k)}e^{t_k A}x_0\) is \(\rho_4(g)L(g^{\mathsf
T}Z_{\pi(k)}g)\), of norm at most \(4\|Z_{\pi(k)}\|_F\) by \eqref{eq:opt-parent-gram}. Each of
the preceding insertions is a sum of four single-slot operators of norm at most
\(\|Z_j\|_{\op}\). Orthogonal exponentials preserve the norms, and the \(k!\) terms cancel the
simplex volume \(1/k!\). This proves the \(4^k\) bound for every \(k\ge1\). Integrating
\(D^2F(tX)[X,Z]\) over the unit interval gives
\[
 \|DF(X)[Z]-LZ\|\le16\|X\|_F\|Z\|_F.
\]

\noindent\emph{A horizontal chart on the stated ball.} Set \(r_V=\|V-I\|_F\le1/16\). The case \(r_V=0\) takes \(X=Z=0\).
Let \(e^{\pm i\theta_j}\), \(0\le\theta_j\le\pi\), be the eigenvalues of
\(V\), one pair for each rotation plane. Then
\[
\begin{gathered}
\|V-I\|_F^2=\sum_j2|e^{i\theta_j}-1|^2=8\sum_j\sin^2(\theta_j/2),\\
\|\log V\|_F^2=2\sum_j\theta_j^2 .
\end{gathered}
\]
Hence \(a_j:=\sin(\theta_j/2)\le r_V/2\) for each \(j\). Since
\(\arcsin(a)\le a/\sqrt{1-a^2}\) on \([0,1)\),
\[
 \theta_j=2\arcsin a_j\le2a_j/\sqrt{1-r_V^2/4}.
\]
Squaring and
summing gives
\begin{equation}
\|\log V\|_F\le r_V/\sqrt{1-r_V^2/4}<2r_V.
 \label{eq:opt-chart-chord}
\end{equation}
Minimize \(\|\log W\|_F\) on the compact set
\[
 W\in VG_0,\ \|W-I\|_F\le2r_V.
\]
The real
principal logarithm is smooth there. At a point with equality in the radius constraint,
\[
 \|\log W\|_F\ge\|W-I\|_F=2r_V,
\]
by \(2\sin(\theta/2)\le\theta\) in the same two
identities, while \(V\) is an interior competitor satisfying
\eqref{eq:opt-chart-chord}. A minimizer \(W=e^X\) thus lies in the interior and
\(\|X\|_F\le\|\log V\|_F\).
For \(Y\in\mathfrak g_0\) vary \(W(s)=We^{sY}=e^{X(s)}\). Differentiating at zero and using
invariance of the Frobenius product gives
\[
\begin{aligned}
Y&=W^{\mathsf T}W'(0) =\int_0^1 e^{-aX}X'(0)e^{aX}\,da,\\
\langle X,Y\rangle_F
 &=\int_0^1\langle e^{aX}X e^{-aX},X'(0)\rangle_F\,da\\
 &=\tfrac12\partial_s\|X(s)\|_F^2\big|_{s=0} =0 .
\end{aligned}
\]
Hence \(X\in\mathfrak m\). Since \(W\in VG_0\), \(h=e^{-X}V\in G_0\).
Multiplication by \(e^{-X}\) is an isometry, and
\(\|e^X-I\|_{\op}\le\|X\|_{\op}\). So
\[
\|h-I\|_{\op}=\|V-e^X\|_{\op}\le r_V+\|X\|_F<3r_V .
\]
Thus \(-1\) is not an eigenvalue of \(h\), and \(Z=\log h\) is real skew.
The inequality \(|\theta|\le(\pi/2)|e^{i\theta}-1|\) on \([-\pi,\pi]\)
gives
\[
 \|Z\|_{\op}\le(\pi/2)\|h-I\|_{\op}<3\pi r_V/2.
\]
Now \(e^{d\rho_4(Z)}x_0=\rho_4(h)x_0=x_0\), and
\[
 \|d\rho_4(Z)\|_{\op}\le4\|Z\|_{\op}<6\pi r_V\le3\pi/8.
\]
The skew matrix \(d\rho_4(Z)\) has eigenvalues \(i\mu\) with
\(|\mu|<2\pi\), and \(e^{i\mu}=1\) forces \(\mu=0\). So \(x_0\) lies in
its kernel: \(LZ=d\rho_4(Z)x_0=0\) and \(Z\in\mathfrak g_0\).
Finally \(V=e^Xh=e^Xe^Z\) with
\[
 \|X\|_F\le\|\log V\|_F\le\frac{r_V}{\sqrt{1-r_V^2/4}} .
\]
The minimization is only a proof of chart existence.
\end{proof}
\begin{alem}[Refinement from an independently obtained frame]
\label{lem:opt-refinement}
Let $\rho_0=1/256$ and let $\mathcal B_{\rho_0}$ be the closed radius-$\rho_0$ ball in $\mathfrak m$.
Suppose the frame $R_{\rm c}\in SO(d)$ is independent of the refinement bank and
\[
 \begin{gathered}
 R_*=R_{\rm c}e^{X_*}g,\qquad g\in G_0,\\
 X_*\in\mathfrak m,\qquad \|X_*\|_F\le\rho_0/2 .
 \end{gathered}
\]
Let $y=x_*+\xi$ be the bank mean from $K$ independent copies per quartet and set
$s=\log(8/\delta)$ for $0<\delta<1$. There is a universal $C$ such that, if
$K\ge C(D+s)$, the following holds with probability at least $1-\delta/2$.
On $\mathcal B_{\rho_0}$ the score
\[
 \phi(X)=\langle y,\Lambda^4(R_{\rm c})F(X)\rangle
\]
is $5$-strongly concave and $18$-smooth. Its unique maximizer $\widehat X$ satisfies
\[
 \|F(\widehat X)-F(X_*)\|_2^2\le C(D+s)/K .
\]
Projected gradient ascent approximates it to any prescribed precision in polynomial time.
\end{alem}
\begin{proof}
\proofdisplays
\noindent\emph{One noise bound for the chart.}
In the seed basis let $\mathcal W$ span the quartets with at least three axes in one block.
Counting four axes in the block or three inside and one outside gives
\[
 p=\dim\mathcal W=B\left[\binom84+\binom83(d-8)\right].
\]
Write $J_4$ for the isometric antisymmetric embedding into $(\R^d)^{\otimes4}$:
each sorted basis quartet maps to the signed sum of its permutations divided by $\sqrt{4!}$.
For four operator-norm contractions $M_i$ and $z$ in the unit ball of $\mathcal W$, put
\[
 \mathcal T_\xi(M_1,\ldots,M_4,z)
 =\langle J_4\xi,(\otimes_i M_i)J_4z\rangle .
\]
Let $a=H(\xi)$ be the supremum of $|\mathcal T_\xi|$ over these five unit balls. Each fixed tuple
tests $\xi$ against $b=J_4^*(\otimes_i M_i)J_4z$, with $\|b\|_2\le1$.
Each measured sign $s_{Il}$ has mean $x_{*,I}$. For any fixed $b$,
independence and Hoeffding's lemma~\cite{hoeffding1963} give
\begin{equation}
\begin{aligned}
 \E e^{u\langle b,\xi\rangle}
 &=\prod_{I,l}\E e^{u b_I(s_{Il}-x_{*,I})/K}\\
 &\le\exp\!\left(\frac{u^2\|b\|_2^2}{2K}\right).
\end{aligned}
\label{eq:opt-linear-noise}
\end{equation}
Cover each of the five unit balls by an $h$-net in its own norm.
Disjoint balls of radius $h/2$ around separated points give at most
$(1+2/h)^q$ tuples~\cite{vershynin2018}, where $q=4d^2+p$.
Replace the five arguments one at a time. By multilinearity each replacement costs at most $ha$,
so $(1-5h)a$ is at most the maximum on the net. Choose $5h=1/2$.
The two-sided tail from \eqref{eq:opt-linear-noise} and a union bound now give
\[
 \begin{gathered}
 \Prob\{a>b_\xi\}\le2e^{-s},\qquad
 b_\xi=2\sqrt{2(q\log21+s)/K}.
 \end{gathered}
 \tag{\ref*{eq:opt-linear-noise}a}\label{eq:opt-tensor-concentration}
\]
Here $21=1+2/h$. Since $p\le7d^2$ and $D=d(4d-25)/8\ge7d^2/64$ for $d\ge8$,
$b_\xi=O(\sqrt{(D+s)/K})$ with an absolute constant. This event controls all frames simultaneously.

\noindent\emph{Noise in the first two derivatives.}
Put $A=d\Lambda^4(X)$, $S_U=d\Lambda^4(U)$, $g_t=e^{tX}$ and
$\eta=\Lambda^4(R_{\rm c})^{\mathsf T}\xi$. With $L=DF(0)$, the vector
$LZ=S_Zx_0$ belongs to $\mathcal W$: each insertion changes only one axis
of a seed quartet. Lemma~\ref{lem:opt-geometry} gives $\|LZ\|_2\le4\|Z\|_F$.
The exponential derivative at order one is
\[
 DF(X)U=\Lambda^4(g_1)\int_0^1 L(g_t^{\mathsf T}Ug_t)\,dt .
\]
Pairing with $\eta$ uses four tensor factors equal to $R_{\rm c}g_1$ in $\mathcal T_\xi$,
so its absolute value is at most $4a\|U\|_F$.

At order two the integrand is the sum of the ordered words
\[
 \begin{gathered}
 e^{t_0A}S_Ue^{t_1A}S_Ve^{t_2A}x_0,\qquad
 e^{t_0A}S_Ve^{t_1A}S_Ue^{t_2A}x_0 .
 \end{gathered}
\]
Integrate over $t_i\ge0$, $\sum_i t_i=1$, using $dt_1dt_2$.
The right insertion gives the rotated native vector
\[
 S_Ve^{t_2A}x_0=\Lambda^4(g_{t_2})z,\qquad
 z=L(g_{t_2}^{\mathsf T}Vg_{t_2}).
\]
Under $J_4$, the other insertion acts once in each of four slots.
In slot $i$ the matrix is $R_{\rm c}g_{t_0}Ug_{t_1+t_2}$, of norm at most
$\|U\|_F$; the other three matrices are $R_{\rm c}g_1$.
Thus each slot contributes at most $a\|U\|_F\|z\|_2$.
Since $\|z\|_2\le4\|V\|_F$, each ordered word contributes at most
$16a\|U\|_F\|V\|_F$ before integration. Each is integrated over a simplex
of area $1/2$. Together with the first derivative bound this gives
\begin{equation}
\begin{aligned}
 |\langle\eta,DF(X)U\rangle|&\le4a\|U\|_F,\\
 |\langle\eta,D^2F(X)[U,V]\rangle|&\le16a\|U\|_F\|V\|_F .
\end{aligned}
\label{eq:opt-refinement-noise}
\end{equation}

\noindent\emph{Curvature and its margin.}
Write $F_X=F(X)$, $F_*=F(X_*)$, $J_X=DF(X)|_{\mathfrak m}$.
The rotated data are $F_*+\eta$. Differentiating $\|F_X\|_2^2=14B$ twice yields
\[
\begin{aligned}
 \langle F_X,J_XU\rangle&=0,\\
 \langle F_X,D^2F(X)[U,V]\rangle&=-\langle J_XU,J_XV\rangle .
\end{aligned}
\]
Consequently $D\phi(X)[U]=\langle F_*-F_X+\eta,J_XU\rangle$ and
\begin{equation}
\begin{aligned}
 D^2\phi(X)[U,V]&=-\langle J_XU,J_XV\rangle\\
 &\quad+\langle F_*-F_X,D^2F(X)[U,V]\rangle\\
 &\quad+\langle\eta,D^2F(X)[U,V]\rangle .
\end{aligned}
\label{eq:opt-hessian}
\end{equation}
Lemma~\ref{lem:opt-geometry} bounds $\|DF\|\le4$ and $\|D^2F\|\le16$.
On the horizontal ball, subtracting $\|(J_X-L)U\|_2$ from
$\|LU\|_2$ gives the lower bound in
\[
 (\sqrt7-16\rho_0)\|U\|_F\le\|J_XU\|_2\le4\|U\|_F .
\]
Since $\|X-X_*\|_F\le3\rho_0/2$, we have $\|F_X-F_*\|_2\le6\rho_0$.
By \eqref{eq:opt-refinement-noise} the two remainder terms in
\eqref{eq:opt-hessian} therefore have combined absolute value at most
$\beta\|U\|_F^2$ at $V=U$, where $\beta=16(6\rho_0+a)$. Thus
\[
 \begin{gathered}
 [(\sqrt7-16\rho_0)^2-\beta] I
 \preceq -\nabla^2\phi(X)\preceq(16+\beta) I .
 \end{gathered}
\]
Choose the bank size so that $b_\xi\le 1/16-6\rho_0$.
The right-hand side is positive at $\rho_0=1/256$,
so a universal $C$ in $K\ge C(D+s)$ suffices. On the event
\eqref{eq:opt-tensor-concentration}, $\beta\le1$ and
$\sqrt7-16\rho_0>5/2$. It follows that $5I\preceq-\nabla^2\phi(X)\preceq18I$
throughout the ball. The event fails with probability at most $2e^{-s}=\delta/4$.

\noindent\emph{Statistical error and arithmetic error.}
Write $\mu=5$, $M=18$ and let $\widehat X$ maximize $\phi$ on $\mathcal B_{\rho_0}$.
Existence follows from compactness and uniqueness from strong concavity.
Put $E=\widehat X-X_*$. First-order optimality at $\widehat X$ gives
$\langle\nabla\phi(\widehat X),E\rangle_F\ge0$, even on the boundary.
Integrate the curvature lower bound along the segment from $X_*$ to $\widehat X$:
\[
\begin{aligned}
 \mu\|E\|_F^2
 &\le\langle\nabla\phi(X_*)-\nabla\phi(\widehat X),E\rangle_F\\
 &\le\langle\nabla\phi(X_*),E\rangle_F\le4a\|E\|_F .
\end{aligned}
\]
The last inequality uses \eqref{eq:opt-refinement-noise} at $X_*$.
For the iteration, put $T(X)=\Pi_{\mathcal B_{\rho_0}}(X+\nabla\phi(X)/M)$.
Optimality gives $T(\widehat X)=\widehat X$. The derivative of the unprojected
step is $I+\nabla^2\phi/M$, with spectrum in $[0,\lambda]$, where
$\lambda=1-\mu/M=13/18$. Integration on a segment and nonexpansiveness of projection
give $\|T(X)-T(Y)\|_F\le\lambda\|X-Y\|_F$.
For stored iterates in the ball with per-step error at most $\eta_X$, the distance
to $\widehat X$ thus obeys $e_{j+1}\le\lambda e_j+\eta_X$.
Here $e_j=\|\widetilde X_j-\widehat X\|_F$. Starting at zero gives $e_0\le\rho_0$. Together with the preceding
statistical estimate this proves
\[
\begin{gathered}
 \|F(\widehat X)-F_*\|_2^2\le (16a/\mu)^2,\\
 \|\widetilde X_k-\widehat X\|_F
 \le\lambda^k\rho_0+\eta_X\frac{1-\lambda^k}{1-\lambda}.
\end{gathered}
\tag{\ref*{eq:opt-hessian}a}\label{eq:opt-maximizer-displacement}
\]
The first line uses $\|DF\|\le4$ and gives the stated risk by
\eqref{eq:opt-tensor-concentration}, enlarging $C$ if necessary.
For chart precision $z>0$ take $\lambda^k\rho_0\le z/2$ and
$\eta_X\le(1-\lambda)z/2$. The exponential and gradient series act on
matrices of dimension $\binom d4$ with generator norm at most $4\rho_0$.
The factorial tail bound and the horizontal projection and rounding estimates
under \emph{Refinement and output precision} below give the prescribed
per-step accuracy in polynomial bit work.
\end{proof}
\subsection{Local decoding and corrections of a partial frame}
On one block put $\varphi=x^\omega$ and
$C(O)=\wedge^3(O)C(I_8)$ for $O\in O(8)$, where the $56$ rows
are sorted triples and
\begin{equation}
\begin{gathered}
C(I_8)_{T,s}=\varphi(e_{t_1},e_{t_2},e_{t_3},e_s),\\
 C(I_8)^{\mathsf T}C(I_8)=7I_8 .
\end{gathered}
\label{eq:opt-cofactor}
\end{equation}
To prove the normalization, let \(v\) be a unit vector and \(e_1=v,e_2,\ldots,e_8\) an
orthonormal basis. The number
\(\|C(I_8)v\|^2=\sum_T\varphi(e_{t_1},e_{t_2},e_{t_3},v)^2\) is the squared norm of the
contraction of \(\varphi\) with \(v\). It is the same in every orthonormal basis.
By Lemma~\ref{lem:closure}, the pair matrix of \(\varphi\) in the new basis is an
orthogonal conjugate of \(Q_\omega\). We keep the name \(Q_\omega\). It obeys
\(Q_\omega^2=3I-2Q_\omega\) by~\eqref{eq:rec-local-spectrum}. For fixed \(j\),
\[
 \sum_{p<q}|\varphi(v,e_j,e_p,e_q)|^2=\|Q_\omega(v\wedge e_j)\|^2.
\]
Alternation gives
\(\langle v\wedge e_j,Q_\omega(v\wedge e_j)\rangle=0\). Each triple away from \(v\)
occurs three times in the pair sum. Hence
\[
 \begin{aligned}
3\|C(I_8)v\|^2
 &=\sum_{j=2}^8\sum_{p<q}|\varphi(v,e_j,e_p,e_q)|^2\\
&=\sum_{j=2}^8\langle v\wedge e_j,(3I-2Q_\omega)(v\wedge e_j)\rangle=21 .
\end{aligned}
\]
Polarization gives the second identity in~\eqref{eq:opt-cofactor}.
Consequently
\begin{equation}
\begin{gathered}
 C(O)^+=C(O)^{\mathsf T}/7,\\
 \|C(M)\|_{\op}\le\sqrt7\|M\|_{\op}^3,\\[3pt]
\begin{aligned}
 &\|C(M)-C(M')\|_{\op}\\
 &\quad \le3\sqrt7\max(\|M\|_{\op},\|M'\|_{\op})^2\|M-M'\|_{\op}.
\end{aligned}
\end{gathered}
\label{eq:opt-cofactor-bounds}
\end{equation}
Here $C(M)=\wedge^3(M)C(I_8)$ also defines the map for nonorthogonal $M$.
The last bound follows by restricting the telescoping identity
\[
\begin{aligned}
M^{\otimes3}-{M'}^{\otimes3}
 &=(M-M')\otimes M\otimes M\\
 &\quad+M'\otimes(M-M')\otimes M\\
 &\quad+M'\otimes M'\otimes(M-M')
\end{aligned}
\]
to the alternating tensors.

\noindent\emph{Local decoder.}
Given $z\in\R^{70}$, put
\[
 J_-=\diag(-1,1,\ldots,1).
\]
When
$\langle z,\star z\rangle<0$, replace $z$ by $\wedge^4(J_-)z$. This negates
the entries $z_I$ with $1\in I$. Remember that the reflection was used. In the resulting orientation,
take a top real unit eigenvector $w$ of
$E^\dagger(\sum_I z_IB_I)E$, with $w^{\mathsf T}v_0\ge0$.
Equivalently, minimize its local parent. Set
\begin{equation}
\begin{gathered}
\theta=\arccos(w^{\mathsf T}v_0),\\
 A=\frac{\theta}{\sin\theta}(wv_0^{\mathsf T}-v_0w^{\mathsf T}),
 \qquad
  e^Av_0=w ,
\end{gathered}
\label{eq:rec-local-rotation}
\end{equation}
with its continuous value at $\theta=0$.
For $\theta>0$, put $u=(w-\cos\theta\,v_0)/\sin\theta$. Then
$Av_0=\theta u$ and $Au=-\theta v_0$, so the exponential series gives
$e^Av_0=\cos\theta\,v_0+\sin\theta\,u=w$.
Equation~\eqref{eq:rec-spin-map} converts $A$ to its Majorana rotation $e^X$.
The skew matrix $X$ has entries $X_{pq}=c_{pq}=\tfrac12\Tr(G_{pq}^{\mathsf T}A)$
for $p<q$. Then $\exp(\tfrac12\sum_{p<q}X_{pq}\gamma_p\gamma_q)E=Ee^A$, so the
state $Ew$ has quartic vector $\wedge^4(e^X)\varphi$.
Left-multiply that rotation by $J_-$ if the reflection was used.
Denote the resulting orthogonal matrix by $\mathcal D(z)$.
This constant-dimensional computation has a fixed tie rule;
the small-error event has a simple top eigenvalue.
\begin{alem}[Stable local decoder and frame correction]
\label{lem:rec-decoder}
Let $\mathcal T$ be the tangent space at $\varphi$ to its $SO(8)$ orbit under $\Lambda^4$, and let
$\mathcal H=\{h\in O(8):\Lambda^4(h)\varphi=\varphi\}$.
Every $h\in\mathcal H$ is proper since an improper orthogonal map exchanges the two Hodge sectors.
Let $z=c\,\Lambda^4O(\varphi+\nu+\tau)$, with $c>0$, $O\in O(8)$,
$\nu\perp\mathcal T$, $\tau\in\mathcal T$, $\|\nu\|\le1/2$ and $\|\tau\|\le1$.
Then $\det\mathcal D(z)=\det O$ and
\begin{equation}
 \min_{h\in\mathcal H}\|\mathcal D(z)-Oh\|_F
 \le \frac{4\|\tau\|}{16-4\sqrt2\|\nu\|}.
 \label{eq:opt-decoder}
\end{equation}
Let $R_*=[R_1,\ldots,R_B]$ be the true orthogonal frame. Match mutually orthogonal eight-frames $q_b$ to distinct $R_b$, $1\le b\le j$. With $P_b=q_bq_b^{\mathsf T}$ write $e_b$ and $\eta_b$ for the Frobenius and operator norms of $(I-P_b)R_b$, respectively.
For the measured alternating tensor $Y=Y(x_*)+Y_\xi$, where $Y_\xi=Y(\xi)$ is its noise tensor, the local update is
\begin{equation}
\begin{gathered}
 z_b=[Y(q_{b,I})]_I,\qquad \widehat O_b=\mathcal D(z_b),\\
 V_b=[Y(q_{b,T},\cdot)]_T,\\
 X_b^\perp=(I-P_b)V_b C(\widehat O_b)/7,\\
 S'_b=(X_b^\perp)^{\mathsf T}X_b^\perp,\\
 X_b=X_b^\perp+q_b\widehat O_b(I-S'_b)^{1/2}.
\end{gathered}
\label{eq:opt-block-update}
\end{equation}
Here $I$ ranges over sorted quartets and the columns of $V_b$ are indexed by sorted triples.
Define $n_b=[Y_\xi(q_{b,I})]_I$ and
$N_b=(I-P_b)[Y_\xi(q_{b,T},\cdot)]_T C(\widehat O_b)/7$.
Suppose $e_b\le1/10$, $\|n_b\|\le3/100$ and $\|N_b\|_F\le1/1000$.
Then $S'_b\preceq I/64$. There exist $h_b\in\mathcal H$ such that
\begin{equation}
\begin{gathered}
 \|X_b-R_bh_b\|_F\le A_b,\\
 A_b=3\eta_b^2e_b+\tfrac12\|n_b\|+
 \tfrac8{\sqrt{63}}\|N_b\|_F .
\end{gathered}
\label{eq:opt-correction-error}
\end{equation}
Set $X=[X_1,\ldots,X_j]$, $R_h=[R_1h_1,\ldots,R_jh_j]$ and $A=(\sum_b A_b^2)^{1/2}$. Then $\|X-R_h\|_F\le A$. If $A\le1/4$, the polar factor $Q^+=[q_b^+]$ exists. Both the full frame error and the block leakage satisfy
\[
 \begin{gathered}
 \|Q^+-R_h\|_F\le\frac{2A}{2-A}\le\tfrac87 A,\\
 \Big(\sum_b\|(I-q_b^+q_b^{+\mathsf T})R_b\|_F^2\Big)^{1/2}
 \le\|Q^+-R_h\|_F .
 \end{gathered}
\]
\end{alem}
\begin{proof}
\proofdisplays
\noindent\emph{The score detects tangent error.}
Let $v=v_0$ be the real unit spinor of the paired seed. On its real even spinor space put $\mathcal M(z)=\sum_I z_I B_I$ and $P=vv^{\mathsf T}$. Write $z_\pm=(z\pm\star z)/2$ and let $\epsilon_I$ be the sign of $(I,I^c)$.
This representation of self-dual forms is classical~\cite[Secs.~3--4]{katzshnider2008}. The trace with the even projector $(I+\Pi)/2$ retains only identical or complementary quartets. Hence
\[
\begin{gathered}
 \Tr_+(B_I B_J)=8(\delta_{IJ}+\epsilon_I\delta_{I^c,J}),\\
 \langle\mathcal M(z),\mathcal M(w)\rangle_F=16\langle z_+,w_+\rangle .
\end{gathered}
\tag{\ref*{eq:opt-correction-error}a}\label{eq:opt-score-trace}
\]
For a real unit spinor $w$ the block expansion~\eqref{eq:block-expansion} and covariance give
$\mathcal M(\varphi_w)=16ww^{\mathsf T}-2I$, where $\varphi_w$ is the quartic vector of $Ew$. Differentiate at $w=v$.
All velocities $u\perp v$ occur by transitivity in Lemma~\ref{lem:rec-local}, so
$\mathcal M(\mathcal T)$ consists exactly of $uv^{\mathsf T}+vu^{\mathsf T}$
with $u\perp v$. The orbit is self-dual and has constant norm, hence
$\mathcal T$ is self-dual and perpendicular to $\varphi$. By
\eqref{eq:opt-score-trace}, $\mathcal M(\nu)$ is orthogonal to this off-diagonal space. Thus
\[
\begin{gathered}
 (I-P)\mathcal M(\nu)v=0,\quad
 \mathcal M(\tau)=bv^{\mathsf T}+vb^{\mathsf T},\quad b\perp v,\\
 2\|b\|^2=16\|\tau\|^2.
\end{gathered}
\]
Since $\star\Lambda^4O=(\det O)\Lambda^4O\star$, the Hodge score multiplied by $(\det O)/c^2$ is
$\|\varphi+\tau+\nu_+\|^2-\|\nu_-\|^2\ge(\sqrt{14}-1/2)^2-1/4>0$.
The decoder therefore uses the correct orientation. In that orientation the score is spin-conjugate to $c\mathcal M(\varphi+\nu+\tau)$.

In $\R v\oplus v^\perp$ the normalized native score has the form
\[
 \begin{pmatrix} a&b^{\mathsf T}\\ b&D\end{pmatrix},\qquad
 g=a-\lambda_{\max}(D)\ge16-4\sqrt2\|\nu\|>0 .
\]
Indeed the unperturbed gap is $16$, and the change in any diagonal difference is at most
$\|\mathcal M(\nu)\|_F\|P-uu^{\mathsf T}\|_F\le4\sqrt2\|\nu\|$ for unit $u\perp v$.
Interlacing gives $\lambda_1\ge a>\lambda_{\max}(D)\ge\lambda_2$.
Write the top unit vector as $v\cos\alpha+y$, $y\perp v$, with nonnegative cosine.
Its lower equation gives
\[
 (\lambda_1 I-D)y=b\cos\alpha,\qquad
 \tan\alpha\le\|b\|/g .
\]
Indeed a zero cosine would force $y=0$. Lift the plane rotation through $\alpha$ by~\eqref{eq:rec-generator-gram}--\eqref{eq:rec-spin-map}. Its physical skew generator $X_\alpha$ has Frobenius norm $\sqrt2\alpha$, also when $\alpha=0$. The lifted rotation followed by $O$ and the decoder give the same oriented spinor ray after the common Hodge reflection. Therefore $h=e^{-X_\alpha}O^{\mathsf T}\mathcal D(z)$ fixes $\varphi$.
The chord is at most the generating arc:
$\|\mathcal D(z)-Oh\|_F\le\sqrt2\alpha\le4\|\tau\|/g$.
This proves \eqref{eq:opt-decoder} and explains exact recovery when $\tau=0$.

\noindent\emph{Complementary quartets suppress tangent leakage.}
Fix $b$ and write $q=q_b$, $P_q=qq^{\mathsf T}$, $e=e_b$ and $\eta=\eta_b$.
Set $W_\beta=q^{\mathsf T}R_\beta$ and take the polar decomposition $W_b=OM$.
Orthogonality gives
\[
\begin{gathered}
 M^2=I-R_b^{\mathsf T}(I-P_q)R_b,\quad
 \sqrt{1-\eta^2}I\preceq M\preceq I,\\
 \sum_{\beta\ne b}W_\beta W_\beta^{\mathsf T}
 =I-OM^2O^{\mathsf T}.
\end{gathered}
\]
Thus $\|W_\beta\|_{\op}\le\eta$ for $\beta\ne b$ and their squared Frobenius norms sum to $e^2$.
Restricting the true four-form to $q$ gives $z_b=\sum_\beta(\Lambda^4W_\beta)\varphi+n_b$.
In an oriented eigenbasis of $M$ write $p_I=\prod_{i\in I}m_i$.
Since $p_Ip_{I^c}=\det M$ and $\epsilon_I\varphi_{I^c}=\varphi_I$,
\begin{equation}
 \big[(\Lambda^4M\varphi)_+-\sqrt{\det M}\,\varphi\big]_I
 =\tfrac12(\sqrt{p_I}-\sqrt{p_{I^c}})^2\varphi_I .
\label{eq:opt-complement-square}
\end{equation}
Each $\sqrt{p_I}\in[1-\eta^2,1]$, so the norm on the left is at most $\sqrt{14}\,\eta^4/2$.
For $\beta\ne b$ the product of the four largest singular values of $W_\beta$ is at most $\eta^2\|W_\beta\|_F^2/4$: AM--GM on their squares gives
\[
 \prod_{i=1}^4s_i\le\left(\tfrac14\sum_{i=1}^4s_i^2\right)^2
 \le\tfrac{\eta^2}{4}\sum_{i=1}^4s_i^2 .
\]
Their total contribution in norm is at most $\sqrt{14}\,\eta^2 e^2/4$.

Unrotate $z_b$ by $O$ and subtract $\varphi$. The full perturbation has norm at most
\[
 \sqrt{14}\,[1-(1-\eta^2)^2+\eta^2 e^2/4]+\|n_b\|<0.11
\]
by the spectrum of $\Lambda^4M$ and the stated input tolerances.
Let $\tau_b$ be the tangent projection of this perturbation and $\nu_b$ its orthogonal remainder.
The tangent projection removes both the anti-self-dual part and the multiple of $\varphi$ in~\eqref{eq:opt-complement-square}. Applying~\eqref{eq:opt-decoder} gives
\[
\begin{gathered}
 \|\tau_b\|\le\sqrt{14}(\eta^4/2+\eta^2e^2/4)+\|n_b\|,\\
 \delta_b:=\|\widehat O_b-Oh_b\|_F
 \le\tfrac1{10}\eta^2e+\tfrac13\|n_b\|.
\end{gathered}
\tag{\ref*{eq:opt-complement-square}a}\label{eq:opt-local-tangent-error}
\]
Indeed $\eta\le e\le1/10$, $\sqrt{14}<4$ and
$4/(16-2\sqrt2)<1/3$. The Hodge score in absolute value exceeds
$(\sqrt{14}-0.11)^2-0.11^2>1$ and the eigengap exceeds $16-4\sqrt2(0.11)>1$.
Here the scale is $c=1$.

\noindent\emph{Triple measurements recover the transverse columns.}
Write $C_0=C(I_8)$ as in~\eqref{eq:opt-cofactor}. Expanding the fourth argument in the true frame gives
\[
 \begin{gathered}
 V_b=\sum_\beta R_\beta C(W_\beta)^{\mathsf T}+V_{\rm noise},\\
 T_\beta=C(W_\beta)^{\mathsf T}C(\widehat O_b)/7,\\
 X_b^\perp=(I-P_q)\sum_\beta R_\beta T_\beta+N_b .
 \end{gathered}
\]
Indeed, the coefficient of $R_\beta e_s$ in the noiseless column $T$ is
\[
 C(W_\beta)_{T,s}
 =\varphi(W_\beta^{\mathsf T}e_{t_1},W_\beta^{\mathsf T}e_{t_2},W_\beta^{\mathsf T}e_{t_3},e_s).
\]
Let $J_3$ be the isometric antisymmetric embedding. Its entries on the cofactor are
$(J_3C_0)_{ijks}=\varphi_{ijks}/\sqrt6$. Contracting over $j,k,s$ gives
\[
\begin{gathered}
 \tfrac16\sum_{jks}\varphi_{ijks}\varphi_{i'jks}
 =(C_0^{\mathsf T}C_0)_{ii'}=7\delta_{ii'},\\
 \|(W\otimes I\otimes I)J_3 C_0\|_F^2=7\|W\|_F^2,\\
 \|C(W)\|_F\le\sqrt7\|W\|_{\op}^2\|W\|_F .
\end{gathered}
\]
Since $\|C(\widehat O_b)\|_{\op}=\sqrt7$, the other blocks together contribute at most $\eta^2e$:
before projection their squared norm is $\sum_{\beta\ne b}\|T_\beta\|_F^2\le\eta^4e^2$.
For the matched block $C(Oh_b)=C(O)h_b$ by stabilizer invariance.
Replacing $C(W_b)$ by $C(O)$ costs at most
$\sqrt7[1-(1-\eta^2)^{3/2}]$ in operator norm.
Replacing $C(\widehat O_b)$ by $C(Oh_b)$ costs at most $3\sqrt7\,\delta_b$ by telescoping the three orthogonal slots.
Hence $\|T_b-h_b\|_{\op}\le 1-(1-\eta^2)^{3/2}+3\delta_b$.
Put $H_b=(I-P_q)R_bh_b$. Subtracting $H_b$ from the transverse update gives
\[
\begin{aligned}
 k_b&:=\|X_b^\perp-H_b\|_F\\
 &\le e[1-(1-\eta^2)^{3/2}+3\delta_b+\eta^2]+\|N_b\|_F\\
 &\le\tfrac{13}{5}\eta^2e+\tfrac1{10}\|n_b\|+\|N_b\|_F<1/100 .
\end{aligned}
\tag{\ref*{eq:opt-complement-square}b}\label{eq:opt-off-column-error}
\]
The last line uses $1-(1-s)^{3/2}\le3s/2$, \eqref{eq:opt-local-tangent-error} and $e\le1/10$. In particular $\|X_b^\perp\|_{\op}\le\eta+k_b<1/8$, proving the square-root guard.

\noindent\emph{Completion and polar projection.}
Write $S(L)=(I-L^{\mathsf T}L)^{1/2}$. From $S(H_b)=h_b^{\mathsf T}Mh_b$,
\[
\begin{gathered}
 X_b =X_b^\perp+q\widehat O_b S(X_b^\perp),\\
 R_bh_b=H_b+qOh_b S(H_b).
\end{gathered}
\]
For $\|L\|_{\op},\|L'\|_{\op}\le r_\perp<1$, the difference $D=S(L)-S(L')$ satisfies
$S(L)D+DS(L')=L'^{\mathsf T}L'-L^{\mathsf T}L$.
In the two eigenbases each divisor is at least $2\sqrt{1-r_\perp^2}$; the numerator has Frobenius norm at most $2r_\perp\|L-L'\|_F$.
Thus
\[
 \|S(L)-S(L')\|_F\le\frac{r_\perp}{\sqrt{1-r_\perp^2}}\,\|L-L'\|_F .
 \tag{\ref*{eq:opt-complement-square}c}\label{eq:opt-completed-columns}
\]
Use $r_\perp=1/8$ and $c_\perp=(1-r_\perp^2)^{-1/2}$.
The in-plane and transverse differences are orthogonal after holding the rotation fixed. Allowing the rotation error separately gives
\[
 \begin{aligned}
 \|X_b-R_bh_b\|_F
 &\le \delta_b+\sqrt{k_b^2+\|S(X_b^\perp)-S(H_b)\|_F^2}\\
 &\le\delta_b+c_\perp k_b .
 \end{aligned}
\]
Now $c_\perp=8/\sqrt{63}<11/10$. The coefficients of $\eta^2 e$ and $\|n_b\|$ are at most
$1/10+(13/5)c_\perp<3$ and $1/3+c_\perp/10<1/2$, respectively. This proves \eqref{eq:opt-correction-error}. Writing $E=\|(e_b)_b\|_2$, $\nu_{\rm loc}=\|(\|n_b\|)_b\|_2$ and $\nu_{\rm off}=\|(\|N_b\|_F)_b\|_2$, Minkowski gives
\[
 A\le3E^3+\tfrac12\nu_{\rm loc}+c_\perp\nu_{\rm off},
\]
since $\sum_b\eta_b^4e_b^2\le E^6$.

For $A\le1/4$, $\sigma_{\min}(X)\ge1-A>1/2$.
Write $X=U\Sigma V^{\mathsf T}$ (thin SVD), $Q^+=UV^{\mathsf T}$ and
$c_i=(U^{\mathsf T}R_hV)_{ii}\in[-1,1]$.
The trace expansion gives the identities
\[
\begin{gathered}
 \|Q^+-R_h\|_F^2=\sum_i 2(1-c_i),\\
 \|X-R_h\|_F^2=\sum_i\big[(\sigma_i-1)^2+2\sigma_i(1-c_i)\big].
\end{gathered}
\]
For $t=1-c_i\in[0,2]$, the difference between the second summand and
$(1+\sigma_i)^2 t/2$ is $(\sigma_i-1)^2(1-t/2)\ge0$. Thus
\[
 \|Q^+-R_h\|_F\le
 \frac{2\|X-R_h\|_F}{1+\sigma_{\min}(X)}
 \le\frac{2A}{2-A}.
\tag{\ref*{eq:opt-complement-square}d}\label{eq:opt-polar-error}
\]
The same expansion with any orthonormal comparator shows $Q^+$ minimizes the distance to $X$ (each $c_i\le1$, with equality at $Q^+$). Finally $(I-q_b^+q_b^{+\mathsf T})R_bh_b=(I-q_b^+q_b^{+\mathsf T})(R_bh_b-q_b^+)$. Projection contracts, and summing squares proves the leakage bound.
\end{proof}
\subsection{Range refresh and the hidden block algebra}
Banks $A$ and $P$ have empirical means $y_A=x_*+\xi_A$ and
$y_P=x_*+\xi_P$.
Write $\ket{b,p}=R_*\ket{8(b-1)+p}$ for $1\le b\le B$, $1\le p\le8$,
\[
 W_b=\spn\{\ket{b,p}:1\le p\le8\}, \quad P_b=\sum_{p=1}^8\ketbra{b,p}{b,p},
\]
and $q=28B$.
Let $P$ project onto
$\mathcal E=\bigoplus_{b=1}^B\Lambda^2W_b$ in pair space.
For an orthonormal skew basis $A_1,\ldots,A_q$ of this range define
\begin{equation}
\begin{gathered}
 \mathscr T S=([A_k,S])_{k=1}^q,\\
 \mathscr C S=\mathscr T^*\mathscr T S
             =-\sum_{k=1}^q[A_k,[A_k,S]],\quad S=S^{\mathsf T}.
\end{gathered}
\label{eq:opt-commutant}
\end{equation}
The sum is invariant under an orthogonal change of the $A_k$.
All matrix norms below are spectral or Frobenius norms as indicated;
norms of maps on symmetric matrices explicitly name their domain norm.
\begin{alem}[Independent-bank range refresh]
\label{lem:opt-range-refresh}
Put $s=\log(1024/\delta)$ and $t_*=2.4\cdot10^{-7}$.
Let $C$ be a sufficiently large universal constant. Two independent banks
with $K_A=K_P=\lceil C(n^2+s)\rceil$ copies per quartet suffice for the
following event, of probability at least $1-\delta/16$.
Let the columns of $U_A$ be orthonormal eigenvectors of $Q(y_A)$ for all
eigenvalues with $|\lambda|>1/2$. Then $U_A$ has $q$ columns, and
$Q(y_P)U_A$ has rank $q$.
Let $U'$ be an orthonormal basis of its range and $P'=U'U'^{\mathsf T}$.
Choose orthonormal bases $(A_k)$ of $\mathcal E$ and $(A'_k)$ of
$\operatorname{ran}P'$ with $\langle A_j,A'_k\rangle_F=0$ for $j\ne k$ and
$\langle A_k,A'_k\rangle_F\ge0$. These are principal bases. Identify them
with skew matrices, and set
\begin{equation}
D_k=A'_k-A_k,\qquad
 t^2=\left\|\sum_{k=1}^qD_k^{\mathsf T}D_k\right\|_{\op}.
\label{eq:opt-stacked-error}
\end{equation}
Then $t\le t_*$. The construction uses no true block or frame.
For $B=1$ block discovery is unnecessary.
\end{alem}
\begin{proof}
\proofdisplays
\noindent\emph{The pilot.} In the orthonormal pair basis let the orthonormal columns of \(U\)
span the true pair space \(\mathcal E=\bigoplus_b\Lambda^2 W_b\), and \(P=UU^{\mathsf T}\). On
each \(28\)-dimensional summand Lemma~\ref{lem:rec-local} gives \(Q_\omega=I-4P_7\), with
\(P_7\) an orthogonal rank-seven projector. On each block, products of two distinct
Majoranas and words of odd local parity have zero expectation. Thus all cross-block seed quartets
vanish. Covariance~\eqref{eq:rec-q-equivariance} gives
\[
\begin{gathered}
Q_*:=Q(x_*)=(\wedge^2R_*)Q(x_0)(\wedge^2R_*)^{\mathsf T},\\
 \operatorname{spec}Q_*=\{-3^{[7B]},1^{[21B]},0^{[m-q]}\},\qquad
  q=28B .
\end{gathered}
\]
For \(v,w\in\Lambda^2\R^d\), the definition of the flattening gives
\begin{equation}
\begin{gathered}
\langle v,Q(\xi)w\rangle=\langle v\wedge w,\xi\rangle,\\
 (v\wedge w)_I=\sum_{J\sqcup L=I,\;|J|=|L|=2}
                   \operatorname{sgn}(J,L)\,v_J w_L .
\end{gathered}
\label{eq:opt-pair-exterior-test}
\end{equation}
The coefficient of each quartet in \(v\wedge w\) is a sum over its six ordered pair splits.
Cauchy--Schwarz followed by summing over quartets gives \(\|v\wedge w\|\le\sqrt6\|v\|\|w\|\).
A \(1/4\)-net on each of the two unit balls in \(\R^m\) has at most \(9^{2m}\) pairs. For a
bilinear form of supremum \(S\), moving both arguments to their closest net points costs at
most \(S/4+S/4\). Equations~\eqref{eq:opt-linear-noise} and~\eqref{eq:opt-pair-exterior-test},
followed by a union bound, therefore give
\begin{equation}
\begin{aligned}
h_X&:=\|Q(\xi_X)\|_{\op}\\
 &\le 2\sqrt{12(2m\log9+s)/K_X},
 \quad
  X=A,P,
\end{aligned}
\label{eq:opt-matrix-noise}
\end{equation}
each except with probability \(2e^{-s}\), \(s=\log(1024/\delta)\). Here \(m=\binom d2\) and \(d=2n\).
Take \(h_A,h_P\le1/8\) by increasing the universal bank constant. Weyl's inequality yields
precisely \(q\) pilot eigenvalues of magnitude greater than \(1/2\). For their matrix
\(\widehat\Lambda\), projecting the eigen-equation gives
\[
\begin{gathered}
(I-P)U_A\widehat\Lambda=(I-P)Q(\xi_A)U_A,\\
 \|(I-P)U_A\|_{\op}\le h_A/(1-h_A)\le1/7 .
\end{gathered}
\]
Set \(M=U^{\mathsf T}Q(y_P)U_A\). Since the smallest singular value of \(Q_*|_{\mathcal E}\) is one,
\[
\begin{aligned}
\sigma_{\min}(M)
 &\ge\sigma_{\min}(U^{\mathsf T}U_A)-h_P \\
 &\ge\sqrt{1-1/49}-1/8>0.86 .
\end{aligned}
\]
Here
\[
 M=(U^{\mathsf T}Q_*U)(U^{\mathsf T}U_A)+U^{\mathsf T}Q(\xi_P)U_A,
\]
because \(Q_*=PQ_*P\).
Also
\[
 \sigma_{\min}(U^{\mathsf T}U_A)^2=1-\|(I-P)U_A\|_{\op}^2\ge1-1/49.
\]
In particular the refreshed columns have full rank.

\noindent\emph{The graph of the refreshed range.}
Put \(Y=(I-P)Q(\xi_P)U_A M^{-1}\), \(S=(I+Y^{\mathsf T}Y)^{-1/2}\).
In this proof \(S\) is this \(q\times q\) matrix.
The parallel and perpendicular components of the same observed matrix give
\begin{equation}
\begin{gathered}
Q(y_P)U_A=(U+Y)M,\qquad
  U^{\mathsf T}Y=0,\\
 U'=(U+Y)S,\qquad
  (U')^{\mathsf T}U'=I .
\end{gathered}
\label{eq:opt-refreshed-graph}
\end{equation}
Only \(Q(y_P)U_A\) is computed by the learner. The graph coordinates are proof variables.

\noindent\emph{Noise after conditioning.} Identify a pair \(\ket{i\wedge j}\) with
\((\ketbra ij-\ketbra ji)/\sqrt2\). Define \(V_v z=z_{\rm skew}v\);
then \(V_v^*w=w\wedge v/\sqrt2\) and \(\|V_v\|_{\op}\le\|v\|/\sqrt2\).
Here \(w\wedge v\) is the pair vector with coordinates \(w_iv_j-w_jv_i\), \(i<j\). Its squared
norm is \(\|w\|^2\|v\|^2-\langle w,v\rangle^2\).
For \(\|v\|=1,\ \|Z\|_F=1\), with \(Z\in\R^{d\times q}\),
\[
 \begin{aligned}
\langle Z,V_v(I-P)Q(\xi_P)U_A\rangle_F
 &=\langle (I-P)V_v^* ZU_A^{\mathsf T},Q(\xi_P)\rangle_F\\
&=\langle Q^*((I-P)V_v^* ZU_A^{\mathsf T}),\xi_P\rangle .
\end{aligned}
\]
Each quartet occurs six times in \(Q\), giving \(\|Q(a)\|_F^2=6\|a\|^2\).
Hence \(\|Q^*(N)\|\le\sqrt6\,\|N\|_F\) for every matrix \(N\). Also
\[
 \|(I-P)V_v^*ZU_A^{\mathsf T}\|_F\le\|V_v\|_{\op}\|Z\|_F\le1/\sqrt2.
\]
The last coefficient therefore has norm at most \(\sqrt3\). Conditional on bank \(A\), it is fixed for
fixed \(v,Z\). The same two-net argument as in \eqref{eq:opt-matrix-noise}, now in dimensions
\(d,dq\), proves
\begin{equation}
\begin{aligned}
\tau&:=\sup_{\|v\|=1}\|V_v(I-P)Q(\xi_P)U_A\|_F \\
 &\le 2\sqrt{6((d+dq)\log9+s)/K_P}.
\end{aligned}
\label{eq:opt-power-concentration}
\end{equation}
The conditional failure probability is at most \(2e^{-s}\), uniformly in \(U_A\). The
data-dependent factor \(M^{-1}\) enters only after \eqref{eq:opt-power-concentration}.

\noindent\emph{The stacked norm in the statement.}
Within each eight-space, seven normalized coordinate skews meet each axis.
Thus \(\sum_k A_k^{\mathsf T}A_k=(7/2)I_d\) for the columns of \(U\), after any orthogonal
change of range basis. By \eqref{eq:opt-refreshed-graph}, \(U^{\mathsf T}U'=S\). A common
orthogonal change \(W\) diagonalizing \(S\) makes \(UW,U'W\) paired principal bases.
Every principal pair has the form \((UO,U'O)\) with \(O\) orthogonal, because \(S\) is positive
definite. The sum \(\sum_kD_k^{\mathsf T}D_k\) does not change under \(O\). So \(t\) is well
defined. Frobenius
invariance under \(W\) gives, exactly,
\begin{equation}
\begin{aligned}
t&=\sup_{\|v\|=1}\|V_v(U'-U)\|_F \\
 &\le \sqrt{7/2}\,\|I-S\|_{\op}+\tau\|M^{-1}\|_{\op}\|S\|_{\op}\\
 &\le \tfrac12\sqrt{7/2}\,(1.17h_P)^2+1.17\tau .
\end{aligned}
 \label{eq:opt-power-stacked}
\end{equation}
Here
\[
\begin{gathered}
 U'-U=U(S-I)+YS\\
 \text{and}\quad \sup_{\|v\|=1}\|V_vU\|_F=\sqrt{7/2}.
\end{gathered}
\]
Also
\[
\begin{gathered}
 \|M^{-1}\|_{\op}\le1/0.86\le1.17, \quad \|Y\|_{\op}\le1.17h_P\\
 \text{and}\quad \|S\|_{\op}\le1.
\end{gathered}
\]
Finally \(\|I-S\|_{\op}\le\|Y\|_{\op}^2/2\), by \(1-(1+x)^{-1/2}\le x/2\).
Since \(m,d+dq=O(n^2)\), the stated \(K_A,K_P\) make \eqref{eq:opt-power-stacked} at most
\(t_*/2\). The spare half permits finite range evaluation. All three failure probabilities sum
to \(6e^{-s}<\delta/16\). For \(B=1\) Algorithm~\ref{alg:paired-quartic} uses the whole eight-space.
\end{proof}
\begin{alem}[Stability of the block commutant]
\label{lem:opt-commutant-stability}
Let $P'$ be a rank-$q$ pair-space projector. Let $(A_k)$ and $(A'_k)$ be
orthonormal bases of $\mathcal E$ and $\operatorname{ran}P'$ whose stacked
error $t$ in~\eqref{eq:opt-stacked-error} obeys $t\le1/40$.
Define $\mathscr C'$ by~\eqref{eq:opt-commutant} using $A'_k$.
Its $B$ lowest eigenvectors span $\widehat{\mathcal K}_{\rm com}$, with projector
$\widehat\Pi$. Write $\Pi$ for the projector onto
$\mathcal K_{\rm com}=\spn\{P_b/\sqrt8:1\le b\le B\}$ and put
\[
 b=2t,\qquad s=\sqrt7-b,\qquad e=8\sqrt{7/2}\,t+4t^2 .
\]
The first $B$ eigenvalues of $\mathscr C'$ are at most $b^2$ and
the remaining ones are at least $s^2$. Moreover $\widehat\Pi I_d=I_d$ and
\begin{equation}
\begin{aligned}
 \|\mathscr C'-\mathscr C\|_{F\to F}&\le b(2\sqrt8+b),\\
 \|\mathscr C'-\mathscr C\|_{\op\to\op}&\le e,\\
 \|\widehat\Pi-\Pi\|_{F\to F}&\le b/s,\\
 \|\widehat\Pi-\Pi\|_{\op\to\op}
 &\le\frac{3(5/6)^2e}{1-(5/6)e}.
\end{aligned}
\label{eq:opt-commutant-stability}
\end{equation}
\end{alem}
\begin{proof}
\proofdisplays
\noindent\emph{The true commutant.}
Let \(P_b\) be the mutually orthogonal rank-eight projectors onto the true blocks. Define two
maps on symmetric matrices:
\[
{\cal D}S=\sum_bP_bSP_b,\qquad
  \Pi S=\tfrac18\sum_b\Tr(P_bS)P_b .
\]
The range of \(\Pi\) is \(\mathcal K_{\rm com}\), so \(\Pi=\Pi_{\mathcal K_{\rm com}}\).
They are orthogonal projectors for the Frobenius inner product and contractions for the matrix
operator norm. The contour integral below needs complex symmetric \(S\). For such \(S\), block pinching retains only the diagonal blocks,
and \(|\Tr(P_bS)|\le8\|S\|_{\op}\), so the latter contraction bounds still hold.
Within one block let
\[
 E_{ij}=\ketbra ij,\ A_{ij}=(E_{ij}-E_{ji})/\sqrt2, \quad i<j.
\]
Using
\(E_{ij}SE_{kl}=S_{jk}E_{il}\),
\[
\begin{gathered}
A_{ij}^2=-\tfrac12(E_{ii}+E_{jj}),\\
2A_{ij}SA_{ij}=S_{ji}E_{ij}-S_{jj}E_{ii}-S_{ii}E_{jj}+S_{ij}E_{ji},\\
\sum_{i<j}A_{ij}^2=-7P_b/2,\\
 2\sum_{i<j}A_{ij}SA_{ij}=P_bSP_b-\Tr(P_bS)P_b.
\end{gathered}
\]
Expanding
\[
 -[A,[A,S]]=-A^2S+2ASA-SA^2
\]
and summing yields the exact splitting
\begin{equation}
\mathscr C=7(I-{\cal D})+8({\cal D}-\Pi).
 \label{eq:opt-exact-commutant}
\end{equation}
The nullspace is \(\operatorname{ran}\Pi\). If \(\mathscr T S=([A_k,S])_k\), then \(\mathscr
T^*\mathscr T=\mathscr C\) and \(\|\mathscr T\|_{F\to F}=\sqrt8\).

\noindent\emph{The commutator perturbation.}
Write $D_k=A'_k-A_k$ and $E S=([D_k,S])_k$, so $\mathscr T'=\mathscr T+E$.
Skew symmetry gives $\sum_kD_kD_k^*=\sum_kD_k^*D_k$, of norm $t^2$.
Left and right multiplication therefore satisfy
\[
 \sum_k\|D_k S\|_F^2,\ \sum_k\|S D_k\|_F^2\le t^2\|S\|_F^2.
\]
The triangle inequality gives $\|E\|_{F\to F}\le b$. Expanding the square,
\begin{equation}
\begin{aligned}
 \mathscr C'-\mathscr C&=\mathscr T^*E+E^*\mathscr T+E^*E,\\
 \|\mathscr C'-\mathscr C\|_{F\to F}&\le b(2\sqrt8+b).
\end{aligned}
\label{eq:opt-commutant-expansion}
\end{equation}

\noindent\emph{Matrix operator norm.}
Set $Vv=(A_kv)_k$, $V'v=(A'_kv)_k$ and $G=V^*V$.
Then $\|V\|=\sqrt{7/2}$, $\|V'-V\|=t$ and
\[
 \mathscr C S=GS+SG-2V^*(I_q\otimes S)V .
\]
Expanding about $V$ gives, for every $Z$,
\[
 \|V'^*ZV'-V^*ZV\|\le(2\|V\|t+t^2)\|Z\|.
\]
Use $Z=I$ for $G$ and $Z=I_q\otimes S$ for the compression. This gives,
also for complex symmetric $S$,
\[
 \|\mathscr C'-\mathscr C\|_{\op\to\op}\le e .
 \tag{\ref*{eq:opt-commutant-expansion}a}\label{eq:opt-commutant-matrix-norm}
\]

\noindent\emph{The singular gap.}
By~\eqref{eq:opt-exact-commutant}, $\mathscr T$ has $B$ zero singular
values and every other singular value is at least $\sqrt7$. The singular-value
min--max principle bounds the two clusters of $\mathscr T'$ by $b$ from
above and $s$ from below. Squaring gives the eigenvalue bounds.
For $v\in\mathcal K_{\rm com}$, the images of the two right singular subspaces
under $\mathscr T'$ are orthogonal, hence
\[
 s\|(I-\widehat\Pi)v\|_F\le\|\mathscr T'v\|_F
 =\|Ev\|_F\le b\|v\|_F .
\]
The two equal-rank orthogonal projectors have the same largest
principal-angle sine in either direction. Thus
\[
 \|\widehat\Pi-\Pi\|_{F\to F}
 =\|(I-\widehat\Pi)\Pi\|_{F\to F}\le b/s .
 \tag{\ref*{eq:opt-commutant-expansion}b}\label{eq:opt-low-projector-frobenius}
\]
Every commutator annihilates $I_d$, so $I_d$ lies in the measured low space.

\noindent\emph{One resolvent on the spectral-norm space.}
Express the resolvent of \eqref{eq:opt-exact-commutant} in the contractions \(I,{\cal D},\Pi\).
For \(|z|=3\), collecting equal coefficients gives
\[
\begin{gathered}
\begin{aligned}
 &R(z)=(z-\mathscr C)^{-1}\\
 &\quad =\frac{I}{z-7}
  +\left(\frac1{z-8}-\frac1{z-7}\right){\cal D}
  +\left(\frac1z-\frac1{z-8}\right)\Pi,
\end{aligned}\\[3pt]
 \|R(z)\|_{\op\to\op}\le\tfrac14+\tfrac1{4\cdot5}+\tfrac8{3\cdot5} =\tfrac56 .
\end{gathered}
\]
For \(t\le1/40\), \(b^2<3<s^2\) and \((5/6)e<1\). The resolvent identity and the
Neumann series therefore apply on this circle. Its length is \(6\pi\), so the low-projector
difference satisfies
\[
\begin{gathered}
\widehat\Pi-\Pi =\frac1{2\pi i}\oint_{|z|=3}(z-\mathscr C')^{-1}
           (\mathscr C'-\mathscr C)R(z)\,dz,\\
\|\widehat\Pi-\Pi\|_{\op\to\op}
 \le\frac{3(5/6)^2e}{1-(5/6)e}.
\end{gathered}
\]
At $t\le1/40$ the four bounds in~\eqref{eq:opt-commutant-stability}
are at most $11.42t$, $15.1t$, $0.771t$ and $46t$, respectively.
The spectral guards follow from $b^2\le1/400$ and $s^2>6.7$.
For the pencils below we use the last two bounds in the form $0.9t$ and $46t$.
\end{proof}
\subsection{Measured-kernel filtering of Gaussian pencils}
\begin{alem}[A polynomial list of accurate block candidates]
\label{lem:opt-candidates}
For $B\ge2$, the two range banks and the Gaussian pencils of Algorithm~\ref{alg:paired-quartic} return exactly one eight-dimensional candidate for each true block, with probability at least $1-\delta/8$. On this event a candidate with orthonormal basis $U$ and true block projector $P_b$ satisfies
\[
 \|UU^{\mathsf T}-P_b\|_{\op}<0.01,\quad
 \|(I-P_b)U\|_F<0.0041 .
\]
The list takes polynomial classical work. For $B=1$ the sole candidate is the entire Majorana space.
\end{alem}
\begin{proof}
\proofdisplays
Suppose $B\ge2$. Condition on both banks on the range event of Lemma~\ref{lem:opt-range-refresh}. The true block projectors $P_c$ are mutually orthogonal, have rank eight, and sum to $I$. In the space of symmetric matrices with its Frobenius inner product, let $Se_c=P_c/\sqrt8$ and $\Pi=SS^*$. Thus $S^*S=I_B$. Lemma~\ref{lem:opt-commutant-stability} gives a measured projector $\widehat\Pi$ of rank $B$, fixing $I$, with
\[
\begin{gathered}
 \|\widehat\Pi-\Pi\|_{F\to F}\le\beta:=0.9t,\\
 \|\widehat\Pi-\Pi\|_{\op\to\op}\le\kappa:=46t,\\
 t\le t_*:=2.4\cdot10^{-7}.
\end{gathered}
\]

\noindent\emph{What a passing pencil certifies.}
Let $P=UU^{\mathsf T}$ have rank eight. Write $\alpha_c=\Tr(PP_c)/8$ and $\sigma=\|(I-\Pi)P\|_F/\sqrt8$. Since $\sum_c\alpha_c=1$, $\alpha_c\ge0$, and $\Pi P=\sum_c\alpha_c P_c$,
\begin{equation}
 \sigma^2=1-\sum_c\alpha_c^2\ge 1-\max_c\alpha_c .
 \label{eq:opt-candidate-residual}
\end{equation}
For $b$ maximizing $\alpha_c$, expand the two traces and apply~\eqref{eq:opt-candidate-residual} to obtain
\begin{equation}
\begin{gathered}
 \|(I-P_b)U\|_F^2=8(1-\alpha_b)\le8\sigma^2,\\
 \|P-P_b\|_F^2=16(1-\alpha_b)\le16\sigma^2.
\end{gathered}
\label{eq:opt-kernel-filter}
\end{equation}
For a pencil $Z$ let $g(Z)=\lambda_8(Z)-\lambda_9(Z)$ in descending eigenvalue order, and take $P$ from its top eight-space. With $\sigma_0=5863t_*$ the algorithm tests
\begin{equation}
\begin{gathered}
 g(Z)\ge 1/32,\qquad \widehat\sigma(P)\le\sigma_0,\\
 \widehat\sigma(P)=\|(I-\widehat\Pi)P\|_F/\sqrt8 .
\end{gathered}
\label{eq:opt-pencil-tests}
\end{equation}
Because $\|P\|_F=\sqrt8$, the two residuals differ by at most $\beta$. With outward error at most $\varepsilon=10^{-6}$, every accepted candidate therefore has $\sigma\le\sigma_0+\beta+2\varepsilon<0.001410$. Equation~\eqref{eq:opt-kernel-filter} gives the claimed leakage and the stronger distance $\|P-P_b\|_F<0.0057$.

\noindent\emph{Why every block can be selected.}
For a standard Gaussian $G$ in symmetric Frobenius space, the ideal pencil is $Z=\widehat\Pi G$. Set $J=\widehat\Pi S$, $K=J^*J$. Then
\[
 I-K=[(I-\widehat\Pi)S]^*[(I-\widehat\Pi)S]
 \preceq\beta^2 I.
\]
Thus $J$ has rank $B$ and $JK^{-1}J^*=\widehat\Pi$. Take independent centered Gaussians
\[
 \xi\sim N(0,I_B),\qquad
 \zeta\sim N(0,K^{-1}-I_B).
\]
Put $\mu=\xi+\zeta$. The covariance of $J\mu$ is $JK^{-1}J^*$, so $J\mu$ has exactly the law of $Z$.

Fix a block $b$ and put $u=\sqrt{2\log B}$, $v=4u/3$. Require $v\le\xi_b\le v+1$ and $|\xi_c|\le u$ for $c\ne b$. For a standard normal $W$, shifting the tail integral gives $\Prob(|W|>u)\le e^{-u^2/2}=1/B$. Independence therefore bounds the probability of these restrictions below by
\[
 \frac{e^{-(v+1)^2/2}}{\sqrt{2\pi}}(1-1/B)^{B-1}
 \ge \frac{e^{-11/2}}{\sqrt{2\pi}\,B^2}.
\]
Here $(v+1)^2\le9v^2/8+9$ and $(1-1/B)^{B-1}\ge e^{-1}$. Each coordinate of $\zeta$ has variance at most $s^2=\beta^2/(1-\beta^2)$. Since $s\le1/48$, the same tail estimate and a union bound give
\[
 \Prob(\|\zeta\|_\infty>u/24)\le B e^{-2u^2}=B^{-3}.
\]
For $s=0$ the event is empty. The two restrictions on $\xi$ and $\zeta$ are independent. Their joint probability is at least $(1-B^{-3})e^{-11/2}/(\sqrt{2\pi}B^2)\ge q_B/2$, with $q_B$ as in~\eqref{eq:opt-probe-count}. On this event, for $c\ne b$,
\[
 \ell:=u/4\le\mu_b-\mu_c\le 4u=:L.
\]
Indeed the perturbation of each difference is at most $u/12$, and $u>1$.

\noindent\emph{The separated block passes both tests.}
Work with $Z=J\mu$. Since $\widehat\Pi I=I$, centering at $\mu_b$ gives
\[
 \begin{gathered}
 N_0=\sum_c(\mu_b-\mu_c)P_c,\\
 N=\mu_b I-\sqrt8 Z=\widehat\Pi N_0 .
 \end{gathered}
\]
The kernel of $N_0$ is precisely $\operatorname{ran} P_b$, with gap at least $\ell$ above zero. As $\Pi N_0=N_0$, the perturbation has norm at most $\eta_Z:=\kappa L$. Write $P$ for the bottom eight-projector of $N$, equivalently the top eight-projector of $Z$. Weyl's inequality gives $g(Z)\ge(\ell-2\eta_Z)/\sqrt8$. On $\operatorname{ran}(I-P)$ the eigenvalues of $N$ are at least $\ell-\eta_Z$. Applying this to $(I-P)NP_b=(I-P)(N-N_0)P_b$ gives
\[
 \|P-P_b\|_{\op}=\|(I-P)P_b\|_{\op}
 \le \frac{\eta_Z}{\ell-\eta_Z}=:\omega .
\]
The equality uses equal ranks. The difference has rank at most sixteen, hence $\|P-P_b\|_F\le4\omega$. As $\Pi P_b=P_b$,
\[
 \widehat\sigma(P)\le\sqrt2\,\omega+\beta,\qquad
 \omega=\frac{16\kappa}{1-16\kappa}.
\]
At $t\le t_*$ one has $\omega\le 740t$ and $g(Z)\ge u(1-32\kappa)/(4\sqrt8)>1/16$. In particular $\widehat\sigma(P)\le1100t$, which leaves room for rounding in both tests.

On the finite-sampling coupling event the pencil error is at most $h=10^{-10}$ in operator norm. The estimates in \emph{Finite computations and guards} below give gap loss at most $2h$ and normalized residual change below $3\cdot10^{-9}$ (from the top-projector bound $4h/(g(Z)-h)$ and the measured-kernel approximation). Thus
\[
 \tfrac1{16}-2h>\tfrac1{32}+\varepsilon,\qquad
 1100t_*+3\cdot10^{-9}<\sigma_0-\varepsilon .
\]
Every favorable trial passes~\eqref{eq:opt-pencil-tests}. It represents the specified $b$: its ideal projector is within $4\omega$ of $P_b$, and its rounded projector moves by at most $4h/(g(Z)-h)$ in Frobenius norm. These two errors sum to less than $0.0057$, whereas distinct true projectors are distance four apart.

\noindent\emph{Coverage and stopping.}
Given the banks, the ideal pencils are independent. For the prescribed cap
\begin{equation}
\begin{gathered}
 T=\left\lceil 2q_B^{-1}\log(32B/\delta)\right\rceil,\qquad
 q_B=\frac{e^{-20}}{\sqrt{2\pi}\,B^2},
\end{gathered}
 \label{eq:opt-probe-count}
\end{equation}
the conditional probability of missing any block is at most $B(1-q_B/2)^T\le\delta/32$. Two passing projectors assigned to the same block are within $2(0.0057)$ of each other. Those assigned to different blocks are at least $4-2(0.0057)$ apart. The duplicate threshold $1/2$ separates them even with distance error $1/100$. Thus the first representative of every block is kept, and stopping at the $B$th insertion gives the same list as the full scan on this event. Later stages use independent random streams.

The range, coverage and finite-coupling failures total at most $\delta/16+\delta/32+\delta/64<\delta/8$. The cap~\eqref{eq:opt-probe-count} is polynomial in $B$ and $\log(1/\delta)$, and the finite implementation uses polynomial work per pencil.
\end{proof}
\subsection{Uniform correction noise and global initialization}
\begin{alem}[Joint noise control on data-selected orthogonal blocks]
\label{lem:opt-joint-noise}
For one bank of $K$ shots per quartet, with probability at least
$1-\delta/8$, simultaneously for every orthogonal frame
$Q=[q_1,\ldots,q_B]$ and every $O_b\in O(8)$,
\begin{equation}
\begin{aligned}
\max\biggl\{
 &\biggl(\sum_{b=1}^B\|n_b\|_2^2\biggr)^{1/2},\\
 &\biggl(\sum_{b=1}^B\|N_b(O_b)\|_F^2\biggr)^{1/2}
 \biggr\}\\
&\qquad\le C\sqrt{\frac{n^2+\log(1024/\delta)}K},
\end{aligned}
\label{eq:opt-joint-noise}
\end{equation}
where $n_b$ and $N_b(O_b)$ are the noise terms defined after
\eqref{eq:opt-block-update}, with $\hat O_b$ replaced by $O_b$.
They are formed from the error tensor $Y_\xi=Y(\xi)$ of the bank,
$\xi=y-x_*$.
The same estimate holds for every partial orthogonal block frame.
\end{alem}
\begin{proof}
\proofdisplays
\noindent\emph{Uniform projection onto the measured coordinates.}
Recall the alternating embedding \(J_4\), the native quartet space \(\mathcal W\), and
\(H(\xi)\) from the proof of Lemma~\ref{lem:opt-refinement}. For a fixed orthogonal \(Q\) and
a unit \(z\in\mathcal W\),
\[
\begin{gathered}
\langle z,(\Lambda^4Q)^{\mathsf T}\xi\rangle
 =\langle J_4\xi,Q^{\otimes4}J_4z\rangle,\\
 \sup_{Q\in O(d)}\|\Pi_{\mathcal W}(\Lambda^4Q)^{\mathsf T}\xi\|\le H(\xi).
\end{gathered}
\]
For the bank \(W\), equation~\eqref{eq:opt-tensor-concentration} applies with \(K=K_W\) and
\(s=\log(1024/\delta)\). It covers every frame at once.

\noindent\emph{Local and transverse entries of the same table.}
Complete any partial block frame to \(Q\in O(d)\), with its \(b\)-th eight-space in the
\(b\)-th native position. Write \(\widetilde\xi=(\Lambda^4Q)^{\mathsf T}\xi\) and let \(I_b\)
be the eight native axes of block \(b\). Let \(J_b\) be the \(d\times56\) matrix with one column for each sorted
local triple \(T\). Its column \(T\) and its norm are
\[
\begin{gathered}
(J_b)_{\cdot,T}=(I-q_bq_b^{\mathsf T})Y_\xi(q_{b,t_1},q_{b,t_2},q_{b,t_3},\cdot),\\
 \|J_b\|_F^2=\sum_{\substack{T\subset I_b,\;|T|=3\\ s\notin I_b}}
                   |\widetilde\xi_{\mathrm{sort}(T\cup\{s\})}|^2 .
\end{gathered}
\]
The second line holds because the columns of \(Q\) outside block \(b\) are an
orthonormal basis of \(\operatorname{ran}(I-q_bq_b^{\mathsf T})\).
Sorting contributes a sign before squaring. All entries counted for \(J_b\) have exactly three
axes in block \(b\), and all entries counted for \(n_b\) have four. A quartet cannot contain
three axes in two disjoint blocks. Thus
\[
\sum_{b\ {\rm found}}\big(\|n_b\|^2+\|J_b\|_F^2\big)
 \le\|\Pi_{\mathcal W}\widetilde\xi\|^2\le H(\xi)^2 .
\]
For an arbitrary \(O_b\in O(8)\), equations~\eqref{eq:opt-cofactor}--\eqref{eq:opt-cofactor-bounds} give
\[
\begin{gathered}
N_b(O_b)=J_b C(O_b)/7,\quad
  C(O_b)^{\mathsf T}C(O_b)=7I_8,\\
\sum_{b\ {\rm found}}\|N_b(O_b)\|_F^2\le\tfrac17\sum_{b\ {\rm found}}\|J_b\|_F^2 .
\end{gathered}
\]
Now $d=2n$, $p\le7d^2$, and $2e^{-s}<\delta/8$.
On the single event $\mathcal E=\{H(\xi)\le b_\xi\}$, with $b_\xi$ the threshold
in~\eqref{eq:opt-tensor-concentration}, the preceding identities give
\[
\begin{aligned}
\sup_{Q,\{O_b\}}\max\Bigl\{
 &\bigl(\sum_b\|n_b(Q)\|^2\bigr)^{1/2},\\
 &\bigl(\sum_b\|N_b(Q,O_b)\|_F^2\bigr)^{1/2}
 \Bigr\}\le b_\xi.
\end{aligned}
\]
This proves~\eqref{eq:opt-joint-noise}. Evaluating the same supremum at
$Q=Q(\xi)$ and $O_b=O_b(\xi)$ covers every admissible data-selected
frame and local rotation. Partial frames are completed orthogonally as
above. Thus the event covers every correction using this bank, without
conditioning on its selected frame or taking a union bound over corrections.
\end{proof}
\begin{alem}[Initialization without a supplied approximation]
\label{lem:opt-initialization}
Three nonadaptive banks $A,P,W$, each with
$\lceil C(n^2+\log(1024/\delta))\rceil$ copies per quartet, and internal
random bits produce in polynomial time an orthogonal frame $R_{\rm c}$.
Except with probability at most $\delta/2$ over the banks and the bits,
$R_{\rm c}\in SO(d)$ and $R_*=R_{\rm c}e^{X_*}g$ with $X_*\in\mathfrak m$,
$\|X_*\|_F\le\rho_0/2$ and $g\in G_0$. The frame is a function of these
banks and bits only.
\end{alem}
\begin{proof}
\proofdisplays
\noindent\emph{Block subspaces and oriented frames.}
For \(B>1\), work on the joint range and coverage event of Lemma~\ref{lem:opt-candidates}.
It gives one orthonormal candidate \(U\) for each true block. Denote the true
eight-frame of its matched block by \(R_{[c]}\) and its projector by
\(P_c^\star=R_{[c]}R_{[c]}^{\mathsf T}\). By~\eqref{eq:opt-kernel-filter} and the bound \(\sigma<0.001410\) in the
proof of Lemma~\ref{lem:opt-candidates},
\[
 e_U:=\|(I-P_c^\star)U\|_F\le\sqrt8\,\sigma<0.0041
\]
for every retained
candidate on the joint event. Put \(\tau=0.0041\).
For \(B=1\) the candidate is the whole eight-space and \(e_U=0\).
Choose the universal multiplier in \(K_W\) so that the right side of
\eqref{eq:opt-joint-noise} is at most \(\nu_*\). On that event
\(\max(\nu_{\rm loc},\nu_{\rm off})\le\nu_*\) simultaneously at every partial
frame and local orientation used with bank \(W\).
Let \(J\) be the ordered list of matched true blocks at a given stage,
\(R_J=[R_b]_{b\in J}\), \(P_b^\star=R_bR_b^{\mathsf T}\), and
\[
 \mathcal H_J=\{\operatorname{diag}_{b\in J}(h_b):h_b\in O(8),\
(\Lambda^4h_b)\varphi=\varphi\}.
\]
For an orthonormal frame \(Q=[q_b]_{b\in J}\) and any matrix \(A\) of the same size, set
\[
\begin{gathered}
E_J(Q)^2=\sum_{b\in J}\|(I-P_b^\star)q_b\|_F^2,\\
 \Delta_J(A)=\min_{H\in\mathcal H_J}\|A-R_JH\|_F .
\end{gathered}
\]
The first quantity measures subspace leakage. The second measures distance to a true
frame up to the symmetries of the paired block. These are quantities for the proof.
The local stabilizer is closed in \(O(8)\), so the minimum is attained. For orthonormal
eight-frames \(q,R\), either squared
leakage equals \(8-\|q^{\mathsf T}R\|_F^2\), so \(E_J\) agrees with the leakage in
Lemma~\ref{lem:rec-decoder}. Projecting \(q_b-R_bh_b\) onto \(\operatorname{ran}(I-P_b^\star)\)
gives \(E_J(Q)\le\Delta_J(Q)\) after summing squares and minimizing.
For \(D\in\mathcal H_J\), right multiplication gives
\[
 \|AD-R_JH\|_F=\|A-R_JHD^{\mathsf T}\|_F,
\]
hence \(\Delta_J(AD)=\Delta_J(A)\).
The triangle inequality gives
\(|\Delta_J(A)-\Delta_J(A')|\le\|A-A'\|_F\).
Let \(u\) bound the leakage entering a correction and \(v\) the full-frame error
after it. Write \(\eta_I,\eta_C\) for the whole-frame insertion and correction
allowances. Let \(a_3,a_\ell,a_o\) be the coefficients in~\eqref{eq:opt-correction-error},
and let \(r_{\rm ch}\) be the admissible frame radius
in Lemma~\ref{lem:opt-geometry}. Put \(S_{\rm c}(t)=a_3t^3+(a_\ell+a_o)\nu_*\) for \(t\ge0\).
Take \(u,\nu_*\) within the leakage and noise bounds of Lemma~\ref{lem:rec-decoder}.
The proof requires the positive parameters to satisfy
\[
\begin{gathered}
 v+\tau<1/2,\quad
  \sqrt{2v^2+\tau^2}+\eta_I\le u,\quad
  S_{\rm c}(u)<1/4,\\
 2S_{\rm c}(u)+\eta_C\le v\le u,\quad
  v\le r_{\rm ch},\quad
  \frac{v}{\sqrt{1-v^2/4}}\le\rho_0/2 .
\end{gathered}
\]
The values in Table~\ref{tab:opt-init-constants} satisfy these conditions and the
discovery, local correction and storage bounds used here. The insertion condition
can be checked exactly by verifying \(u>\eta_I\) and \(2v^2+\tau^2<(u-\eta_I)^2\).

\noindent\emph{Insertion controls leakage.}
Suppose \(E=E_J(Q)\le v\), initially zero. For a new true block \(c\notin J\), write
\(R=R_{[c]}\), \(P_f=QQ^{\mathsf T}\), and \(P_U=UU^{\mathsf T}\).
Since \(P_c^\star\preceq I-P_b^\star\) for \(b\in J\),
\[
\begin{aligned}
 \|P_fR\|_F^2
 &=\sum_{b\in J}\|R^{\mathsf T}q_b\|_F^2 \le E^2,\\
 \|Q^{\mathsf T}U\|_{\op}
 &=\|P_fU\|_{\op}\\
 &\le\|P_fR\|_F\|R^{\mathsf T}U\|_{\op}+e_U \le E+e_U<1/2 .
\end{aligned}
\]
The second chain uses \(P_c^\star=RR^{\mathsf T}\) and the split
\(U=P_c^\star U+(I-P_c^\star)U\), with \(\|R^{\mathsf T}U\|_{\op}\le1\).
Thus \(V=(I-P_f)U\) satisfies
\[
 \sigma_{\min}(V)^2=1-\|Q^{\mathsf T}U\|_{\op}^2>3/4,
\]
and the insertion passes.
Let \(P_+\) project onto \(\operatorname{ran}V\).
The enlarged projector \(P_f+P_+\) fixes \(U\), since \(U=P_fU+V\).
Its complement \(P_\perp=I-P_f-P_+\) satisfies \(P_\perp\preceq I-P_U\) and
\(P_\perp P_f=0\). Pythagoras for \(I-P_+=P_f+P_\perp\) gives
\[
\begin{aligned}
\|(I-P_+)R\|_F^2
 &=\|P_fR\|_F^2+\|P_\perp R\|_F^2 \\
 &\le E^2+\|(I-P_U)R\|_F^2 =E^2+e_U^2,\\
E_{\rm ins}^2
 &=E^2+\|(I-P_+)R\|_F^2
 \le2E^2+e_U^2 .
\end{aligned}
\]
Here \(E_{\rm ins}\) is the leakage of the exact enlarged frame \([Q,q_+]\),
with \(q_+\) any orthonormal basis of \(\operatorname{ran}V\). Its new term is
\[
 \|(I-P_c^\star)q_+\|_F^2=8-\|R^{\mathsf T}q_+\|_F^2=\|(I-P_+)R\|_F^2.
\]
The map \(Q\mapsto[(I-P_b^\star)q_b]_{b\in J}\) is a Frobenius contraction for any
fixed list, including the enlarged list. Storing the orthonormal insertion within
\(\eta_I\) of its exact frame therefore gives leakage at most
\(\sqrt{2v^2+\tau^2}+\eta_I\le u\).

\noindent\emph{Each correction controls the full frame.}
Use \(J\) now for the enlarged list, \(j=|J|\), and evaluate the exact
update~\eqref{eq:opt-block-update} at the actual stored frame with \(E=E_J(Q)\le u\).
Since each local leakage is at most \(E\) and each local noise norm is at most
\(\nu_*\), Lemma~\ref{lem:rec-decoder} applies. Its exact completed columns
\(X=[X_b]_{b\in J}\) satisfy
\[
\begin{gathered}
\Delta_J(X)\le S_{\rm c}(E)\le S_{\rm c}(u)<1/4,\\
  \sigma_{\min}(X)\ge1-S_{\rm c}(u)>3/4 .
\end{gathered}
\]
The absolute Hodge score and the local eigengap exceed one, and the squared transverse operator norms
are below \(1/2\), by~\eqref{eq:opt-local-tangent-error} and~\eqref{eq:opt-off-column-error}.
Let \(A\) be the computed columns. Up to a right stabilizer in each block, allow
Frobenius error \(h_\ell\) per block and error \(h_P\) for computing and storing
the polar factor of \(A\). Choose these allowances so that
\(\sqrt B\,h_\ell<1/4\) and \(2\sqrt B\,h_\ell+h_P\le\eta_C\).
Then some \(D\in\mathcal H_J\)
satisfies \(\|A-XD\|_F^2\le jh_\ell^2\). Right invariance and the distance inequality give
\[
\begin{gathered}
\Delta_J(A)\le S_{\rm c}(E)+\sqrt j\,h_\ell,\\
 \sigma_{\min}(A)\ge1-S_{\rm c}(u)-\sqrt B\,h_\ell>1/2 .
\end{gathered}
\]
The polar factor \(P=\operatorname{polar}(A)\) is a nearest orthonormal frame
by the argument of~\eqref{eq:opt-polar-error}. Each \(R_JH\) is orthonormal, hence
\(\|P-A\|_F\le\Delta_J(A)\). The stored orthonormal output \(Q^+\) then satisfies
\begin{equation}
\begin{aligned}
E_J(Q^+)
 &\le\Delta_J(Q^+) \\
 &\le\|Q^+-P\|_F+\|P-A\|_F+\Delta_J(A)\\
 &\le h_P+2S_{\rm c}(E)+2\sqrt j\,h_\ell\\
 &\le 2S_{\rm c}(u)+\eta_C \le v .
\end{aligned}
 \label{eq:opt-recursion}
\end{equation}
One correction already gives~\eqref{eq:opt-recursion}. Because \(v\le u\),
the bound persists through the two further corrections of
Algorithm~\ref{alg:paired-quartic} and supplies the hypothesis for the next
insertion. The uniform event for bank \(W\) covers every reuse.

\noindent\emph{Entry into the chart.}
At the final correction, \(\Delta_J(R_{\rm c})\le v\) gives
\(\|R_{\rm c}-R_*g_0\|_F\le v\) for a permutation of the true eight-axis blocks
followed by local stabilizers. Their product \(g_0\) fixes the sum of the native
four-forms \(x_0\) in Corollary~\ref{cor:global-parent} and is proper.
Each block swap has even parity, and an improper local map exchanges the two Hodge
sectors and cannot fix the nonzero self-dual \(\varphi\).
An orthogonal matrix of determinant \(-1\) has eigenvalue \(-1\). Thus opposite
components of \(O(d)\) have operator distance at least two, and \(R_{\rm c}\in SO(d)\).
Put \(V=R_{\rm c}^{\mathsf T}R_*g_0\in SO(d)\). Then
\[
 \|V-I\|_F=\|R_*g_0-R_{\rm c}\|_F\le v\le r_{\rm ch}.
\]
Lemma~\ref{lem:opt-geometry} gives \(R_{\rm c}^{\mathsf T}R_*g_0=e^{X_*}e^Z\) with
\(e^Z\in G_0\). Since \(s/\sqrt{1-s^2/4}\) increases on \([0,r_{\rm ch}]\),
the coordinates required by Lemma~\ref{lem:opt-refinement} satisfy
\[
\begin{gathered}
R_*=R_{\rm c}e^{X_*}g,\qquad
  g=e^Zg_0^{-1}\in G_0,\\
 X_*\in\mathfrak m,\qquad
  \|X_*\|_F\le\frac{v}{\sqrt{1-v^2/4}}\le\rho_0/2 .
\end{gathered}
\]
The list event of Lemma~\ref{lem:opt-candidates} includes finite pencil
sampling and fails with probability at most \(\delta/8\). The bank-\(W\)
event fails with probability at most \(2e^{-s}=\delta/512\). The
deterministic numerical allowances are justified below. The combined failure
is at most \(\delta/4\le\delta/2\). The three banks each have \(\lceil C(n^2+\log(1024/\delta))\rceil\) copies per
quartet. Capped trials, insertion, cofactor arithmetic and polar computations all have
polynomial cost. Every operation in initialization precedes opening the independently
collected bank \(R\).
\end{proof}
\begin{table}[H]
\centering
\setlength{\belowcaptionskip}{4pt}
\caption{One sufficient choice of constants.}\label{tab:opt-init-constants}
\setlength{\tabcolsep}{3pt}
\begin{tabular}{@{}p{.32\columnwidth}p{.27\columnwidth}p{.34\columnwidth}@{}}
\toprule
Role & Parameters & Exact values \\
\midrule
Input leakage and output frame error & \(u,\ v\) & \(2^{-6},\ 2^{-12}\) \\
Candidate leakage and bank noise & \(\tau,\ \nu_*\) & \(0.0041,\ 5\cdot10^{-5}\) \\
Insertion and correction storage & \(\eta_I,\ \eta_C\) & \(10^{-6},\ 10^{-8}\) \\
Block and polar error allocation & \(\sqrt B\,h_\ell,\ h_P\) & \(\eta_C/100,\ \eta_C/2\) \\
\addlinespace
Correction coefficients (Lemma~\ref{lem:rec-decoder}) &
  \(a_3,\ a_\ell,\ a_o\) & \(3,\ 1/2,\ 8/\sqrt{63}\) \\
Chart bounds (Lemmas~\ref{lem:opt-geometry} and~\ref{lem:opt-refinement}) &
  \(r_{\rm ch},\ \rho_0\) & \(1/16,\ 2^{-8}\) \\
\bottomrule
\end{tabular}
\end{table}

\noindent\emph{Proof of the upper bound in Theorem~\ref{thm:reconstruction}.}
Precollect the three initialization banks and an independent refinement bank.
For $s=\log(1024/\delta)$ prescribe
\begin{equation}
\begin{gathered}
K_A=K_P=K_W=\lceil C(n^2+s)\rceil,\\
 K_R=\left\lceil C(n^2+s)/\epsilon^2\right\rceil .
\end{gathered}
\label{eq:opt-bank-budgets}
\end{equation}
Parts I and II supply \(R_{\rm c}\) with failure probability at most \(\delta/2\) by
Lemma~\ref{lem:opt-initialization}. In Part III let \(x_{\rm c}(X)=\wedge^4(R_{\rm c})F(X)\).
Since \(\|x_{\rm c}(X)\|^2=14B\), the score fitted by Algorithm~\ref{alg:paired-quartic} obeys
\[
 \tfrac12\|y_R-x_{\rm c}(X)\|^2
 =\tfrac12(\|y_R\|^2+14B)-\langle y_R,x_{\rm c}(X)\rangle.
\]
Condition on a successful initialization transcript. The independent bank \(R\) retains its original law.
Here
\[
 K_R\ge2C(n^2+s)\ge C(D+\log(8/\delta)),
\]
because \(\epsilon^2<1/2\),
\(D\le2n^2\) and \(s\ge\log(8/\delta)\). So the bank hypothesis of
Lemma~\ref{lem:opt-refinement} holds, and its bound holds a fortiori with the
larger \(s\) used here. Let \(\widehat X\) be the exact maximizer on the chart ball.
Lemma~\ref{lem:opt-refinement} and~\eqref{eq:paired-certificate-proof} give, for the
corresponding pure state \(\rho_{\widehat X}\),
\[
 \trd(\rho,\rho_{\widehat X})^2
 \le\frac{\|F(X_*)-F(\widehat X)\|^2}{28}
 \le\frac{C'(D+s)}{28K_R}\le\epsilon^2/4
\]
except with conditional probability at most \(\delta/2\). The universally chosen \(C\)
in~\eqref{eq:opt-bank-budgets} suffices since \(D\le2n^2\). Optimization and circuit rounding
together cost at most \(\epsilon/2\) in trace distance, as shown below.
Let \(\hat R\) be the product of the rounded output plane rotations and
\(\hat\rho=U_{\hat R}P_0U_{\hat R}^\dagger\) the returned state. The triangle
inequality gives \(\trd(\rho,\hat\rho)\le\epsilon\) with total failure at most \(\delta\).
Counting every predetermined observation yields
\begin{equation}
\begin{aligned}
N_{\rm copies}
 &=r(K_A+K_P+K_W+K_R) \\
 &\le C''\,\frac{n^4(n^2+\log(1/\delta))}{\epsilon^2}.
\end{aligned}
\label{eq:opt-total-copies}
\end{equation}
Here \(r\le2n^4/3\). Since \(\epsilon^2<1/2\), the three initialization banks
together need at most \(3K_R/2\) copies per quartet, up to rounding. For
\(\delta\le1/8\),
\[
 s=\log1024+\log(1/\delta)\le\tfrac{13}{3}\log(1/\delta).
\]
Hence \(C''=9C\) suffices, with \(C\) the constant of
\eqref{eq:opt-bank-budgets}.

\noindent\emph{Finite computations and guards.}
Reading the records, computing the projectors and list, correcting frames
and constrained optimization take polynomially many operations.
For \(B>1\), let \(T\) be the pencil cap in~\eqref{eq:opt-probe-count}, and set
\[
\begin{gathered}
p_{\rm G}=d(d+1)/2,\qquad
  b=\sqrt{p_{\rm G}},\\
 s_{\rm G}=\sqrt{2\log(128T/\delta)},\qquad
  M=b+s_{\rm G},\\
 \eta_{\rm num}=10^{-10}.
\end{gathered}
\]
In the symmetric Frobenius-orthonormal basis the \(p_{\rm G}\) coefficients of each
ideal \(G_j\) are independent standard normal variables. To certify the norm comparison,
put \(x=s_{\rm G}-1/2\). Here \(d\ge16\) gives \(b\ge\sqrt{136}>6\). Also \(T\ge1\) and
\(\delta\le1/8\) give \(s_{\rm G}\ge\sqrt{2\log1024}>3.7\), so \(x\ge2\).
Using \(\E e^{t\|G_j\|_F^2}=(1-2t)^{-p_{\rm G}/2}\) at \(2t=1-b^2/(b+x)^2\) gives
\[
\begin{aligned}
&\Prob\{\|G_j\|_F \ge b+x\} \\
&\quad\le\exp\!\left[-x^2/2-bx+b^2\log(1+x/b)\right],\\
&bx-b^2\log(1+x/b) =\int_0^x\frac{bz}{b+z}\,dz \\
&\quad\ge\frac{b x^2}{2(b+x)} \ge\frac{3x}{4}.
\end{aligned}
\]
As \(x^2/2+3x/4\ge s_{\rm G}^2/2\), the probability that any ideal norm exceeds
\(M-1/2\) is at most \(\delta/128\).
Choose a dyadic mesh with
\[
 \eta_{\rm num}/(4\sqrt{p_{\rm G}})\le h_{\rm G}\le\eta_{\rm num}/(2\sqrt{p_{\rm G}})
\]
on \([-M,M]\), with overflow symbols.
In the algorithm \(M\) is replaced by a dyadic number in \([M,M+1/8]\).
The rounding below describes the ideal Gaussian coefficient in the coupling;
the algorithm draws the mesh point from the approximate cell law. Round each coefficient to its nearest mesh point.
Discard overflow or a norm above \(M+h_{\rm G}\sqrt{p_{\rm G}}/2\), advancing the
pencil counter. On the norm event all coefficients are interior to the mesh and
rounding costs at most \(h_{\rm G}\sqrt{p_{\rm G}}/2\) in Frobenius norm. Thus the
comparison has margin at least \(1/2\), and interval error below \(1/4\) suffices.
Approximate each unconditioned rounded cell law within total variation
\(\delta/(128Tp_{\rm G})\), including overflow. Coupling the entire scan costs at most
\(\delta/128\), giving the \(\delta/64\) allowance in Lemma~\ref{lem:opt-candidates}.
The coupling compares the two discrete cell laws.
On the coupled good event, compute the projector in the symmetric basis with
\[
 \|\widetilde\Pi-\hat\Pi\|_{F\to F}\le\alpha_{\rm G}
\le\eta_{\rm num}/[4(M+1)].
\]
Then
\[
\begin{gathered}
 \|\widetilde G-G\|_F\le h_{\rm G}\sqrt{p_{\rm G}}/2,\\[3pt]
\begin{aligned}
 &\|\widetilde\Pi\widetilde G-\hat\Pi G\|_{\op}\\
 &\quad \le h_{\rm G}\sqrt{p_{\rm G}}/2
      +\alpha_{\rm G}(M+h_{\rm G}\sqrt{p_{\rm G}}/2)
 <\eta_{\rm num}.
\end{aligned}
\end{gathered}
\]
There are \(O(1+M/h_{\rm G})\) cells per coefficient. Integrating the Taylor polynomial of
\(e^{-x^2/2}\) on each cell in \([-M,M]\) gives its Gaussian mass with a factorial remainder
bounded by \(e^{M^2/2}(M^2/2)^{j+1}/(j+1)!\) times the interval length. A degree polynomial in
\(M^2\) and the required log precision suffices. The overflow probability is the complement of
the finite cell masses. Compute each normalized cell mass to error at most
\(\epsilon_{\rm c}=\delta/(512\,Tp_{\rm G}L)\), where \(L\) is the number of
cells including overflow. Clipping negative approximations and normalizing
changes the total variation by at most \(2L\epsilon_{\rm c}\). For \(k=\lceil\log_2(1/\epsilon_{\rm c})\rceil\), round cumulative
probabilities to the nearest multiple of \(2^{-k}\). Endpoints stay fixed,
so the added total variation is at most \(L2^{-k}/2\). A uniform
\(k\)-bit integer selects a cell on this grid. The total is below \(\delta/(128\,Tp_{\rm G})\).
This gives a discrete sampler with the stated error in polynomial bit work.
Certify the rank-$q$ pilot and refreshed ranges and the $B$-dimensional
low commutant space using their separated spectra. The good-event margins are
\[
\begin{gathered}
\sigma_{\min}(Q(y_P)U_A)>0.86>\tfrac12,\\
 \lambda_B(\mathscr C')\le0.29<1<6<6.7\le\lambda_{B+1}(\mathscr C').
\end{gathered}
\]
To see how the gaps give a computable \(\widetilde\Pi\), let \(P_A=U_AU_A^{\mathsf T}\) and
form \(Q(y_P)P_AQ(y_P)\). Its nonzero spectrum is bounded below by \(0.86^2\); its support
projector is \(P'\). For the laboratory skew basis \(S_a\), \(1\le a\le m\), expand
\(A'_k=\sum_a U'_{ak}S_a\) in~\eqref{eq:opt-commutant}. For any pair-space error \(D_P\), the
resulting change is linear:
\[
\begin{gathered}
\mathscr C(P'+D_P)-\mathscr C(P')
 =-\sum_{a,b}(D_P)_{ab}\operatorname{ad}_{S_a}\operatorname{ad}_{S_b},\\
 \|\mathscr C(P'+D_P)-\mathscr C(P')\|_{F\to F}\le 2m\|D_P\|_F .
\end{gathered}
\]
Indeed \(\operatorname{ad}_S(X)=[S,X]\), \(\|\operatorname{ad}_{S_a}\|_{F\to F}\le
2\|S_a\|_{\op}=\sqrt2\), and \(\sum_{ab}|(D_P)_{ab}|\le m\|D_P\|_F\). At the pilot threshold,
the retained eigenvalue magnitudes are at least \(7/8\) and the others are at most \(1/8\).
Thus residual-certified symmetric eigensolves at successively inverse-polynomial tolerances
control \(\widetilde\Pi\) to \(\alpha_{\rm G}\), including the absolute-magnitude pilot
cutoff. The rational bank means have polynomial bit length. Only the gap to the
excluded subspace enters the projector estimate.
For each pencil certify~\eqref{eq:opt-pencil-tests} to outward error
$<10^{-6}$; discard ambiguities and increment the counter. The favorable gap \(>1/16\) and
residual margin \((5863-1100)t_*>10^{-3}\) suffice.
For a favorable gap $g\ge1/16$, the rounded pencil changes its top-eight
projector by at most $4\eta_{\rm num}/(g-\eta_{\rm num})$ in Frobenius norm~\cite{davis1970}.
Its normalized residual changes by at most
\[
 \sqrt2\eta_{\rm num}/(g-\eta_{\rm num})+\alpha_{\rm G}<3\cdot10^{-9},
\]
and Weyl's inequality changes the gap by at most $2\eta_{\rm num}$.
Evaluate the candidate residual on the stored orthonormal basis. A projector change
from \(P_1\) to \(P_2\) changes the normalized measured residual by at most
\[
 \|(I-\hat\Pi)(P_1-P_2)\|_F/\sqrt8\le\|P_1-P_2\|_F/\sqrt8.
\]
The gap of a passing pencil permits this contribution and residual evaluation together
to fit within the stated outward tolerance.
Deduplicate at the threshold $1/2$ of Algorithm~\ref{alg:paired-quartic},
with distances computed to error $<1/100$. Frobenius distances are $<0.012$
within blocks and $>3.98$ between blocks; require exactly $B$ candidates.
Insertion requires $\|Q_{\rm found}^{\mathsf T}U\|_{\op}\le1/2$, as in
Algorithm~\ref{alg:paired-quartic}. This is the same as
$\sigma_{\min}((I-P_{\rm found})U)\ge\sqrt3/2$.
Each correction requires an absolute Hodge score and a top local eigengap greater than one,
$\|X_b^\perp\|_{\op}^2\le1/2$ for every block, and
$\sigma_{\min}([X_1,\ldots,X_j])\ge1/2$.
The insertion and correction estimates give strict good-event margins.
A backward residual divided by the external gap controls each retained
projector~\cite{davis1970}. For a symmetric \(k\times k\) matrix \(S\), let
\(E=\sum_{i\ne j}|S_{ij}|^2\) and \(b=\max_{i\ne j}|S_{ij}|\), with \(k\ge2\).
A Jacobi rotation reducing \(b\) by half preserves the sum of squares on every
other pair of affected entries. Hence
\[
 E^+\le E-2\big[b^2-(b/2)^2\big]
 \le\bigl(1-3/[2k(k-1)]\bigr)E .
\]
Polynomially many rotations certify the requested residual. An ambiguous test at
the precision cap has its specified discard or fallback.

\noindent\emph{Computing and storing the corrected frame.}
For the \(d\times8\) completed column blocks of Lemma~\ref{lem:opt-initialization},
entrywise error \(h_\ell/\sqrt{8d}=\eta_C/(100d)\) suffices for Frobenius error
\(h_\ell\). Each of the 56 columns of \(V_b\) in
\eqref{eq:opt-block-update} has norm at most \(\sqrt r\), since its pairing with any unit
vector contracts the measured quartic tensor with an exterior product of norm at most one.
Thus \(\|V_b\|_{\op}\le\sqrt{56r}\). By~\eqref{eq:opt-cofactor-bounds} it suffices to
compute each local orientation, up to its stabilizer, within \(c h_\ell/(1+\sqrt r)\)
for a sufficiently small universal \(c\), and to evaluate the remaining block operations
within a fixed fraction of \(h_\ell\). The local eigengap controls the eigenline.
In estimating the error, choose the sign of the exact eigenvector closest to the
computed one. Either exact sign represents the same ray, so their lifts differ by a
right stabilizer \(h\). Since \(C(Oh)=C(O)h\) and \(S(Ah)=h^{\mathsf T}S(A)h\) for
\(S(A)=(I-A^{\mathsf T}A)^{1/2}\), equation~\eqref{eq:opt-block-update} sends
\(X_b\) to \(X_bh\), as required in~\eqref{eq:opt-recursion}. For \(t=w^{\mathsf T}v_0\),
the coefficient in~\eqref{eq:rec-local-rotation} is \(\arccos(t)/\sqrt{1-t^2}\).
It is smooth across \(t=0\) and has expansion \(1+(1-t)/3+O((1-t)^2)\) at \(t=1\).
Each local square-root argument is at least \(I/2\), and the computed
polar input has smallest singular value greater than \(1/2\) on the initialization event.
The square-root bound~\eqref{eq:opt-completed-columns}
and the inverse identity \(A^{-1}-B^{-1}=A^{-1}(B-A)B^{-1}\) control the remaining
matrix functions on these intervals.
Compute \((I-S_b')^{1/2}\) by its binomial series in \(S_b'\). The
coefficients after the first sum to one in absolute value and
\(\|S_b'\|_{\op}\le1/2\), so the tail past degree \(p\) has operator norm at
most \(2^{-p-1}\).
For a polar input \(A\) satisfying the singular-value guard, scale by
a power of two \(D_s\in[2\sqrt d,4\sqrt d)\). Its completed ideal blocks have
orthonormal columns, hence \(\|A\|_{\op}\le\sqrt d+\sqrt B\,h_\ell\le D_s\).
With certified smallest singular value at least \(1/4\) including column error,
\(B_s=A/D_s\) satisfies
\[
 \|I-B_s^{\mathsf T}B_s\|_{\op}\le q_s=1-1/(16D_s^2).
\]
In \(B_s(B_s^{\mathsf T}B_s)^{-1/2}\) the binomial series past degree \(p\) has
Frobenius norm at most \(\sqrt d\,q_s^{p+1}/(1-q_s)\). Thus \(O(d\log(d/z))\)
terms suffice for tolerance \(z/16\), \(0<z\le1\).
To store an orthonormal \(d\times l\) frame \(Y\), start from a rational approximation
within \(z/16\) in Frobenius norm. Eliminate entries below the diagonal in column
order, pairing the diagonal row with each lower row, including a zero lower entry.
This takes at most \(G=d(d-1)/2\) plane rotations. Choose a power of two
\(z/[32(1+G)]\le\gamma\le z/[16(1+G)]\).
For dyadic \(t\in[-1,1]\), use either common sign on the rational pair
\[
\begin{gathered}
(c_t,s_t)=\big((1-t^2)/(1+t^2),\,2t/(1+t^2)\big),\\
  c_t^2+s_t^2=1 .
\end{gathered}
\]
The angle \(2\arctan t\) has derivative at most two, and each active pair of entries
has norm at most two. A grid of spacing at most \(\gamma/8\) therefore contains a
rotation with nonnegative upper pivot and absolute lower residual at most \(\gamma\).
Check these inequalities by rational comparison and set the lower residual to zero.
Deletion only decreases the affected column norm. All accumulated rotations form
an exactly orthogonal matrix \(U\). If the resulting array is \([H;0]\), with \(H\)
upper triangular, then
\[
\begin{gathered}
e:=\|UY-[H;0]\|_F\le z/16+G\gamma\le z/8,\\
 \|H-H^{-\mathsf T}\|_F\le\frac{\|H^{\mathsf T}H-I\|_F}{1-e}
                              \le\frac{2e+e^2}{1-e}.
\end{gathered}
\]
The first inequality sums the deleted residuals under isometries. The second uses
\(\sigma_{\min}(H)\ge1-e\). Put \(D_H=\operatorname{diag}(\operatorname{sgn}H_{ii})\).
By triangularity and \(|b-\operatorname{sgn}b|\le|b-1/b|\) for \(b\ne0\),
\(\|H-D_H\|_F\le\|H-H^{-\mathsf T}\|_F\). The exactly orthonormal
\(U^{\mathsf T}[D_H;0]\) is therefore within \(3e/(1-e)<z\) of \(Y\).
For a proper square target the first \(d-1\) pivots are positive. The last is positive
as well, since \(e<1\) preserves the determinant sign along the segment from \(UY\)
to \(H\). Thus \(D_H=I\) and at most \(G\) rotations suffice. At insertion perform
the encoding in the stored orthogonal complement. Re-encode the full polar frame at
each correction. Each grid has polynomial size in \(d,1/z\), and the product of the
rational denominators has \(O(G\log((1+G)/z))\) bits. An exhausted precision guard
uses the prescribed discard or fallback.

\noindent\emph{Refinement and output precision.}
Put
\[
 k=\lceil10\log(16/\epsilon^2)\rceil \quad\text{and}\quad \eta_X=\epsilon/128.
\]
By~\eqref{eq:opt-maximizer-displacement}, the chart error from the finite iteration is
at most \(\rho_0\lambda^k+\eta_X/(1-\lambda)\), where \(\lambda=13/18\).
Since \(\lambda^{10}\le e^{-10(1-\lambda)}<e^{-1}\), \(\rho_0\lambda^k<\epsilon/16\).
The derivative bound and the parent certificate turn a chart displacement of size
\(z\) into trace error at most \(\kappa z\), with \(\kappa=4/\sqrt{28}\).
For any unit-Frobenius skew direction, differentiating the \(j\)-th exponential term
inserts its generator in \(j\) positions. The generator has operator norm at most four,
\(\|d\Lambda^4(X)\|_{\op}\le4\rho_0<1/2\), and \(\|y_R\|\|x_0\|\le\sqrt{14Br}\).
Thus the gradient tail after degree \(p\) is bounded in dual Frobenius norm by
\[
 4\sqrt{14Br}\sum_{j\ge p}(4\rho_0)^j/j!\le4e\sqrt{14Br}\,2^{-p},
\]
where
\(e\) here is Euler's number. Choose the degree to make this at most \(18\eta_X/8\).
Extend the score by the same exponential formula to all skew matrices. Its ambient
Frobenius gradient has entries \((\nabla_{\rm skew}\phi)_{ij}
=D\phi(X)[E_{ij}-E_{ji}]/2\) for \(i<j\), with \(E_{ij}\) the matrix units.
Rounding each skew gradient entry within \(18\eta_X/(8\sqrt{2m})\) costs another
\(18\eta_X/8\), where \(m=\binom d2\). After the \(1/18\) step, allow
\(\eta_X/4\) for rounding to a fixed skew grid before horizontal projection.
This projection keeps the crossing entries and uses \((I-Q_\omega)/4\) in each
internal pair space by~\eqref{eq:rec-local-spectrum}. For the resulting \(Z\ne0\),
round \(\|Z\|_F\) upward on the largest power-of-two mesh of spacing at most
\(\eta_X/4\), obtaining
\(s_Z\in[\|Z\|_F,\|Z\|_F+\eta_X/4]\), and store
\(Z\min(1,\rho_0/s_Z)\). Store zero at \(Z=0\).
The stored point lies in the ball and differs from exact radial projection
by at most \(\eta_X/4\). These allowances sum to less than \(\eta_X\) per step.
The fixed skew grid and the dyadic rescaling bound the stored bit lengths afresh
at each iteration.
Let \(Y=R_{\rm c}e^{\widetilde X_k}\), and compute a proper orbital factorization
\(Y'\) with \(\|Y'-Y\|_F\le\epsilon/100\) before rounding the output angles.
The bounded exponential series supplies the encoding input to Frobenius tolerance
\(\epsilon/1600\).
For \(Z=\log(Y^{\mathsf T}Y')\), the angle--chord inequality
\(|\theta|\le(\pi/2)|e^{i\theta}-1|\) on \([-\pi,\pi]\) gives
\(\|Z\|_F\le\pi\epsilon/200\). Integrating \(\Lambda^4(e^{tZ})LZ\) and using
\(\|LZ\|\le4\|Z\|_F\) from Lemma~\ref{lem:opt-geometry} bounds this trace error by
\(\kappa\pi\epsilon/200\). Use at most \(G=d(d-1)/2\) plane rotations
\cite[Section~5.1]{golub2013}, choosing positive pivots in the proper factorization.
For an angular error \(h\) in a consistent spin lift, the relative gate is
\(\exp(h\gamma_p\gamma_q/2)\), whose operator distance to \(I\) is
\(2|\sin(h/4)|\le|h|/2\). Telescoping the gate products bounds the ket error,
and hence the trace error, by \(G h_\theta/2\) when \(|h|\le h_\theta=\epsilon/(2G)\).
The total computational contribution to trace distance is bounded by
\[
\begin{aligned}
&\kappa\!\left(\rho_0\lambda^k+\frac{\eta_X}{1-\lambda}+\frac{\pi\epsilon}{200}\right)
 +\frac{G h_\theta}{2}\\
&\quad\le \epsilon\!\left(\frac1{16}+\frac{18}{5\cdot128}+\frac1{50}+\frac14\right)
 =\frac{577}{1600}\epsilon<\frac{\epsilon}{2}.
\end{aligned}
\]
This proves polynomial implementability, not an optimal bit or gate cost.
\subsection{A matching lower bound for fixed-basis quartic readout}
\begin{alem}[Necessary copies, including adaptive quartet choices]
\label{lem:opt-lower}
Fix $0<\epsilon_{\max}<1/\sqrt2$. Let $n=4B$ with $B\ge1$,
$0<\epsilon\le \epsilon_{\max}$ and $0<\delta\le1/8$.
A learner measures one raw fixed-basis quartic sign $B_I$ on each
fresh copy and uses at most $N_{\rm copies}$ copies.
Each quartet may be chosen from the complete earlier transcript and
from private randomness.
The estimate $\hat\rho$ may be any density operator specified by the
transcript; success means $1-\bra\psi\hat\rho\ket\psi\le \epsilon^2$.
Suppose the learner succeeds with probability at least $1-\delta$ for
every $\ketbra\psi\psi\in\mathcal O$. Then
\begin{equation}
 N_{\rm copies}\ge c_{\epsilon_{\max}}\,
 \frac{n^4(n^2+\log(1/\delta))}{\epsilon^2},
 \label{eq:opt-lower}
\end{equation}
where $c_{\epsilon_{\max}}>0$ depends only on $\epsilon_{\max}$.
Classical computation is unrestricted.
\end{alem}
\begin{proof}
In this proof $L$, $H$, $F$ and $g$ denote the constant, rotation,
fidelity and group number defined below, not the earlier objects.

\noindent\emph{One real axis per block.}
Put $d=2n=8B$ and choose one Majorana coordinate $e_j$ in each block.
By Corollary~\ref{cor:global-parent} the seed quartic is
$x_0=\sum_{j=1}^B\phi_j$, with $\phi_j$ supported in block $j$.
Rotate only the selected axes by $O\in SO(B)$. Let $x(O)$ be the quartic
tensor of $U_O\ket{\Omega_B}$, with $U_O$ the Gaussian lift. Each term
of $\phi_j$ contains at most one selected axis.
Write $\phi_j=e_j\wedge\zeta_j+\phi_j^*$, where neither $\zeta_j$ nor
$\phi_j^*$ contains $e_j$. Put $x_*=\sum_j\phi_j^*$ and
$\widehat O=O\oplus I_{d-B}$; $U_O$ is the lift of $\widehat O$.
Lemma~\ref{lem:closure} gives $x(O)=(\wedge^4\widehat O)x_0$, and
$\widehat Oe_j=\sum_iO_{ij}e_i$. Hence
\begin{equation}
\begin{gathered}
x(O)=x_*+\sum_{i,j=1}^B O_{ij}\,e_i\wedge\zeta_j,\\
 \|x(O)\|_1 \le LB(1+\|O\|_{1\to1}),\\
 \|x(O)-x(O')\|_1 \le L|O-O'|_{\Sigma}.
\end{gathered}
 \label{eq:opt-linear}
\end{equation}
Here $|M|_\Sigma=\sum_{ij}|M_{ij}|$, and $\|M\|_{1\to1}$ is the maximum
absolute column sum. Equation~\eqref{eq:block-expansion} gives
$\|\zeta_j\|_1\le\|\phi_j\|_1\le L:=\sqrt{14\binom84}$ by Cauchy--Schwarz
in the exterior coordinate basis.
The forms $x_*$ and $e_i\wedge\zeta_j$, over all pairs $(i,j)$, have
pairwise disjoint coordinate supports. Therefore
$\|x(O)\|_1=\|x_*\|_1+\sum_{i,j}|O_{ij}|\,\|\zeta_j\|_1$,
with $\|x_*\|_1\le LB$ and $\|\zeta_j\|_1\le L$. The difference
$x(O)-x(O')$ is bounded in the same way.
Apply a common rotation $H$ to obtain $\ket{\psi_O}=U_HU_O\ket{\Omega_B}$.
Take
\[
 H_{st}=\sqrt{2/(d+1)}\sin(\pi st/(d+1)), \quad 1\le s,t\le d,
\]
negating one row if needed. Finite sine orthogonality gives $H\in SO(d)$.
Each four-entry row has norm at most $\sqrt{8/d}$, so Hadamard's inequality
gives $|\det H[I,J]|\le64/d^2$.
Put $\mu_I(O)=\bra{\psi_O}B_I\ket{\psi_O}$, the laboratory bias. Up to a
sign, $U_HU_O$ is the lift of $H\widehat O$, so
Lemma~\ref{lem:closure} gives $\mu(O)=(\Lambda^4H)x(O)$, with matrix
entries $\det H[I,J]$. Hence
\begin{equation}
\begin{gathered}
 \mu_I(O)=[(\Lambda^4H)x(O)]_I,\quad |\mu_I(O)|\le\frac{64}{d^2}\|x(O)\|_1,\\
 |\mu_I(O)-\mu_I(O')|\le\frac{64L}{d^2}|O-O'|_\Sigma .
\end{gathered}
\label{eq:opt-flat}
\end{equation}
For the rotation-plane angles $\theta_j$ of $O^{\mathsf T}O'$, the common
$H$ cancels and
\begin{equation}
\begin{gathered}
|\braket{\psi_O}{\psi_{O'}}|
 =\prod_j|\cos(\theta_j/2)|
 \le e^{-\|O-O'\|_F^2/16},\\
 \|O-O'\|_F^2=8\sum_j\sin^2(\theta_j/2).
\end{gathered}
\label{eq:opt-overlap}
\end{equation}
Indeed each plane lifts to
\[
 \cos(\theta_j/2)I\pm
\sin(\theta_j/2)\gamma(v_j)\gamma(w_j)
\]
for orthonormal plane axes
$v_j,w_j$, where $\gamma(v)=\sum_i v_i\gamma_i$ in Majorana coordinates.
The planes are orthogonal, so the plane lifts commute and
$U_O^\dagger U_{O'}$ is their product up to a sign. Each plane contributes
\[
 2(1-\cos\theta_j)^2+2\sin^2\theta_j=8\sin^2(\theta_j/2)
\]
to
$\|I-O^{\mathsf T}O'\|_F^2=\|O-O'\|_F^2$.
For any $k>0$ distinct axes $z_1,\ldots,z_k$ from these planes, put
$V_{i\ell}=\langle e_i,z_\ell\rangle$, where $e_1,\ldots,e_B$ are the
selected axes, one in each block. Orthogonality makes their Clifford
product alternating, so
\[
 \gamma(z_1)\cdots\gamma(z_k)
 =\sum_{i_1<\cdots<i_k}\det[V_{i_p\ell}]_{p,\ell=1}^k
   \,\gamma(e_{i_1})\cdots\gamma(e_{i_k}).
\]
Every summand is odd in a touched block and has zero seed expectation
by block parity. Only the product of cosines survives. Use $1-z\le e^{-z}$
on its square.

\noindent\emph{A sign cube of separated states.}
In $g$ disjoint groups of $s\ge2$ selected axes, $gs\le B$, form a skew
matrix $A_u$ with entries $au_{ij}$ above each group diagonal,
$u_{ij}\in\{-1,1\}$, and all other independent entries zero. Put
\begin{equation}
\begin{gathered}
d_0=g\binom s2,\quad
  \kappa=\frac{128}{1-2\epsilon_{\max}^2},
 \quad
  a^2=\frac{\kappa\epsilon^2}{d_0},\\
\begin{aligned}
 O_u&=(I+A_u)(I-A_u)^{-1}\\
 &=2(I-A_u)^{-1}-I .
\end{aligned}
\end{gathered}
\label{eq:opt-hard-cube}
\end{equation}
Skew symmetry and the path $tA_u$, $0\le t\le1$, give $O_u\in SO(B)$.
Suppose $a(s-1)\le1/4$. The absolute row and column sums of $A_u$ are
at most $1/4$, hence those of $R_u=(I-A_u)^{-1}$ are at most $4/3$ by
the Neumann series. Thus $\|O_u\|_{1\to1}\le(5/4)(4/3)=5/3$.
Writing $h=h(u,v)$ for Hamming distance and $u\sim v$ for $h=1$ gives
\begin{equation}
\begin{gathered}
O_u-O_v=2R_u(A_u-A_v)R_v,\\
|O_u-O_v|_\Sigma\le 2(4/3)^2\,4a\qquad
 (u\sim v),\\
\|O_u-O_v\|_F\ge 2(4/5)^2\|A_u-A_v\|_F
                 \ge\sqrt{8a^2h}.
\end{gathered}
 \label{eq:opt-cayley}
\end{equation}
The second line uses the column bound on $R_u$ and row bound on $R_v$.
For the third, multiply the first line by $I-A_u$ and $I-A_v$ on its
two sides. Their operator norms are at most $5/4$, and each flipped
bit contributes $8a^2$ to the squared Frobenius norm of $A_u-A_v$.
Let $P_u=\ketbra{\psi_{O_u}}{\psi_{O_u}}$ and choose $\hat u$ maximizing
$\Tr(\hat\rho P_v)$ over all vertices. On success,
\begin{equation}
\begin{aligned}
 2(1-\epsilon^2)
 &\le\Tr[\hat\rho(P_u+P_{\hat u})] \le 1+|\braket{\psi_{O_u}}{\psi_{O_{\hat u}}}|\\
 &\le 1+e^{-a^2h(u,\hat u)/2}.
\end{aligned}
\label{eq:opt-hamming-decoder}
\end{equation}
The middle bound follows from the largest eigenvalue of the two-ket Gram
matrix. Hence
\[
 h/d_0\le-2\log(1-2\epsilon^2)/(\kappa\epsilon^2)
\le1/32,
\]
since $-\log(1-z)\le z/(1-z)$. For uniform $u$, including
all measurement randomness, $\E h/d_0\le1/32+\delta\le5/32$.

\noindent\emph{Information in one quartet.}
Take $g=\lceil32\kappa\rceil$ and, for $B\ge4g$, $s=\lfloor B/g\rfloor$.
Then $d_0=\Theta_{\epsilon_{\max}}(n^2)$ and
$a(s-1)\le\epsilon\sqrt{2\kappa/g}\le1/4$.
Indeed $s\ge\tfrac34B/g$ and $s-1\ge\tfrac34s$ give
$B^2/(5g)\le d_0\le B^2/(2g)$, and
\[
 a^2(s-1)^2=2\kappa\epsilon^2(s-1)/(gs)\le2\kappa\epsilon^2/g\le1/16.
\]
Equations~\eqref{eq:opt-linear} and \eqref{eq:opt-flat} with
$\|O_u\|_{1\to1}\le5/3$ give $|\mu_I|\le8L/(3B)<1/2$ since $L<32$ and
$g\ge4096$. Every edge has bias difference $O(a/d^2)$ uniformly in $I$.
Write $p_{u,I}(\pm)=(1\pm\mu_I(O_u))/2$. For $N=N_{\rm copies}$
and neighbors $u,v$,
\begin{equation}
\begin{aligned}
D_{\rm KL}(p_{u,I}\|p_{v,I})
 &\le\chi^2(p_{u,I}\|p_{v,I}) \\
 &=\frac{(\mu_I(O_u)-\mu_I(O_v))^2}{1-\mu_I(O_v)^2},\\
D_{\rm KL}(\Prob_u\|\Prob_v)
 &=\sum_{t=1}^N\E_u D_{\rm KL}(p_{u,I_t}\|p_{v,I_t}) \\
 &\le CNa^2/d^4 .
\end{aligned}
 \label{eq:opt-transcript}
\end{equation}
Here $C=\tfrac43(8192L/9)^2$: by~\eqref{eq:opt-flat}
and~\eqref{eq:opt-cayley} an edge changes each bias by at most
$(64L/d^2)(128a/9)$, and $1-\mu_I(O_v)^2\ge3/4$.
Further, $\Prob_u$ is the law of the private seed, chosen quartets and signs.
At any common history the choice rule is the same under both targets, and
the next sign comes from a fresh copy. The seed law cancels, and early
stops can be padded with unused measurements. For each bit, pair vertices
differing only in that bit. The sum of their two
testing errors is at least $1-\TV(\Prob_u,\Prob_v)$. Averaging over pairs
and bits, then applying Pinsker's inequality, gives
\begin{equation}
\begin{aligned}
 \frac{\E h}{d_0}&\ge\frac{1-\max_{u\sim v}\TV(\Prob_u,\Prob_v)}2\\
 &\ge\frac{1-\sqrt{CNa^2/(2d^4)}}2 .
\end{aligned}
 \label{eq:opt-assouad}
\end{equation}
Comparison with $5/32$ forces $N=\Omega_{\epsilon_{\max}}(n^6/\epsilon^2)$.

\noindent\emph{Confidence and small dimensions.}
In two selected axes compare $I$ with a rotation $O_\theta$ such that
$\cos(\theta/2)=(1-2\epsilon^2)^2$, $0<\theta<\pi$.
By~\eqref{eq:opt-overlap} their overlap is strictly below
$1-2\epsilon^2$. The projector inequality in~\eqref{eq:opt-hamming-decoder}
therefore makes the acceptable sets disjoint.
Here $\|O_\theta\|_{1\to1}=|\cos\theta|+|\sin\theta|\le\sqrt2$, so
\eqref{eq:opt-linear} and~\eqref{eq:opt-flat} give
$|\mu_I|\le(1+\sqrt2)L/B<1/2$ for both states when $B\ge4g$. Also
$|O_\theta-I|_\Sigma\le4\theta=O(\epsilon)$ since
\[
 \theta^2/(2\pi^2)\le1-\cos(\theta/2)=4\epsilon^2(1-\epsilon^2).
\]
By~\eqref{eq:opt-flat} and the chain rule in~\eqref{eq:opt-transcript},
the transcript divergence is at most $C'N\epsilon^2/d^4$ for $B\ge4g$,
with $C'=\tfrac43(512\sqrt2\,\pi L)^2$.
The learner induces a binary test with both errors at most $\delta$.
Let $p\ge1-\delta$ and $q\le\delta$ be the probabilities, under the
first and the second state, that the estimate is acceptable for the
first state. The binary divergence increases in $p$ and decreases in
$q$ for $q\le p$.
Data processing to that decision therefore gives divergence at least
\[
 (1-2\delta)\log((1-\delta)/\delta)\ge\tfrac23\log(1/\delta),
\]
hence
$N=\Omega(d^4\log(1/\delta)/\epsilon^2)$.
For every $B\ge1$, let $P$ be the seed projector and let $Q$ be its
image under $e^{iu(2a_1^\dagger a_1-I)}$. The first mode has occupation
probability $1/2$, so the overlap is $\cos u$. Set $f=\epsilon^2$.
Choose any $0<c=\cos u<1-2f$. For every density operator $\sigma$,
\[
 \Tr[\sigma(P+Q)]\le\lambda_{\max}(P+Q)=1+c<2(1-f).
\]
The acceptable reports for $P$ and $Q$ are disjoint. Classify a report
as $P$ precisely when it is acceptable for $P$. Both error probabilities
are at most $\delta$. A procedure with a maximum of $N$ copies defines
a test on $N$ copies by ignoring unused ones. The squared overlap is
$c^{2N}$, so the minimum sum of testing errors is
$1-\sqrt{1-c^{2N}}$~\cite{nielsen2010}. Consequently
$c^{2N}\le4\delta(1-\delta)$. Taking $c\uparrow1-2f$ gives
\begin{equation}
 N\ge\frac{\log[1/(4\delta(1-\delta))]}{-2\log(1-2\epsilon^2)} .
 \label{eq:lower-confidence}
\end{equation}
Here $-2\log(1-2\epsilon^2)\le4\epsilon^2/(1-2\epsilon_{\max}^2)$
and $\log[1/(4\delta)]\ge\frac13\log(1/\delta)$ for $\delta\le1/8$.
This covers the omitted
dimensions $B<4g$ since $n$ is then bounded in terms of $\epsilon_{\max}$ alone
and $\log(1/\delta)\ge\log 8$.
Indeed $n<16g$ there, and
\[
 n^2+\log(1/\delta)\le(n^2/\log8+1)\log(1/\delta).
\]
For $B\ge4g$ use $\max(x,y)\ge(x+y)/2$ on the cube bound and the
two-point bound.
Taking the larger of the dimension and confidence bounds
proves~\eqref{eq:opt-lower}.
The cube is used only for $B\ge4g\ge2^{14}$. For smaller $B$ the
constant absorbs the dimension, so $c_{\epsilon_{\max}}$ is an order
constant, not a practical one.
\end{proof}

\noindent\emph{Completion of Theorem~\ref{thm:reconstruction}.}
The nonadaptive upper bound~\eqref{eq:opt-total-copies} matches the
adaptive lower bound~\eqref{eq:opt-lower} in $n,\epsilon,\delta$.
Since nonadaptive learners are adaptive learners, both minimax budgets
have this order. Constants need not be uniform as
$\epsilon_{\max}\uparrow1/\sqrt2$.
The lower bound uses infidelity, equal to squared trace distance for
the constructed pure output; it does not cover arbitrary mixed estimates
judged only by trace distance $\epsilon$.
A laboratory frame $G\in SO(d)$ that is given, known in advance and
not random is absorbed into the hard family: replace $H$ by
$G^{\mathsf T}H$. A frame drawn at random by the learner, several
frames, state-dependent prerotations, extra outcomes per copy and
collective measurements lie outside this lower-bound model.
\subsection{The necessary budget without a readout restriction}
\begin{alem}[Worst-case necessary copies]
\label{lem:necessary-copies}
Let $n\ge4$, $0<f\le1/4$ and $0<\delta\le1/8$.
Suppose a learner uses at most $N_{\rm copies}$ copies of each unknown
state and returns a density operator $\widehat\rho$ with
$1-\bra\psi\widehat\rho\ket\psi\le f$ with probability at least
$1-\delta$ for every state in the paired orbit on $4B$ modes,
where $B=\lfloor n/4\rfloor$ and the remaining modes are in a fixed vacuum.
Then, for a universal constant $c>0$,
\begin{equation}
 N_{\rm copies}\ge c\,\frac{n^2+\log(1/\delta)}{f}.
 \label{eq:necessary-copies}
\end{equation}
The bound permits arbitrary collective measurements on target copies
and classical computation. It applies to Theorem~\ref{thm:reconstruction}
with $n=4B$ and $f=\epsilon^2$ (hence $0<\epsilon\le1/2$), and to
the target-copy general family in the earlier version.
It does not supply the lower bound for the reference-access model of
Theorem~\ref{thm:active-main}; that proof is in Section~\ref{gen:app:lower}.
\end{alem}
\begin{proof}
Put $B=\lfloor n/4\rfloor$, $\epsilon^2=f$ and
$\epsilon_{\max}=1/2$. Extend the rotations on $4B$ modes by the
identity on the vacuum spectators, preserving every overlap.
In~\eqref{eq:opt-hard-cube} this fixes $\kappa=256$. Set $g=8\kappa$.
For $B\ge2g$, use $g$ groups of $s=\lfloor B/g\rfloor\ge2$ axes.
With $d_0=g\binom s2$ and $a^2=\kappa f/d_0$,
\[
 a^2(s-1)^2=\frac{2\kappa f}{g}\frac{s-1}{s}
 \le\frac f4\le\frac1{16}.
\]
Thus $a(s-1)\le1/4$, and the decoder
in~\eqref{eq:opt-hamming-decoder} gives $\E h/d_0\le5/32$.
For neighbors $u,v$,
\[
 (I-A_u)^{\mathsf T}(I-A_u)=I+A_u^{\mathsf T}A_u
\succeq I,
\]
hence $\|R_u\|_{\op}\le1$. By~\eqref{eq:opt-cayley},
$\|O_u-O_v\|_F\le2\|A_u-A_v\|_F$. Put
$F_{uv}=|\braket{\psi_{O_u}}{\psi_{O_v}}|^2$ and $N=N_{\rm copies}$.
Equation~\eqref{eq:opt-overlap} and the pure-state trace-distance formula give
\[
\begin{aligned}
1-F_{uv}
 &=1-\prod_j(1-\sin^2(\theta_j/2)) \\
 &\le\sum_j\sin^2(\theta_j/2) \\
 &=\tfrac18\|O_u-O_v\|_F^2
 \le4a^2,\\
\trd(P_u^{\otimes N},P_v^{\otimes N})
 &=\sqrt{1-F_{uv}^{\,N}}
 \le\sqrt{N(1-F_{uv})}\\
 &\le2a\sqrt N .
\end{aligned}
\]
Any protocol using at most $N$ copies induces a measurement on all $N$
by ignoring unused copies. Its outcome total variation is at most this
trace distance. The bit-pairing inequality in~\eqref{eq:opt-assouad} therefore
gives $5/32\ge(1-2a\sqrt N)/2$, so $N\ge(11/16)^2/(4a^2)$.
Since $s\ge B/(2g)$ and $B\ge n/8$, we have
\[
 d_0\ge\frac{gs^2}{4}\ge\frac{B^2}{16g}
 \ge\frac{n^2}{1024g}.
\]
Substitution of $a^2=\kappa f/d_0$ yields $N=\Omega(n^2/f)$.
For every $B\ge1$, the two-state test~\eqref{eq:lower-confidence} gives
\[
 \begin{gathered}
 -2\log(1-2f)\le\frac{4f}{1-2f}\le8f,\\
 N\ge\frac{\log(1/\delta)}{24f},
 \end{gathered}
\]
using $\log(1/(4\delta))\ge\tfrac13\log(1/\delta)$.
If $B<2g$, then $n<8g$ and $\log(1/\delta)\ge\log8$,
so this last bound also implies $N=\Omega(n^2/f)$ with a universal
constant. Combining the dimension and confidence bounds proves
\eqref{eq:necessary-copies}. Each paired input block has two particles in
either branch, so the padded family also satisfies the earlier target-copy promise
for $r_\star\ge2$.
\end{proof}
\clearpage
\section{Learning with reference preparations}
\label{app:general-reference}
We prove the upper and lower bounds of Theorem~\ref{thm:active-main}.
The state family is fixed in Section~\ref{gen:sec:setting} and the three
preparation requests in Section~\ref{sec:access}. All reference runs use the same unknown
circuit. All preparations, including rejected candidates and retries,
are counted. We first construct a preparation with explicit cost
$\mathcal C$ in~\eqref{gen:eq:count}. Section~\ref{gen:app:verifier}
then tests the complete output. Section~\ref{gen:app:lower} gives
independent dimension and confidence lower bounds.

\subsection{Frame conventions}
\noindent\emph{Conventions and transport.}
To preserve the readout signs in the reference-access construction, this
appendix uses the reflected Majorana basis $\gref=S_0\gamma$, where
$S_0=\bigoplus_{j=1}^{n}\operatorname{diag}(1,-1)$ and $\gamma$ is the basis
of Appendix~\ref{app:fermions}. If $R_{\rm A}$ denotes the same circuit in
that basis, then
\[
\begin{gathered}
R=S_0R_{\rm A}S_0,\qquad
\Gamma=S_0\Gamma_{\rm A} S_0,\\
U^\dagger\gref U=R\gref,\qquad
R\in SO(2n),\quad S_0^{\mathsf T}S_0=I.
\end{gathered}
\]
Here and in Algorithm~\ref{alg:general-reference}, $R$ and its estimates
are represented in this reflected basis. We choose a Spin lift $U_R$ for
each rotation, suppressing the overall phase of its physical circuit. Conjugation by $S_0$ preserves
all norms, probabilities and physical preparations; it adds no access
requirement. The scalar $\gamma$ in the frame estimates measures covariance error.

\noindent\emph{Reference covariance.}
With $\gref_{2j-1}=a_j+a_j^\dagger$ and $\gref_{2j}=i(a_j-a_j^\dagger)$,
\begin{equation}
\begin{gathered}
 U_R^\dagger\gref_aU_R=\sum_{c=1}^{2n}R_{ac}\gref_c,\qquad
 \Gamma_{ab}=\tfrac i2\langle[\gref_a,\gref_b]\rangle,\\
 J=\begin{pmatrix}0&1\\-1&0\end{pmatrix},\quad J_n=\bigoplus_{j=1}^{n}J,\quad
 \Gamma_x=\bigoplus_{j=1}^{n}(1-2x_j)J .
\end{gathered}
\label{gen:eq:conventions}
\end{equation}
$\Gamma$ is the covariance of a state. It is real and skew. The state
$U_R\ket x$ has covariance $R\Gamma_xR^{\T}$, which is skew and orthogonal.
A $2n\times2n$ matrix $M$ has the $2\times2$ blocks $M_{[jk]}$ of the modes
$j,k$. The blocks $M_{[jj]}$ are its \emph{mode-diagonal} blocks, and
$M_\od$ is $M$ with these blocks set to zero. $J^{(j)}$ is the matrix with
the block $J$ at $[jj]$ and zeros elsewhere. $\skw M=(M-M^{\T})/2$.
$\polar(M)=M(M^{\T}M)^{-1/2}$ is the orthogonal polar factor of an invertible
real matrix $M$, $\Log$ the principal matrix logarithm,
$\mathrm{rot}(\phi)$ the rotation of the plane by the angle $\phi$, and
$n_j=a_j^\dagger a_j$ unless stated otherwise. We write
$\|\cdot\|_F$ and $\|\cdot\|_{\op}$ for the Frobenius and operator norms, and
\begin{equation}
\begin{gathered}
 \infid(\psi,\varphi)=1-|\braket\psi\varphi|^2,\\
 \dist(\hat R,R)=\min_{T\in\SO(2)^{n}}\|\hat R^{\T}RT-I\|_F .
\end{gathered}
\label{gen:eq:distances}
\end{equation}
Here $\SO(2)^n$ is the group of mode phases $\bigoplus_j\mathrm{rot}(\phi_j)$.
Let $\cK$ be the span of the Fock kets corresponding to the patterns of $\Xi$,
and let $\Pi$ be the orthogonal projector onto $\cK$.
A mode phase applied to $\ket\Xi$ changes only the phases of the $c_{bl}$ and
the global phase; the result is called a \emph{canonical input}. It has the
sets and the weights of $\ket\Xi$ and lies in $\cK$. Logarithms are natural.

\subsection{Frames, gauges and rounding}
\label{gen:app:frames}
\begin{alem}[Gauges]
\label{gen:lem:gauge}
Let $\hat R,R\in\SO(2n)$. Let $T_*\in\SO(2)^n$ attain the minimum in
\eqref{gen:eq:distances} and
\begin{equation}
\begin{gathered}
 S=\hat R^{\T}RT_*,\qquad E=S-I,\\
 A=\skw S,\qquad d=\|E\|_F .
\end{gathered}
 \label{gen:eq:gauge}
\end{equation}
\begin{enumerate}[label=(\alph*)]
\item $E=A-\tfrac12E^{\T}E$. Every block $S_{[jj]}$ is symmetric with
nonnegative trace, and $A_{[jj]}=0$.
\item $\|E\|_{\op}\le d/\sqrt2$, $\|A\|_F\le d$ and $\|A\|_{\op}\le d/\sqrt2$.
\item Let the \emph{frame angle} be
\[
\begin{aligned}
 &\vartheta(\hat R,R)=\min\{\|Y\|_F:\ Y\text{ real skew},\ \\
 &\qquad e^{Y}\in\hat R^{\T}R\,\SO(2)^n\}.
\end{aligned}
\]
Then
\begin{equation}
 \vartheta\le\varsigma(d)\,d,\qquad
 \varsigma(d)=\frac{\arcsin(d/\sqrt8)}{d/\sqrt8}\quad(d\le\sqrt8).
 \label{gen:eq:varsigma}
\end{equation}
$\varsigma$ is increasing, $\varsigma(0)=1$ and
$\varsigma(d)\le1+d^2/40$ for $d\le1$.
\item If $\vartheta<\sqrt2\,\pi$, every $Y$ that attains $\vartheta$ has
$Y_{[jj]}=0$ for all $j$.
\end{enumerate}
\end{alem}
\begin{proof}
\noindent\emph{(a).}
$S$ is orthogonal, so $E+E^{\T}+E^{\T}E=0$. For
$T_\phi=T_*\bigoplus_j\mathrm{rot}(\phi_j)$,
\[
 \|\hat R^{\T}RT_\phi-I\|_F^2=4n-2\sum_j\Tr(S_{[jj]}\mathrm{rot}(\phi_j)).
\]
The left side is minimal at $\phi=0$. So every
$\Tr(S_{[jj]}\mathrm{rot}(\phi))$ is maximal at $\phi=0$, which makes
$S_{[jj]}$ symmetric with nonnegative trace.

\noindent\emph{(b).}
$S\in\SO(2n)$, so its eigenvalues are $e^{\pm i\theta_m}$, $m=1,\dots,n$,
$|\theta_m|\le\pi$, in conjugate pairs; the eigenvalue $-1$ has even
multiplicity because $\det S=1$. $E=S-I$ is normal, so
$d^2=2\sum_m|e^{i\theta_m}-1|^2$. Hence
$\|E\|_{\op}\le d/\sqrt2$. $A$ is the skew part of $E$.

\noindent\emph{(c).}
Let $Y$ be a real skew logarithm of $S$ with the eigenvalues
$\pm i\theta_m$, $|\theta_m|\le\pi$ (the principal logarithm if $-1$ is not an
eigenvalue of $S$). It is admissible, because $e^Y=S=\hat R^{\T}RT_*$. Then
$\|Y\|_F^2=2\sum_m\theta_m^2$
and $|e^{i\theta_m}-1|=2\sin(|\theta_m|/2)\le d/\sqrt2$. Since
$\theta/(2\sin(\theta/2))$ is increasing on $(0,\pi]$,
$|\theta_m|\le\varsigma(d)\cdot2\sin(|\theta_m|/2)$. Summing over $m$,
\[
 \|Y\|_F^2\le\varsigma(d)^2\cdot2\sum_m(2\sin(|\theta_m|/2))^2
=\varsigma(d)^2d^2.
\]
The series of $\arcsin$ gives the last claim.

\noindent\emph{(d).}
Let $(Y,T)$ attain $\vartheta$. The eigenvalues of $Y$ come in pairs, so
$\|Y\|_{\op}\le\vartheta/\sqrt2<\pi$, and $Y$ is the principal logarithm of
$e^Y$, which is smooth near $e^Y$. For $D=J^{(j)}$ the matrices
$Y(r)=\Log(e^Ye^{rD})$ are admissible, so $\|Y(r)\|_F^2$ is minimal at $r=0$.
From $e^{-Y}\tfrac{d}{dr}e^{Y(r)}|_0=D$ and
$e^{-Y}\,d(e^Y)[X]=\int_0^1e^{-sY}Xe^{sY}ds$,
\[
 \langle Y,D\rangle=\int_0^1\langle e^{sY}Ye^{-sY},Y'(0)\rangle\,ds
 =\langle Y,Y'(0)\rangle=0 .
\]
So $Y_{2j-1,2j}=0$.
\end{proof}

\begin{alem}[Rounding]
\label{gen:lem:round}
\begin{enumerate}[label=(\alph*)]
\item For Hermitian $X,Y$ and $\mathsf f:\R\to\mathbb C$ with
$|\mathsf f(s)-\mathsf f(s')|\le|s-s'|$:
\[
 \|\mathsf f(X)-\mathsf f(Y)\|_F\le\|X-Y\|_F.
\]
\item Let $B$ be real and skew and $h=-iB$. Then $h$ is Hermitian,
$\polar(I+B)=e^{i\arctan h}\in\SO(2n)$, the map $B\mapsto\polar(I+B)$ is
$1$-Lipschitz in the Frobenius norm, and
\begin{equation}
 \|\skw\polar(I+B)-B\|_F\le\tfrac12\|B\|_{\op}^2\|B\|_F .
 \label{gen:eq:skw-pol}
\end{equation}
\item Let $S\in\SO(2n)$ with $\|S-I\|_{\op}\le\tfrac12$ and $A=\skw S$. Then
\begin{equation}
\begin{gathered}
 \|\polar(I+A)-S\|_F\le\tfrac12\|A\|_{\op}^2\|A\|_F,\\
 \|\Log S-A\|_F\le\tfrac15\|A\|_{\op}^2\|A\|_F .
\end{gathered}
 \label{gen:eq:round}
\end{equation}
\item Let $Q$ be orthogonal and $\|M-Q\|_F=m<1$. Then
\[
 \|\polar(M)-Q\|_F\le-\log(1-m)
\]
and $\det\polar(M)=\det Q$.
\end{enumerate}
\end{alem}
\begin{proof}
\noindent\emph{(a).}
With spectral decompositions $X=\sum_i\lambda_iP_i$, $Y=\sum_j\mu_jQ_j$,
\[
\begin{aligned}
 &\|\mathsf f(X)-\mathsf f(Y)\|_F^2\\
 &\quad=\sum_{i,j}|\mathsf f(\lambda_i)-\mathsf f(\mu_j)|^2\Tr(P_iQ_j)\\
 &\quad\le\sum_{i,j}|\lambda_i-\mu_j|^2\Tr(P_iQ_j)=\|X-Y\|_F^2.
\end{aligned}
\]

\noindent\emph{(b).}
$I+B=I+ih$ and $(I+ih)^\dagger(I+ih)=I+h^2$, so
\[
 \polar(I+B)=(I+ih)(I+h^2)^{-1/2}=e^{i\arctan h}.
\]
It is real and orthogonal,
and its determinant is positive along the path $rB$, $0\le r\le1$. The
function $e^{i\arctan s}$ has a derivative of modulus $1/(1+s^2)\le1$, so (a)
gives the Lipschitz claim. Its skew part is $ih(I+h^2)^{-1/2}$, and
$|s-s(1+s^2)^{-1/2}|\le|s|^3/2$.

\noindent\emph{(c).}
$S$ is normal with eigenvalues $e^{i\theta}$, $|e^{i\theta}-1|\le\tfrac12$,
so $|\theta|<\pi/2$. The Hermitian matrix $h_A=-iA$ has the eigenvalues
$\sin\theta$ on the same eigenvectors, and $\|A\|_{\op}\le\tfrac12$. Hence
$S=e^{i\arcsin h_A}$, $\Log S=i\arcsin h_A$ and
$\polar(I+A)=e^{i\arctan h_A}$. For $|s|\le\tfrac12$,
\begin{equation}
 |\arcsin s-\arctan s|\le\tfrac12|s|^3,\qquad
 |\arcsin s-s|\le\tfrac15|s|^3 .
 \label{gen:eq:scalar}
\end{equation}
The first holds because the derivative $(1-s^2)^{-1/2}-(1+s^2)^{-1}$ of the
difference lies between $0$ and $\tfrac32s^2$. The second holds because
$(\arcsin s-s)/s^3$ has a series with positive coefficients, so it is at most
its value $8(\pi/6-\tfrac12)<\tfrac15$ at $s=\tfrac12$. Apply
\eqref{gen:eq:scalar} to every eigenvalue.

\noindent\emph{(d).}
On invertible matrices the polar factor is differentiable. Let $M=QH$ be a
polar decomposition, $H$ positive definite. From $dM=dQ\,H+Q\,dH$, with the
skew matrix $\Omega=Q^{\T}dQ$ and the symmetric $dH$:
$H\Omega+\Omega H=Q^{\T}dM-dM^{\T}Q$. The solution
\[
 \Omega=\int_0^\infty e^{-sH}(Q^{\T}dM-dM^{\T}Q)e^{-sH}ds
\]
has
$\|\Omega\|_F\le\|dM\|_F/\sigma_{\min}(M)$. On the path $Q+r(M-Q)$,
$0\le r\le1$, the smallest singular value is at least $1-rm$. So
\[
 \|\polar(M)-Q\|_F\le\int_0^1m\,(1-rm)^{-1}dr=-\log(1-m).
\]
The path consists
of invertible matrices, so the determinant of the polar factor is constant.
\end{proof}

\begin{alem}[Generator of a corrected frame]
\label{gen:lem:generator}
Let $S\in\SO(2n)$ with $\|S-I\|_F=d_0$ and $A=\skw S$. Let $B$ be real and
skew, $P=\polar(I+B)$ and $d'=\|P-S\|_F$. If $d_0+d'\le1/\sqrt2$, then
$Y=\Log(P^{\T}S)$ is real and skew, $e^{Y}=P^{\T}S$,
$\|Y\|_F\le\varsigma(d')\,d'$, and
\begin{equation}
 \|Y-(A-B)\|_F\le(d_0+d')\,d'+\tfrac12(d_0+d')^3 .
 \label{gen:eq:generator}
\end{equation}
\end{alem}
\begin{proof}
\noindent\emph{Logarithm.}
$S'=P^{\T}S\in\SO(2n)$ and
\[
 \|S'-I\|_F=\|S-P\|_F=d'.
\]
As in
Lemma~\ref{gen:lem:gauge}(b),(c), $\|S'-I\|_{\op}\le d'/\sqrt2\le\tfrac12$ and
$\|Y\|_F\le\varsigma(d')d'$. Let $A'=\skw S'$. By \eqref{gen:eq:round},
\[
 \|Y-A'\|_F\le\tfrac15\|A'\|_{\op}^2\|A'\|_F\le d'^3/10.
\]

\noindent\emph{Skew part.}
\[
 S'-I=P^{\T}(S-P)=(S-P)+(P-I)^{\T}(S-P).
\]
So
\[
 A'=A-\skw P+\skw\bigl((P-I)^{\T}(S-P)\bigr),
\]
and the last term has norm at
most $\|P-I\|_{\op}\,d'$. By \eqref{gen:eq:skw-pol},
\[
 \|\skw P-B\|_F\le\tfrac12\|B\|_{\op}^2\|B\|_F.
\]
Hence
\[
 \|Y-(A-B)\|_F\le\frac{d'^3}{10}+\|P-I\|_{\op}\,d'
 +\tfrac12\|B\|_{\op}^2\|B\|_F .
\]

\noindent\emph{Size of $B$.}
\[
\begin{aligned}
 \|P-I\|_{\op}&\le\|S-I\|_{\op}+\|S'-I\|_{\op}\\
 &\le(d_0+d')/\sqrt2\le\tfrac12
\end{aligned}
\]
and $\|P-I\|_F\le d_0+d'$. The eigenvalues of $P$ are $e^{i\psi}$ with
$\tan\psi$ an eigenvalue of $h=-iB$, and
$2|\sin(\psi/2)|\le\tfrac12$ gives $|\psi|<\pi/6$. On this range
$|\tan\psi|\le\tfrac98\cdot2|\sin(\psi/2)|$, because the quotient is increasing and its value at $\pi/6$ is
$\cos(\pi/12)/\cos(\pi/6)<9/8$. So
$\|B\|_{\op}\le\tfrac98\|P-I\|_{\op}$ and $\|B\|_F\le\tfrac98\|P-I\|_F$.
Inserting, the right side is at most
\[
 \tfrac1{10}d'^3+\tfrac1{\sqrt2}(d_0+d')d'+\tfrac{729}{2048}(d_0+d')^3,
\]
which is below the right side of \eqref{gen:eq:generator}.
\end{proof}

\subsection{Pure Gaussian states}
\label{gen:app:gaussian}
Throughout this section $b_1,\dots,b_n$ are the particle--hole modes of a
pattern $x$, defined below (for $x=0$ they are the $a_j$), with the
Majorana operators $\gref^b_{2j-1}=b_j+b_j^\dagger$ and
$\gref^b_{2j}=i(b_j-b_j^\dagger)$. The covariance of a state in these modes
is $\Gamma^b_{\mu\nu}=\tfrac i2\langle[\gref^b_\mu,\gref^b_\nu]\rangle$, as
in \eqref{gen:eq:conventions}. A \emph{pure Gaussian state} is a state
$U_{R'}\ket{x'}$ with $R'\in\SO(2n)$ and a pattern $x'$. For a pattern $x$
put
\[
 \Theta_x=\bigoplus_{j=1}^n\diag(1,\,1-2x_j),\qquad
 \Theta_x\Gamma_x\Theta_x=J_n .
\]
The \emph{particle--hole modes of $x$} are $b_j=a_j$ if $x_j=0$ and
$b_j=a_j^\dagger$ if $x_j=1$. A state with covariance $\Gamma$ in the modes
$a_j$ has covariance $\Gamma^b=\Theta_x\Gamma\Theta_x$ in the modes $b_j$, and
$\ket x$ is the vacuum of the $b_j$. In this section
$n_j=b_j^\dagger b_j$.

\begin{alem}[Pure Gaussian states]
\label{gen:lem:gaussian}
Let $\psi$ be a pure Gaussian state and
\begin{equation}
\begin{gathered}
 D_{jk}=\langle b_j^\dagger b_k\rangle,\qquad
 \Delta_{jk}=\langle b_jb_k\rangle,\\
 T_{jk}=\langle b_j^\dagger b_k^\dagger\rangle=-\overline{\Delta_{jk}},\qquad
 \gamma=\|\Gamma^b-J_n\|_F .
\end{gathered}
\label{gen:eq:DKT}
\end{equation}
\begin{enumerate}[label=(\alph*)]
\item $\Gamma^b=J_n+L_\Delta+L_D$ with the $2\times2$ blocks
\begin{equation}
\begin{gathered}
 (L_\Delta)_{jk}=-2\begin{pmatrix}\mathrm{Im}\Delta_{jk}&\mathrm{Re}\Delta_{jk}\\
                \mathrm{Re}\Delta_{jk}&-\mathrm{Im}\Delta_{jk}\end{pmatrix},\\
 (L_D)_{jk}=-2\begin{pmatrix}\mathrm{Im}D_{jk}&\mathrm{Re}D_{jk}\\
                -\mathrm{Re}D_{jk}&\mathrm{Im}D_{jk}\end{pmatrix}\quad(j\ne k),\\
 (L_\Delta)_{jj}=0,\qquad (L_D)_{jj}=-2D_{jj}J .
\end{gathered}
\label{gen:eq:LKLD}
\end{equation}
$L_\Delta$ anticommutes with $J_n$, $L_D$ commutes with $J_n$, and
$\|L_\Delta\|_F^2=8\|\Delta\|_F^2$, $\|L_D\|_F^2=8\|D\|_F^2$.
\item $D-D^2=\Delta^\dagger\Delta$ and
\[
 \Tr D=\|\Delta\|_F^2+\|D\|_F^2=\gamma^2/8.
\]
\item If $\gamma\le1/2$, then $\|\Delta\|_F\le\gamma/\sqrt8$ and
$\|D\|_F\le\tfrac54\|\Delta\|_F^2$.
\item Let $Y$ be the number of occupied $b$ modes in an occupation readout,
the number of \emph{clicks}. Then $\E Y=\gamma^2/8$ and
$\Var Y=2\|\Delta\|_F^2\le2\,\E Y$.
\end{enumerate}
\end{alem}
\begin{proof}
\noindent\emph{Blocks.}
For $j\ne k$ the Majorana operators of the two modes anticommute, so
$\Gamma^b_{\mu\nu}=i\langle\gref^b_\mu\gref^b_\nu\rangle$ for their
indices $\mu,\nu$. With
$\langle b_jb_k^\dagger\rangle=-\overline{D_{jk}}$,
\[
\begin{aligned}
 \Gamma^b_{2j-1,2k-1}
 &=i\langle(b_j+b_j^\dagger)(b_k+b_k^\dagger)\rangle\\
 &=i\,(\Delta_{jk}-\overline{\Delta_{jk}})+i\,(D_{jk}-\overline{D_{jk}})\\
 &=-2\,\mathrm{Im}\Delta_{jk}-2\,\mathrm{Im}D_{jk},\\
 \Gamma^b_{2j-1,2k}
 &=-\langle(b_j+b_j^\dagger)(b_k-b_k^\dagger)\rangle\\
 &=-2\,\mathrm{Re}\Delta_{jk}-2\,\mathrm{Re}D_{jk} .
\end{aligned}
\]
The entries $(2j,2k-1)$ and $(2j,2k)$ follow in the same way and give
\eqref{gen:eq:LKLD}. For $j=k$, $\Gamma^b_{2j-1,2j}=1-2D_{jj}$. A block
$\bigl(\begin{smallmatrix}p&q\\q&-p\end{smallmatrix}\bigr)$ anticommutes with
$J$ and a block
$\bigl(\begin{smallmatrix}p&q\\-q&p\end{smallmatrix}\bigr)$ commutes with $J$;
the two kinds are orthogonal in the Frobenius inner product. Each block of
$L_\Delta$ has squared norm $8|\Delta_{jk}|^2$, each block of $L_D$ has $8|D_{jk}|^2$.

\noindent\emph{Projector.}
The matrix of the second moments of $(b_j^\dagger)_j,(b_j)_j$ is
\[
 \mathcal G=\begin{pmatrix}D&\Delta^\dagger\\\Delta&I-D^{\T}\end{pmatrix}.
\]
For the vacuum of the $b_j$ it is $\diag(0,I)$, a projector. A Gaussian
unitary conjugates $\mathcal G$ by a unitary matrix. The operator
$b_j+b_j^\dagger$ exchanges $b_j$ and $b_j^\dagger$ and changes the sign of
the other $b_k,b_k^\dagger$, so it also conjugates $\mathcal G$ by a unitary
matrix. Every pure Gaussian state $U_{R'}\ket{x'}$ arises from the vacuum
$\ket x$ of the $b_j$ by such operations: the operators $b_j+b_j^\dagger$
with $x_j\ne x'_j$ map $\ket x$ to $\pm\ket{x'}$, and $U_{R'}$ is a Gaussian
unitary. So $\mathcal G^2=\mathcal G$,
whose upper left block is $D^2+\Delta^\dagger\Delta=D$. Taking the trace,
$\Tr D=\|D\|_F^2+\|\Delta\|_F^2$, and by (a) this is $\gamma^2/8$.

\noindent\emph{Small $\gamma$.}
By (b), $\|\Delta\|_F^2\le\gamma^2/8$ and $\|D\|_{\op}\le\Tr D\le1/32$. From
$D(I-D)=\Delta^\dagger\Delta$,
\[
 \|D\|_F\le\tfrac{32}{31}\|\Delta\|_{\op}\|\Delta\|_F\le\tfrac54\|\Delta\|_F^2.
\]

\noindent\emph{Clicks.}
Wick's theorem for a Gaussian state~\cite{bravyi} gives, for $j\ne k$,
\[
 \langle n_jn_k\rangle=D_{jj}D_{kk}-|D_{jk}|^2+|\Delta_{jk}|^2.
\]
Hence
\[
\begin{aligned}
 \Var Y&=\sum_j(D_{jj}-D_{jj}^2)+\sum_{j\ne k}\bigl(|\Delta_{jk}|^2-|D_{jk}|^2\bigr)\\
 &=\Tr D-\|D\|_F^2+\|\Delta\|_F^2 ,
\end{aligned}
\]
which is $2\|\Delta\|_F^2$ by (b).
\end{proof}

\subsection{The certificate}
\label{gen:app:certificate}
A \emph{click run} at $\hat R$: draw $x$ uniformly, take $U_R\ket x$, apply
$U_{\hat R}^\dagger$, read the occupations $y$. The run clicks if $y\ne x$.
Let $\bar p$ be the probability of a click, and $d=\dist(\hat R,R)$.

\begin{alem}[Click law]
\label{gen:lem:click}
For all $\hat R,R\in\SO(2n)$,
\begin{equation}
 \frac{d^2\,\max\{1,\,2-d^2/4\}}{16+4d^2}\;\le\;\bar p\;\le\;\frac{d^2}8 .
 \label{gen:eq:click}
\end{equation}
For $0<u\le1$, every frame with $d^2\ge u$ has
\[
 \bar p\ge\mathsf g(u):=\frac u8\cdot\frac{1-u/8}{1+u/4} .
\]
\end{alem}
\begin{proof}
\noindent\emph{One pattern.}
Use the gauge \eqref{gen:eq:gauge}. The rotated state is $U_S\ket x$ up to a
phase, because $T_*$ commutes with $\Gamma_x$. In the particle--hole modes of
$x$ it is pure Gaussian with
\[
 \gamma(x)=\|S\Gamma_xS^{\T}-\Gamma_x\|_F=\|[S,\Gamma_x]\|_F\le2d.
\]
Let $Y$
be the number of modes with $y_j\ne x_j$. By Lemma~\ref{gen:lem:gaussian}(d),
$\E Y=\gamma(x)^2/8$ and $\Var Y\le2\,\E Y$. The circuit is even, so $Y$ is
even. Hence the click probability $p(x)=\Prob\{Y\ge2\}$ satisfies
$p(x)\le\E Y/2=\gamma(x)^2/16$ and, by
$\Prob\{Y>0\}\ge(\E Y)^2/\E Y^2$,
\[
 p(x)\ge\frac{\E Y}{2+\E Y}=\frac{\gamma(x)^2}{16+\gamma(x)^2}
 \ge\frac{\gamma(x)^2}{16+4d^2}.
\]

\noindent\emph{Average over patterns.}
The block $[jk]$ of $[S,\Gamma_x]$ is
$\sigma_kS_{[jk]}J-\sigma_jJS_{[jk]}$ with the independent uniform signs
$\sigma_j=1-2x_j$. For $j\ne k$ its mean square is $2\|S_{[jk]}\|_F^2$. For
$j=k$ it is $\|[S_{[jj]},J]\|_F^2=2(\lambda_j-\lambda_j')^2$, where
$\lambda_j\ge\lambda_j'$ are the eigenvalues of the symmetric block
$S_{[jj]}$. The rows of $S$ have norm $1$, so
$\sum_{j\ne k}\|S_{[jk]}\|_F^2=2n-\sum_j(\lambda_j^2+\lambda_j'^2)$. Hence
\[
 \E_x\gamma(x)^2=\sum_j\bigl(4-4\lambda_j\lambda_j'\bigr),\qquad
 d^2=\sum_j\bigl(4-2\lambda_j-2\lambda_j'\bigr).
\]
Put $\alpha_j=1-\lambda_j$ and $\beta_j=1-\lambda_j'$. Both lie in $[0,2]$,
and $\alpha_j+\beta_j\le2$ by Lemma~\ref{gen:lem:gauge}(a). Then
\[
\begin{gathered}
 \E_x\gamma^2=\sum_j4(\alpha_j+\beta_j-\alpha_j\beta_j)\\
 \text{and}\quad d^2=\sum_j2(\alpha_j+\beta_j).
\end{gathered}
\]
So $\E_x\gamma^2\le2d^2$ and
\[
 \E_x\gamma^2\ge\sum_j4(\alpha_j+\beta_j)(1-\tfrac14(\alpha_j+\beta_j))\ge d^2.
\]
Also $4\alpha_j\beta_j\le(\alpha_j+\beta_j)^2$ and
$\alpha_j+\beta_j\le d^2/2$ give $\E_x\gamma^2\ge2d^2(1-d^2/8)$. With the two
bounds on $p(x)$ this is \eqref{gen:eq:click}.

\noindent\emph{Threshold.}
Put $v=d^2$ and
\[
 h(v)=v(2-v/4)/(16+4v)=\mathsf g(v).
\]
Then
\[
 \begin{gathered}
 h'(v)=\frac{32-8v-v^2}{(16+4v)^2}>0
                 \qquad (0\le v\le8/3),\\
 \frac{v}{16+4v}\ge\frac1{10}>\frac7{80}=h(1)
                 \qquad (v\ge8/3).
 \end{gathered}
\]
Apply the first line when $u\le v\le8/3$ and the second when $v\ge8/3$.
\end{proof}

\begin{alem}[Click test]
\label{gen:lem:test}
Let $\rho>1$, $0<u\le1$, $0<q_{\rm t},\delta'<1$,
$p_{\rm a}=u/(8\rho)$ and $p_{\rm r}=\mathsf g(u)$, and assume
$p_{\rm r}>p_{\rm a}$. Let $\mathsf T(u,\delta')$ take $m$ click runs and
accept if the number of clicks is at most $k$, where $(m,k)$ satisfy
\begin{equation}
\begin{gathered}
 \Prob\{\mathrm{Bin}(m,p_{\rm a})>k\}\le q_{\rm t},\\
 \Prob\{\mathrm{Bin}(m,p_{\rm r})\le k\}\le\delta' .
\end{gathered}
 \label{gen:eq:test}
\end{equation}
Then a frame with $d^2\le u/\rho$ is rejected with probability at most
$q_{\rm t}$, and a frame with $d^2\ge u$ is accepted with probability at most
$\delta'$. With $a_1=\sqrt{p_{\rm r}/p_{\rm a}}-1$ and
$a_2=1-\sqrt{p_{\rm a}/p_{\rm r}}$, the choice
\begin{equation}
 \begin{gathered}
 m=\Bigl\lceil\max\Bigl\{\frac{(2+a_1)\log(1/q_{\rm t})}{a_1^2\,p_{\rm a}},\
 \frac{2\log(1/\delta')}{a_2^2\,p_{\rm r}}\Bigr\}\Bigr\rceil,\\
 k=\bigl\lfloor m\sqrt{p_{\rm a}p_{\rm r}}\bigr\rfloor
\end{gathered}
 \label{gen:eq:test-size}
\end{equation}
satisfies \eqref{gen:eq:test}; every larger integer $m$, with its corresponding
threshold $k$, does as well. Let $0<u_\rho\le1$ with
$\rho_\rho:=\rho(1-u_\rho/8)/(1+u_\rho/4)>1$. Then for every
$u\le u_\rho$ a sufficient size is
\[
 m=\Bigl\lceil\frac cu\log\frac1{\min\{q_{\rm t},\delta'\}}\Bigr\rceil,
 \qquad
 c=\frac{8\rho(1+\sqrt{\rho_\rho})}{(\sqrt{\rho_\rho}-1)^2}.
\]
This also bounds the length of the shortest feasible test.
\end{alem}
\begin{proof}
The clicks of the $m$ runs are independent with probability $\bar p$, so
their number is $\mathrm{Bin}(m,\bar p)$, which is stochastically increasing
in $\bar p$. By Lemma~\ref{gen:lem:click}, $d^2\le u/\rho$ gives
$\bar p\le p_{\rm a}$ and $d^2\ge u$ gives $\bar p\ge p_{\rm r}$. The
Chernoff bounds
\[
\begin{gathered}
 \Prob\{\mathrm{Bin}(m,p)\ge(1+a)mp\}\le e^{-a^2mp/(2+a)}\\
 \text{and}\quad \Prob\{\mathrm{Bin}(m,p)\le(1-a)mp\}\le e^{-a^2mp/2}
\end{gathered}
\]
with the threshold
\[
 m\sqrt{p_{\rm a}p_{\rm r}}=(1+a_1)mp_{\rm a}=(1-a_2)mp_{\rm r}
\]
give
\eqref{gen:eq:test-size}. Finally put
$\varrho=p_{\rm r}/p_{\rm a}=\rho(1-u/8)/(1+u/4)$. It decreases in $u$,
so $\varrho\ge\rho_\rho$ for $u\le u_\rho$. With $a_1=\sqrt\varrho-1$,
$a_2=a_1/\sqrt\varrho$ and $p_{\rm a}=u/(8\rho)$,
\[
\begin{gathered}
 \frac{2+a_1}{a_1^2\,p_{\rm a}}
 =\frac{8\rho}{u}\,\frac{1+\sqrt\varrho}{(\sqrt\varrho-1)^2},\\
 \frac{2}{a_2^2\,p_{\rm r}}
 =\frac{8\rho}{u}\,\frac{2}{(\sqrt\varrho-1)^2}.
\end{gathered}
\]
Both are at most $c/u$: the functions $(1+x)/(x-1)^2$ and $2/(x-1)^2$
decrease for $x>1$, and $2\le1+\sqrt{\rho_\rho}$.
\end{proof}

\subsection{Matching records and the start frame}
\label{gen:app:matching}
\begin{alem}[Matching record]
\label{gen:lem:matching}
Let a state have covariance $\Gamma$. Draw a uniformly random perfect matching
$\mathfrak m$ of $\{1,\dots,2n\}$, measure the commuting observables
$i\gref_p\gref_q$, $\{p,q\}\in\mathfrak m$, $p<q$, with outcomes
$z_{pq}\in\{\pm1\}$, and put $\hat\Gamma=(2n-1)\mathsf M$, where
$\mathsf M_{pq}=z_{pq}$, $\mathsf M_{qp}=-z_{pq}$ for matched pairs and
$\mathsf M_{pq}=0$ otherwise. Then $\E\hat\Gamma=\Gamma$, the matrix
$\mathsf M$ is orthogonal, and $\|\hat\Gamma\|_F^2=2n(2n-1)^2$. If the state
is pure Gaussian and $n\ge2$, then with
\[
 \begin{gathered}
 \sM^2=\frac{4n(n-1)^2(2n-1)}{2n-3},\\
 \omega_n=\frac{2(2n-1)(3n-2)}{2n-3},
\end{gathered}
\]
\begin{enumerate}[label=(\alph*)]
\item for every real skew orthogonal $C$,
\[
\begin{aligned}
 \E\Tr\bigl((\mathsf MC)^2\bigr)
 &=\frac{2n}{2n-1}+\frac{2\Tr((\Gamma C)^2)-\langle\Gamma,C\rangle_F^2}{(2n-1)(2n-3)}\\
 &\ge-\frac{2n}{(2n-1)(2n-3)} ;
\end{aligned}
\]
\item for every real skew orthogonal $C$,
$\E\|\tfrac12[C,\hat\Gamma]\|_F^2\le\sM^2$, with equality for $C=\pm\Gamma$;
\item for every real skew $V$, $\E\langle V,\hat\Gamma\rangle_F^2\le\omega_n\|V\|_F^2$.
\end{enumerate}
\end{alem}
\begin{proof}
\noindent\emph{Measurement.}
The observables of disjoint pairs commute. Number the pairs of $\mathfrak m$
as $\{p_j,q_j\}$, $p_j<q_j$. Let $W$ be the permutation matrix with
$W_{2j-1,p_j}=W_{2j,q_j}=1$. Since $1-2n_j=i\gref_{2j-1}\gref_{2j}$,
\eqref{gen:eq:conventions} gives
$U_W^\dagger(1-2n_j)U_W=i\gref_{p_j}\gref_{q_j}$. So $U_W$ followed by the
occupation readout $y$ measures all pairs at once, with
$z_{p_jq_j}=1-2y_j$. If $\det W=-1$, exchange $p_1$ and $q_1$ in $W$. Then
$W\in\SO(2n)$, the first observable becomes
$i\gref_{q_1}\gref_{p_1}=-i\gref_{p_1}\gref_{q_1}$, and
$z_{p_1q_1}=-(1-2y_1)$. The other outcomes are unchanged.

\noindent\emph{Mean and norm.}
$\E[z_{pq}\mid\{p,q\}\in\mathfrak m]=\Gamma_{pq}$, and a given pair is matched
with probability $1/(2n-1)$. $\mathsf M$ has one entry $\pm1$ in every row
and column.

\noindent\emph{(a).}
Write $\zeta_{pq}=\mathsf M_{pq}$ for a matched pair, so
$\zeta_{qp}=-\zeta_{pq}$. For two different matched pairs Wick's theorem
gives
\[
 \E[\zeta_{pp'}\zeta_{rr'}]=\Gamma_{pp'}\Gamma_{rr'}-\Gamma_{pr}\Gamma_{p'r'}
+\Gamma_{pr'}\Gamma_{p'r}.
\]
Let $\pi$ be the matching as an involution. Then
\[
 \Tr(\mathsf MC\mathsf MC)=\sum_{p,r}\zeta_{p,\pi p}C_{\pi p,r}
\zeta_{r,\pi r}C_{\pi r,p}.
\]
The terms with $r=p$ give
$\sum_pC_{\pi p,p}^2$, with mean $\|C\|_F^2/(2n-1)=2n/(2n-1)$. The terms with
$r=\pi p$ vanish. For four prescribed distinct indices $p,p',r,r'$, the probability
that $\pi p=p'$ and $\pi r=r'$ is $1/((2n-1)(2n-3))$: the two partners are
chosen from $2n-1$ and then $2n-3$ possibilities. The remaining terms
therefore contribute this probability times the sum over distinct indices of
\[
 \bigl[\Gamma_{pp'}\Gamma_{rr'}-\Gamma_{pr}\Gamma_{p'r'}
 +\Gamma_{pr'}\Gamma_{p'r}\bigr]\,C_{p'r}C_{r'p} .
\]
The bracket vanishes if two of the four indices coincide, so the sum may run
over all indices. Its three parts are $\Tr(\Gamma C\Gamma C)$,
$\Tr(\Gamma C\Gamma C)$ and $\Tr(\Gamma C)\langle\Gamma,C\rangle_F
=-\langle\Gamma,C\rangle_F^2$. For the inequality let $O=\Gamma C$, a real
orthogonal matrix with eigenvalues $e^{i\theta_k}$, and
$c_k=\cos\theta_k$. Then
\[
\begin{aligned}
 2\Tr(O^2)-(\Tr O)^2
 &=4\sum_kc_k^2-4n-\Bigl(\sum_kc_k\Bigr)^2\\
 &\ge(4-2n)\sum_kc_k^2-4n\ge4n-4n^2,
\end{aligned}
\]
because
$(\sum_kc_k)^2\le2n\sum_kc_k^2$ and $\sum_kc_k^2\le2n$.

\noindent\emph{(b).}
\[
\begin{aligned}
 \|\tfrac12[C,\hat\Gamma]\|_F^2
 &=\tfrac14\|C\hat\Gamma-(C\hat\Gamma)^{\T}\|_F^2\\
 &=\tfrac12\|\hat\Gamma\|_F^2-\tfrac12(2n-1)^2\Tr((\mathsf MC)^2).
\end{aligned}
\]
With (a) the
mean is at most
\[
 n(2n-1)^2+n(2n-1)/(2n-3)=\sM^2.
\]

\noindent\emph{(c).}
$\langle V,\hat\Gamma\rangle_F=2(2n-1)\sum_{e\in\mathfrak m}V_e\zeta_e$. The
squares give
\[
 \E\sum_eV_e^2=\|V\|_F^2/(2(2n-1)).
\]
As in (a), the mixed terms
give
\[
 \tfrac14[\langle V,\Gamma\rangle_F^2-2\Tr((V\Gamma)^2)]/((2n-1)(2n-3)).
\]
So
\[
\begin{aligned}
 \E\langle V,\hat\Gamma\rangle_F^2&=2(2n-1)\|V\|_F^2\\
 &\quad+\tfrac{2n-1}{2n-3}[\langle V,\Gamma\rangle_F^2-2\Tr((V\Gamma)^2)],
\end{aligned}
\]
and
$\langle V,\Gamma\rangle_F^2\le2n\|V\|_F^2$,
$-2\Tr((V\Gamma)^2)\le2\|V\|_F^2$.
\end{proof}

\begin{alem}[Norm of a mean]
\label{gen:lem:mean}
Let $Y_1,\dots,Y_N$ be independent copies of a centered Hilbert-space
vector $Y$. For some $\sigma^2>0$, suppose $\E\|Y\|^2\le\sigma^2$,
$\|Y\|^2\le\mathsf b^2$ almost surely, and
$\E|\langle V,Y\rangle|^2\le\omega\|V\|^2$ for all $V$. Then for
the mean $\overline Y$ and every $a>0$,
\begin{equation}
 \Prob\Bigl\{\|\overline Y\|^2\ge(1+a)\frac{\sigma^2}N\Bigr\}
 \le\frac{2\omega/\sigma^2+\mathsf b^2/(N\sigma^2)}{a^2} .
 \label{gen:eq:mean-tail}
\end{equation}
\end{alem}
\begin{proof}
$\|\sum_iY_i\|^2=\sum_i\|Y_i\|^2+\sum_{i\ne j}\langle Y_i,Y_j\rangle$. The two
sums are uncorrelated, because
\[
 \E[\|Y_i\|^2\langle Y_i,Y_j\rangle]=\langle\E[\|Y_i\|^2Y_i],\E Y_j\rangle=0.
\]
In the variance of the second sum only the pairs $(k,l)=(i,j)$ and $(j,i)$
contribute, each at most $\E|\langle Y,Y'\rangle|^2$, with $Y'$ an
independent copy. With $\sigma_{\rm t}^2=\E\|Y\|^2$,
$\E|\langle Y,Y'\rangle|^2\le\omega\sigma_{\rm t}^2$ and
$\Var\|Y\|^2\le\mathsf b^2\sigma_{\rm t}^2$. So
\[
 \Var\|\overline Y\|^2\le(\mathsf b^2/N+2\omega)\sigma^2/N^2.
\]
Since
$\E\|\overline Y\|^2=\sigma_{\rm t}^2/N\le\sigma^2/N$, Chebyshev's inequality
gives \eqref{gen:eq:mean-tail}.
\end{proof}

\begin{alem}[Plane estimator]
\label{gen:lem:planes}
For $i=1,\dots,N$ draw a uniform pattern $x^{(i)}$ and take the matching
record $\hat\Gamma^{(i)}$ of $U_R\ket{x^{(i)}}$. With $\sigma_j=1-2x_j$ put
$\hat P_j=\tfrac1N\sum_i\sigma_j^{(i)}\hat\Gamma^{(i)}$. Let $\hat w_j$ be a
unit eigenvector of the Hermitian matrix $i\hat P_j$ for its largest
eigenvalue, chosen by a fixed rule if it is not unique, $\mathsf W$ the real
matrix with the columns
$\sqrt2\,\mathrm{Re}\,\hat w_j$, $-\sqrt2\,\mathrm{Im}\,\hat w_j$
$(j=1,\dots,n)$, and $\hat R_0=\polar(\mathsf W)$, with the sign of one column
changed if the determinant is $-1$, and $\hat R_0=I$ if $\mathsf W$ is
singular. For $0<t<\tfrac12$ and $0<u\le1$ let
\[
 \begin{gathered}
 c_t=\frac4{(1-t)^2\bigl(1+\sqrt{1-t^2/(1-t)^2}\bigr)},\\
 x_*(u,t)=\frac{(1-e^{-\sqrt u})^2}{c_t} .
\end{gathered}
\]
If
\[
 a=x_*(u,t)\,N/(4n^2(n-1))-1>0,
\]
then
\begin{equation}
\begin{aligned}
 &\Prob\{\dist(\hat R_0,R)^2>u\}\\
 &\quad\le4n^2\exp\Bigl(-\frac{Nt^2}{2(2n-1)^2+4nt/3}\Bigr)\\
 &\qquad+\frac{\omega_n/(4n^2(n-1))+n/((n-1)N)}{a^2} .
\end{aligned}
 \label{gen:eq:planes}
\end{equation}
\end{alem}
\begin{proof}
\noindent\emph{Planes.}
Let $r_a$ be the columns of $R$ and
\[
 P_j=RJ^{(j)}R^{\T}
=r_{2j-1}r_{2j}^{\T}-r_{2j}r_{2j-1}^{\T}.
\]
Since
$\Gamma_x=\sum_k\sigma_kJ^{(k)}$ with independent centered signs,
$\E[\sigma_j\hat\Gamma]=P_j$. The Hermitian matrix $iP_j$ has the eigenvalues
$1,-1,0$, and $w_j=(r_{2j-1}-ir_{2j})/\sqrt2$ is its unit eigenvector for $1$.
Let $\Delta_j=\hat P_j-P_j$ and $X=\sum_j\|\Delta_jw_j\|^2$.

\noindent\emph{The sum $X$.}
$X=\|\overline\xi\|^2$ for the mean of the vectors
$\xi=((\sigma_j\hat\Gamma-P_j)w_j)_{j\le n}$. They are centered, and by
Lemma~\ref{gen:lem:matching},
\[
 \E\|\xi\|^2=n((2n-1)^2-1)=4n^2(n-1)
\]
and
$\|\xi\|^2\le4n^3$. For $V=(V_j)$ the pairing $\langle V,\xi\rangle$ equals
$\langle W,\hat\Gamma\rangle$ up to a constant, with
$W=\sum_j\sigma_j\overline{V_j}w_j^{\T}$. Only the skew part of $W$ counts.
Its mean square over the patterns is
$\tfrac12\|V\|^2-\tfrac12\sum_j|V_j^{\T}w_j|^2\le\tfrac12\|V\|^2$, because
the $w_j$ are orthonormal and the signs are independent. Its real and
imaginary parts are real skew. So Lemma~\ref{gen:lem:matching}(c), applied for
each pattern, gives $\E|\langle V,\xi\rangle|^2\le\tfrac12\omega_n\|V\|^2$.
Lemma~\ref{gen:lem:mean} bounds $\Prob\{X>x_*\}$ by the second term of
\eqref{gen:eq:planes}.

\noindent\emph{Operator norms.}
The summands of $\Delta_j$ are independent, centered, of operator norm at
most $2n$, and $\E[(\cdot)(\cdot)^{\T}]\preceq(2n-1)^2I$. The matrix Bernstein
inequality~\cite{tropp} gives
\[
 \Prob\{\|\Delta_j\|_{\op}\ge t\}\le4n\exp(-Nt^2/(2(2n-1)^2+4nt/3)).
\]
The
union over $j$ is the first term of \eqref{gen:eq:planes}. Let $\mathcal O$ be
the event that $\|\Delta_j\|_{\op}\le t$ for all $j$.

\noindent\emph{One plane.}
Let $\theta_j$ be the angle between $\hat w_j$ and $w_j$, and
$\hat Q=I-\hat w_j\hat w_j^\dagger$. From $(i\hat P_j-1)w_j=i\Delta_jw_j$ and
Weyl's inequality (the second eigenvalue of $i\hat P_j$ is at most
$\|\Delta_j\|_{\op}$): on $\mathcal O$,
\[
 \sin\theta_j=\|\hat Qw_j\|\le\|\Delta_jw_j\|/(1-t)\le t/(1-t).
\]

\noindent\emph{All planes.}
A phase of $\hat w_j$ rotates the two columns of mode $j$ in their plane.
This is a mode phase and does not change the distance. Choose the phases
with $\langle w_j,\hat w_j\rangle\ge0$. Then
\[
\begin{aligned}
 \|\mathsf W-R\|_F^2&=2\sum_j\|\hat w_j-w_j\|^2\\
 &=4\sum_j\sin^2\theta_j/(1+\cos\theta_j)\le c_tX
\end{aligned}
\]
on $\mathcal O$.

\noindent\emph{Polar factor.}
If $X\le x_*$ on $\mathcal O$, then
$\|\mathsf W-R\|_F\le1-e^{-\sqrt u}<1$, and Lemma~\ref{gen:lem:round}(d) gives
$\dist(\hat R_0,R)\le\sqrt u$ and $\det\polar(\mathsf W)=1$.
\end{proof}

\begin{alem}[Matching correction]
\label{gen:lem:mcorr}
Let
\[
 d_{\rm c}=\dist(\hat R_{\rm c},R)\le1/\sqrt2.
\]
Take $N$ fresh reference
runs with uniform patterns $x_i$ and matching records $\hat\Gamma_i$, and put
\[
 B=\frac1N\sum_i\tfrac12\bigl[\Gamma_{x_i},\,
 \hat R_{\rm c}^{\T}\hat\Gamma_i\hat R_{\rm c}\bigr],\qquad
 \hat R_0=\hat R_{\rm c}\,\polar(I+B).
\]
Then $B$ is real and skew,
\begin{equation}
 \begin{gathered}
 \dist(\hat R_0,R)\le\|B-\E B\|_F+\tfrac12d_{\rm c}^3,\\
 \E\|B-\E B\|_F^2\le\frac{\sM^2}N ,
\end{gathered}
 \label{gen:eq:mcorr}
\end{equation}
and for every $a>0$,
\begin{equation}
 \begin{aligned}
 &\Prob\Bigl\{\|B-\E B\|_F^2\ge(1+a)\frac{\sM^2}N\Bigr\}\\
 &\qquad\le\frac{\omega_n/\sM^2+8n(2n-1)^2/(N\sM^2)}{a^2} .
\end{aligned}
 \label{gen:eq:mcorr-tail}
\end{equation}
\end{alem}
\begin{proof}
Use the gauge \eqref{gen:eq:gauge} of $\hat R_{\rm c}$. For one run let
$b=\tfrac12[\Gamma_x,\hat R_{\rm c}^{\T}\hat\Gamma\hat R_{\rm c}]$. The record
is unbiased and $T_*$ commutes with $\Gamma_x$, so
$\E[b\mid x]=\tfrac12[\Gamma_x,S\Gamma_xS^{\T}]$. Insert $S=I+E$ and use
$\Gamma_x^2=-I$, $E-E^{\T}=2A$:
\[
 \tfrac12[\Gamma_x,S\Gamma_xS^{\T}]
 =A+\Gamma_xA\Gamma_x
 -\tfrac12\bigl(E\Gamma_xE^{\T}\Gamma_x-\Gamma_xE\Gamma_xE^{\T}\bigr).
\]
For uniform $x$ the signs are independent and centered, and $A_{[jj]}=0$. So
$\E_x\Gamma_xA\Gamma_x=0$, and with $X=\bigoplus_jJE_{[jj]}J$,
\begin{equation}
 \E_x\,\tfrac12[\Gamma_x,S\Gamma_xS^{\T}]=A-\tfrac12\bigl(EX-XE^{\T}\bigr).
 \label{gen:eq:cancel}
\end{equation}
By Lemma~\ref{gen:lem:gauge}(a),(b),
\[
 \|X\|_{\op}=\max_j\tfrac12\|(E^{\T}E)_{[jj]}\|_{\op}\le d_{\rm c}^2/4.
\]
So
$\|\E B-A\|_F\le d_{\rm c}^3/4$. By Lemma~\ref{gen:lem:round}(b),(c) and
$\|S-I\|_{\op}\le\tfrac12$,
\[
 \dist(\hat R_0,R)\le\|\polar(I+B)-S\|_F\le\|B-A\|_F+d_{\rm c}^3/4.
\]
This is
the first claim.

\noindent\emph{Moments.}
$\|b\|_F=\|\tfrac12[\hat R_{\rm c}\Gamma_x\hat R_{\rm c}^{\T},\hat\Gamma]\|_F$,
so $\E\|b\|_F^2\le\sM^2$ by Lemma~\ref{gen:lem:matching}(b), and
$\|b\|_F\le\|\hat\Gamma\|_F$, so $\|b-\E b\|_F^2\le8n(2n-1)^2$. For real skew
$V$, $\langle V,b\rangle=\langle V',\hat R_{\rm c}^{\T}\hat\Gamma\hat R_{\rm c}\rangle$
with $V'=\tfrac12(V\Gamma_x-\Gamma_xV)$, and
\[
 \E_x\|V'\|_F^2=\tfrac12\|V\|_F^2-\sum_jV_{2j-1,2j}^2\le\tfrac12\|V\|_F^2.
\]
So
$\E\langle V,b\rangle^2\le\tfrac12\omega_n\|V\|_F^2$ by
Lemma~\ref{gen:lem:matching}(c). Lemma~\ref{gen:lem:mean} gives
\eqref{gen:eq:mcorr-tail}.
\end{proof}

\subsection{From the frame to the state}
\label{gen:app:transfer}
For a real skew matrix $Y$ put
\begin{equation}
\begin{gathered}
 \dd(Y)=\tfrac14\sum_{a,b}Y_{ab}\gref_a\gref_b,\qquad H_Y=i\,\dd(Y),\\
 e^{-\dd(Y)}\gref_ae^{\dd(Y)}=\sum_b(e^{Y})_{ab}\gref_b .
\end{gathered}
 \label{gen:eq:quadratic}
\end{equation}
$H_Y$ is Hermitian, and by \eqref{gen:eq:conventions}
$U_{e^Y}=\pm e^{\dd(Y)}=\pm e^{-iH_Y}$. For a $2\times2$ block $M$ let
$\mathsf P^+M$ be its projection on the matrices
$\bigl(\begin{smallmatrix}p&q\\q&-p\end{smallmatrix}\bigr)$ and
$\mathsf P^-M$ its projection on the matrices
$\bigl(\begin{smallmatrix}p&q\\-q&p\end{smallmatrix}\bigr)$. Let $\pi_{jk}^+$
be the probability that the modes $j,k$ have equal occupations in a pattern
of $\Xi$, drawn with the weights of $\Xi$, and $\pi_{jk}^-=1-\pi_{jk}^+$.
For real skew $Y$ put
\begin{equation}
 \begin{aligned}
 \|Y\nP^2&=\frac12\sum_{j<k}\Bigl(\pi_{jk}^+\|\mathsf P^+Y_{[jk]}\|_F^2
 +\pi_{jk}^-\|\mathsf P^-Y_{[jk]}\|_F^2\Bigr)\\
 &\le\frac14\|Y_\od\|_F^2 .
\end{aligned}
 \label{gen:eq:phi-norm}
\end{equation}
This is a Hilbert seminorm. It depends only on $Y_\od$. It counts, for each
pair of modes, only the component that the input can excite.

\begin{alem}[Quadratic operators on patterns]
\label{gen:lem:operators}
Let $Y$ be real and skew, $H_\od=H_{Y_\od}$ and $H_{\rm d}=H_Y-H_\od$.
\begin{enumerate}[label=(\alph*)]
\item $H_{\rm d}$ is diagonal in the occupation basis, and
\[
 H_\od\ket x=\sum_{j<k}h_{jk}(x)\ket{x\oplus e_j\oplus e_k}
\]
with
$|h_{jk}(x)|^2=\tfrac12\|\mathsf P^{+}Y_{[jk]}\|_F^2$ if $x_j=x_k$ and
$\tfrac12\|\mathsf P^{-}Y_{[jk]}\|_F^2$ otherwise.
\item $\Pi H_\od\Pi=0$. For a unit vector $\psi=\sum_xc_x\ket x\in\cK$,
\begin{equation}
 \|H_\od\psi\|^2=\sum_x|c_x|^2\sum_{j<k}|h_{jk}(x)|^2\le\tfrac14\|Y_\od\|_F^2 .
 \label{gen:eq:variance}
\end{equation}
If $|c_x|^2$ are the pattern weights of $\Xi$, then
$\|H_\od\psi\|=\|Y\nP$.
\end{enumerate}
\end{alem}
\begin{proof}
\noindent\emph{(a).}
$H_Y=\tfrac i2\sum_{a<b}Y_{ab}\gref_a\gref_b$. A mode-diagonal term is
\[
 \tfrac i2Y_{2j-1,2j}\gref_{2j-1}\gref_{2j}=\tfrac12Y_{2j-1,2j}(1-2n_j).
\]
For $j<k$ collect the four terms with $a$ in mode $j$ and $b$ in mode $k$.
With the Jordan--Wigner sign $s_j(x)=\pm1$ and $\epsilon_j=2x_j-1$,
$\gref_{2j-1}\ket x=s_j\ket{x\oplus e_j}$ and
$\gref_{2j}\ket x=i\epsilon_js_j\ket{x\oplus e_j}$. With the entries
$Y_{11},Y_{12},Y_{21},Y_{22}$ of $Y_{[jk]}$ this gives
\[
 |h_{jk}(x)|^2=\tfrac14\bigl[(Y_{11}-\epsilon_j\epsilon_kY_{22})^2
 +(\epsilon_kY_{12}+\epsilon_jY_{21})^2\bigr].
\]
For $\epsilon_j\epsilon_k=1$ the two brackets are the coordinates of
$\mathsf P^+Y_{[jk]}$, of squared norm
$\tfrac12[(Y_{11}-Y_{22})^2+(Y_{12}+Y_{21})^2]$; for
$\epsilon_j\epsilon_k=-1$ those of $\mathsf P^-Y_{[jk]}$.

\noindent\emph{(b).}
$H_\od$ maps a pattern to states that differ from it in exactly two modes,
and two patterns of $\Xi$ differ in none or in at least six. So
$\Pi H_\od\Pi=0$. The states $\ket{x\oplus e_j\oplus e_k}$ with different
$(x,\{j,k\})$ are different, because their patterns $x$ would differ in at
most four modes. This gives the identity in \eqref{gen:eq:variance}. The bound
follows from $\|\mathsf P^\pm M\|_F\le\|M\|_F$ and
$\sum_{j<k}\|Y_{[jk]}\|_F^2=\tfrac12\|Y_\od\|_F^2$.
\end{proof}

\begin{alem}[Frame to state]
\label{gen:lem:frame-state}
Let $\vartheta=\vartheta(\hat R,R)\le\pi$ and
$\ket{\Xi'}=U_{\hat R}^\dagger\ket\Psi$. There is a canonical input
$\tilde\Xi$ such that, with $P_{\rm can}=\ket{\tilde\Xi}\bra{\tilde\Xi}$,
\[
 \begin{gathered}
 \varepsilon=\infid(\Xi',\tilde\Xi),\qquad
 \ell=\|(I-\Pi)\Xi'\|^2,\\ \sigma=\|(\Pi-P_{\rm can})\Xi'\|,
\end{gathered}
\]
\begin{enumerate}[label=(\alph*)]
\item $\varepsilon\le\sin^2(\vartheta/2)$;
\item $\varepsilon=\ell+\sigma^2$;
\item $\sigma\le2\sin^2(\vartheta/4)\le\vartheta^2/8$;
\item for every unit vector $\hat\Xi\in\cK$:
\[
 \infid(\Xi',\hat\Xi)\le\ell+\bigl(\infid(\tilde\Xi,\hat\Xi)^{1/2}+\sigma\bigr)^2.
\]
\end{enumerate}
\end{alem}
\begin{proof}
\noindent\emph{Gauge.}
Let $(Y,T)$ attain $\vartheta$. By Lemma~\ref{gen:lem:gauge}(d), $Y=Y_\od$. Put
$\tilde\Xi=U_T^\dagger\Xi$ and $H=H_Y$. Then
\[
 \ket{\Xi'}=\pm U_{\hat R^{\T}R}\ket\Xi=\pm e^{-iH}\ket{\tilde\Xi},
\]
and
$\Pi H\Pi=0$ by Lemma~\ref{gen:lem:operators}.

\noindent\emph{Speed limit.}
Let $\psi_s=e^{-isH}\tilde\Xi$ and $p(s)=|\braket{\tilde\Xi}{\psi_s}|^2$.
The variance $v^2$ of $H$ is the same in every $\psi_s$, and by
\eqref{gen:eq:variance} $v\le\|Y\|_F/2=\vartheta/2$. The uncertainty relation for
$H$ and $P_{\rm can}$ gives
\[
 |p'(s)|=|\langle[H,P_{\rm can}]\rangle_{\psi_s}|
\le2v\sqrt{p(1-p)}.
\]
Let $\phi(s)=\arccos\sqrt{p(s)}\in[0,\pi/2]$. It is continuous with
$\phi(0)=0$, and where $0<\phi<\pi/2$ the display gives $|\phi'|\le v$.
Fix $s\in[0,1]$ with $\phi(s)>0$, and let $s_0$ be the last zero of $\phi$
before $s$. If $\phi<\pi/2$ on $(s_0,s]$, then
$\phi(s)\le v(s-s_0)\le s\vartheta/2$. Otherwise the first point
$s_1\in(s_0,s]$ with $\phi(s_1)=\pi/2$ has
$\pi/2\le vs_1\le s\vartheta/2$. In both cases, since
$s\vartheta/2\le\pi/2$,
$1-p(s)=\sin^2\phi(s)\le\sin^2(s\vartheta/2)$ for $0\le s\le1$; this
is (a).

\noindent\emph{Orthogonal components (b).}
Since $P_{\rm can}\le\Pi$,
\[
\begin{gathered}
 (I-P_{\rm can})\psi_1=(I-\Pi)\psi_1+(\Pi-P_{\rm can})\psi_1,\\
 (I-\Pi)(\Pi-P_{\rm can})=0,\\ \varepsilon=\|(I-P_{\rm can})\psi_1\|^2=\ell+\sigma^2.
\end{gathered}
\]

\noindent\emph{Quadratic return (c).}
Lemma~\ref{gen:lem:operators} gives $\Pi H\Pi=0$ and
$\|H\Pi\|_{\op}\le\vartheta/2$. Hence
\[
 (\Pi-P_{\rm can})H\Pi=(\Pi-P_{\rm can})\Pi H\Pi=0.
\]
Since $(\Pi-P_{\rm can})\psi_0=0$,
\[
\begin{gathered}
\tfrac d{ds}(\Pi-P_{\rm can})\psi_s
 =-i(\Pi-P_{\rm can})H(I-\Pi)\psi_s,\\
(\Pi-P_{\rm can})\psi_1
 =-i\int_0^1(\Pi-P_{\rm can})H(I-\Pi)\psi_s\,ds.
\end{gathered}
\]
Self-adjointness of $H$ and part (a) imply
\[
\begin{aligned}
\|(\Pi-P_{\rm can})H(I-\Pi)\|_{\op}
 &\le\|\Pi H(I-\Pi)\|_{\op}\\
 &=\|(I-\Pi)H\Pi\|_{\op}\le\vartheta/2,\\
\|(I-\Pi)\psi_s\|
 &\le\|(I-P_{\rm can})\psi_s\|\\
 &\le\sin(s\vartheta/2).
\end{aligned}
\]
Hence, also at $\vartheta=0$,
\[
\begin{aligned}
\sigma&\le\frac\vartheta2\int_0^1\sin(s\vartheta/2)\,ds\\
 &=1-\cos(\vartheta/2)
 =2\sin^2(\vartheta/4)\le\vartheta^2/8.
\end{aligned}
\]

\noindent\emph{Reconstruction error (d).}
For $\hat P=\ket{\hat\Xi}\bra{\hat\Xi}\le\Pi$, put
$\alpha=\langle\tilde\Xi|\psi_1\rangle$ and $\xi=(\Pi-P_{\rm can})\psi_1$.
Then
\[
\begin{gathered}
\Pi\psi_1=\alpha\tilde\Xi+\xi,\qquad
|\alpha|\le1,\quad\|\xi\|=\sigma,\\
\begin{aligned}
\infid(\Xi',\hat\Xi)
 &=\ell+\|(\Pi-\hat P)\psi_1\|^2\\
 &\le\ell+
 \left(\sqrt{\infid(\tilde\Xi,\hat\Xi)}+\sigma\right)^2.
\end{aligned}
\end{gathered}
\]
The last step uses
\[
\begin{gathered}
 (\Pi-\hat P)\psi_1=(\Pi-\hat P)(\alpha\tilde\Xi+\xi)\\
 \text{and}\quad \|(\Pi-\hat P)\tilde\Xi\|=\infid(\tilde\Xi,\hat\Xi)^{1/2}.
\end{gathered}
\]
\end{proof}

\begin{alem}[Leak in any gauge]
\label{gen:lem:leak}
Let $Y$ be real and skew with $e^{Y}=\hat R^{\T}RT$ for a mode phase $T$. Then
the leak of $\ket{\Xi'}=U_{\hat R}^\dagger\ket\Psi$ satisfies
\[
 \ell^{1/2}\le\|Y_\od\nP\,\bigl(1+\tfrac14\|Y_\od\|_F\bigr).
\]
\end{alem}
\begin{proof}
Let $\tilde\Xi=U_T^\dagger\Xi$, so $\Xi'=\pm e^{-iH_Y}\tilde\Xi$, and
$H_Y=H_\od+H_{\rm d}$. $H_{\rm d}$ commutes with $\Pi$, and
$\Pi H_\od\Pi=0$. Let $\psi_s=e^{-isH_Y}\tilde\Xi$ and
$\mathsf u_s=(I-\Pi)\psi_s$. Then
\[
 \mathsf u_s'=-i(I-\Pi)H_Y(I-\Pi)\mathsf u_s-iH_\od\Pi\psi_s.
\]
The first term
does not change the norm. So
$\tfrac d{ds}\|\mathsf u_s\|\le\|H_\od\Pi\psi_s\|$. Let
$\varphi_s=e^{-isH_{\rm d}}\tilde\Xi$. It lies in $\cK$ and has the pattern
weights of $\Xi$, so $\|H_\od\varphi_s\|=\|Y\nP$ by
Lemma~\ref{gen:lem:operators}. By Duhamel's formula,
$\|\psi_s-\varphi_s\|\le s\|Y\nP$, and
$\|H_\od\Pi\|_{\op}\le\tfrac12\|Y_\od\|_F$ by \eqref{gen:eq:variance}. Integrate
over $[0,1]$:
\[
 \begin{aligned}
 \|H_\od\Pi\psi_s\|
 &\le\|H_\od\varphi_s\|+
       \|H_\od\Pi\|_{\op}\|\psi_s-\varphi_s\|\\
 &\le\|Y\nP\left(1+\tfrac s2\|Y_\od\|_F\right),\\
 \sqrt\ell&\le\int_0^1\|H_\od\Pi\psi_s\|\,ds
 \le\|Y\nP\left(1+\tfrac14\|Y_\od\|_F\right).
 \end{aligned}
\]
\end{proof}

\subsection{Pairing records and corrections}
\label{gen:app:pairing}
Throughout, the $b_j$ are the particle--hole modes of a pattern $x$
(Section~\ref{gen:app:gaussian}), and $n_j=b_j^\dagger b_j$.

\begin{alem}[Probes]
\label{gen:lem:probe}
A \emph{probe} on $n'$ active modes is an antisymmetric matrix $G$ that is
unitary on the active modes, vanishes elsewhere, and has $|G_{jk}|^2\in\{0,g^2\}$.
A pair $\{j,k\}$ is \emph{usable} if $G_{jk}\ne0$. Select the $n'$
active modes uniformly, with $n'$ even and $6\le n'\le n$.
\begin{enumerate}[label=(\alph*)]
\item \emph{Fourier probe.} For even $n'$ let
\[
\begin{gathered}
 W_{uv}=(2/n')^{1/2}e^{4\pi iuv/n'}\quad(0\le u,v<n'/2)\\
 \text{and}\quad
 G^0=\bigl(\begin{smallmatrix}0&W\\-W^{\T}&0\end{smallmatrix}\bigr),
\end{gathered}
\]
applied
with a uniformly random numbering of the active modes. Then $g^2=2/n'$.
\item \emph{Flat probe.} Let $n'=q+1$ with an odd prime power $q$, let
$q-1=2^{\mathsf e}r$ with $r$ odd, and let $\eta$ be a multiplicative
character of the field with $q$ elements of order $2^{\mathsf e}$. Then
\[
 G=\frac1{\sqrt q}\begin{pmatrix}0&\mathbf 1^{\T}\\-\mathbf 1&(\eta(y-x))_{x,y}\end{pmatrix}
\]
is a probe with $g^2=1/(n'-1)$: every pair is usable.
\item A given pair of modes is usable with probability
$\piu=n'/(g^2n(n-1))$. So $\piu=1$ for a flat probe with $n'=n$,
$\piu=(n-2)/n$ for a flat probe with $n'=n-1$ and odd $n$, and $\piu\ge3/7$ for the full Fourier
probe and $n\ge6$.
\end{enumerate}
\end{alem}
\begin{proof}
\noindent\emph{(a).}
$W$ is the unitary Fourier matrix of order $n'/2$, and it is symmetric.

\noindent\emph{(b).}
$\eta(-1)=-1$, because $-1=\mathsf z^{(q-1)/2}$ for a generator $\mathsf z$ and
$\eta(\mathsf z)$ is a primitive root of unity of order $2^{\mathsf e}$. So
$\eta(x-y)=-\eta(y-x)$ and $G$ is antisymmetric, with $\eta(0)=0$ on the
diagonal. The rows are orthogonal: $\sum_y\overline{\eta(y-x)}=0$, and for
$x\ne x'$,
\[
 1+\sum_y\eta(y-x)\overline{\eta(y-x')}=1-1=0:
\]
with
$t=(y-x)/(y-x')$, which takes every value except $1$, the sum is
$\sum_{t\ne1}\eta(t)=-\eta(1)=-1$. Every row has squared norm $q$ before
the normalization.

\noindent\emph{(c).}
For a fixed numbering the number of usable ordered pairs is
$\|G\|_F^2/g^2=n'/g^2$, and all $n(n-1)$ ordered pairs are equally likely.
For the Fourier probe $\piu=(n')^2/[2n(n-1)]$.
\end{proof}

\noindent\emph{The record.}
Fix a probe and an amplitude $0<s\le1/\sqrt2$. Put $c=(1-s^2)^{1/2}$,
$\theta=\arcsin s$, and call $\tau=s^2n'$ the \emph{strength}. Draw
$\chi_1,\dots,\chi_n$ independently and uniformly from
$\{2\pi r/5:r=0,1,2,3,4\}$. For every
record the active subset, its numbering and these
phases are drawn afresh and independently; for a flat probe with $n'=n$
the numbering is fixed. With
\[
\begin{gathered}
 \mathsf K=\tfrac12\sum_{l,m}\bigl(G_{lm}b_l^\dagger b_m^\dagger
     -\overline{G_{lm}}\,b_mb_l\bigr),\qquad
 F_\chi=P_\chi e^{\theta\mathsf K}P_\chi^\dagger,\\
 P_\chi=e^{i\sum_j\chi_jn_j},\qquad \varkappa=1+s^2-s^2g^2,
\end{gathered}
\]
apply the even Gaussian circuit $F_\chi$, read the occupations $y$ of the
$b_j$, and put
\begin{equation}
 \begin{gathered}
 \hat T_{jk}=\frac{e^{-i(\chi_j+\chi_k)}\,y_jy_k}{\piu\,c\,s\,G_{jk}\,\varkappa}
 \quad\text{if }G_{jk}\ne0,\\
 \hat T_{jk}=0\quad\text{otherwise.}
\end{gathered}
 \label{gen:eq:record}
\end{equation}
$\hat T$ is antisymmetric. With $\hat\Delta=-\overline{\hat T}$ and $L$ of
\eqref{gen:eq:LKLD}, the \emph{pairing record} is
\begin{equation}
 \hat A=-\Theta_x\,L_{\hat\Delta}\,J_n\,\Theta_x .
 \label{gen:eq:A-record}
\end{equation}
It is real and skew, $\hat A_{[jj]}=0$, and
$\|\hat A\|_F^2=8\|\hat T\|_F^2$.

\begin{alem}[Pairing record]
\label{gen:lem:record}
Let $\psi$ be a pure Gaussian state with $\gamma$, $T$, $D$ as in
\eqref{gen:eq:DKT}. With the expectation over the probe and the outcome,
\begin{equation}
\begin{gathered}
 \|\E\hat T-T\|_F\le\frac{3\,\tau\,\gamma^3}{4\sqrt8\,\piu\,n(n-1)},\\
 \E\|\hat T\|_F^2\le\frac{n(n-1)}{\piu}\cdot
 \frac{\bigl(1+\frac{\gamma^2}{8\tau}\bigr)\bigl(1+\tau+\frac{\gamma^2}8\bigr)}{c^2}
 +\frac{\gamma^2}{2\piu^2} .
\end{gathered}
 \label{gen:eq:record-bounds}
\end{equation}
Let moreover $x$ be uniform on $\{0,1\}^n$ and independent of the probe,
let the state be $\psi_x=U_S\ket x$ for a fixed $S\in\SO(2n)$, let the record
be taken in the particle--hole modes of $x$, and let $\gamma(x)\le\gamma$ for
every $x$. Then, with the expectation also over $x$,
\begin{equation}
\begin{gathered}
 \E\|\hat A\nP^2\le\tilde V:=\frac{n(n-1)}{\piu}\,\mathsf v_\Xi
 +\frac{\gamma^2}{\piu^2},\\
 \mathsf v_\Xi=1+\frac{2(\tau+\gamma^2/8)^2+\tfrac14\gamma^2(1+\gamma/2)^2}{c^2\,\tau} .
\end{gathered}
 \label{gen:eq:Vphi}
\end{equation}
For $\gamma\to0$ and then $\tau\to0$, the two bounds, multiplied by
$\piu$, tend to $n(n-1)$.
\end{alem}
\begin{proof}
\noindent\emph{Probe identity.}
Let $g_j^\dagger=\sum_mG_{jm}b_m^\dagger$ and $F=e^{\theta\mathsf K}$. From
$G^{\T}=-G$ and $G\overline G=-I$ on the active modes,
$[\mathsf K,b_j]=-g_j^\dagger$ and $[\mathsf K,g_j^\dagger]=b_j$. So, for an active mode $j$,
$F^\dagger b_jF=c\,b_j+s\,g_j^\dagger$, and with
$\mathsf N_j=c^2n_j+s^2(1-\sum_l|G_{jl}|^2n_l)$,
\[
\begin{aligned}
 F^\dagger n_jF=\mathsf N_j
 &+cs\sum_{m}G_{jm}b_j^\dagger b_m^\dagger
 +cs\sum_{l}\overline{G_{jl}}\,b_lb_j\\
 &+s^2\sum_{l\ne m}\overline{G_{jl}}G_{jm}\,b_lb_m^\dagger .
\end{aligned}
\]
Give $b_m^\dagger$ the charge $e_m$. Conjugation by $P_\chi$ multiplies a
term of charge $\mathsf q$ by $e^{i\mathsf q\cdot\chi}$. The charges in
$F^\dagger n_jn_kF$ have entries between $-2$ and $2$, so the average of
$e^{i(\mathsf q-e_j-e_k)\cdot\chi}$ over fifth roots is $1$ for
$\mathsf q=e_j+e_k$ and $0$ otherwise. The charge $e_j+e_k$ arises from a
creating pair term times a neutral term. The other products with this charge
carry the coefficients $G_{jj}=G_{kk}=0$. Hence, for a usable pair,
\[
\begin{gathered}
\begin{aligned}
 &\E_\chi\bigl[e^{-i(\chi_j+\chi_k)}F_\chi^\dagger n_jn_kF_\chi\bigr]=cs\,G_{jk}\bigl(b_j^\dagger b_k^\dagger\mathsf N_k
        +\mathsf N_jb_j^\dagger b_k^\dagger\bigr)\\
 &\quad =cs\,G_{jk}\,b_j^\dagger b_k^\dagger
   \Bigl[\varkappa-s^2\sum_{l\ne j,k}\omega_l\,n_l\Bigr],
\end{aligned}\\[3pt]
 \omega_l=|G_{jl}|^2+|G_{kl}|^2 .
\end{gathered}
\]

\noindent\emph{Mean.}
Take the expectation in $\psi$, divide as in \eqref{gen:eq:record} and multiply by
the probability $\piu$ of being usable: $\E\hat T_{jk}=T_{jk}-\beta_{jk}$,
where $\beta_{jk}$ is an average of
$(s^2/\varkappa)\sum_l\omega_l\langle b_j^\dagger b_k^\dagger n_l\rangle$. By
Wick's theorem
\[
 \langle b_j^\dagger b_k^\dagger n_l\rangle
=T_{jk}D_{ll}-T_{jl}D_{kl}+D_{jl}T_{kl}.
\]
With $\omega_l\le2g^2$,
$\varkappa\ge1$ and the Cauchy--Schwarz inequality,
\[
 \|\beta\|_F\le2s^2g^2\|T\|_F(\Tr D+2\|D\|_F).
\]
Now
$s^2g^2=\tau/(\piu n(n-1))$ by Lemma~\ref{gen:lem:probe}(c),
$\|T\|_F^2\le\Tr D=\gamma^2/8$ by Lemma~\ref{gen:lem:gaussian}(b), and
$\|D\|_F\le\Tr D$ because $D$ is positive semidefinite.

\noindent\emph{Second moment.}
$\|\hat T\|_F^2\le2\sum_{j<k}y_jy_k/(\piu^2c^2s^2g^2)$. Condition on the active probe and its numbering. For a fixed phase
vector let $D^F,\Delta^F$ belong to $F_\chi\psi$ and $\nu=\Tr D^F$. By Wick's
theorem and Lemma~\ref{gen:lem:gaussian}(b),
\[
 \E\sum_{j<k}y_jy_k\le\tfrac12(\nu^2+\|\Delta^F\|_F^2)\le\tfrac12(\nu^2+\nu).
\]
Let $\mathcal A$ index the active modes. Summing $F^\dagger n_jF$ leaves
each inactive $n_j$ unchanged and cancels the terms with $l\ne m$, since
$(G^\dagger G)_{lm}=\mathbf 1_{\{l\in\mathcal A\}}\delta_{lm}$. Hence
\[
 \begin{gathered}
 \nu=\nu_0+2cs\,\mathrm{Re}\,Z,\qquad
 \nu_0\le\tau+\frac{\gamma^2}8,\\
 Z=2\sum_{j<m}G_{jm}T_{jm}e^{i(\chi_j+\chi_m)} .
\end{gathered}
\]
Indeed, with $\tau=s^2n'$,
\[
 \nu_0=\tau+(1-2s^2)\sum_{j\in\mathcal A}D_{jj}
             +\sum_{j\notin\mathcal A}D_{jj}\le\tau+\Tr D.
\]
The inequality holds because $D_{jj}\ge0$ and $1-2s^2\le1$.
Over the fifth roots, $\E Z=\E Z^2=0$ and
\[
 \E|Z|^2=4\sum_{j<m}|G_{jm}T_{jm}|^2\le g^2\gamma^2/4.
\]
So
$\E\nu^2\le\nu_0^2+c^2s^2g^2\gamma^2/2$ and
\[
 \begin{gathered}
 \E\|\hat T\|_F^2\le\frac{\nu_0^2+\nu_0}{\piu^2c^2s^2g^2}+\frac{\gamma^2}{2\piu^2},\\
 \frac{\nu_0^2+\nu_0}{\tau}\le\Bigl(1+\frac{\gamma^2}{8\tau}\Bigr)
 \Bigl(1+\tau+\frac{\gamma^2}8\Bigr).
\end{gathered}
\]
With $1/(\piu^2s^2g^2)=n(n-1)/(\piu\tau)$ this is the second line of
\eqref{gen:eq:record-bounds}, uniformly in the selected modes and their numbering.

\noindent\emph{The seminorm.}
The block $[jk]$ of $L_{\hat\Delta}J_n$ is of the type
$\bigl(\begin{smallmatrix}p&q\\q&-p\end{smallmatrix}\bigr)$ with squared norm
$8|\hat T_{jk}|^2$, and $\Theta_x$ multiplies its second row by $\sigma_j$ and
its second column by $\sigma_k$. So $\hat A_{[jk]}$ lies in the range of
$\mathsf P^+$ if $x_j=x_k$ and of $\mathsf P^-$ otherwise, and
\[
 \begin{aligned}
 \|\hat A\nP^2&=4\sum_{j<k}\pi_{jk}^{\pm(x)}|\hat T_{jk}|^2\\
 &=\frac{4}{\piu^2c^2s^2g^2\varkappa^2}\sum_{j<k\ \rm usable}\pi_{jk}^{\pm(x)}y_jy_k,
\end{aligned}
\]
where $\pm(x)$ is the sign of $\sigma_j\sigma_k$. With $V=sG^\chi$,
$G^\chi_{jm}=G_{jm}e^{i(\chi_j+\chi_m)}$, the probe maps $b_j$ to
$cb_j+(Vb^\dagger)_j$ for the active modes $j$ and acts as the identity
on the remaining modes. For active indices the second moments become
\[
 D'=c^2D+c\,TV^{\T}+c\,\overline V\Delta+\overline V(I-D^{\T})V^{\T}
\]
and
$\Delta'=-cV+\mathsf M$ with the antisymmetric
\[
 \mathsf M=c^2\Delta+c(VD+D^{\T}V)+VTV^{\T}.
\]
Wick's theorem gives
$\E y_jy_k\le D'_{jj}D'_{kk}+|\Delta'_{jk}|^2$. Over the fifth roots only the
terms of $\mathsf M$ with the phase of $V_{jk}$ survive in
$\E_\chi[\overline{V_{jk}}\mathsf M_{jk}]$; they are
$c\,|V_{jk}|^2(D_{jj}+D_{kk})\ge0$. So
\[
 \E_\chi|\Delta'_{jk}|^2\le c^2s^2|G_{jk}|^2+\E_\chi|\mathsf M_{jk}|^2.
\]
Now
collect three parts. First, each sign $\pm(x)$ has probability $\tfrac12$ and
$\pi^++\pi^-=1$, and a pair is usable with probability $\piu$, independently
of $x$. So
\[
 \E\sum_{\rm usable}\pi_{jk}^{\pm(x)}c^2s^2g^2
=\tfrac14n(n-1)\piu\,c^2s^2g^2.
\]
Second,
$\sum_{j<k}\pi_{jk}|\mathsf M_{jk}|^2\le\tfrac12\|\mathsf M\|_F^2$ and
\[
 \|\mathsf M\|_F\le\|\Delta\|_F+2s\|D\|_F\le(\gamma/\sqrt8)(1+\gamma/2);
\]
the first step uses $\|V\|_{\op}=s$, $\|T\|_F=\|\Delta\|_F$ and $c^2+s^2=1$,
the second $2s\le\sqrt2$ and $\|D\|_F\le\Tr D=\gamma^2/8$.
Third, usable pairs are active, so
\[
 \sum_{j<k\ \rm usable}\pi_{jk}D'_{jj}D'_{kk}
\le\tfrac12\bigl(\sum_{\rm active}D'_{jj}\bigr)^2\le\tfrac12\nu^2,
\]
with
$\E\nu^2\le\nu_0^2+c^2s^2g^2\gamma^2/2$ as above. Multiply by
$4/(\piu^2c^2s^2g^2)$ and use $1/(\piu^2s^2g^2)=n(n-1)/(\piu\tau)$:
\[
 \begin{aligned}
 \E\|\hat A\nP^2\le\frac{n(n-1)}{\piu}\Bigl[1
 &+\frac{\gamma^2(1+\gamma/2)^2}{4c^2\tau}\\
 &+\frac{2(\tau+\gamma^2/8)^2}{c^2\tau}\Bigr]
 +\frac{\gamma^2}{\piu^2}.
\end{aligned}
\]
\end{proof}

\noindent\emph{Two functions of the probe.}
With $n'$ and $\piu$ of the probe put
\begin{equation}
\begin{aligned}
 \mathsf V(\gamma,\tau)
 &=\frac{8n(n-1)}{\piu}
   \frac{(1+\gamma^2/(8\tau))(1+\tau+\gamma^2/8)}{1-\tau/n'}
   \\
 &\quad+\frac{4\gamma^2}{\piu^2},\\
 c_3(\tau)&=\frac12+\frac{6\tau}{\piu n(n-1)} .
\end{aligned}
\label{gen:eq:V}
\end{equation}

\begin{alem}[Pairing correction]
\label{gen:lem:pcorr}
Let
\[
 d=\dist(\hat R,R)\le1/\sqrt2,
\]
with the gauge $S$, $A$ of
\eqref{gen:eq:gauge}. Let $A_{\rm av}$ be the mean of $N$ independent pairing records
\eqref{gen:eq:A-record} of the states $U_{\hat R}^\dagger U_R\ket x$, with
independent uniform patterns, independent probe randomness and the strength
$\tau\le\tfrac12$. Let
\[
 d'=\|\polar(I+A_{\rm av})-S\|_F.
\]
Then, with $c_3$, $\mathsf V$ of
\eqref{gen:eq:V} and $\tilde V$ of \eqref{gen:eq:Vphi} at $\gamma=2d$,
\[
\begin{gathered}
 \dist\bigl(\hat R\,\polar(I+A_{\rm av}),R\bigr)\le d',\\
 d'\le\|A_{\rm av}-\E A_{\rm av}\|_F+c_3\,d^3,\\
 \|\E A_{\rm av}-A\|_F\le\bigl(c_3-\tfrac14\bigr)d^3,\\
 \E\|A_{\rm av}-\E A_{\rm av}\|_F^2\le\frac{\mathsf V(2d,\tau)}N,\\
 \E\|A_{\rm av}-\E A_{\rm av}\nP^2\le\frac{\tilde V}N .
\end{gathered}
\]
Both variance bounds increase with $\gamma$, so $2d$ may be replaced by any
larger value.
\end{alem}
\begin{proof}
\noindent\emph{One record.}
The rotated state has covariance $S\Gamma_xS^{\T}$, because
$\hat R^{\T}R=ST_*^{\T}$ and $T_*$ commutes with $\Gamma_x$. In the
particle--hole modes of $x$, $\Gamma^b=\Theta_xS\Gamma_xS^{\T}\Theta_x$ and
$\gamma(x)=\|[E,\Gamma_x]\|_F\le2d$. By Lemma~\ref{gen:lem:gaussian}(a),
$L_\Delta=\tfrac12(\Gamma^b+J_n\Gamma^bJ_n)$, so
\[
 -\Theta_xL_\Delta J_n\Theta_x=\tfrac12[\Gamma_x,S\Gamma_xS^{\T}].
\]
With
$\|L_\Delta\|_F=\sqrt8\|\Delta\|_F$ and \eqref{gen:eq:record-bounds},
\[
 \begin{gathered}
 \bigl\|\E[\hat A\mid x]-\tfrac12[\Gamma_x,S\Gamma_xS^{\T}]\bigr\|_F
 \le\frac{3\tau(2d)^3}{4\,\piu\,n(n-1)},\\
 \E\|\hat A\|_F^2\le\mathsf V(2d,\tau).
\end{gathered}
\]

\noindent\emph{Average and rounding.}
By \eqref{gen:eq:cancel} and $\|X\|_{\op}\le d^2/4$ (proof of
Lemma~\ref{gen:lem:mcorr}),
\[
 \|\E\hat A-A\|_F\le d^3/4+6\tau d^3/(\piu n(n-1)).
\]
Since $\|S-I\|_{\op}\le d/\sqrt2\le\tfrac12$, Lemma~\ref{gen:lem:round}(b),(c)
give
\[
\begin{aligned}
 d'&\le\|A_{\rm av}-A\|_F+\tfrac12\|A\|_{\op}^2\|A\|_F\\
 &\le\|A_{\rm av}-A\|_F+d^3/4.
\end{aligned}
\]
The corrected frame has the distance at most
$d'$, because $\hat R\,\polar(I+A_{\rm av})$ and $R$ differ by
$\polar(I+A_{\rm av})^{\T}S$ in the gauge $T_*$. The last claim is
\eqref{gen:eq:Vphi}; the centered mean of $N$ independent records has, in every
Hilbert seminorm, a second moment at most that of one record over $N$.
\end{proof}

\subsection{Input structure}
\label{gen:app:structure}
\begin{alem}[Structure from occupation records]
\label{gen:lem:structure}
Let $\Lambda=\log(64n^2/\delta)$ and let $N=N_{\rm s}\ge16\Lambda n^2$ records
$y^{(1)},\dots,y^{(N)}\in\{0,1\}^n$ of $\ket\Xi$ be given, $m_j$ the number
of records with mode $j$ occupied. Let $S_j\subseteq[N]$ be those records and set
\[
 k_0=\max\{2,\lceil\log(16n^2/\delta)/\log6\rceil\}.
\]
\begin{enumerate}[label=\arabic*.]
\item Counts $m_j>N-k_0$ form the occupied core $\hat C_1$; counts
$m_j<k_0$ form the vacant core $\hat C_0$. Among the remaining modes,
collapse identical supports $S_j$ into one atom $a$ with support $S_a$
and count $m_a=|S_a|$.
\item Call an atom frequent if $m_a\ge N/(2n)$. On frequent atoms
relate each atom to itself and distinct atoms with disjoint supports
to each other. Require an equivalence relation. Form its classes and
store, for each class, the union of its frequent supports.
\item Assign each other atom to the unique class whose stored union it
misses. Keep these unions fixed throughout. In each completed class
$\cA_b$ choose a largest-count atom $G_b$, breaking ties by lexicographic mode order. List the atoms of each
class in that order.
\end{enumerate}
An inconsistent relation, a missing or ambiguous assignment, an atom with
fewer than three modes, or a class containing different branch parities
makes the records malformed. Empty retained families are allowed.
A retained block may consist of $G_b$ alone. It then has no partner in the
readout, and its recovered factor is $\ket{F_{G_b}}$.
There is an event $\cS$ with $\Prob(\cS)\ge1-\delta/4$ on which
the following holds. Every retained atom is a branch of $\ket\Xi$. The atoms
of a retained block are branches of one block of $\ket\Xi$, and different
retained blocks lie in different blocks of $\ket\Xi$. The modes of the
branches of that block that are not retained lie in $\hat C_0$. Every other
block of
$\ket\Xi$ is \emph{frozen}: $\hat C_1$ contains its heaviest branch and
$\hat C_0$ its other modes. Further $C_1\subseteq\hat C_1$ and
$C_0\subseteq\hat C_0$. For a retained block $b$ let $\tau_b$ be the total
weight of its branches that are not retained, and $w_{b,\max}$ the largest weight of its block. For
a frozen block with heaviest branch $G$ let $\tau_b=1-w_G$. Then
\[
\begin{gathered}
 \sum_b\tau_b\le\frac{8\Lambda n}{3N},\qquad
 w_{G_b}\ge\frac1{4L_b},\qquad w_a\le4\,w_{G_b},\\
 w_{G_b}\ge w_{b,\max}-\mathsf r(N),\qquad
 \mathsf r(N)=2\sqrt{\frac{2\Lambda}N}+\frac{4\Lambda}{3N}
\end{gathered}
\]
for every branch $a$ of the block of $\ket\Xi$ that contains the retained
block $b$. The sum runs over all blocks of $\ket\Xi$, and $L_b$ counts
the branches of the block of $\ket\Xi$. Whenever the records are not
malformed, the retained atoms number at most $n/3$. On $\cS$ the numbers of \eqref{gen:eq:structure-numbers} satisfy $B_{\rm r}\le n/6$, $\kappa_w\le4n/3$, $S_{\rm br}+\kappa_w\le5n/3-2$, $\hat L\le n/3$ and $w_{\min}\ge3/(4n)$.
\end{alem}
\begin{proof}
\noindent\emph{Conditional row supports.}
Each record selects one branch of each true block, independently across
blocks, with probabilities $w_a$. Modes of one branch have identical
supports. Distinct branches of a block have disjoint supports.
True core occupations are constant. Given its size, a branch support
is uniform among subsets of that size. More generally, conditioning on
all counts of one block preserves row symmetry. The union of any
branches selected by those counts is therefore uniform given the counts.

For supports $S,T$ from different blocks of sizes $m,l$,
\[
 \Prob(S\cap T=\varnothing\mid m,l)
 =\frac{\binom{N-m}{l}}{\binom Nl}\le(1-m/N)^l\le e^{-ml/N}.
\]
The numerator is zero if $l>N-m$. Here $k_0\le\Lambda$,
$N/(2n)\ge k_0$ and $N\ge6k_0$. If both counts are eligible,
equality of supports requires $m=l$, and then has probability at most
\[
 \binom Nm^{-1}\le\binom N{k_0}^{-1}
 \le(k_0/N)^{k_0}\le6^{-k_0}.
\]
If both supports are frequent, disjointness has probability at most
$e^{-N/(4n^2)}\le e^{-\Lambda}$.

For a true block $b$ let $H_b$ be the union of the row supports of branches with
count at least $N/(2n)$, before core filtering. Its complement satisfies
\[
 N-|H_b|\le L_b N/(2n)\le N/6
\]
since each branch uses at least three distinct modes. Condition on the
counts in $b$ and on an independent support $S$ from another block
with $|S|=m\ge k_0$. By row symmetry,
\[
 \Prob(S\cap H_b=\varnothing\mid S,\text{ counts in }b)
 =\frac{\binom{N-|H_b|}{m}}{\binom Nm}\le6^{-m}.
\]

\noindent\emph{The common success event.}
For each branch the binomial Chernoff and Bernstein bounds give,
with joint failure probability at most $5e^{-\Lambda}$,
\begin{equation}
\begin{gathered}
 Nw_a\ge8\Lambda\ \Rightarrow\ m_a\ge\tfrac12Nw_a,\\
 N(1-w_a)\ge8\Lambda\ \Rightarrow\ N-m_a\ge\tfrac12N(1-w_a),\\
 m_a\le\max\{2Nw_a,\,6\Lambda\},\\
 \Bigl|\frac{m_a}N-w_a\Bigr|\le\sqrt{\frac{2w_a\Lambda}N}+\frac{2\Lambda}{3N}.
\end{gathered}
\label{gen:eq:binomial}
\end{equation}
For the third bound, dominate the count by a binomial of mean
$\max\{Nw_a,3\Lambda\}$ and use its doubled-mean tail.
There are at most $n/3$ true branches. Excluding all cross-block
eligible coincidences, cross-block frequent disjoint pairs, misses of
another block's $H_b$ by an eligible support, and violations
of~\eqref{gen:eq:binomial}, costs at most
\[
 2n^2\,6^{-k_0}+(n^2+5n/3)e^{-\Lambda}<\delta/4.
\]
Denote the complementary event by $\cS$.

\noindent\emph{Identification on $\cS$.}
If a branch has count $>N-k_0$, every other branch of that block has
count $<k_0$. Its complement bound in~\eqref{gen:eq:binomial} gives
lost weight $<8\Lambda/N<1/2$, so this branch is the unique heaviest.
In every unfrozen block the largest count obeys
\[
 M_b\ge N/L_b\ge3N/n>N/(2n).
\]
It is eligible, and its block's frequent union is exactly $H_b$.
Frequent supports from distinct blocks intersect; those of distinct
branches of one block are disjoint. Thus the frequent classes are
the unfrozen blocks. Every eligible light support misses its own class
and intersects every other class. The assignment is unique, including
when a block retains just one branch. True core modes go to their
respective cores, and eligible supports from different blocks never merge.

\noindent\emph{Missing mass and reference weights.}
An omitted branch of an unfrozen block has $m_a<k_0$, hence
$Nw_a<8\Lambda$. Charge once per omitted branch and once per frozen
block to obtain $\sum_b\tau_b\le 8\Lambda n/(3N)$.
In a retained block $m_{G_b}=M_b>6\Lambda$, whence
\[
 w_{G_b}\ge M_b/(2N)\ge1/(2L_b).
\]
If $Nw_a\ge8\Lambda$, then
$w_a\le2m_a/N\le2M_b/N\le4w_{G_b}$.
Otherwise $w_a<8\Lambda/N\le w_{G_b}$.
Applying the last bound of~\eqref{gen:eq:binomial} to $G_b$ and to a
largest-weight branch gives $w_{b,\max}-w_{G_b}\le\mathsf r(N)$.

Every true block uses at least six modes. Thus $B_{\rm r}\le n/6$,
$\hat L\le n/3$, and $w_{\min}\ge3/(4n)$ under the empty-family
convention of~\eqref{gen:eq:structure-numbers}. Using the weaker
$1/w_{G_b}\le4L_b$ retained in the statement gives
\[
 \begin{aligned}
 \kappa_w&\le4n/3,\\
 S_{\rm br}+\kappa_w&\le\sum_{b\ {\rm retained}}(5L_b-2)\le5n/3-2 .
 \end{aligned}
\]
For the empty family the left side is zero. On every well-formed
record, the retained atoms are disjoint sets of at least three modes,
so their number is at most $\lfloor n/3\rfloor$.
\end{proof}

\subsection{Readout}
\label{gen:app:readout}
\begin{alem}[Local coherence records]
\label{gen:lem:readout}
Assume the event $\cS$. Let $b$ be a retained block with the chosen atom
$G=G_b$, weight $w_G$, retained weight $W_b=1-\tau_b$, and let
$\tilde c_l$ be the amplitudes of the canonical input $\tilde\Xi$ of
Lemma~\ref{gen:lem:frame-state}, with $\varepsilon,\ell,\sigma$ as there. For a
partner $a\in\cA_b\setminus\{G\}$ the settings below measure two observables
$\cX_a^{(1)},\cX_a^{(2)}$ with values in $\{-1,0,1\}$. Let $\hat w_{G_b}$ be the
fraction of the $N_{\rm w}$ records of the reserved batch in step~1 of
$\mathsf{Fit}$ in Algorithm~\ref{alg:general-reference} in which
the modes of $G_b$ are occupied. Take $N_\phi$ copies of $\ket\Psi$ for each
of the two settings, apply $U_{\hat R}^\dagger$ and the setting, and read the
occupations $y_s$, $y_s'$. Let
\[
 \begin{gathered}
 \hat\kappa_a=\frac1{2N_\phi}\sum_{s=1}^{N_\phi}
   \bigl(\cX_a^{(1)}(y_s)+\mathrm{i}\,\cX_a^{(2)}(y'_s)\bigr),\\
 \hat v_{bG}=\hat w_{G_b},\qquad \hat v_{ba}=\hat\kappa_a\quad(a\ne G),\end{gathered}
\]
and $\ket{\hat\omega_b}=\sum_{a\in\cA_b}\hat v_{ba}\ket{F_a}/\|\hat v_b\|$
($\ket{F_G}$ if $\hat v_b=0$). Define
$\tilde v_b=\overline{\tilde c_G}(\tilde c_G;\tilde c_a)_a$,
$e_b=\hat w_{G_b}-w_G$ and $s_a=\hat\kappa_a-\E\hat\kappa_a$.
Conditional on the structure and frame, there are a scalar $\mu_b$ and a
deterministic bias $\mathfrak B_b$ such that
\[
\begin{gathered}
 \hat v_b-\mu_b\tilde v_b=(e_b;s_a)_a+\mathfrak B_b,\\
 \|\tilde v_b\|^2=w_GW_b,\qquad
 |\mu_b|\ge\lambda_b:=1-\varepsilon-\sigma/\sqrt{w_G},\\
 \frac{\|\mathfrak B_b\|^2}{w_G}
 \le2w_G\varepsilon^2+3\|\eta^{(b)}\|^2
                +\frac{3\sigma^4/4+3\ell^2}{w_G}.
\end{gathered}
\]
Here $\sum_b\|\eta^{(b)}\|^2\le\sigma^2$, the noise is centered and
$\E e_b^2=w_G(1-w_G)/N_{\rm w}$. For every $\eta'>0$,
\begin{equation}
\begin{aligned}
 \mathsf S_b&:=\frac{\sum_a\E|s_a|^2}{w_G}\\
 &\le\frac1{2N_\phi}\Bigl[(1+\eta')(|\cA_b|-2+w_G^{-1})\\
 &\hspace{1.5em}+(|\cA_b|-1)(\ell+(1+1/\eta')\sigma^2)w_G^{-1}\Bigr].
\end{aligned}
\label{gen:eq:readout-var}
\end{equation}
For $\vartheta\le2$, the total squared bias satisfies
\[
 \sum_b\|\mathfrak B_b\|^2/w_{G_b}\le c_B\vartheta^4/16,
\]
with
$c_B$ of~\eqref{gen:eq:structure-numbers}.
\end{alem}
\begin{proof}
\noindent\emph{Settings.}
Order the modes of $F_G$ and of $F_a$ increasingly. Match the first
$\min\{|F_G|,|F_a|\}$ modes of the two sets in pairs $(G_u,a_u)$, and pair
the remaining modes of the larger set consecutively, $(e,e')$ with $e<e'$.
Their number is even because the branches have equal parity. Put
\[
\begin{gathered}
 \text{matched pair:}\quad x_u=a_{G_u}^\dagger a_{a_u}+\mathrm{h.c.},\\
 o_u=n_{a_u}-n_{G_u};\\
 \text{remaining pair:}\quad x_u=a_e^\dagger a_{e'}^\dagger+\mathrm{h.c.},\\
 o_u=n_en_{e'}-(1-n_e)(1-n_{e'}).
\end{gathered}
\]
The Gaussian rotations
\[
\begin{gathered}
 v_u=\exp(-\tfrac\pi4(a_{G_u}^\dagger a_{a_u}-\mathrm{h.c.}))\\
 \text{and}\quad v_u=\exp(\tfrac\pi4(a_e^\dagger a_{e'}^\dagger-\mathrm{h.c.}))
\end{gathered}
\]
give
$v_u^\dagger o_uv_u=x_u$. Let $P^\varnothing$ be the projector on the states
in which the other retained branches of $b$ are empty, and
$\varepsilon_{Ga}=\pm1$ the sign with
$\varepsilon_{Ga}\prod_ux_u\ket{F_a}=\ket{F_G}$. In the ordered block
product of Section~\ref{gen:sec:setting}, this sign is independent of the
spectator occupations. The even operator $\prod_ux_u$ commutes with every
product $X$ of spectator creation operators, and
\[
\begin{gathered}
 a_{F_a}^\dagger X=(-1)^{|X||F_a|}Xa_{F_a}^\dagger\\
 \text{and}\quad a_{F_G}^\dagger X=(-1)^{|X||F_G|}Xa_{F_G}^\dagger,
\end{gathered}
\]
where $|F_a|$ and $|F_G|$
have equal parity. Applying $\prod_uv_u$ and
recording $\varepsilon_{Ga}P^\varnothing(y)\prod_uo_u(y)$ measures
\[
 \begin{gathered}
 \cX_a^{(1)}=\varepsilon_{Ga}\,P^\varnothing\prod_ux_u,\\
 \cX_a^{(2)}=e^{i\pi n_{a_1}/2}\,\cX_a^{(1)}\,e^{-i\pi n_{a_1}/2}.
\end{gathered}
\]
For the second quadrature the mode phase $e^{-i\pi n_{a_1}/2}$ is applied
first. The observables of different blocks are even and act on disjoint
modes. One copy serves one partner in every block: $2(\hat L-1)$ settings.

\noindent\emph{The transition operator.}
On an occupation pattern with $\cX_a^{(1)}\ne0$, $\cX_a^{(1)}$ flips the
occupation of $a_1$. So $\cX_a^{(2)}=\mp i\cX_a^{(1)}$ on patterns with $a_1$
occupied or empty, and
$K_a:=\tfrac12(\cX_a^{(1)}+i\cX_a^{(2)})=\cX_a^{(1)}n_{a_1}$. On branches,
$K_a\ket{F_a}=\ket{F_G}$ and $K_a\ket{F_l}=0$ for $l\ne a$. $K_a$ and
$K_a^\dagger$ map patterns of $\Xi$ to patterns of $\Xi$ or to $0$, so
$K_a$ commutes with $\Pi$. Let $Q_a=(\cX_a^{(1)})^2$,
$Q_a^{\rm dom}=Q_an_{a_1}$ and $Q_a^{\rm ran}=Q_a(1-n_{a_1})$. For $a\ne a'$,
$Q_a^{\rm dom}Q_{a'}^{\rm dom}=0$, because $Q_{a'}$ requires the branch $a$ to
be empty.

\noindent\emph{Decomposition.}
Write $\Xi'=\mathsf a\tilde\Xi+\sigma_{\rm v}+\mathsf u$ with
$\mathsf a=\braket{\tilde\Xi}{\Xi'}$, $\sigma_{\rm v}=(\Pi-P_{\rm can})\Xi'$ and
$\mathsf u=(I-\Pi)\Xi'$. Let $\psi_l$ be the unit vector with
$P_{b=l}\tilde\Xi=\tilde c_l\psi_l$, where $P_{b=l}$ projects on the states
in which block $b$ shows the branch $l$. Then $K_a\psi_a=\psi_G$. With
$\kappa_a=\langle\Xi'|K_a|\Xi'\rangle$,
$\tilde\kappa_a=\overline{\tilde c_G}\tilde c_a$,
$\eta_l=\braket{\psi_l}{\sigma_{\rm v}}$,
$r_a=\langle\sigma_{\rm v}|K_a|\sigma_{\rm v}\rangle$ and
$\ell_a=\langle\mathsf u|K_a|\mathsf u\rangle$,
\[
 \kappa_a=\mu\,\tilde\kappa_a+\overline{\mathsf a}\,\overline{\tilde c_G}\,\eta_a
 +r_a+\ell_a,\qquad
 \mu=|\mathsf a|^2+\mathsf a\,\overline{\eta_G}/\overline{\tilde c_G} .
\]
The mixed terms between $\Pi\Xi'$ and $\mathsf u$ vanish. Since
$|\mathsf a|^2=1-\varepsilon$ and $|\eta_G|\le\sigma$, $|\mu|\ge\lambda_b$.

\noindent\emph{Sizes.}
$\|\eta^{(b)}\|^2:=|\eta_G|^2+\sum_a|\eta_a|^2$ is the squared projection of
$\sigma_{\rm v}$ on $\spn\{\psi_l\}$. Since $\sigma_{\rm v}\perp\tilde\Xi$,
it is at most the squared projection on the span of all $\psi_l$ of the block
minus $\tilde\Xi$. These spaces are orthogonal for different blocks, so
$\sum_b\|\eta^{(b)}\|^2\le\sigma^2$. Next
$|r_a|\le\|P_{b=G}\sigma_{\rm v}\|\|P_{b=a}\sigma_{\rm v}\|$, so
$\sum_a|r_a|^2\le\sigma^4/4$. And
$|\ell_a|\le\|Q_a^{\rm ran}\mathsf u\|\|Q_a^{\rm dom}\mathsf u\|$, so
$\sum_a|\ell_a|^2\le\ell^2$.

\noindent\emph{Angle.}
Within this block write $\hat v=\hat v_b$ and $\tilde v=\tilde v_b$, so
$\|\tilde v\|^2=w_GW_b$. When $\lambda_b>0$,
$\infid(\tilde\omega_b,\hat\omega_b)=\tau_b+W_b\sin^2\angle(\hat v,\tilde v)$
and $\sin\angle(\hat v,\tilde v)\le\|\hat v-\mu\tilde v\|/(|\mu|\|\tilde v\|)$.
For $\hat v=0$, interpret the angle using the prescribed output
$\ket{F_G}$; the inequality still holds because its right side equals one.
The column error separates as
\[
\begin{gathered}
 \hat v-\mu\tilde v=(e_b;s_a)_a+\mathfrak B_b,\quad (\mathfrak B_b)_G=w_G(1-\mu),\\
 (\mathfrak B_b)_a=\overline{\mathsf a}\,\overline{\tilde c_G}\eta_a+r_a+\ell_a,\\
 |1-\mu|^2\le2\varepsilon^2+2|\eta_G|^2/w_G,\\
 \|\mathfrak B_b\|^2
 \le2w_G^2\varepsilon^2+3w_G\|\eta^{(b)}\|^2
                    +\tfrac34\sigma^4+3\ell^2.
\end{gathered}
\]
The random terms are centered and $\E e_b^2=w_G(1-w_G)/N_{\rm w}$.
For the total bias use $\varepsilon,\ell\le\vartheta^2/4$,
$\sigma\le\vartheta^2/8$ and $\vartheta\le2$:
\[
 \begin{aligned}
 \sum_b\frac{\|\mathfrak B_b\|^2}{w_{G_b}}
 &\le\left[\frac{B_{\rm r}}8+\frac3{64}
       +\left(\frac3{1024}+\frac3{16}\right)\kappa_w\right]\vartheta^4\\
 &\le\frac{c_B\vartheta^4}{16}.
 \end{aligned}
\]

\noindent\emph{Variance.}
$\E|s_a|^2\le\langle Q_a\rangle_{\Xi'}/(2N_\phi)$, and $Q_a$ commutes with
$\Pi$. So
\[
\begin{aligned}
 \langle Q_a\rangle_{\Xi'}
 &\le(1+\eta')(w_G+w_a)\\
 &\quad+(1+\tfrac1{\eta'})\|Q_a\sigma_{\rm v}\|^2+\|Q_a\mathsf u\|^2.
\end{aligned}
\]
On patterns of $\Xi$, $Q_a=P_{b=G}+P_{b=a}$, so
\[
 \sum_a\|Q_a\sigma_{\rm v}\|^2\le(|\cA_b|-1)\sigma^2;
\]
also
$\sum_a\|Q_a\mathsf u\|^2\le(|\cA_b|-1)\ell$ and
\[
 \sum_a(w_G+w_a)\le(|\cA_b|-2)w_G+1.
\]
This is \eqref{gen:eq:readout-var} with
$\mathsf S_b=\sum_a\E|s_a|^2/w_G$. 
\end{proof}

\subsection{Leakage and simultaneous readout}
\label{gen:app:tail}
A \emph{learned pattern} is $\hat C_1\cup\bigcup_bF_{a_b}$ with one retained
atom $a_b\in\cA_b$ for every retained block.
Let $\hat\Pi$ project onto the span $\hat\cK$ of the learned patterns and
put $\Xi'=U_{\hat R}^\dagger\Psi$. For a unit vector $\hat\Xi\in\hat\cK$,
write $\hat P=\ket{\hat\Xi}\bra{\hat\Xi}$ and set
\[
 \hat\ell=\|(I-\hat\Pi)\Xi'\|^2,\qquad
 r_{\rm in}=\|(\hat\Pi-\hat P)\Xi'\|.
\]
The returned state $\hat\Xi$ lies in $\hat\cK$. On $\cS$ every learned
pattern is a pattern of $\Xi$, so $\hat\Pi\le\Pi$.

\begin{alem}[Split and leak test]
\label{gen:lem:split}
\begin{enumerate}[label=(\alph*)]
\item For every unit vector $\hat\Xi\in\hat\cK$,
\begin{equation}
 \infid(\Psi,U_{\hat R}\hat\Xi)=\hat\ell+r_{\rm in}^2 .
 \label{gen:eq:split}
\end{equation}
For $\hat\ell<1$, this also gives $r_{\rm in}^2=(1-\hat\ell)\hat\theta$,
where $\hat\theta=\infid(\hat\Xi,\hat\Pi\Xi'/\|\hat\Pi\Xi'\|)$.
\item A copy of the leak test is a miss with probability $\hat\ell$. If
$\hat\ell\le f_1/\rho'$, the test rejects with probability at most
$q_{\rm t}$. If $\hat\ell\ge f_1$, it accepts with probability at most
$\delta_\ell$.
\item On $\cS$, with $\ell$, $\sigma$ and $\tilde\Xi$ of
Lemma~\ref{gen:lem:frame-state} and $\tau_b$ of Lemma~\ref{gen:lem:structure},
\[
\begin{gathered}
 \hat\ell\le\ell+\bigl(\sqrt{\tau_{\rm tot}}+\sigma\bigr)^2,\\
 \tau_{\rm tot}=1-\prod_b(1-\tau_b)\le\sum_b\tau_b .
\end{gathered}
\]
\item On $\cS$, for every unit vector $\hat\Xi\in\hat\cK$ and
$\tilde\Xi_\parallel=\hat\Pi\tilde\Xi/\|\hat\Pi\tilde\Xi\|$,
\begin{equation}
 r_{\rm in}\le
 \sin\angle(\hat\Xi,\tilde\Xi_\parallel)+\sigma .
 \label{gen:eq:inside-bound}
\end{equation}
\end{enumerate}
\end{alem}
\begin{proof}
\noindent\emph{(a).}
Since $\hat P\le\hat\Pi$, the two summands in
$I-\hat P=(I-\hat\Pi)+(\hat\Pi-\hat P)$ are projectors with orthogonal
ranges. Hence
\[
 1-|\braket{\hat\Xi}{\Xi'}|^2
 =\|(I-\hat P)\Xi'\|^2=\hat\ell+r_{\rm in}^2 .
\]
For $\hat\ell<1$, normalizing $\hat\Pi\Xi'$ gives the second identity.

\noindent\emph{(b).}
$\hat\Pi$ is diagonal in the occupation basis, so the number of misses is
$\mathrm{Bin}(m_\ell,\hat\ell)$. Its lower tail decreases with
$\hat\ell$, so the two inequalities follow from~\eqref{gen:eq:leak-size}.

\noindent\emph{(c).}
$\hat\ell=\|(I-\Pi)\Xi'\|^2+\|(\Pi-\hat\Pi)\Xi'\|^2$. Write
$\Pi\Xi'=\braket{\tilde\Xi}{\Xi'}\tilde\Xi+\xi$ with $\|\xi\|=\sigma$.
Then $\|(\Pi-\hat\Pi)\Xi'\|\le\|(\Pi-\hat\Pi)\tilde\Xi\|+\sigma$. The
canonical input is a product over the blocks with the weights of $\Xi$.
$\hat\Pi$ keeps the retained branches of a retained block and the heaviest
branch of a frozen block, of weight $1-\tau_b$ in both cases. So
$\|\hat\Pi\tilde\Xi\|^2=\prod_b(1-\tau_b)$.

\noindent\emph{(d).}
The product in the proof of (c) is positive, because every factor
$1-\tau_b$ is the weight of at least one branch. So
$\|\hat\Pi\tilde\Xi\|>0$. Since $\hat\Pi\le\Pi$,
$\hat\Pi\Xi'=\braket{\tilde\Xi}{\Xi'}\hat\Pi\tilde\Xi+\hat\Pi\xi$ with
$\|\xi\|=\sigma$. Thus
\[
\begin{aligned}
 r_{\rm in}&=\|(I-\hat P)\hat\Pi\Xi'\|\\
 &\le\|(I-\hat P)\tilde\Xi_\parallel\|+\sigma\\
 &=\sin\angle(\hat\Xi,\tilde\Xi_\parallel)+\sigma .
\end{aligned}
\]
\end{proof}

\begin{alem}[Readout in tail form]
\label{gen:lem:readout-tail}
Assume $\cS$ and a frame within the level $u_{\rm c}$, and let
$\tilde\Xi_\parallel=\hat\Pi\tilde\Xi/\|\hat\Pi\tilde\Xi\|$. With
probability at least $1-3\delta/8$ over the copies of step~5 of
$\mathsf{Fit}$ in Algorithm~\ref{alg:general-reference} and the $N_{\rm w}$ copies,
\[
 \sin\angle(\hat\Xi,\tilde\Xi_\parallel)\le
 \frac{r_\theta+\tfrac14\sqrt{c_B}\,\varsigma_{\rm c}^2u_{\rm c}}
      {\lambda\sqrt{1-\bar\tau}} ,
\]
where $\varsigma_{\rm c}=\varsigma(\sqrt{u_{\rm c}})$, $c_B\le\bar c_B$ by the structure bounds, $r_\theta=H$ of \eqref{gen:eq:tail-parts} and $\bar\tau$ of
\eqref{gen:eq:input-batches}. Here $\lambda>0$ by
Lemma~\ref{gen:lem:closed-budgets}.
\end{alem}
\begin{proof}
If no block has a retained partner, $\hat\cK$ is one-dimensional and the
sine is zero. Otherwise $r_\theta>0$. In the notation of
Lemma~\ref{gen:lem:readout}, the frame has
$d\le\bar d$, so $\vartheta\le\varsigma_{\rm c}\bar d$,
$\varepsilon\le\varepsilon_{\rm c}$, $\sigma\le\sigma_{\rm c}$, and
$|\mu_b|\ge\lambda$ for every retained block.

\noindent\emph{One vector.}
Let $Z$ be the vector with the entries $e_b/\sqrt{w_{G_b}}$ and
$s_a/\sqrt{w_{G_b}}$, over all retained blocks $b$ and partners $a$. It is a
sum of independent centered vectors, one for every copy. A copy of the second
batch of step~1 of $\mathsf{Fit}$ changes every $e_b$ by
$(\mathbf 1\{G_b\text{ occupied}\}-w_{G_b})/N_{\rm w}$, so its vector has
norm at most $\sqrt{\kappa_w}/N_{\rm w}$. A copy of its step~5 serves one partner
in every block and changes $s_a$ by a centered value of $\cX_a/(2N_\phi)$,
so its vector has norm at most $\sqrt{\kappa_w}/N_\phi$. The second moments
add up to
\[
 \E\|Z\|^2=\sum_b(1-w_{G_b})/N_{\rm w}+\sum_b\mathsf S_b\le\bar V^2,
\]
by the
variance bound~\eqref{gen:eq:readout-var} with $\eta'=1$, summed over
blocks using~\eqref{gen:eq:structure-bounds}. Here
$\bar V^2=\bar B/N_{\rm w}+D_\phi/N_\phi$ and
$\bar{\mathsf b}=\sqrt{\bar\kappa}/\min(N_{\rm w},N_\phi)$.
Equation~\eqref{gen:eq:tail-parts} yields
$\bar V^2\le1/t_{\rm rd}$ and $\bar{\mathsf b}\le b_0/t_{\rm rd}$.

\noindent\emph{Vector Bernstein inequality.}
For independent centered vectors $\xi_i$ of a Hilbert space with
$\|\xi_i\|\le\mathsf b$ and $\sum_i\E\|\xi_i\|^2\le V^2$,
\[
 \Prob\Bigl\{\Bigl\|\sum_i\xi_i\Bigr\|\ge r\Bigr\}
 \le2\exp\Bigl(-\frac{r^2}{2V^2+\tfrac23\mathsf b\,r}\Bigr)
\]
\cite[Theorem~3.4]{pinelis}, followed by
\[
 (1+x)\log(1+x)-x\ge x^2/[2(1+x/3)].
\]
Markov's inequality and the exponential bound applied to the same
vector give
\[
 \begin{gathered}
\Prob\{\|Z\|\ge H\}\le\min\left\{\frac1{t_{\rm rd}H^2},\,
 2e^{-t_{\rm rd}H^2/(2+2b_0H/3)}\right\}\\ \le\delta_{\rm rd}.\end{gathered}
\]
Indeed each entry supplies its own sufficient scale in
\eqref{gen:eq:tail-parts}; either suffices.

\noindent\emph{Angle.}
By the paragraph ``Angle'' of Lemma~\ref{gen:lem:readout},
\[
 \sin\angle(\hat v_b,\tilde v_b)\le
\|\hat v_b-\mu_b\tilde v_b\|/(|\mu_b|\|\tilde v_b\|)
\]
with
$\|\tilde v_b\|^2=w_{G_b}W_b$, and $\hat v_b-\mu_b\tilde v_b$ is the random
part $(e_b;s_a)$ plus a bias with the entries $w_G(1-\mu)$ and
$\overline{\mathsf a}\,\overline{\tilde c_G}\eta_a+r_a+\ell_a$. Divide by
$\sqrt{w_{G_b}}$ and collect the blocks:
\[
 \sum_bW_b\sin^2\angle(\hat v_b,\tilde v_b)\le\lambda^{-2}\|Z+\mathfrak B\|^2 .
\]
By the paragraphs ``Sizes'' and ``Angle'',
\[
\begin{aligned}
 \|\mathfrak B\|^2
 &\le\sum_b\bigl[2w_G\varepsilon^2+3\|\eta^{(b)}\|^2
   +(\tfrac34\sigma^4+3\ell^2)/w_G\bigr]\\
 &\le c_B\vartheta^4/16.
\end{aligned}
\]
On the event
$\|Z\|<r_\theta$ the left side is at most
$\lambda^{-2}(r_\theta+\tfrac14\sqrt{c_B}\vartheta^2)^2$.

\noindent\emph{Product.}
$\hat\Xi$ and $\tilde\Xi_\parallel$ are products over the retained blocks
and agree on the other modes. So
$\infid(\hat\Xi,\tilde\Xi_\parallel)
\le\sum_b\sin^2\angle(\hat v_b,\tilde v_b)$, and
$W_b=1-\tau_b\ge1-\bar\tau$ on $\cS$.
\end{proof}

\subsection{Finite reconstruction schedule}
\label{gen:app:schedule}
This section proves the guarantee for a single call to $\mathsf{Fit}$.
Here $f,\delta$ denote the local infidelity and failure parameters of
that call. Every ceiling is retained in its finite cost. The same design constants are used for every request.

\subsubsection{Design and deterministic bounds}
\begin{table}[!ht]
\centering
\caption{Fixed calibration of the reference-assisted learner.}
\label{gen:tab:design}
\setlength{\tabcolsep}{3pt}\footnotesize
\begin{tabular}{rrrrrrrrr}
\toprule
$\rho$&$\rho'$&$r$&$q_c$&$q$&$q_{\rm t}$&$\theta$&$d_c$&$t$\\
\midrule
$3/2$&$5/4$&$3/2$&$1/4$&$1/5$&$1/50$&$2/5$&$2/5$&$11/50$\\
\bottomrule
\end{tabular}
\end{table}
Put $u_0=1/16$, $f_1=f_2=f/2$, $\tau_0=\sqrt{u_0/2}$ and
$\bar c_3=1/2+\tau_0/3$. The following inequalities ensure that each correction has positive margin:
\begin{equation}
\begin{gathered}
 \rho,r>1,\quad 0<q_c,q,q_{\rm t}<1,\quad 1<\rho'\le16,\\
 0<\theta<1,\quad 0<t<1/2,\quad 0<d_c\le1/\sqrt2,\\
 \rho_0:=\rho\frac{1-u_0/8}{1+u_0/4}>1,\qquad
 \beta_0:=\bar c_3u_0\sqrt{r\rho}\le\tfrac12,\\
 \bar q:=q_c+q+2q_{\rm t}<1,\qquad
 g_0:=\sqrt{u_0/\rho}-\tfrac12d_c^3>0.
\end{gathered}
\label{gen:eq:design}
\end{equation}
For Table~\ref{gen:tab:design}, $\rho_0=381/260$, $\bar q=49/100$,
$\beta_0=3\bar c_3/32$ and $g_0=1/\sqrt{24}-4/125$, as required.
For the retained blocks define the actual quantities
\begin{equation}
\begin{gathered}
 B_{\rm r}=\#\{\text{retained blocks}\},\qquad
 \kappa_w=\sum_b w_{G_b}^{-1},\\
 S_{\rm br}=\sum_b(|\cA_b|-2),\qquad c_B=4B_{\rm r}+1+4\kappa_w,\\
 \hat L=\max_b|\cA_b|,\qquad w_{\min}=\min_b w_{G_b}.
\end{gathered}
\label{gen:eq:structure-numbers}
\end{equation}
For an empty retained family set $B_{\rm r}=\kappa_w=S_{\rm br}=0$,
$\hat L=w_{\min}=1$ when describing the actual quantities.
The weights in this display are unknown. In the schedule and its cost,
use the following deterministic upper bounds, and the stated lower bound
for $w_{\min}$, supplied by Lemma~\ref{gen:lem:structure}:
\begin{equation}
\begin{gathered}
 B_{\rm r}\le n/6,\quad \kappa_w\le4n/3,\quad
 S_{\rm br}+\kappa_w\le P:=5n/3-2,\\
 \hat L\le n/3,\qquad w_{\min}\ge3/(4n),\qquad c_B\le6n+1.
\end{gathered}
\label{gen:eq:structure-bounds}
\end{equation}
For the schedule, denote these deterministic bounds by
\[
 \begin{gathered}
 (\bar B,\bar\kappa,\bar L)=(n/6,4n/3,n/3),\\
 \underline w=3/(4n),\qquad \bar c_B=6n+1.
 \end{gathered}
\]
These deterministic values fix the batches and levels. An empty retained family can omit the readout; its error vector
is zero. This does not alter the preceding schedule. In proofs using
actual quantities the same upper bounds remain valid.

Let $Q$ be the largest odd prime power below $n$. On every pairing
record choose a uniform subset of $n'=Q+1$ modes and use the flat probe
of Lemma~\ref{gen:lem:probe}. Thus $g^2=1/Q$ and
\[
 \piu=\frac{Q(Q+1)}{n(n-1)},\qquad \piu n(n-1)\ge30.
\]
Since $n\ge6$, the largest power of five below $n$ exceeds $(n-1)/5$.
Consequently $\piu>1/25$ and $n'\ge6$. Equation~\eqref{gen:eq:A-record}
uses this inclusion probability and $s^2g^2=\tau/[\piu n(n-1)]$
for probe amplitude $s$. A field, generator and
character can be constructed in polynomial time in $n$ by finite-field
enumeration.
Define
\[
 V(v)=\mathsf V(2\sqrt v,\sqrt{v/2}),
\]
with $\mathsf V$ and $c_3$ of \eqref{gen:eq:V}.
Let $\tilde V(v)$ be~\eqref{gen:eq:Vphi} at the same arguments and
$\varpi=n(n-1)/\piu$. This is a variance scale, not the dimension of the
frame space. Put $z_0=\varsigma(\sqrt{u_0})$, with $\varsigma$ from
\eqref{gen:eq:varsigma}, and
\[
\begin{gathered}
 v_0=\frac{8(1+\tau_0)(1+\tau_0+\tau_0^2)}{1-\tau_0/6}
          +\frac{16\tau_0^2}{9},\\
 \tilde v_0=1+
 \frac{2\tau_0[(1+\tau_0)^2+(1+\sqrt2\tau_0)^2]}{1-\tau_0/6}
          +\frac{4\tau_0^2}{9},\\
 h_0=\sqrt{u_0/\rho},\qquad
 \kappa_0=\tfrac12(\sqrt{u_0}+h_0)+h_0^2,\\
 \chi_0=1+\tfrac14\varsigma(h_0)h_0,\qquad
 \bar A=\chi_0(1+\kappa_0\sqrt{v_0}),\\
 C_R=\chi_0[\tfrac78+\bar c_3(\tfrac12+\kappa_0)].
\end{gathered}
\]

\subsubsection{Levels, batches and tests}
The three stopping scales and the final cap are
\[
\begin{gathered}
 u_{\rm in}^{-1}
 =\max\{(\sqrt{\bar c_B}+1)\sqrt{2/f},\;2/\sqrt{\underline w},\;\bar L\},\\
 u_R^{-1}=(C_R/\theta)^{2/3}(4\rho'/(qf))^{1/3},\\
 u_K^{-1}=(4\bar c_3^2\rho/u_{\rm in})^{1/3},\qquad u_k=u_0r^{-k},\\
 J-1=\min\{k\ge0:u_k\le\min(u_R,u_K)\},\\
 u=u_{J-1},\qquad u_{\rm c}=\min(u,u_{\rm in}).
\end{gathered}
\]
At the cap write $\bar d=\sqrt{u_{\rm c}}$,
$\varsigma_{\rm c}=\varsigma(\bar d)$ and
\begin{equation}
\begin{gathered}
 \varepsilon_{\rm c}=\sin^2(\varsigma_{\rm c}\bar d/2),\qquad
 \sigma_{\rm c}=2\sin^2(\varsigma_{\rm c}\bar d/4),\\
 \lambda=1-\varepsilon_{\rm c}-\sigma_{\rm c}/\sqrt{\underline w},\\
 h=\sqrt{u_{\rm c}/\rho},\qquad
 \kappa_d=\tfrac12(\sqrt u+h)+h^2,\\
 \chi=1+\tfrac14\varsigma(h)h,\quad A_\ell=\chi(1+\kappa_d\sqrt{V(u)/\tilde V(u)}),\\
 R_\ell=\chi[\tfrac12c_3(\sqrt{u/2})+\tfrac78+
                       \kappa_dc_3(\sqrt{u/2})]u^{3/2}.
\end{gathered}
\label{gen:eq:caps}
\end{equation}
Lemma~\ref{gen:lem:structure} is used with $\delta/2$ in place of $\delta$
throughout: $\Lambda=\log(128n^2/\delta)$,
$k_0=\max\{2,\lceil\log(32n^2/\delta)/\log6\rceil\}$. The structure batch is
\begin{equation}
\begin{gathered}
 N_{\rm s}=\left\lceil\max\left\{
 \frac{8\Lambda n}{3(\sqrt{f/(4\rho')}-\sigma_{\rm c})^2},
 16\Lambda n^2\right\}\right\rceil,\\
 \bar\tau=\frac{8\Lambda n}{3N_{\rm s}},\quad
 f_{\rm miss}=(\sqrt{\bar\tau}+\sigma_{\rm c})^2,\quad
 f_\ell=\frac{f_1}{\rho'}-f_{\rm miss}.
\end{gathered}
\label{gen:eq:input-batches}
\end{equation}
The first frame uses $x_*=x_*(d_c^2,t)$ of Lemma~\ref{gen:lem:planes},
$a_a=\sqrt{2/(9q_c)}$, $a_b=\sqrt{1/(9q_c)}$ and
\[
\begin{aligned}
 N_a&=\left\lceil\max\left\{
 \begin{gathered}
 \frac{(1+a_a)4n^2(n-1)}{x_*},\\
 \frac{2(2n-1)^2+4nt/3}{t^2}\log\frac{16n^2}{q_c}
 \end{gathered}
 \right\}\right\rceil,\\
 N_b&=\left\lceil\frac{(1+a_b)\sM^2}{g_0^2}\right\rceil,
 \qquad N_0=N_a+N_b.
\end{aligned}
\]
For $1\le k<J$ put $\tau_k=\sqrt{u_{k-1}/2}$ and
$\tau_J=\sqrt{u/2}$. The pairing records of stage $k$ use the probe amplitude
$s=(\tau_k/n')^{1/2}$, that is, the strength $\tau=\tau_k$ in
Lemma~\ref{gen:lem:record}. The correction batches are
\[
\begin{aligned}
 N_k&=\left\lceil\frac{V(u_{k-1})}
 {q_c(\sqrt{u_k/\rho}-c_3(\tau_k)u_{k-1}^{3/2})^2}\right\rceil,\\
\end{aligned}
\]
For the final batch set $h=\sqrt{u_{\rm c}/\rho}$ and
\[
\begin{aligned}
 C_F&=V(u)/(h-c_3(\tau_J)u^{3/2})^2,\\
 K&=A_\ell\sqrt{\tilde V(u)/f_\ell},\qquad b=R_\ell/\sqrt{f_\ell},\\
 w&=q_c+q,\qquad d=w-b^2,\\
 N_J&=\left\lceil\left(\frac{Kb+\sqrt{dC_F+wK^2}}{d}\right)^2\right\rceil .
\end{aligned}
\]
For $k=0,\dots,J$, the attempt caps and individual test errors are
\begin{equation}
 \begin{gathered}
 A_k=\left\lceil\frac{\log(2^{J-k+4}/\delta)}{\log(1/\bar q)}\right\rceil,\\
 \delta_k=\frac{\delta\,2^{k-J-4}}{A_k},\qquad
 \delta_\ell=\frac{\delta}{4A_J}.
 \end{gathered}
\label{gen:eq:confidences}
\end{equation}
Write $F_{m,p}(k)=\Prob\{\mathrm{Bin}(m,p)\le k\}$. For endpoints
$a<b$ and error allowances $\alpha,\beta$, let $\mathsf{BinTest}(a,b;\alpha,\beta)$
return the lexicographically first nonnegative pair $(m,k)$ satisfying
\[
 F_{m,a}(k)\ge1-\alpha,\qquad F_{m,b}(k)\le\beta .
\]
The test accepts counts at most $k$. Increasing $p$ increases the
count in stochastic order, so checking these two endpoints suffices.
In a sequential implementation of a pass/fail test, let $Z_s$ be the
count after $s$ of $m$ observations. The final count lies in
$[Z_s,Z_s+m-s]$. Reject when $Z_s>k$ and accept when $Z_s+m-s\le k$.
Either decision agrees with every completion of the batch. With fresh
draws for subsequent stages this uses at most the scheduled runs.
Batches used for estimation are always acquired in full.
For $v=u_i$ when $i<J$, and $v=u_{\rm c}$ at $i=J$, use
\begin{equation}
 \begin{gathered}
 p_a=v/(8\rho),\quad p_r=v(1-v/8)/[8(1+v/4)],\\
 (m_i,k_i^{\rm click})=\mathsf{BinTest}(p_a,p_r;q_{\rm t},\delta_i).
 \end{gathered}
 \label{gen:eq:click-schedule}
\end{equation}
The final leakage count has schedule
\begin{equation}
 (m_\ell,k_\ell)=\mathsf{BinTest}(f_1/\rho',f_1;q_{\rm t},\delta_\ell).
 \label{gen:eq:leak-size}
\end{equation}
The sufficient size in~\eqref{gen:eq:test-size} bounds the search for
each test. Binomial probabilities follow by the adjacent-mass recurrence.

\noindent\emph{Readout allocation.}
Write $\vartheta_c=\varsigma(\bar d)\bar d$ and set
\begin{equation}
\begin{aligned}
 H={}&\lambda\sqrt{1-\bar\tau}(\sqrt{f_2}-\sigma_{\rm c})
       -\tfrac14\sqrt{\bar c_B}\vartheta_c^2,\\
 w_\phi={}&2(\lfloor n/3\rfloor-1),\qquad \delta_{\rm rd}=3\delta/8,\\
 D_\phi={}&P+\tfrac12(\bar L-1)(\varepsilon_{\rm c}+2\sigma_{\rm c}^2)\bar\kappa,\\
 T={}&\sqrt{\bar B}+\sqrt{w_\phi D_\phi},\\
 c_w={}&T\sqrt{\bar B},\qquad c_\phi=T\sqrt{D_\phi/w_\phi},\\
 b_0={}&\sqrt{\bar\kappa}/\min(c_w,c_\phi),\\
 t_{\rm rd}={}&H^{-2}\min\{1/\delta_{\rm rd},(2+2b_0H/3)\log(2/\delta_{\rm rd})\},\\
 N_{\rm w}={}&\lceil c_w t_{\rm rd}\rceil,\qquad N_\phi=\lceil c_\phi t_{\rm rd}\rceil .
\end{aligned}
\label{gen:eq:tail-parts}
\end{equation}
The integer $w_\phi$ bounds the
number of quadratures across all blocks. Cauchy--Schwarz gives
\[
 (N_{\rm w}+w_\phi N_\phi)
 (\bar B/N_{\rm w}+D_\phi/N_\phi)\ge T^2,
\]
with equality at the chosen ratio before rounding.

\subsubsection{The last correction}
\begin{alem}[State-weighted final correction]
\label{gen:lem:last-correction}
Let $u_{\rm c}\le u\le u_0$ and suppose the incoming frame has squared
distance at most $u$. Take $N$ fresh pairing records with the strength
$\tau_J=\sqrt{u/2}$. With $h,\kappa_d,\chi,A_\ell,R_\ell$ of
\eqref{gen:eq:caps}, let $d'$ be the distance of Lemma~\ref{gen:lem:pcorr} of the
corrected frame, in the gauge of the incoming frame, and $\ell$ the leak
$\|(I-\Pi)\Xi'\|^2$ of the corrected frame. For $g>0$,
\begin{equation}
 \Prob\{d'\le h,\ \ell>g\}
 \le\frac{(A_\ell\sqrt{\tilde V(u)/N}+R_\ell)^2}{g}.
 \label{gen:eq:last-leak}
\end{equation}
The bound is joint; it does not condition on $d'\le h$.
\end{alem}
\begin{proof}
Let $Q=\polar(I+A_{\rm av})$ and $Y=\Log(Q^{\T}S)$ in the
gauge of Lemma~\ref{gen:lem:gauge}. Put
$Z_F=\|A_{\rm av}-\E A_{\rm av}\|_F$ and
$Z_\Xi=\|A_{\rm av}-\E A_{\rm av}\nP$.
Lemmas~\ref{gen:lem:generator}, \ref{gen:lem:leak} and~\ref{gen:lem:pcorr} give,
on $d'\le h$,
\begin{align*}
 d'&\le Z_F+c_3u^{3/2},\qquad
 \|Y\|_F\le\varsigma(h)h,\\
 \|Y_\od\nP
 &\le Z_\Xi+\tfrac12(c_3-\tfrac14)u^{3/2}\\
 &\quad+\tfrac12(\sqrt u+d')d'+\tfrac14(\sqrt u+d')^3\\
 &\le Z_\Xi+\kappa_d d'+(\tfrac12c_3+\tfrac78)u^{3/2},\\
 \sqrt\ell
 &\le\chi\|Y_\od\nP
 \le\chi(Z_\Xi+\kappa_dZ_F)+R_\ell.
 \end{align*}
Here $\sqrt u+h\le1/2$, so Lemma~\ref{gen:lem:generator} applies; convexity gives
$(\sqrt u+d')^3/4\le u^{3/2}+(d')^3$ and $(d')^3\le h^2d'$.
The seminorm bound used here is $\|M_\od\nP\le\|M\|_F/2$.
Minkowski's inequality and the record variances imply
\[
 \left(\E\left[(\chi(Z_\Xi+\kappa_dZ_F)+R_\ell)^2\right]\right)^{1/2}
 \le A_\ell\sqrt{\tilde V(u)/N}+R_\ell.
\]
Markov's inequality on the joint event proves~\eqref{gen:eq:last-leak}.
\end{proof}

\subsubsection{Verification of the schedule}
\noindent\emph{Terminology.}
Stage $0$ is the first frame, stage $k$ $(1\le k<J)$ the $k$th correction
and stage $J$ the last correction. The incoming frame of stage $k\ge1$ is
\emph{good} if it is within the level $u_{k-1}$. An attempt of stage $k<J$
meets its \emph{inner frame level} if its candidate is within $u_k/\rho$. An
attempt of stage $J$ meets its inner frame level if $d'\le h$, which bounds
the candidate's squared frame distance by $(d')^2\le u_{\rm c}/\rho$
(Lemma~\ref{gen:lem:pcorr}), and its \emph{inner leakage level} if
$\hat\ell\le f_1/\rho'$. A history is \emph{good} for stage $k$ if the
structure event $\cS$ holds and every earlier accepted frame is within the
level of the test that accepted it.

\begin{alem}[Closed budgets]
\label{gen:lem:closed-budgets}
The schedule is well defined. From a good incoming frame, a nonfinal
attempt fails its inner frame level with probability at most $q_c$.
On the structure event $\cS$, a final attempt from a good incoming frame
fails either its inner frame level or its inner leakage
level with probability at most $q_c+q$. Moreover $\lambda>0$ and
\begin{equation}
 \left[\frac{r_\theta+\sqrt{\bar c_B}\varsigma_{\rm c}^2u_{\rm c}/4}
 {\lambda\sqrt{1-\bar\tau}}+\sigma_{\rm c}\right]^2\le f_2.
 \label{gen:eq:closed-budget}
\end{equation}
\end{alem}
\begin{proof}
\noindent\emph{Uniform record bounds.}
The single probe construction gives $\piu n(n-1)\ge30$ and $n'\ge6$.
Substitute $\gamma=2\sqrt v$ and $\tau=\sqrt{v/2}\le\tau_0$ into
\eqref{gen:eq:V} and~\eqref{gen:eq:Vphi}. Since $\gamma^2/8=\tau^2$,
\begin{align*}
 c_3&\le\bar c_3,\qquad \varpi\le\tilde V(v)\le\tilde v_0\varpi,
 \qquad V(v)\le v_0\varpi,\\
 A_\ell&\le\bar A,
 \qquad R_\ell\le C_Ru^{3/2}.
\end{align*}
All bounding expressions increase with $\tau,h,u$.

\noindent\emph{Initialization.}
The two tail bounds of Lemmas~\ref{gen:lem:planes} and~\ref{gen:lem:mcorr}
have variance ratios bounded by
\[
 \frac{\omega_n}{4n^2(n-1)}\le\frac{22}{405},\qquad
 \frac{\omega_n}{\sM^2}\le\frac4{75}.
\]
Both ratios decrease with $n\ge6$. Since $x_*\le1/5$ and $g_0\le1/4$,
\begin{align*}
 \frac n{(n-1)N_a}&\le\frac1{3000},\qquad
 \frac{8n(2n-1)^2}{N_b\sM^2}\le\frac1{900},\\
 \frac{22}{405}+\frac1{3000}&<\frac1{18},\qquad
 \frac4{75}+\frac1{900}<\frac1{18}.
\end{align*}
The exponential part is at most $q_c/4$; the two remaining parts are
at most $(18a_a^2)^{-1}=q_c/4$ and $(18a_b^2)^{-1}=q_c/2$.
Outside these events the plane estimator returns a frame at distance at
most $d_c\le1/\sqrt2$, and the matching correction has
\[
 \|B-\E B\|_F^2<(1+a_b)\sM^2/N_b\le g_0^2.
\]
Then Eq.~\eqref{gen:eq:mcorr} gives the
distance at most $g_0+d_c^3/2=\sqrt{u_0/\rho}$.
Thus the first frame misses $u_0/\rho$ with probability at most $q_c$.

\noindent\emph{Correction levels.}
For $1\le k<J$,
\begin{align*}
 \sqrt{u_k/\rho}-c_3u_{k-1}^{3/2}
 &\ge(1-\beta_0)\sqrt{u_k/\rho}>0,\\
 c_3u^{3/2}&\le\tfrac12\sqrt{u_{\rm c}/\rho}.
\end{align*}
For the second line use $u\le u_K$ if $u_{\rm c}=u_{\rm in}$, and
$\bar c_3u_0\sqrt\rho\le\beta_0$ otherwise. Lemma~\ref{gen:lem:pcorr}
and Markov's inequality give the failure bound $q_c$ for each nonfinal correction.

\noindent\emph{Missing mass and the cap.}
The integral representation of $\arcsin$ gives
\[
 z_0^2\le\frac{128}{127},\qquad
 \varepsilon_{\rm c}\le\frac{z_0^2u_{\rm c}}4,\qquad
 \sigma_{\rm c}\le\frac{z_0^2u_{\rm c}}8.
\]
Since
\[
 u_{\rm c}\le\min\{1/16,\sqrt{\underline w}/2,
\sqrt{f_2}/(1+\sqrt{\bar c_B})\}
\]
and $\rho'\le16$,
\begin{align*}
 \lambda&\ge\lambda_0:=1-5z_0^2/64\ge117/127,\\
 \sigma_{\rm c}&\le\tfrac12\sqrt{f/(4\rho')},\qquad
 \bar\tau\le\frac1{6n}\le\frac1{36},\\
 f_{\rm miss}&\le f/(4\rho'),\qquad f_\ell\ge f/(4\rho')>0.
\end{align*}
The level $u_R$ gives $R_\ell\le\theta\sqrt{q f_\ell}$, hence
$b^2\le\theta^2q<w$. From Lemmas~\ref{gen:lem:pcorr} and
\ref{gen:lem:last-correction}, the two disjoint failure events satisfy
\[
\begin{aligned}
 \Prob(d'>h)&\le C_F/N_J,\\
 \Prob(d'\le h,\ \ell>f_\ell)&\le(K/\sqrt{N_J}+b)^2 .
\end{aligned}
\]
The sum of the two upper bounds is at most $w$ precisely when
$dN_J-2Kb\sqrt{N_J}-(C_F+K^2)\ge0$.
The stated $N_J$ is the ceiling of the squared positive root.
Lemma~\ref{gen:lem:split} transfers this joint control to the
observable inner leakage level.

\noindent\emph{Readout margin.}
To verify $H\ge(3/5)\sqrt{f_2}$, substitute
$r_\theta=(3/5)\sqrt{f_2}$ temporarily in the left side
of~\eqref{gen:eq:closed-budget}.
Set $x=\sqrt{\bar c_B}/(1+\sqrt{\bar c_B})\in[0,1]$. The square root in
\eqref{gen:eq:closed-budget}, divided by $\sqrt{f_2}$, for this substitution is at most
\begin{align*}
 \frac{3/5+z_0^2x/4}{\lambda_0\sqrt{35/36}}
       +\frac{z_0^2(1-x)}8
 &\le\frac{3/5+z_0^2/4}{\lambda_0\sqrt{35/36}}<1.
\end{align*}
The first expression increases with $x$. For the final strict inequality,
the squared numerator and denominator obey the exact comparison
\[
 36(3\cdot127+5\cdot32)^2
 <35(5\cdot117)^2.
\]
Thus $H>0$. Restoring $r_\theta=H$ gives~\eqref{gen:eq:closed-budget} directly by its definition.
\end{proof}

\subsubsection{Probability and total cost}
\noindent\emph{Accuracy of one reconstruction.}
The click test has good-candidate rejection at most $q_{\rm t}$ and
bad-candidate acceptance at most $\delta_k$. Indeed
$p_r/p_a\ge\rho_0>1$ and~\eqref{gen:eq:click-schedule} enforces the two
endpoint requirements of Lemma~\ref{gen:lem:test}. The leak test has the
corresponding errors $q_{\rm t},\delta_\ell$ by Lemma~\ref{gen:lem:split}.
By Lemma~\ref{gen:lem:closed-budgets}, conditional on $\cS$ and a good incoming
history, rejection is therefore at most
$q_c+q_{\rm t}$ at a nonfinal stage and at most
$\bar q=q_c+q+2q_{\rm t}$ at the final stage.

Let $\mathcal F_{k,j}$ contain the revealed records before attempt $j$
of stage $k$, excluding the reserved input-weight outcomes. The candidate
is fixed before its fresh test records are drawn. If $E_{k,j}$ denotes
a false click acceptance, the conditional bound gives
\[
\Prob(E_{k,j})
 =\E[\Prob(E_{k,j}\mid\mathcal F_{k,j},\text{candidate})]
 \le\delta_k.
\]
At a good incoming history, successive conditioning on rejected attempts
gives rejection of all $A_k$ attempts with probability at most
$\bar q^{A_k}$: retries use fresh records and the same incoming frame.
By~\eqref{gen:eq:confidences},
\[
\begin{gathered}
 \sum_{k=0}^J A_k\delta_k
 =\frac\delta{16}\sum_{i=0}^J2^{-i}<\frac\delta8,\\
 \bar q^{A_k}\le\delta\,2^{k-J-4},\qquad
\sum_{k=0}^J\bar q^{A_k}<\frac\delta8,\quad A_J\delta_\ell=\frac\delta4.
\end{gathered}
\]
Let $\mathcal H$ be the event that structure recovery succeeds, all accepted
tests are sound, no good stage exhausts its attempts, and readout
concentration holds. Apply the bounds at the first failure and average over
the preceding records. The weight batch is uninspected and the readout
copies are fresh, so the conditional readout guarantee also averages
without change. Lemma~\ref{gen:lem:structure}, used with $\delta/2$, gives
\begin{align*}
\Prob(\mathcal H^c)
 &\le\frac\delta8+\sum_{k=0}^JA_k\delta_k
       +\sum_{k=0}^J\bar q^{A_k}\\
 &\quad+A_J\delta_\ell+\frac{3\delta}8\\
 &<\frac\delta8+\frac\delta8+\frac\delta8
       +\frac\delta4+\frac{3\delta}8=\delta.
\end{align*}
On $\mathcal H$, $\hat\ell\le f_1$ and
\[
 r_{\rm in}
 \le\sin\angle(\hat\Xi,\tilde\Xi_\parallel)+\sigma_{\rm c}
 \le\sqrt{f_2},
\]
by \eqref{gen:eq:inside-bound} with $\sigma\le\sigma_{\rm c}$,
Lemma~\ref{gen:lem:readout-tail} and \eqref{gen:eq:closed-budget}.
Equation~\eqref{gen:eq:split} gives $X\le f$.

\noindent\emph{Unconditional expected cost.}
Let $T_k\le A_k$ count the attempts made at stage $k$, with $T_k=0$
if that stage is not reached. Let stage $k$ be reached with a good incoming
frame (for $k=J$: also on $\cS$), and let the conditional rejection be at
most $q_*$. Given the records $\mathcal F_{k,1}$,
\[
\begin{aligned}
\E[T_k\mid\mathcal F_{k,1}]
 &=\sum_{j=1}^{A_k}\Prob\{T_k\ge j\mid\mathcal F_{k,1}\}\\
 &\le\sum_{j=0}^{A_k-1}q_*^j\le(1-q_*)^{-1}.
\end{aligned}
\]
If $\mathcal B_k$ denotes reaching stage $k$ with an incoming frame that is
not within the level $u_{k-1}$ (for $k=J$: with such a frame or outside
$\cS$), the deterministic cap yields
\[
\E T_k\le(1-q_*)^{-1}+A_k\Prob(\mathcal B_k).
\]
Before the final stage the reference-frame guarantees do not depend on
successful input-structure recovery. Therefore
\[
\begin{aligned}
 \Prob(\mathcal B_k)
 &\le\sum_{i=0}^{k-1}A_i\delta_i
 <\delta\,2^{k-J-4}\quad(1\le k<J),\\
 \Prob(\mathcal B_J)
 &\le\Prob(\mathcal S^c)+\sum_{i=0}^{J-1}A_i\delta_i <\delta/8+\delta/16.
\end{aligned}
\]
The resulting attempt factors are
\begin{equation}
\begin{gathered}
 \mathfrak a_0=(1-q_c-q_{\rm t})^{-1},\\
 \mathfrak a_k=\mathfrak a_0+A_k\delta2^{k-J-4}\quad(1\le k<J),\\
 \widehat{\mathfrak a}_J=(1-\bar q)^{-1}+3A_J\delta/16.
\end{gathered}
\label{gen:eq:attempt-factors}
\end{equation}
Multiplying these factors by the deterministic per-attempt costs gives
$\E N\le\mathcal C(n,f,\delta;D)$, where
\begin{equation}
\begin{aligned}
 \mathcal C(n,f,\delta;D)={}&N_{\rm s}+N_{\rm w}
                          +w_\phi N_\phi\\
 &+(N_0+m_0)\mathfrak a_0
 +\sum_{k=1}^{J-1}(N_k+m_k)\mathfrak a_k\\
 &+(N_J+m_J+m_\ell)\widehat{\mathfrak a}_J .
\end{aligned}
\label{gen:eq:count}
\end{equation}
Counting the leakage test even when the click test rejects only enlarges
this bound. Early termination likewise cannot increase it.

\noindent\emph{Order of the count.}
For $f\le1/n$ and $\delta\le1/2$, define the fixed constants
\begin{align*}
 c_{\rm in}&=\sqrt2(\sqrt{37}+1)/\sqrt6,\\
 c_U&=\max\left\{\begin{gathered}6^{-1/3}/u_0,\\
 r(C_R/\theta)^{2/3}(4\rho'/q)^{1/3},\\
 r(4\bar c_3^2\rho c_{\rm in})^{1/3}\end{gathered}\right\}.
\end{align*}
The stopping rule and~\eqref{gen:eq:structure-bounds} imply
\[
\begin{gathered}
 u_{\rm in}^{-1}\le c_{\rm in}\sqrt{n/f},\qquad
 u^{-1}\le U:=c_Uf^{-1/3},\\
 u_{\rm c}^{-1}\le U_c:=\max\{U,c_{\rm in}\sqrt{n/f}\},\\
 J=O_D(1+\log(1/f)).
\end{gathered}
\]
Since $\piu>1/25$, $\varpi<25n^2$. Directly from the batches,
\begin{align*}
 N_0&=O_D(n^3+n^2\log n),\\
 \sum_{k=1}^{J-1}N_k&=O_D(n^2U+J),\\
 N_J&=O_D(n^2/f+n^2U_c),\\
 N_{\rm s}&=O_D((n/f+n^2)\log(n/\delta)),\\
 \mathcal C_\theta&:=N_{\rm w}+w_\phi N_\phi\le T^2t_{\rm rd}+1+w_\phi .
\end{align*}
The cap $u_{\rm c}\le1/\bar L$ and the bounds on
$\varepsilon_{\rm c},\sigma_{\rm c}$ give $D_\phi=O(n)$, so
$T^2=O(n^2)$. For fixed $\delta$, the first entry defining $t_{\rm rd}$
and $H\ge(3/5)\sqrt{f_2}$ give $\mathcal C_\theta=O_\delta(n^2/f)$.
For arbitrary confidence the exponential entry keeps both counts
polynomial in $n,1/f,\log(1/\delta)$. The pooled correction fits
within the count obtained by assigning $q_c$ and $q$ separately:
$N\ge C_F/q_c$ and $N\ge K^2/[(1-\theta)^2q]$ suffice.
For the tests, set
$\mathcal L_0=\log(32A_0/(q_{\rm t}\delta))$ and
$\mathcal L_J=\log(16A_J/(q_{\rm t}\delta))$. The geometric sums give
\begin{align*}
 \sum_{k=1}^{J-1}\frac1{u_k}&\le\frac r{r-1}\frac1u,\qquad
 \sum_{j=0}^{\infty}\frac j{r^j}=\frac r{(r-1)^2},\\
 m_0&=O_D(\mathcal L_0+J),\\
 \sum_{k=1}^{J-1}m_k&=O_D((\mathcal L_0+1)U+J),\\
 m_J+m_\ell&=O_D(\mathcal L_J(U_c+f^{-1})),\\
 \mathcal L_J&=O_D(\log(1/\delta)),\\
 \mathcal L_0&=O_D(\log(1/\delta)+\log(1+\log(1/f))).
\end{align*}
\noindent\emph{Cost at fixed confidence.}
Fix the local failure parameter $\delta=\zeta\in(0,1/2]$; the
verification route takes $\zeta=1/4$. For $i=J-k$,
\eqref{gen:eq:confidences} gives $A_k=O_{D,\zeta}(1+i)$, so
$(1+i)2^{-i}\le2$
bounds every factor in~\eqref{gen:eq:attempt-factors}, including $k=J$.
Put $L_f=\log(1/f)$. On $nf\le1$ we have
\[
 \begin{gathered}
 n^3\le n^2/f,\qquad \sqrt{n/f}\le 1/f,\\
 (n/f+n^2)\log(4n)=O(n^2/f),\\
 (1+U)(1+\log(1+L_f))+J=O_D(1/f).
 \end{gathered}
\]
The last inequality uses $U=c_U f^{-1/3}$ and
$f^{2/3}L_f\le3/(2e)$. Substituting the batch and test sums above into
\eqref{gen:eq:count} proves
\begin{equation}
 \mathcal C(n,f,\zeta;D)=O_{D,\zeta}(n^2/f),\qquad 0<f\le1/n .
 \label{gen:eq:pilot-order}
\end{equation}

\noindent\emph{Classical processing.}
Every run produces $n$ occupation bits. Structure recovery compares
occupation columns and co-occurrences in polynomial time. A matching or
pairing record is a $2n\times2n$ matrix; its construction, averaging,
eigendecomposition and polar rounding require polynomially many arithmetic
operations, where one Hermitian eigendecomposition and one polar factor of a
$2n\times2n$ matrix are counted as $O(n^3)$ operations. The retained branch list has at most $n/3$ entries; its
coherence columns and Gaussian gate lists also have polynomial size.
The finite caps $A_k$ and batches above are polynomial in
$n,f^{-1},\log(1/\delta)$, so this applies even on histories that terminate
with the fixed failure output. This is a real-arithmetic count, not a
claim about exact gates at finite bit precision.\hfill$\square$

\subsection{Verification of a proposed preparation}
\label{gen:app:verifier}
The pilot estimates a frame and its branch coefficients. Once these are
specified, a different measurement can test the state they prepare on a
single fresh target. The spectral-gap principle for verification is
standard~\cite{pallister2018verification}. Here the test must use only
the allowed Gaussian gates and occupation readout, even when the proposed
frame or the block decomposition is wrong.

\begin{alem}[A Gaussian verification measurement]
\label{gen:lem:whole-test}
Let $P_* = U_{\hat R}\ketbra{\hat\Xi}{\hat\Xi}U_{\hat R}^\dagger$
be any proposed preparation of the form~\eqref{gen:eq:family}, including
a core alone. There is a randomized known even Gaussian circuit on the
$n$ given modes followed by a single occupation readout whose mean
rejection effect $Q$ satisfies
\begin{equation}
 \gamma_{\rm v}(I-P_*)\preceq Q\preceq I-P_*,\qquad \gamma_{\rm v}=1/5 .
 \label{gen:eq:verifier-gap}
\end{equation}
Its settings and decision rule have polynomial classical cost.
\end{alem}

\begin{proof}
Undo $U_{\hat R}$. Move zero-weight branches to the vacant core and
fold each remaining singleton block into the occupied core. In one
remaining block write $\ket i=\ket{F_i}$, $1\le i\le l$,
$\ket c=\sum_i c_i\ket i$, $w_i=|c_i|^2>0$,
$P=\sum_i\ketbra ii$ and $P_v=\ketbra vv$ for a unit vector $v$.

\noindent\emph{Two-port comparison.}
Let $I_{ij}=\ketbra ii+\ketbra jj$, $v_{ij}=I_{ij}v$ and
$s_{ij}=\|v_{ij}\|^2$. Put $D_{ij}(v)=I_{ij}-P_{v_{ij}/\sqrt{s_{ij}}}$
for $s_{ij}>0$ and $D_{ij}=0$ otherwise. On the coordinate space,
\begin{equation}
 \sum_{i<j}s_{ij}D_{ij}(v)=\|v\|^2 I-\ketbra vv .
 \label{gen:eq:dark-edges}
\end{equation}
Indeed $s_{ij}D_{ij}=s_{ij}I_{ij}-\ketbra{v_{ij}}{v_{ij}}$.
Summation gives diagonal entry $\sum_{j\ne i}|v_j|^2$ and
off-diagonal entry $-v_i\overline{v_j}$. A passive two-mode rotation
followed by occupation readout measures this orthogonal line when the
two ports contain one particle.

\noindent\emph{Reduction to one particle.}
For $l\ge3$ take $d=\lfloor l/2\rfloor+1$. Select one branch with
probabilities $w_i$ and draw $d-1$ others uniformly without replacement.
For $P_A=\sum_{i\in A}\ketbra ii$ and $W_A=\sum_{i\in A}w_i$,
\[
 \Prob(A)=W_A/\binom{l-1}{d-1},\qquad
 t=P_Ac/\sqrt{W_A}.
\]
For $l=2$ set $d=2$, take both branches and put $t=c$.

In every active branch retain a marker $X_i$ if its size is odd, or
two markers $X_i,Y_i$ if it is even. Pair its remaining modes, using on
each pair the even Gaussian rotation generated by
$a_p^\dagger a_q^\dagger-a_q a_p$, at angle $\pi/4$. In an ordered pair,
\[
 \begin{aligned}
 \ket{00}&\longmapsto(\ket{00}+\ket{11})/\sqrt2,\\
 \ket{11}&\longmapsto(-\ket{00}+\ket{11})/\sqrt2 .
 \end{aligned}
\]
For $h$ such pairs the record $y$ contributes $2^{-h/2}$ times
$(-1)^{h_i(y)}$ on branch $i$, where $h_i(y)$ counts the pairs of
$F_i$ read as $00$. For even branches the $Y$ register also contains
one particle. Pad it to $b=2^{\lceil\log_2 d\rceil}$ modes and apply
the real Walsh transform $W_{zi}=b^{-1/2}(-1)^{z\cdot i}$, using
distinct binary column labels for the active branches. The output
port $z$ contributes another known sign and a common factor $b^{-1/2}$.
The odd case needs no $Y$ register. Fixed fermionic reordering signs
are known from the ordered occupation kets. Thus the map from the active
code to the one-particle $X$ register, conditional on the erased
occupations $\eta$, is
\[
 \begin{gathered}
 \ket i\longmapsto \alpha_\eta s_i(\eta)\ket{X_i},\\
 s_i(\eta)\in\{-1,1\},\qquad \sum_\eta|\alpha_\eta|^2=1 .
 \end{gathered}
\]
This identity holds on every superposition. All gates are chosen before
readout. Only the final classical decision uses $\eta$.

\noindent\emph{Interfering equal amplitudes.}
For $l\ge3$ put $u_i=|t_i|^2$ and partition the unit circle into
intervals $I_i$ of lengths $u_i$. Reflect about a random shift $z$:
\[
 \begin{gathered}
 q_{ij}=|I_i\cap(z-I_j)|=q_{ji},\\
 \sum_j q_{ij}=u_i,\qquad \E_z q_{ij}=u_i u_j .
 \end{gathered}
\]
The last identity follows by integrating over $z$ for each point in
$I_i$. The overlap lengths are affine between successive endpoint sums modulo one. Sampling the midpoint of each such interval with
probability its length integrates them exactly. The common refinement
has at most $2d$ cells. Use one port per positive ordered entry and the
passive isometry
\[
 J\ket{X_i}= e^{-\mathrm{i}\arg t_i}\sum_{j:q_{ij}>0}
                     \sqrt{q_{ij}/u_i}\ket{i,j}.
\]
Its columns are orthonormal. At most $d$ additional modes are needed
for $J$; in the even case the Walsh register also needs $b-d$.
The inactive branches supply $3(l-d)\ge d$ modes in the odd case
and $4(l-d)\ge b$ in the even case ($b\le2d-2$).
All gates on the markers and borrowed modes are passive. On the code,
the total occupation of all markers and all inactive modes is one
(odd) or two (even) for an active branch, and at least three otherwise.
Read this number at the end and accept invalid records. On the active
code the borrowed modes start vacant.

Interfere reciprocal ports $(i,j),(j,i)$ on a balanced beam splitter.
Their candidate amplitudes are $\alpha_\eta\sqrt{q_{ij}}$ times
$s_i,s_j$. Reject a click in the output that cancels for the observed
relative sign; ignore self-ports. A valid record has even occupations
on the erased pairs, one $X$ click, one $Y$ click when present, and the
specified total occupation. For an active column $v$ put $x_i=v_i/t_i$. Averaging gives
\[
 \begin{aligned}
 \E_z\tfrac12\sum_{i<j}q_{ij}|x_i-x_j|^2
 &=\tfrac12\sum_{i<j}u_i u_j|x_i-x_j|^2\\
 &=\tfrac12\Bigl(\sum_i u_i|x_i|^2-
                  \bigl|\sum_i u_ix_i\bigr|^2\Bigr)\\
 &=\tfrac12\bigl(\|v\|^2-|\braket tv|^2\bigr).
 \end{aligned}
\]
Figure~\ref{fig:verifier-phase-comparisons} illustrates this identity for
$v_i=t_i e^{i\phi_i}$. A pair contributes zero whenever its two phases
agree. Their sum vanishes only when all active branches share one phase.
\begin{figure}[tbp]
 \centering
 \includegraphics[width=\columnwidth]{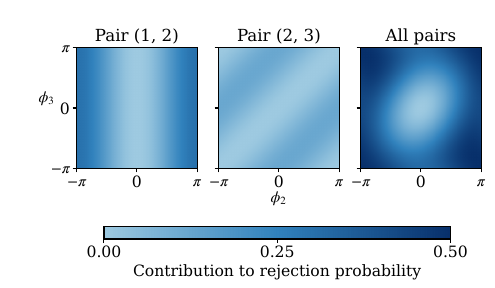}
 \caption{Phase comparisons on three active branches with normalized
 weights $(u_1,u_2,u_3)=(0.5,0.3,0.2)$. Phases are relative to branch~1
 ($\phi_1=0$). The first two panels show individual contributions to
 rejection, zero along $\phi_2=0$ and $\phi_3=\phi_2$, respectively,
 modulo $2\pi$. The last sums all three pairs, including $(1,3)$,
 and gives the rejection probability conditional on this active subset:
 $\tfrac12(1-|\braket tv|^2)$.}
 \label{fig:verifier-phase-comparisons}
\end{figure}
Each pair belongs to $\binom{l-2}{d-2}$ active subsets. Multiplying the
compression $(P_A-P_t)/2$ by $\Prob(A)$ and summing
\eqref{gen:eq:dark-edges} first in $A$ and then over all pairs yields
\[
 PMP=\frac{d-1}{2(l-1)}(P-P_c)\succeq\tfrac14(P-P_c).
\]
For $l=2$ draw one fair sign $e$, measure the marker line orthogonal
to $c_1\ket{X_1}+e c_2\ket{X_2}$ and reject its click only if
$e=s_1s_2$. The correct guess has probability $1/2$, giving
$PMP=(P-P_c)/2$. Both constructions accept $\ket c$ with certainty.
Put $\eta_0=1/4$.

\noindent\emph{Removing off-code interference.}
Choose one representative $j_i\in F_i$ and write
$\hat n_j=a_j^\dagger a_j$. Before the interference test, independently
apply or omit $e^{i\pi(\hat n_j-\hat n_{j_i})}$ with equal probability
for each $j\in F_i\setminus\{j_i\}$. Also apply
$\exp(2\pi\mathrm{i}\,s\sum_i\hat n_{j_i}/(l+1))$ for uniform
$s\in\{0,\ldots,l\}$. Every code word has the same phase character.
An occupation string with that character must be constant within
each $F_i$, and its representative count must be $1\bmod(l+1)$.
The count lies between zero and $l$, so the code is the entire
character space. Averaging conjugations gives an effect $M'$ with
\[
 PM'P=PMP,\qquad PM'(I-P)=0.
\]
Use the direct occupation test of membership in $P$ with probability
$\gamma_{\rm v}=\eta_0/(1+\eta_0)$ and use $M'$ otherwise. The local rejection
effect is
\[
 \begin{aligned}
 Q_b&=\gamma_{\rm v}(I-P)+(1-\gamma_{\rm v})M',\\
 Q_b&\succeq\gamma_{\rm v}(I-P)+(1-\gamma_{\rm v})\eta_0(P-P_c)
       =\gamma_{\rm v}(I-P_c).
 \end{aligned}
\]
Also $Q_b\ket c=0$ and $0\preceq Q_b\preceq I$, whence
$Q_b\preceq I-P_c$.

\noindent\emph{All blocks in one shot.}
Test the core occupations directly and choose independent local settings
in every block. Requiring every test to pass gives mean acceptance
$P_{\rm core}\otimes\bigotimes_b(I-Q_b)$ in the ordered occupation
basis, since the block operations are even. Its eigenvalue is one on
the proposed product and at most $1-\gamma_{\rm v}$ on the orthogonal
complement. This includes a core alone. Conjugation by $U_{\hat R}$
proves~\eqref{gen:eq:verifier-gap} on the full Fock space. The circuit
uses disjoint pairs, mode phases and passive rotations on at most $n$
modes. Constructing their matrices, sampling settings and processing
the final string take polynomially many operations.
\end{proof}

\noindent\emph{Confidence and expected cost.}
Freeze a proposal with loss $X$. The rejection bits on fresh targets
are independent with probability $p$ satisfying $X/5\le p\le X$.
Set $a=f/5$ and take $(m,k)=\mathsf{BinTest}(a/2,a;1/3,\delta)$.
The test exists: at threshold $\lfloor3ma/4\rfloor$ the two
Chernoff bounds are $e^{-ma/20}$ and $e^{-ma/32}$, so it suffices to take
\[
 m_{\rm cap}=\left\lceil a^{-1}\max\{20\log3,32\log(1/\delta)\}\right\rceil .
\]
A fresh $\mathsf{Fit}(f/10,1/4)$ proposal has $X\le a/2$ with
probability at least $3/4$. The probability of accepting such a
proposal is at least $(3/4)(2/3)=1/2$. The probability of accepting
one with $X>f$ is at most $\delta/4$, since producing one already
has probability at most $1/4$. These statements hold conditional on
every history reaching this round.

For an integer horizon $0\le j\le j_{\max}$, set
\begin{equation}
 j_{\max}=\lceil\log_2(1/\delta)\rceil,\qquad
 z_j=\min\{1/2,\delta(1+2^j)/2\}.
 \label{gen:eq:validation}
\end{equation}
With horizon $j$, start with terminal failure risk $e_0\le z_j$.
Adding one fresh round gives $e_{i+1}\le\delta/4+e_i/2$, whence
\[
 e_j\le\delta/2+2^{-j}(z_j-\delta/2)\le\delta .
\]
The probability of reaching round $i+1$ is at most $2^{-i}$,
and the terminal fit is reached with probability at most $2^{-j}$.
Write $C(g,z)=\mathcal C(n,g,z;D)$ for the cost bound in~\eqref{gen:eq:count}
at infidelity $g$ and failure probability $z$. Charging each reached
round for its fit and test, and the terminal fit when reached, gives
\begin{equation}
 \begin{aligned}
 U_j={}&2(1-2^{-j})[C(f/10,1/4)+m]\\
       &+2^{-j}C(f,z_j),\\
 \E N\le{}&\min_{0\le j\le j_{\max}} U_j .
 \end{aligned}
 \label{gen:eq:finite-verified}
\end{equation}
Choose the smallest minimizing horizon before acquiring data. Since $z_{j_{\max}}=1/2$,
\[
 \min_j U_j\le2[C(f/10,1/4)+m]+C(f,1/2).
\]
For $f\le1/n$, the fixed-confidence bound~\eqref{gen:eq:pilot-order}
and $m\le m_{\rm cap}$ give~\eqref{gen:eq:bound}. The $1+j_{\max}$
candidate horizons and the finite binomial searches require polynomial
classical work in the stated model. Every $z_j\ge\min(\delta,1/2)$,
so even the smallest terminal confidence budget preserves that bound.

\subsection{Proof of the run lower bound}
\label{gen:app:lower}
This proof concerns every learner in the stated access model and does not
use Algorithm~\ref{alg:general-reference}. For the dimension term, take the input
without blocks and without occupied core: $B=0$ and $C_1=\emptyset$, so the
input is $\ket0$, which can be disclosed to the learner.
The target on this subfamily is $U_R\ket0$ with an arbitrary $R\in\SO(2n)$. The proof allows
more than the access model: any measurement with finitely many outcomes of
the single state supplied, private randomness whose law does not depend on
$R$, and any density operator $\hat\rho$ as output. For the vacuum orbit put
\[
\begin{gathered}
 d=n(n-1),\qquad \Psi_0=\ket0\bra0,\\
 X=1-\bra0U_R^\dagger\hat\rho\,U_R\ket0 ,
\end{gathered}
\]
and $\E_R$, $\Prob_R$ refer to the learner run on the rotation $R$. The
vacuum-orbit bound holds for every $n\ge2$. In the vacuum argument the letters $d$, $q$, $Q$,
$V$, $P_{\rm pair}$, $\kappa$, $\chi$ and $\vartheta$ have the meanings given here, not
those of Sections~\ref{gen:app:frames}--\ref{gen:app:schedule}.

\subsubsection{Pair rotations}
For each pair of modes $j<k$ let $E_{jk,1}$ and $E_{jk,2}$ be the real skew
matrices whose only nonzero blocks are $[jk]$ and $[kj]$, with
\[
 (E_{jk,1})_{[jk]}=\frac{1}{\sqrt2}
 \bigl(\begin{smallmatrix}1&0\\0&-1\end{smallmatrix}\bigr),\qquad
 (E_{jk,2})_{[jk]}=\frac{1}{\sqrt2}
 \bigl(\begin{smallmatrix}0&1\\1&0\end{smallmatrix}\bigr),
\]
and the block $[kj]$ equal to minus the transpose of the block $[jk]$.
Number these $d$ matrices $E_1,\dots,E_d$. They are orthonormal for
$\langle A,C\rangle=\tfrac12\Tr(A^{\T}C)$. For $v\in\R^d$ put
\[
 E_v=\sum_{c=1}^{d}v_cE_c,\qquad\text{so}\qquad \|E_v\|_F^2=2|v|^2 .
\]
Let
\[
\begin{gathered}
 \ket{jk}=a_j^\dagger a_k^\dagger\ket0,\qquad
 P_{\rm pair}=\sum_{j<k}\ket{jk}\bra{jk},\\
 \dd_c=\dd(E_c),\qquad \ket{\mathsf t_c}=\dd_c\ket0 .
\end{gathered}
\]
The \emph{pair states} $\ket{jk}$ are orthonormal, and $P_{\rm pair}$ is the
projector on their span.

\begin{alem}[Pair rotations]
\label{gen:lem:pair-rotations}
\begin{enumerate}[label=(\alph*)]
\item With $\kappa=1/\sqrt2$ for $E_{jk,1}$ and $\kappa=-i/\sqrt2$ for
$E_{jk,2}$,
\[
\begin{gathered}
 \dd_c=\kappa\,a_j^\dagger a_k^\dagger-\overline\kappa\,a_ka_j,\\
 \ket{\mathsf t_c}=\kappa\ket{jk},\qquad
 \dd_c\ket{\mathsf t_c}=-\tfrac12\ket0 .
\end{gathered}
\]
\item For every vector $\chi$,
\[
 \sum_{c=1}^{d}\bigl(2\,\mathrm{Re}\braket{\mathsf t_c}{\chi}\bigr)^2
 =2\,\|P_{\rm pair}\chi\|^2 .
\]
\item With $\Upsilon_c:=[\dd_c,\Psi_0]$,
\[
\begin{gathered}
 \Upsilon_c=\ket{\mathsf t_c}\bra0+\ket0\bra{\mathsf t_c},\\ \sum_{c=1}^{d}[\dd_c,\Upsilon_c]=2\,P_{\rm pair}-d\,\Psi_0 .
\end{gathered}
\]
\end{enumerate}
\end{alem}
\begin{proof}
\noindent\emph{(a).}
Let $E$ have the block
$\bigl(\begin{smallmatrix}\alpha&\beta\\\beta&-\alpha
\end{smallmatrix}\bigr)$ at $[jk]$. Since
$\dd(E)=\tfrac12\sum_{a<b}E_{ab}\gref_a\gref_b$,
\[
\begin{aligned}
 \dd(E)&=\tfrac\alpha2\,
 (\gref_{2j-1}\gref_{2k-1}-\gref_{2j}\gref_{2k})\\
 &\quad+\tfrac\beta2\,
 (\gref_{2j-1}\gref_{2k}+\gref_{2j}\gref_{2k-1}) .
\end{aligned}
\]
The definitions of the $\gref_a$ give
\begin{align*}
 \gref_{2j-1}\gref_{2k-1}-\gref_{2j}\gref_{2k}
 &=2\,(a_ja_k+a_j^\dagger a_k^\dagger),\\
 \gref_{2j-1}\gref_{2k}+\gref_{2j}\gref_{2k-1}
 &=2i\,(a_ja_k-a_j^\dagger a_k^\dagger).
\end{align*}
Hence, with $\kappa=\alpha-i\beta$,
\[
 \dd(E)=\kappa\,a_j^\dagger a_k^\dagger+\overline\kappa\,a_ja_k
 =\kappa\,a_j^\dagger a_k^\dagger-\overline\kappa\,a_ka_j .
\]
Since $(a_j^\dagger a_k^\dagger)^2=0$ and
$a_ka_ja_j^\dagger a_k^\dagger\ket0=\ket0$,
\[
 \dd_c\ket{\mathsf t_c}=\kappa\,\dd_c\ket{jk}
 =-|\kappa|^2\ket0=-\tfrac12\ket0 .
\]

\noindent\emph{(b).}
The two matrices of the pair $j<k$ give
\[
 2\,\mathrm{Re}\braket{\mathsf t_c}{\chi}
 =\sqrt2\,\mathrm{Re}\braket{jk}{\chi}
 \quad\text{and}\quad-\sqrt2\,\mathrm{Im}\braket{jk}{\chi} .
\]
Their squares add up to $2|\braket{jk}{\chi}|^2$. Sum over the pairs.

\noindent\emph{(c).}
For an anti-Hermitian operator $\mathfrak b$ and vectors $\chi,\chi'$,
\[
 [\mathfrak b,\ket\chi\bra{\chi'}]
 =\ket{\mathfrak b\chi}\bra{\chi'}+\ket\chi\bra{\mathfrak b\chi'} .
\]
With $\mathfrak b=\dd_c$ and $\chi=\chi'=\ket0$ this is the formula for
$\Upsilon_c$. Applied once more, with (a),
\[
\begin{aligned}
 [\dd_c,\Upsilon_c]
 &=\ket{\dd_c\mathsf t_c}\bra0+2\,\ket{\mathsf t_c}\bra{\mathsf t_c}
   +\ket0\bra{\dd_c\mathsf t_c}\\
 &=2\,\ket{\mathsf t_c}\bra{\mathsf t_c}-\Psi_0 .
\end{aligned}
\]
For each pair, $\ket{\mathsf t_c}\bra{\mathsf t_c}=\tfrac12\ket{jk}\bra{jk}$
for both matrices. Sum over the $d$ values of $c$.
\end{proof}

\subsubsection{The frame of a rotation}
Let $R\in\SO(2n)$ and let $\hat\rho$ be a density operator. Put
$\tilde\rho=U_R^\dagger\hat\rho\,U_R$ and
\begin{equation}
\begin{gathered}
 X=1-\bra0\tilde\rho\ket0,\qquad Q=\Tr(P_{\rm pair}\tilde\rho),\\
 V=\|P_{\rm pair}\tilde\rho\ket0\|^2,\qquad
 \xi_c=\Tr(\Upsilon_c\tilde\rho) .
\end{gathered}
\label{gen:eq:frame-quantities}
\end{equation}
$X$ is the error of $\hat\rho$ for the target of $R$, and $Q$ is the weight
of $\hat\rho$ on the pair states seen from this target. For a function
$\mathsf F$ on $\SO(2n)$ put
\[
 \partial_c\mathsf F(R)=\frac{\mathrm d}{\mathrm ds}\,
 \mathsf F\bigl(Re^{sE_c}\bigr)\Big|_{s=0} .
\]
By \eqref{gen:eq:conventions} and \eqref{gen:eq:quadratic}, up to a phase,
\[
 U_{Re^{sE_c}}=U_R\,e^{s\dd_c} .
\]

\begin{alem}[Relations and derivatives]
\label{gen:lem:frame}
For a fixed density operator $\hat\rho$:
\begin{enumerate}[label=(\alph*)]
\item $\displaystyle\sum_{c=1}^{d}\xi_c^2=2V,\qquad
V\le(1-X)\,Q,\qquad Q\le X$;
\item $\displaystyle\partial_cX=-\xi_c,\qquad
\sum_{c=1}^{d}\partial_c\xi_c=2Q-d\,(1-X)$.
\end{enumerate}
\end{alem}
\begin{proof}
\noindent\emph{(a).}
By Lemma~\ref{gen:lem:pair-rotations}(c), and then (b) with
$\chi=\tilde\rho\ket0$,
\[
 \xi_c=2\,\mathrm{Re}\bra{\mathsf t_c}\tilde\rho\ket0,\qquad
 \sum_{c=1}^{d}\xi_c^2=2\,\|P_{\rm pair}\tilde\rho\ket0\|^2=2V .
\]
Since $\tilde\rho$ is positive semidefinite,
\[
 |\bra{jk}\tilde\rho\ket0|^2
 \le\bra{jk}\tilde\rho\ket{jk}\,\bra0\tilde\rho\ket0 .
\]
Sum over the pairs: $V\le Q\,(1-X)$. Finally $\Psi_0+P_{\rm pair}$ is a projector
and $\Tr\tilde\rho=1$, so $(1-X)+Q\le1$.

\noindent\emph{(b).}
At the rotation $Re^{sE_c}$, $\tilde\rho$ becomes
$e^{-s\dd_c}\tilde\rho\,e^{s\dd_c}$. Its derivative at $s=0$ is
$[\tilde\rho,\dd_c]$, and for every operator $O$
\[
 \Tr\bigl(O\,[\tilde\rho,\dd_c]\bigr)
 =\Tr\bigl([\dd_c,O]\,\tilde\rho\bigr) .
\]
With $O=\Psi_0$ and with $O=\Upsilon_c$ this gives
\[
\begin{gathered}
 \partial_c\bra0\tilde\rho\ket0=\Tr(\Upsilon_c\tilde\rho)=\xi_c,\quad \partial_c\xi_c=\Tr\bigl([\dd_c,\Upsilon_c]\,\tilde\rho\bigr) .
\end{gathered}
\]
Sum the second relation over $c$ and use
Lemma~\ref{gen:lem:pair-rotations}(c).
\end{proof}

\begin{alem}[Invariance]
\label{gen:lem:haar}
Let $\mathrm dR$ be the invariant probability measure on $\SO(2n)$. Let
$\mathsf F$ be a bounded measurable function on $\SO(2n)$ such that
$s\mapsto\mathsf F(Re^{sE_c})$ is Lipschitz with a constant that does not
depend on $R$. Then $\partial_c\mathsf F(R)$ exists for almost every $R$,
and
\[
 \int\partial_c\mathsf F(R)\,\mathrm dR=0 .
\]
\end{alem}
\begin{proof}
\noindent\emph{The derivative exists almost everywhere.}
A Lipschitz function of $s$ is differentiable at almost every $s$. So the
pairs $(R,s)$ at which $s\mapsto\mathsf F(Re^{sE_c})$ has no derivative form
a set of measure zero in $\SO(2n)\times[0,1]$, by Fubini's theorem. Let $Z$
be the set of rotations at which $\partial_c\mathsf F$ does not exist. Then
there is an $s_0$ such that $Re^{s_0E_c}\in Z$ only for $R$ in a set of
measure zero. By the invariance this set has the measure of $Z$. So $Z$ has
measure zero.

\noindent\emph{The integral.}
The invariance also gives, for every $s$,
\[
 \int\bigl[\mathsf F(Re^{sE_c})-\mathsf F(R)\bigr]\,\mathrm dR=0 .
\]
Divide by $s$ and let $s\to0$. The quotients are bounded by the Lipschitz
constant and converge almost everywhere, so the dominated convergence
theorem gives the claim.
\end{proof}

\subsubsection{The information of the runs}
\begin{alem}[One run]
\label{gen:lem:one-run}
Fix $R\in\SO(2n)$, a request and a measurement $(M_y)$ with finitely many
outcomes, and let $p_y(R)$ be the outcome probabilities. Then
\begin{equation}
 \sum_{c=1}^{d}\ \sum_{y:\,p_y>0}\frac{(\partial_cp_y)^2}{p_y}\le d .
 \label{gen:eq:one-run}
\end{equation}
\end{alem}
\begin{proof}
On the vacuum subfamily an unevolved input has zero derivative.
Otherwise $\ket\psi=U_R\ket x$, with $x=0$ for a target request.
By Lemma~\ref{gen:lem:pair-rotations}, for each $j<k$ the two coordinate generators are
\[
 \kappa a_j^\dagger a_k^\dagger-\bar\kappa a_k a_j,\qquad
 \kappa=1/\sqrt2,\;-i/\sqrt2 .
\]
Both annihilate $\ket x$ if $x_j\ne x_k$. For $x_j=x_k$ put
$\ket{\beta_{jk}}=U_R\ket{x\oplus e_j\oplus e_k}$. Creation on $00$
or annihilation on $11$ gives the two tangents
$\pm\ket{\beta_{jk}}/\sqrt2$ and $\pm i\ket{\beta_{jk}}/\sqrt2$.
The occupation-order signs disappear on squaring. With
$p_y=\bra\psi M_y\ket\psi$ the product rule gives
$\partial_c p_y=2\operatorname{Re}\bra\psi M_y\ket{\partial_c\psi}$.
Pairing real and imaginary parts yields
\begin{equation}
\begin{aligned}
 \sum_c(\partial_c p_y)^2
 &=2\sum_{\substack{j<k\\x_j=x_k}}|\bra\psi M_y\ket{\beta_{jk}}|^2\\
 &\le2p_y\sum_{\substack{j<k\\x_j=x_k}}
             \bra{\beta_{jk}}M_y\ket{\beta_{jk}} .
\end{aligned}
\label{gen:eq:one-run-step}
\end{equation}
The inequality is Cauchy--Schwarz after applying $M_y^{1/2}$; it also
gives zero derivative when $p_y=0$. Divide by $p_y>0$, sum over
outcomes, and use $\sum_y M_y=I$. For $r=|x|$ the result is
\[
 \sum_{c,y:\,p_y>0}\frac{(\partial_cp_y)^2}{p_y}
 \le 2\left[\binom r2+\binom{n-r}2\right]\le n(n-1)=d.
\]
\end{proof}

\noindent\emph{Transcripts.}
Let the learner make at most $N_{\max}$ runs, a fixed number, and fix its
private randomness. The transcript $\mathbf y=(y_1,\dots,y_N)$ has finitely
many values. With $y_{<t}=(y_1,\dots,y_{t-1})$,
\[
\begin{gathered}
 P_R(\mathbf y)=\prod_{t\le N}p_t(y_t;R),\\
 S_c=\partial_c\log P_R(\mathbf y)\quad\text{where }P_R(\mathbf y)>0 .
\end{gathered}
\]
Here $p_t(\cdot\,;R)$ is the outcome law of run $t$; its request and its
measurement are fixed by the private randomness and by $y_{<t}$. Each
factor is a smooth function of $R$.

\begin{alem}[Information of a transcript]
\label{gen:lem:transcript-information}
For every $R$ and every value of the private randomness,
\[
 \E_R\sum_{c=1}^{d}S_c^2\le d\,\E_RN .
\]
\end{alem}
\begin{proof}
The request and the measurement of run $t$, and the event $\{N\ge t\}$,
are functions of $y_{<t}$. Where $P_R(\mathbf y)>0$,
\[
 S_c=\sum_{t\le N}r_{t,c},\qquad r_{t,c}=\partial_c\log p_t(y_t;R) .
\]
A differentiable nonnegative function has zero derivative where it
vanishes, the outcome set is finite and $\sum_yp_t(y;R)=1$. So, given
$y_{<t}$ with $N\ge t$,
\[
 \E_R\bigl[r_{t,c}\mid y_{<t}\bigr]
 =\partial_c\sum_yp_t(y;R)=0 .
\]
For $t<t'$ the product $r_{t,c}\mathbf 1\{N\ge t'\}$ is a function of
$y_{<t'}$, so the mixed terms have expectation zero. By
Lemma~\ref{gen:lem:one-run} for each run,
\begin{align*}
 \E_R\sum_cS_c^2
 &=\sum_{t=1}^{N_{\max}}\E_R\Bigl[\mathbf 1\{N\ge t\}
 \sum_c\E_R\bigl[r_{t,c}^2\mid y_{<t}\bigr]\Bigr]\\
 &\le d\sum_{t=1}^{N_{\max}}\Prob_R\{N\ge t\}=d\,\E_RN .\qedhere
\end{align*}
\end{proof}

\subsubsection{Proof of the bound on the expected number of runs}
In this part $\E$ and $\Prob$ refer to a rotation $R$ drawn from the
invariant measure, the private randomness and the transcript.

\begin{alem}[A bounded number of runs]
\label{gen:lem:two-bounds}
Let a learner make at most $N_{\max}$ runs, let $0<f<1$ and
$\mathsf s=\Prob\{X\le f\}$. If
\[
 q:=\mathsf s\,(d-(d+2)f)-2\,(1-\mathsf s)\sqrt f\ \ge0 ,
\]
then $\E N\ge q^2/(2df)$.
\end{alem}
\begin{proof}
Fix the private randomness. For a transcript $\mathbf y$ let
$X,Q,V,\xi_c$ be the quantities \eqref{gen:eq:frame-quantities} of the
output $\hat\rho(\mathbf y)$, as functions of $R$. Put
\[
 \varrho(X)=\sqrt{\frac{f}{\max\{f,X\}}},\qquad
 \zeta_c=\varrho(X)\,\xi_c .
\]
The function $\varrho$ is Lipschitz and equals one on $[0,f]$. At its kink $X=f$ we set
$\varrho'=0$.

\noindent\emph{The derivative of $\zeta_c$.}
By Lemma~\ref{gen:lem:frame}(b), where $X\ne f$,
\[
 \partial_c\zeta_c
 =\varrho(X)\,\partial_c\xi_c-\varrho'(X)\,\xi_c^2 .
\]
Where $X=f$ and $\xi_c=0$, $\zeta_c$ still has the derivative
$\varrho(f)\,\partial_c\xi_c$ along the curve $s\mapsto Re^{sE_c}$, because
$\varrho$ is Lipschitz and $\xi_c$ vanishes there; this is the same formula
for any value of $\varrho'(f)$. The other points, with $X=f$ and
$\xi_c\ne0$, are isolated on every such curve, because $X$ has the
derivative $-\xi_c\ne0$ there. As in the proof of Lemma~\ref{gen:lem:haar}
they form a set of measure zero in $\SO(2n)$. Let $Z_0$ be the union of
these sets over $c=1,\dots,d$. Outside $Z_0$, by Lemma~\ref{gen:lem:frame},
\begin{equation}
\begin{gathered}
 -\sum_{c=1}^{d}\partial_c\zeta_c=\mathcal J,\\
 \mathcal J:=\varrho(X)\bigl(d\,(1-X)-2Q\bigr)+2\,\varrho'(X)\,V .
\end{gathered}
\label{gen:eq:J}
\end{equation}
Where $X=f$ outside $Z_0$, all $\xi_c$ vanish, so $V=0$ and the value of
$\varrho'(f)$ does not enter $\mathcal J$.

\noindent\emph{Integration by parts.}
Apply Lemma~\ref{gen:lem:haar} to the function
\[
 \mathsf F(R)=\sum_{\mathbf y}P_R(\mathbf y)\,\zeta_c(R,\mathbf y) ,
\]
a finite sum of Lipschitz functions with constants that do not depend on
$R$. Where $P_R(\mathbf y)=0$ its derivative vanishes, and elsewhere
$\partial_cP_R(\mathbf y)=S_c\,P_R(\mathbf y)$. Sum over $c$, average over
the private randomness and use \eqref{gen:eq:J}:
\begin{equation}
 \E\sum_{c=1}^{d}\zeta_c\,S_c=\E\,\mathcal J .
 \label{gen:eq:haar-identity}
\end{equation}
This is the argument of the van Trees inequality~\cite{gilllevit} with the
invariant measure as prior: there is no boundary term and no term from the
prior.

\noindent\emph{Cauchy--Schwarz.}
By Lemma~\ref{gen:lem:frame}(a) and
Lemma~\ref{gen:lem:transcript-information}, averaged over $R$ and the
private randomness,
\[
 \sum_{c=1}^{d}\zeta_c^2=2\,\varrho(X)^2V,\qquad
 \E\sum_{c=1}^{d}S_c^2\le d\,\E N .
\]
So \eqref{gen:eq:haar-identity} gives
\begin{equation}
 \bigl(\E\,\mathcal J\bigr)^2
 \le2\,\E\bigl[\varrho(X)^2V\bigr]\,d\,\E N .
 \label{gen:eq:haar-cs}
\end{equation}

\noindent\emph{The two sides.}
By Lemma~\ref{gen:lem:frame}(a), $V\le(1-X)Q\le X$. So, in all cases,
\[
 \varrho(X)^2V\le\varrho(X)^2X=\min\{X,f\},\qquad
 \E\bigl[\varrho(X)^2V\bigr]\le f .
\]
If $X\le f$, then $\varrho=1$ and $\varrho'=0$, so
\[
 \mathcal J=d\,(1-X)-2Q\ge d-(d+2)X\ge d-(d+2)f .
\]
If $X>f$, then $\varrho'(X)=-\varrho(X)/(2X)$ and $V\le(1-X)X$, so
\begin{align*}
 \mathcal J&=\varrho(X)\Bigl(d\,(1-X)-2Q-\frac VX\Bigr)\\
 &\ge\varrho(X)\bigl((d-1)(1-X)-2X\bigr)\\
 &\ge-2\varrho(X)X=-2\sqrt{fX}\ge-2\sqrt f .
\end{align*}
Therefore
\[
 \E\,\mathcal J\ge\mathsf s\,(d-(d+2)f)-2(1-\mathsf s)\sqrt f=q\ge0 .
\]
Insert this in \eqref{gen:eq:haar-cs}.
\end{proof}

\noindent\emph{Proof of~\eqref{gen:eq:lower}.}
If $(1+\tfrac2d)(f+\delta)=1$, the right side of \eqref{gen:eq:lower} is
zero and there is nothing to prove. Let $(1+\tfrac2d)(f+\delta)<1$, and let
$\bar N$ be the supremum of $\E_RN$ over all rotations. If $\bar N=\infty$
there is nothing to prove. Otherwise fix $\epsilon>0$ and an integer
$N_{\max}\ge\max\{1,\bar N/\epsilon\}$. The stopped learner follows the
given learner; if run $N_{\max}$ is over and the given learner has not
stopped, it stops and returns $\Psi_0$. It makes $N'\le N$ runs. When
$N\le N_{\max}$ both learners return the same state, and by Markov's
inequality, for every $R$,
\[
 \Prob_R\{N>N_{\max}\}\le\epsilon,\qquad
 \Prob_R\{X\le f\}\ge1-\delta-\epsilon
\]
for the stopped learner. Hence $\mathsf s\ge1-\delta-\epsilon$. The number
$q$ of Lemma~\ref{gen:lem:two-bounds} increases with $\mathsf s$, because
$d-(d+2)f\ge0$. Therefore
\begin{align*}
 \frac qd&\ge(1-\delta-\epsilon)\bigl(1-(1+\tfrac2d)f\bigr)
 -\frac{2(\delta+\epsilon)\sqrt f}{d}\\
 &\ge1-\bigl(1+\tfrac2d\bigr)(f+\delta+\epsilon) .
\end{align*}
The last step uses $\sqrt f\le1$ and drops a nonnegative term. For small
$\epsilon$ the right side is positive, and Lemma~\ref{gen:lem:two-bounds}
gives, for the given learner,
\[
 \E N\ge\E N'\ge\frac d{2f}
 \Bigl(1-\bigl(1+\tfrac2d\bigr)(f+\delta+\epsilon)\Bigr)^2 .
\]
Let $\epsilon\to0$. Here $\E N$ is the average of $\E_RN$ over the
invariant measure, so some rotation $R$ has $\E_RN$ at least the right side
of \eqref{gen:eq:lower}.\hfill$\square$

\noindent\emph{Dimension term in \eqref{gen:eq:complexity}.}
Let $n\ge6$, $f\le1/n$ and $\delta\le1/2$. Then
\[
 1+\tfrac2d\le\tfrac{16}{15},\qquad f\le\tfrac16,\qquad
 n(n-1)\ge\tfrac56n^2 .
\]
So \eqref{gen:eq:lower} gives
\[
\begin{gathered}
 \E N\ge\frac5{12}\Bigl(1-\frac{16}{15}\bigl(\tfrac16+\delta\bigr)\Bigr)^2
 \frac{n^2}{f}\ge c_\star\,\frac{n^2}{f},\\
 c_\star=\frac5{12}\Bigl(\frac{13}{45}\Bigr)^2>\frac1{29} .
\end{gathered}
\]
The bracket is at least $\tfrac{13}{45}$ because
$\tfrac{16}{15}(\tfrac16+\tfrac12)=\tfrac{32}{45}$.

\subsubsection{Confidence cost under adaptive stopping}
\noindent\emph{A two-branch subfamily.}
Fix $R=I$, $F_1=\{1,2,3\}$, $F_2=\{4,5,6\}$ and vacant remaining modes.
On the logical qubit $\mathcal V=\spn\{\ket{F_1},\ket{F_2}\}$ use the
uniform pure-state prior
\[
 \begin{aligned}
 \ket\psi&=\sqrt{x}\ket{F_1}+e^{i\phi}\sqrt{1-x}\ket{F_2},\\
 d\pi&=dx\,d\phi/(2\pi),\qquad 0<x<1,\quad 0\le\phi<2\pi .
 \end{aligned}
\]
Input and target requests both supply $\ket\psi$. Reference requests
supply known Fock states. Allow any single-state finite-outcome
measurement and any density operator as the report.

\noindent\emph{Likelihood of a stopped record.}
Fix all private randomness. On a finite terminal path $h$, let $k$ be
the number of input and target requests and let $E_i$ be the observed
effects compressed to $\mathcal V$. The reference outcomes supply a
scalar likelihood $c_h$ independent of $\psi$ on this path. The settings
and stopping decision are fixed by the history, so, with $S_k$ the
symmetric projector on $\mathcal V^{\otimes k}$,
\[
 \begin{aligned}
 K_h&=c_h S_k(E_1\otimes\cdots\otimes E_k)S_k\succeq0,\\
 \Prob(h\mid\psi)&=c_h\prod_{i=1}^k\bra\psi E_i\ket\psi
 =\bra{\psi^{\otimes k}}K_h\ket{\psi^{\otimes k}} .
 \end{aligned}
\]
The empty tensor product at $k=0$ is the scalar one.

\noindent\emph{Posterior failure.}
Compress the reported density operator to $\mathcal V$. Let its eigenvalues
be $\lambda_1\ge\lambda_2\ge0$, with $\lambda_1+\lambda_2\le1$, and $u$
a leading unit eigenvector. Write $y=|\braket u\psi|^2$. The fidelity is
$\lambda_1y+\lambda_2(1-y)$, at most $1/2$ for $y\le1/2$ and at most
$y$ otherwise. Thus $y<z:=1-f$ implies failure when $f<1/2$.

The prior in the basis $(u,u^\perp)$ is again uniform in $y$ and the
relative phase. Write $\ket{D_j}$ for the normalized symmetric ket with
$j$ factors in $u$ among $k$. Phase integration removes the cross terms,
and $|\braket{D_j}{\psi^{\otimes k}}|^2=\binom kj y^j(1-y)^{k-j}$.
For the cap operator $A_z=\int_{y<z}(\ketbra\psi\psi)^{\otimes k}\,d\pi$,
the diagonal entry at $D_j$ is
\[
\begin{aligned}
 a_j(z)&=\binom kj\int_0^z y^j(1-y)^{k-j}\,dy\\
 &=\binom kj z^{j+1}\int_0^1t^j(1-zt)^{k-j}\,dt\\
 &\ge z^{k+1} a_j(1).
\end{aligned}
\]
Indeed $1-zt\ge1-t$ and $z^{j+1}\ge z^{k+1}$.
The beta integral gives $a_j(1)=1/(k+1)$, hence
$A_1=S_k/(k+1)$ and $A_z\succeq z^{k+1}A_1$, including $k=0$.
Since $K_h\succeq0$, Bayes' rule on a path of positive prior probability gives
\[
 \Prob_\pi(X>f\mid h)\ge\frac{\Tr(K_h A_z)}{\Tr(K_h A_1)}
 \ge z^{k+1}.
\]
Finite worst-case expected cost ensures almost-sure termination under
the prior; infinite worst-case cost already implies the bound. Average
over terminal paths and private randomness. With $K$ the random number
of input and target requests, Jensen's inequality gives
\[
 \delta\ge\Prob_\pi(X>f)\ge\E_\pi z^{K+1}
        \ge z^{\E_\pi K+1},\qquad K\le N .
\]
Rearrangement proves~\eqref{gen:eq:confidence-lower}. For $n\ge6$,
$f\le1/n$, $L=\log(1/\delta)\ge\log2$,
\[
 \frac{L}{-\log(1-f)}-1
 \ge\frac{(1-f)L-f}{f}\ge\frac{L}{2f},
\]
since $f\le1/6$ and $L\ge1/2$. The supremum over the full family is
at least each subfamily bound, even if their worst targets differ.
Together with~\eqref{gen:eq:lower} this proves the lower side
of~\eqref{gen:eq:complexity}.\hfill$\square$
\section{Worked examples and learning algorithms}
\label{app:worked-examples}
\subsection{Dimer setting and logical geometry}
\label{app:example-geometry}
For each dimer the occupation order is
\mbox{$(L{\uparrow},L{\downarrow},R{\uparrow},R{\downarrow})$}.
The state in~\eqref{eq:example-target} at $\phi=0$ has one electron on each
site and antisymmetric spin amplitudes. Its phase-varied family remains on
the same paired orbit. For $n$ modes ($n=4B$ for $B$ dimers; $n=6,8,10,12$ in
Section~\ref{app:block-example}), set
\begin{equation}
\begin{aligned}
H_q&=-\sum_{j=1}^{n-1}J_j(a_j^\dagger a_{j+1}+a_{j+1}^\dagger a_j)\\
&\quad-\sum_{j=1}^{n}\mu_j(n_j-\tfrac12)\\
&\quad+\sum_{j=1}^{n-1}\Delta_j(a_j^\dagger a_{j+1}^\dagger+a_{j+1}a_j),\\
J_j&=0.8+0.13\sin j,\\
\Delta_j&=0.55+0.07\cos(2j-1),\\
\mu_j&=0.35+0.11\cos(j-1),\qquad \tau=0.85.
\end{aligned}
\label{eq:example-quench}
\end{equation}
Here $n_j=a_j^\dagger a_j$, $1\le j<n$ for $J_j,\Delta_j$, and
$1\le j\le n$ for $\mu_j$. Hopping moves particles; pairing creates or
removes a pair. This specifies a quadratic-chain benchmark, not a material
or a ground-state preparation; the spin labels of the input do not impose
spin conservation on the quench. The coefficients are fixed formulas, so
the benchmark needs no random seed; they vary from site to site, keep
$\Delta_j\ne0$, and stay away from the point $\Delta_j=J_j$, $\mu_j=0$,
where the quench splits into independent pairs of Majorana operators.
With $\widehat N=\sum_{j=1}^n n_j$,
\[
\begin{gathered}
[\widehat N,H_q]=2\sum_{j=1}^{n-1}\Delta_j
 (a_j^\dagger a_{j+1}^\dagger-a_{j+1}a_j)\ne0,
 \\
  [(-1)^{\widehat N},H_q]=0.
\end{gathered}
\]
For $\phi=0$, direct evaluation gives
$\Var(\widehat N)=0.6064815140$ and $1.5144615602$ for $n=4,8$.

\noindent\emph{Membership and transport.}
Let $V_\phi$ be the passive unitary, fixing the vacuum, whose creation-mode
map is
\[
 (a_1^\dagger,a_2^\dagger,a_3^\dagger,a_4^\dagger)
\mapsto(e^{i\phi}a_1^\dagger,a_4^\dagger,-a_2^\dagger,a_3^\dagger).
\]
Then
\[
\begin{gathered}
 V_\phi\ket\omega
 =\frac{e^{i\phi}a_1^\dagger a_4^\dagger-a_2^\dagger a_3^\dagger}{\sqrt2}\ket0
 =\ket{s_\phi},\\
 \ket{\Psi_\phi^{(B)}}=U_qV_\phi^{\otimes B}\ket{\Omega_B},\quad U_q=e^{-i\tau H_q}.
\end{gathered}
\]
In the Majorana convention of Appendix~\ref{app:fermions}, the only
independent nonzero entries of the real antisymmetric matrix $A$ are
\[
\begin{gathered}
 A_{2j-1,2j}=-\mu_j,\quad A_{2j-1,2j+2}=-(J_j+\Delta_j),\\
 A_{2j,2j+1}=J_j-\Delta_j.
\end{gathered}
\]
The CAR give the complete physical-to-observable transport:
\begin{equation}
\begin{gathered}
 H_q=\frac i4\sum_{p,q=1}^{2n}A_{pq}\gamma_p\gamma_q,\quad i[H_q,\gamma_p]=\sum_{q=1}^{2n}A_{pq}\gamma_q,\\
 U_q^\dagger\gamma U_q=e^{\tau A}\gamma,\quad R_{\rm pair}=e^{\tau A}R_{V_\phi^{\otimes B}}\in SO(2n),\\
 x(\rho_\phi^{(B)})=\bigwedge^4R_{\rm pair}\,x(P_0),\quad \rho_\phi^{(B)}=\ketbra{\Psi_\phi^{(B)}}{\Psi_\phi^{(B)}}.
\end{gathered}
 \label{eq:example-transport}
\end{equation}
Equation~\eqref{eq:example-transport} proves membership in the promised
orbit, including the inter-dimer coupling at $j=4$ when $n=8$.

\noindent\emph{Logical sphere.}
Put $\ket u=\ket{1001}$, $\ket v=\ket{0110}$ and
$\ket\chi=c\ket u+e^{i\alpha}s\ket v$, with
$c=\cos(\theta/2)$ and $s=\sin(\theta/2)$. Its projective metric follows from
\[
\begin{gathered}
d\ket\chi =\tfrac12(-s\ket u+e^{i\alpha}c\ket v)d\theta
                +i e^{i\alpha}s\ket v\,d\alpha,\\
 \langle d\chi\mid d\chi\rangle =\tfrac14d\theta^2+s^2d\alpha^2,\\
 \langle\chi\mid d\chi\rangle =i s^2d\alpha,\\
\begin{aligned}
 ds_{\rm FS}^2 &=\langle d\chi\mid d\chi\rangle
                  -|\langle\chi\mid d\chi\rangle|^2\\
&=\tfrac14(d\theta^2+\sin^2\theta\,d\alpha^2).
\end{aligned}
\end{gathered}
\]
The unit Bloch sphere has metric $4g_{\rm FS}$; the targets follow its
great circle $\theta=\pi/2$, $\alpha=\pi-\phi$.
For a reconstructed four-mode state, undo the quench only for evaluation:
\begin{equation}
\begin{gathered}
 \tilde\rho=U_q^\dagger\hat\rho U_q,\quad p_u=\bra u\tilde\rho\ket u,\\
 p_v=\bra v\tilde\rho\ket v,\quad z=\bra v\tilde\rho\ket u,\\
 w=p_u+p_v,\quad b=(2\Re z,2\Im z,p_u-p_v),\\
 \|b\|^2=(p_u-p_v)^2+4|z|^2\le w^2\le1,\\
 b_\phi=(-\cos\phi,\sin\phi,0).
\end{gathered}
 \label{eq:example-logical}
\end{equation}
Positivity gives $|z|^2\le p_up_v$. Consequently leakage is retained in
the plotted vector, rather than removed by conditioning or normalization.
Figure~\ref{fig:paired-logical} rotates this display by
$\mathsf D=R_x(-\pi/3)$; its three displayed axes are
\[
\begin{gathered}
 e_x=(1,0,0),\quad e_y=(0,1/2,-\sqrt3/2),\\
 e_z=(0,\sqrt3/2,1/2),\quad \mathsf D b=b_xe_x+b_ye_y+b_ze_z.
\end{gathered}
\]
These are observable-coordinate directions, not state kets.
For one block, occupation-basis conjugation $\mathsf K_{\rm occ}$ gives
\[
\begin{gathered}
 \Theta=-\gamma_2\gamma_4\gamma_6\gamma_8\mathsf K_{\rm occ}=-J,\quad \Theta^2=I,\\
 \Theta\ket\omega=\ket\omega,\quad \Theta U_R=U_R\Theta,\\
 \Theta\ket u=\ket v,\quad \Theta\ket v=\ket u.
\end{gathered}
\]
Here $U_R$ is its spin lift; an overall phase does not affect the density
operator. Equation~\eqref{eq:rec-real} therefore implies
\[
 \Theta\tilde\rho\Theta^{-1}=\tilde\rho
 \quad\Longrightarrow\quad
 p_u=p_v
 \quad\Longrightarrow\quad
 b_z=0.
\]
Even inaccurate orbit outputs lie in this plane. Their angle, length and
full-state error still distinguish them from their targets; this logical
projection does not parametrize the full orbit.

\noindent\emph{Tangential and radial errors.}
For each trial use its target phase $\phi$ and the unnormalized vector $b$
in~\eqref{eq:example-logical}. With $\Delta b=b-b_\phi$,
\[
\begin{gathered}
 t_\phi=\partial_\phi b_\phi =(\sin\phi,\cos\phi,0), \quad t_\phi\cdot b_\phi=0,\\
 \|t_\phi\|=\|b_\phi\| =1, \quad e_{\rm t}=\Delta b\cdot t_\phi =b_x\sin\phi+b_y\cos\phi,\\[3pt]
\begin{aligned}
 e_{\rm r}
 &=\Delta b\cdot b_\phi =-b_x\cos\phi+b_y\sin\phi-1\\
 &\le\|b\|\|b_\phi\|-1 \le0,
\end{aligned}\\[3pt]
 \Delta b=e_{\rm t}t_\phi+e_{\rm r}b_\phi+(0,0,b_z),\\
 \|\Delta b\|^2=e_{\rm t}^2+e_{\rm r}^2+b_z^2 =e_{\rm t}^2+e_{\rm r}^2.
\end{gathered}
\]
The last equality uses $b_z=0$ above. Orthogonality of $\mathsf D$ preserves
both components after the display rotation. These are signed projections,
not an angular error or full-state infidelity. The lower panels retain all
trials and phase colours, with common axis limits and no clipping.
\subsection{Quartic records, reconstruction and error checks}
\label{app:example-records}
Let $r=\binom{2n}{4}$. The orbit identities in Appendix~\ref{app:reconstruction}
give, with multiplicities in brackets,
\[
\begin{gathered}
\|x\|^2=14B,\\
 \operatorname{spec}Q(x)=\{-3^{[7B]},1^{[21B]},0^{[\binom{8B}{2}-28B]}\},
 \\
 (n,r,\dim\cO)=\begin{cases}(4,70,7),&B=1,\\(8,1820,78),&B=2.\end{cases}
\end{gathered}
\]
Each of the four banks $b\in\{A,P,W,R\}$ uses separate copies for every
quartet $I\in\mathcal I_4(2n)$. Exact binomial sampling generates the
sufficient statistics of the Born outcomes:
\begin{equation}
\begin{gathered}
 M_{bI}\sim\operatorname{Bin}\!\left(K_b,\frac{1+x_I}{2}\right),\quad y_{bI}=2M_{bI}/K_b-1,\\
 \E y_{bI}=x_I,\quad \Var(y_{bI})=(1-x_I^2)/K_b,\\
 (K_A,K_P,K_W,K_R)=(K,K,K,\lceil K/\epsilon^2\rceil),\\
 N_{\rm copies}=r(3K+\lceil K/\epsilon^2\rceil).
\end{gathered}
 \label{eq:example-records}
\end{equation}
All banks are acquired before reconstruction. Set $\epsilon^2=\delta=1/8$;
then~\eqref{eq:example-records} gives~\eqref{eq:example-copies-main}.
All four banks are acquired and counted, although the one-block
specialization reads only banks $W$ and $R$.
The theorem chooses $K$ through a sufficiently large constant in
\eqref{eq:opt-bank-budgets}; the experiment sweeps $K$ without asserting
that each choice reaches that threshold.

\noindent\emph{Specialization of Algorithm~\ref{alg:paired-quartic}.}
For $B=1$ the candidate space is all eight Majorana directions. The case
$B=2$ ($n=8$) was run only in the accuracy sweep that is not shown; its
candidate scan is recorded here for completeness. For $B=2$,
the code estimates and independently refreshes the quartic range, forms
the measured commutant, and applies the prescribed candidate filters.
Identity lies in its kernel. Thus its two-dimensional low space has an
orthonormal basis $I/4,S$, where $\Tr S=0$ and $\|S\|_F=1$, and
\[
\begin{gathered}
 Z=g_0I/4+g_1S,\quad g_0,g_1\overset{\rm iid}{\sim}N(0,1),\\
 h=\lambda_9(S)-\lambda_8(S).
\end{gathered}
\]
The eigenvalues of $S$ are in ascending order. Only its top and bottom
eight-dimensional eigenspaces can be proposed. If the residual test
of~\eqref{eq:opt-pencil-tests} passes for both eigenspaces and $h>0$, each
sign of $g_1$ passes the gap test of~\eqref{eq:opt-pencil-tests}
with probability $p=\Phi(-1/(32h))$. The full capped scan can therefore
be sampled exactly in real arithmetic by
\[
\begin{gathered}
 W_1\sim\operatorname{Geom}(2p),\quad W_2\sim\operatorname{Geom}(p),\quad W_1\perp W_2,\\
 \Prob\{W_1+W_2\le T\}
 =1-2(1-p)^T+(1-2p)^T.
\end{gathered}
\]
The cap $T$ is exactly~\eqref{eq:opt-probe-count}.
The first accepted sign is uniform. The displayed probability is
inclusion--exclusion for observing both signs in $T$ independent draws;
if the residual test fails for either eigenspace, two distinct candidates
cannot be retained.
This compression preserves the ideal candidate-list law and the full cap.
It does not use true blocks or truncate unsuccessful scans early.
The retained list undergoes the same three corrections after each insertion,
using bank $W$, and the independent bank-$R$ refinement in
Lemma~\ref{lem:opt-refinement}. For the parameters above this is 49 projected
gradient iterations, with step $1/18$ and radius $1/256$. All guards and
fallback outputs are retained. Output plane-rotation angles are rounded
on a grid of spacing $\min\{10^{-12},\epsilon/(8n^2)\}$ and evaluated
using the resulting circuit. Floating-point eigensolvers and ideal-law
pseudorandom sampling implement the real-arithmetic procedure numerically,
not an interval-certified realization of every finite-bit guard.

\noindent\emph{Certificate and numerical protocol.}
For orbit outputs $\sigma=\hat\rho$, Corollary~\ref{cor:global-parent} gives
\[
\begin{aligned}
 K_\rho
 &=BI-\frac1{14}\sum_{I\in\mathcal I_4(2n)}x_I(\rho)B_I \succeq I-\rho,\\[3pt]
 \trd(\rho,\sigma)^2
 &=1-\Tr(\rho\sigma) \le\Tr(K_\rho\sigma)\\
 &=B-\frac{\langle x(\rho),x(\sigma)\rangle}{14} =\frac{\|x(\rho)-x(\sigma)\|^2}{28}.
\end{aligned}
\]
For $B=1$, \eqref{eq:block-expansion} and quartic trace orthogonality sharpen
the diagnostic to the exact identity
\[
 \|x(\rho)-x(\sigma)\|^2
 =16\Tr[(\rho-\sigma)^2]=32\trd(\rho,\sigma)^2.
\]
Hence four-mode infidelity equals $7/8$ of the certificate, up to numerical
error. This remains an independent numerical check, rather than a plotted
axis; both sides use simulator truth, not an additional measurement or an
empirical confidence bound supplied to the learner.
The sphere uses $\phi=2\pi j/24$ ($0\le j<24$), three trials per phase,
and $K\in\{1,2,4\}$. For these three budgets the 72 trials contain 4, 3
and 0 fallbacks, their median infidelities are $0.175$, $0.093$ and
$0.042$, and the infidelity is at most $\epsilon^2=1/8$ in 14, 54 and 71
trials. In all 209 trials without fallback the refinement iterate ends on
the boundary of its chart ball of radius $1/256$ and changes the
infidelity by at most $0.0027$. Seed 2026101501 fixes a run of 744 trials. The sphere
trials are its last 216; they follow an accuracy sweep of 528 trials that
is not shown. Independent
count banks and algorithm seeds are saved. All 744 outputs, including
fallbacks, obeyed the certificate. Fourteen noiseless checks on dimer and
Haar-rotated targets, seven on four modes and seven on eight modes (not
shown), recovered their targets to numerical precision.
The dense Fock-space simulator is
used only to generate outcomes and evaluate error, not as a learning
subroutine. Aggregating counts changes neither their Born law nor the
reported number of copies; simulation time is not experimental shot time.
\subsection{One block with reference preparations}
\label{app:block-example}
\noindent\emph{Setting.}
For six modes, take the quench~\eqref{eq:example-quench} on $n=6$ modes, with the same
coefficient formulas, and put
\[
\begin{gathered}
 \ket u=\ket{111000}=a_1^\dagger a_2^\dagger a_3^\dagger\ket0,\\
 \ket v=\ket{000111}=a_4^\dagger a_5^\dagger a_6^\dagger\ket0 .
\end{gathered}
\]
Then~\eqref{eq:block-example} reads
$\ket{\Xi_\phi}=(\ket u+e^{i\phi}\ket v)/\sqrt2$ and
$\ket{\Psi_\phi}=U_q\ket{\Xi_\phi}$. In~\eqref{gen:eq:family} the cores are
empty and there is one block with two branches:
\[
\begin{gathered}
 B=1,\quad L_1=2,\quad F_{11}=\{1,2,3\},\quad F_{12}=\{4,5,6\},\\
 c_{11}=1/\sqrt2,\qquad c_{12}=e^{i\phi}/\sqrt2 .
\end{gathered}
\]
Both branches have three modes. So their sizes have the same parity and
satisfy~\eqref{gen:eq:size-three}.

\noindent\emph{Circuit.}
Appendix~\ref{app:general-reference} uses the reflected basis
$\gref=S_0\gamma$. Every nonzero entry of the matrix $A$ of
Section~\ref{app:example-geometry}, taken for $n=6$, joins an odd and an
even index, so
$S_0AS_0=-A$. The transport rule $U_q^\dagger\gamma U_q=e^{\tau A}\gamma$
of~\eqref{eq:example-transport} therefore gives the rotation of the circuit
in the reflected basis:
\[
 U_q=U_R,\qquad R=S_0e^{\tau A}S_0=e^{-\tau A}\in\SO(12).
\]
Here $R$ is the rotation of Section~\ref{gen:sec:setting}, not the matrix
$R_{\rm pair}$ of~\eqref{eq:example-transport}. The learner knows neither $H_q$, $R$,
$\phi$ nor the branch sets. The operator $\widehat N^2$ has degree four in
the Majorana operators, and the two branches differ in six modes. So
$\Var(\widehat N)$ of the target does not depend on $\phi$; direct
evaluation gives $\Var(\widehat N)=1.5069017647$.

\noindent\emph{Schedule.}
Take six modes and the request $(f,\delta)=(1/12,1/2)$, one of the requests of
Fig.~\ref{fig:reference-runs}. The public schedule chooses horizon zero
and one refinement. The attempt caps are $(A_0,A_1)=(6,5)$ and
$u_{\rm c}=u_{\rm in}\simeq0.02882$. The batches are
\[
\begin{gathered}
 (N_{\rm s},N_{\rm w},N_\phi)=(9275,1061,2126),\\
 (N_a,N_b,N_1)=(43\,185,41\,255,91\,707),\\
 (m_0,m_1,m_\ell)=(25\,582,45\,564,8349).
\end{gathered}
\]
The corresponding acceptance thresholds are 157, 131 and 312.
With two coherence settings and no retry, the measured total is
\[
\begin{aligned}
 N={}&N_{\rm s}+N_{\rm w}+N_a+N_b+N_1\\
     &+m_0+m_1+m_\ell+2N_\phi=270\,230.
\end{aligned}
\]
The expected cost on any allowed target at this request is bounded
by $519\,092$ after rounding upward.

\noindent\emph{Sampling.}
Every run is drawn from its exact outcome law. A reference run $U_R\ket x$
is Gaussian with covariance $R\Gamma_xR^{\T}$. Its occupation record, after
a known Gaussian circuit, is drawn mode by mode from the rotated covariance
by Wick's theorem. Input records are drawn from the exact pattern law of
$\ket\Xi$, the products of the branch weights. Target copies are not
Gaussian; their outcome laws are computed exactly, on the Fock space of
dimension $2^n$. The runs of one
test are independent, so its count is binomial. The binomial parameter is the probability of a counted event: a miss
for the leakage test or a click for the click test. For
a candidate $\hat R$, the latter is given by Section~\ref{gen:app:certificate},
\[
 \bar p=1-2^{-n}\sum_{x\in\{0,1\}^n}
 \bigl|\bra xU_{\hat R}^\dagger U_R\ket x\bigr|^2 .
\]
In the last stage the simulation runs and counts both tests in every
attempt.

\noindent\emph{Logical vector.}
For an output state $\hat\rho$, undo the quench only for evaluation:
\begin{equation}
\begin{gathered}
 \tilde\rho=U_q^\dagger\hat\rho\,U_q,\qquad z=\bra v\tilde\rho\ket u,\\
 p_u=\bra u\tilde\rho\ket u,\quad p_v=\bra v\tilde\rho\ket v,\quad
 w=p_u+p_v,\\
 b=(2\Re z,\,2\Im z,\,p_u-p_v),\\
 \|b\|^2=(p_u-p_v)^2+4|z|^2\le w^2\le1 .
\end{gathered}
\label{eq:block-logical}
\end{equation}
The inequality uses $|z|^2\le p_up_v$. So weight outside the two branches
shortens $b$; it is not removed by conditioning or normalization. The
states $\ket u,\ket v$ are the poles, and the target $\ket{\Xi_\phi}$ has
\[
 b=b_\phi=(\cos\phi,\sin\phi,0)
\]
on the equator. The true deviations are too small to see on a sphere. So
Fig.~\ref{fig:block-sphere} draws the direction of $b$ at ten times its
angle from the target: with $\cos\vartheta=b\cdot b_\phi/\|b\|$,
$0\le\vartheta\le\pi$, and $s$ the unit vector along
$b-(b\cdot b_\phi)b_\phi$, it shows
\[
 \mathsf D\bigl(\cos(10\vartheta)\,b_\phi+\sin(10\vartheta)\,s\bigr),
\]
with the display rotation $\mathsf D$ and the axes of
Section~\ref{app:example-geometry}; for $\vartheta=0$ it shows $\mathsf Db_\phi$.
Every drawn point lies on the sphere. The magnification concerns the
spheres only. The circle of
Fig.~\ref{fig:paired-logical} runs the other way, since there
$b_\phi=(-\cos\phi,\sin\phi,0)$.

\noindent\emph{Errors and infidelity.}
For each trial use its own target phase $\phi$ and put $\Delta b=b-b_\phi$.
With the unit tangent
$t_\phi=\partial_\phi b_\phi=(-\sin\phi,\cos\phi,0)$,
\[
\begin{gathered}
 e_{\rm t}=\Delta b\cdot t_\phi,\qquad
 e_{\rm r}=\Delta b\cdot b_\phi=b\cdot b_\phi-1\le0,\\
 \Delta b=e_{\rm t}t_\phi+e_{\rm r}b_\phi+(0,0,b_z),\\
 \|\Delta b\|^2=e_{\rm t}^2+e_{\rm r}^2+b_z^2 .
\end{gathered}
\]
Since
$b=(1+e_{\rm r})b_\phi+e_{\rm t}t_\phi+(0,0,b_z)$ and $\|b\|\le1$,
\[
 e_{\rm t}^2+b_z^2\le1-(1+e_{\rm r})^2\le-2e_{\rm r} .
\]
The full-state infidelity is
$X=1-\bra{\Xi_\phi}\tilde\rho\ket{\Xi_\phi}$:
\begin{equation}
\begin{aligned}
 \bra{\Xi_\phi}\tilde\rho\ket{\Xi_\phi}
 &=\tfrac12\bigl(p_u+p_v+2\Re(e^{-i\phi}z)\bigr)\\
 &=\tfrac12(w+b\cdot b_\phi),\\
 X&=\tfrac12(1-w)-\tfrac12e_{\rm r} .
\end{aligned}
\label{eq:block-infidelity}
\end{equation}
Since $w\le1$ and $b\cdot b_\phi\le\|b\|\le w$,
\[
 -\tfrac12e_{\rm r}\le X\le-e_{\rm r} .
\]
The left inequality is an equality if and only if $w=1$. So $X\le f$ gives
$e_{\rm r}\ge-2f$, and $e_{\rm r}\ge-f$ gives $X\le f$.
Write $e_{\rm w}=b_z=p_u-p_v$ for the weight error; the target has equal
weights, so $e_{\rm w}=0$ there. The bounds above combine to
\begin{equation}
 e_{\rm t}^2+e_{\rm w}^2\le-2e_{\rm r}\le4X .
 \label{eq:block-board}
\end{equation}
The lower panels of Fig.~\ref{fig:block-sphere} plot $(e_{\rm t},e_{\rm w})$
at true scale. A trial at distance $r$ from the centre therefore has
$X\ge r^2/4$. These panels do not show $e_{\rm r}$ or the weight $1-w$
outside the two branches, so a point near the centre need not have small $X$. At $n=6$ the fallback output $\ket0$ is even, the
two branches are odd, and $U_q$ preserves parity. So the fallback has
$p_u=p_v=z=0$, hence $b=0$, $e_{\rm t}=e_{\rm w}=0$, $e_{\rm r}=-1$ and
$X=1$. Such an output has no direction and would sit at the centre of a
lower panel; no trial of Fig.~\ref{fig:block-sphere} fell back.

\noindent\emph{Records.}
Figure~\ref{fig:block-sphere} uses $n=6$, $\phi=2\pi j/24$, $0\le j<24$,
and the requests $f=1/(6\cdot2^k)$, $k\in\{0,3,6\}$, with two trials
$t\in\{0,1\}$ at each phase and seed tuple $(2026112804,k,j,t)$. The 144
trials use $\delta=1/2$; the error plots include all outputs and use the
same axes. Every trial selected direct reconstruction, every stage
accepted its first attempt, and no trial fell back. At $f=1/6$, $1/48$
and $1/384$ the median full-state infidelities were $1.12\times10^{-3}$,
$2.92\times10^{-4}$ and $5.18\times10^{-5}$, and the largest were
$3.73\times10^{-3}$, $7.86\times10^{-4}$ and $9.99\times10^{-5}$; the
largest ratios $X/f$ were $0.0225$, $0.0378$ and $0.0384$, rounded upward.
The largest angle between $b$ and its target was $6.4^\circ$, so every
magnified point stays within $64^\circ$ of its target.

\noindent\emph{Larger blocks.}
Figure~\ref{fig:reference-runs} uses the same construction for
$n=6,8,10,12$: the branches are the first and the last $n/2$ modes, the
quench~\eqref{eq:example-quench} acts on $n$ modes, and $R=e^{-\tau A}$.
In each plotted two-branch example the row supports are complements. If one
count is below $k_0$, its partner is frozen. Otherwise the larger count is at least $N_{\rm s}/2$. If both branches
are frequent they form one class; if only one is frequent the other
attaches to it. The reference is the larger-count branch with the stated
tie rule. The plotted run totals charge full test
batches. The prescribed flat probe uses all modes at these four sizes, and the
selected verification horizon is zero at every request. The failure probability is
$\delta=1/2$ and $f=1/(n2^k)$.
\begin{center}
\footnotesize\setlength{\tabcolsep}{3.5pt}
\begin{tabular}{@{}rlrlll@{}}
\hline
$n$ & requests $f$ & exec. & $N/10^6$ & $N/U_0$ & $\max X/f$\\
\hline
6 & $1/6,\dots,1/1536$ & 27 & 0.21--6.41 & 0.52--0.56 & 0.0362\\
8 & $1/8,\dots,1/256$ & 12 & 0.42--2.66 & 0.52--0.57 & 0.0250\\
10 & $1/10,\dots,1/80$ & 4 & 0.78--2.07 & 0.53--0.56 & 0.0218\\
12 & $1/12,\ 1/24$ & 2 & 1.31--1.80 & 0.56--0.56 & 0.0064\\
\hline
\end{tabular}
\end{center}
Every stage accepted its first attempt in all 45 executions.
The bound $U_0$ allows for retries and charges
$2(\lfloor n/3\rfloor-1)$ coherence settings. Each observed block used two.
The observed $Nf/n^2$ ranged from $115.9$ to $950.4$. Trial
$t=0,1,\dots$ of the request $f=1/(n2^k)$, $k=0,1,\dots$, uses
$\phi=2\pi(t+k/2)/3$ and the seed tuple $(2026111601,3,n,k,t)$, with
outcome laws on $2^n$ dimensions.

\noindent\emph{Limits.}
The simulation uses ideal devices and floating-point arithmetic.
All reported reference-assisted trials used $\delta=1/2$ and selected
direct reconstruction with a flat probe. The 45 executions in Fig.~\ref{fig:reference-runs} cover
$n=6,8,10,12$. Figure~\ref{fig:size-scaling} also includes three executions
at $n=16$ and $f=1/16$. These trials do not measure a worst-case failure probability or an asymptotic exponent.

\subsubsection{Reference share of the two-branch example}
\label{app:example-reference-share}
Fix the calibration and $0<\delta\le1/2$. Let $n\ge6$ be even and
$0<f\le1/n$. Write $(I,R,T)=(N_{\rm input},N_{\rm reference},N_{\rm target})$.
The two complementary occupation strings of this example are the only
possible input records. If both survive structure selection they belong to
one block and give one partner for coherence recovery. Retaining fewer
branches can only reduce the number of settings. Thus at most two
quadratures are measured per completed fit.

Consider one fit with tolerance $g\in\{f,f/10\}$. Its failure allowance is
bounded below by a positive constant depending only on $\delta$.
Equation~\eqref{gen:eq:input-batches} and the cap on $\sigma_{\rm c}$ give
\[
 N_{\rm s}=O((n/g+n^2)\log n)=O(n^{3/2}/g).
\]
Here $ng\le1$ and $\log n=O(\sqrt n)$. In~\eqref{gen:eq:tail-parts},
$\bar B=n/6$, $w_\phi=\Theta(n)$ and $D_\phi=O(n)$. Hence
\[
 c_w=O(n^{3/2}),\qquad c_\phi=O(n),\qquad
 t_{\rm rd}=O(1/g),
\]
where the last estimate uses $H\ge(3/5)\sqrt{g/2}$. The final leakage test
takes $m_\ell=O(1/g)$ target copies per attempt and its attempt cap
$A_J=O_\delta(1)$ by~\eqref{gen:eq:confidences}. Thus even allowing retries,
\[
 I+T\le N_{\rm s}+N_{\rm w}+2N_\phi+A_Jm_\ell=O(n^{3/2}/g)
\]
within a fit. The outer horizon is bounded by
$\lceil\log_2(1/\delta)\rceil$ and each outer test uses $O_\delta(1/f)$
targets. Consequently $\E(I+T)=O(n^{3/2}/f)$.

For the reverse estimate on references, the first fit reaches the final
correction on its good event, whose probability is at least one half
for either choice of its failure allowance. On reaching this stage it
acquires the entire estimation batch $N_J$. Its defining budget gives
\[
\begin{gathered}
 w\ge (K/\sqrt{N_J}+b)^2,\\
 N_J\ge K^2/w=\frac{A_\ell^2\widetilde V(u)}{w f_\ell}
       \ge\frac{n(n-1)}{wg}.
\end{gathered}
\]
Indeed $b\ge0$, $A_\ell\ge1$, $\widetilde V(u)\ge n(n-1)$ and
$0<f_\ell\le g$. Therefore $\E R=\Omega(n^2/f)$, and
$\E R\le\E N=O(n^2/f)$ by~\eqref{gen:eq:bound}. In particular,
\[
\begin{aligned}
 0\le 1-\frac{\E R}{\E N}
 &=\frac{\E(I+T)}{\E R+\E(I+T)}\\
 &\le\frac{\E(I+T)}{\E R}=O(n^{-1/2}).
\end{aligned}
\]
This ratio weights each execution by its preparation count.
At $f=1/1536$ and $\delta=1/2$, the six-mode records give a reference
share of $54.5\%$. The first-attempt schedule gives $62.0\%$, $67.7\%$
and $72.0\%$ at eight, ten and twelve modes, respectively, with two
coherence settings and full test batches throughout.
\newpage
\clearpage\onecolumn
\subsection{Learning algorithms}
\label{app:learning-algorithms}
\subsubsection{Copy-optimal nonadaptive paired reconstruction}
The notation and all constants are defined in
Appendix~\ref{app:reconstruction}; in particular,
$s=\log(1024/\delta)$, $\rho_0=1/256$, $q=28B$,
$t_*=2.4\cdot10^{-7}$ and $\sigma_0=5863t_*$.
Each bank $X\in\{A,P,W,R\}$ has one independent batch of $K_X$ copies per
quartet $I$, and $y_X$ is its vector of empirical means. $Q(y)$ is the
flattening of $y$ on pair space and $Y(y)$ is its alternating four-tensor.
A \emph{guard} is a check. Whenever a guard fails, the algorithm stops and
returns the fallback, which is the known paired reference $P_0$.
\begin{algorithm}[H]
\caption{Copy-optimal nonadaptive fixed-basis quartic learner
(Theorem~\ref{thm:reconstruction})}
\label{alg:paired-quartic}
\begin{algorithmic}[1]
\Require Independent copies of one $\rho\in\mathcal O$; $n=4B$;
 $0<\epsilon\le \epsilon_{\max}<1/\sqrt2$, $0<\delta\le1/8$.
\Ensure A Gaussian preparation with trace distance at most $\epsilon$,
 except with probability at most $\delta$.
\AlgPart{Acquisition. No later step changes what is measured.}
\For{each bank $X\in\{A,P,W,R\}$ and each quartet $I$}
 \AlgState{Measure the fixed-basis sign of $B_I$ on $K_X$ fresh copies, with
  $K_X$ from \eqref{eq:opt-bank-budgets}. Store the mean of the signs as
  $(y_X)_I$.}
\EndFor
\AlgPart{Part I. Locate the blocks (banks $A$ and $P$).}
\AlgState{If $B=1$, the full eight-dimensional Majorana space is the sole
 candidate block: go to Part II.}
\AlgState{\emph{Pilot.} Let the columns of $U_A$ be orthonormal eigenvectors of
 $Q(y_A)$ for all eigenvalues with $|\lambda|>1/2$.
 \AlgGuard\ $U_A$ has exactly $q$ columns.}
\AlgState{\emph{Refresh.} Let $U'$ be an orthonormal basis of the range of
 $Q(y_P)U_A$.
 \AlgGuard\ this range has rank $q$, certified by
 $\sigma_{\min}(Q(y_P)U_A)>1/2$.}
\AlgState{\emph{Commutant.} Read the columns of $U'$ as skew matrices
 $A'_1,\ldots,A'_q$ and form $\mathscr C'$ by \eqref{eq:opt-commutant}. Let
 $\hat\Pi$ project onto its $B$ lowest eigenvectors.
 \AlgGuard\ $\lambda_B(\mathscr C')<1$ and $\lambda_{B+1}(\mathscr C')>6$.}
\AlgState{Start with an empty ordered list of representatives.}
\For{$j=1,\ldots,T$ with $T$ from \eqref{eq:opt-probe-count}}
 \AlgState{Draw a rounded Gaussian symmetric matrix $\widetilde G$ by the finite
  sampling rule of Appendix~\ref{app:reconstruction}. On overflow or a norm
  failure, discard this trial.}
 \AlgState{Form $Z=\widetilde\Pi\widetilde G$ with the certified approximation
  $\widetilde\Pi$ of $\hat\Pi$. Let $V$ be the top-eight eigenspace of $Z$.}
 \AlgState{Discard an ambiguous pencil or one failing either test in
  \eqref{eq:opt-pencil-tests}. Otherwise retain its projector only if its
  Frobenius distance from every representative exceeds $1/2$.
  End the search once $B$ representatives have been retained.}
\EndFor
\AlgState{\AlgGuard\ exactly $B$ representatives were found.}
\AlgPart{Part II. Assemble the frame (bank $W$).}
\AlgState{Start with an empty frame $Q_{\rm found}$ of found blocks, with
 projector $P_{\rm found}$.}
\For{each retained candidate, in list order, with orthonormal basis $U$}
 \If{$\|Q_{\rm found}^{\mathsf T}U\|_{\op}\le1/2$}
  \AlgState{\emph{Insert.} Append an orthonormal basis of
   $(I-P_{\rm found})U$ to the frame.}
  \For{three times}
   \AlgState{\emph{Correct.} For every block $b$ of the frame form $z_b$
    by \eqref{eq:opt-block-update}, with $Y=Y(y_W)$.
    \AlgGuard\ the absolute Hodge score $|\langle z_b,\star z_b\rangle|$ and
    the top eigengap of the local decoder exceed one.
    Decode $\hat O_b=\mathcal D(z_b)$.}
   \AlgState{For every $b$ form $X_b^\perp$ by \eqref{eq:opt-block-update}.
    \AlgGuard\ $\|X_b^\perp\|_{\op}^2\le1/2$. Complete it to $X_b$.}
   \AlgState{\AlgGuard\ $\sigma_{\min}([X_b]_b)\ge1/2$. Replace the frame by
    the polar factor of $[X_b]_b$.}
  \EndFor
 \EndIf
\EndFor
\AlgState{\AlgGuard\ $B$ blocks were inserted and the completed frame is in
 $SO(2n)$. Denote it by $R_{\rm c}$.}
\AlgPart{Part III. Refine and compile (bank $R$).}
\AlgState{Put $\widetilde X_0=0$ and
 $\phi(X)=\langle y_R,\wedge^4(R_{\rm c})F(X)\rangle$ on the ball
 $\{X\in\mathfrak m:\|X\|_F\le\rho_0\}$.}
\For{$i=0,\ldots,k-1$ with $k=\lceil10\log(16/\epsilon^2)\rceil$}
 \AlgState{Let $\widetilde X_{i+1}$ be the stored projection of
  $\widetilde X_i+\nabla\phi(\widetilde X_i)/18$ onto this ball.}
\EndFor
\AlgState{Compile $R_{\rm c}e^{\widetilde X_k}$, with $\widetilde X_k$ the last
 stored iterate, into at most $d(d-1)/2$ Gaussian plane rotations and round
 their angles. Optimization and rounding together use trace error at most
 $\epsilon/2$.}
\AlgState{Return the circuit, or the fallback if a numerical precision guard
 fails.}
\end{algorithmic}
\end{algorithm}
\noindent\emph{Correctness and resources.}
Lemmas~\ref{lem:opt-range-refresh}--\ref{lem:opt-commutant-stability}
control the refreshed range and measured commutant.
Lemma~\ref{lem:opt-candidates} produces one accurate candidate per block;
Lemma~\ref{lem:opt-joint-noise} permits all corrections to reuse bank $W$.
Lemma~\ref{lem:opt-initialization} supplies the local chart required by
Lemma~\ref{lem:opt-refinement}, whose bank is independent.
The certificate and finite-precision budget give
\eqref{eq:reconstruction-guarantee}; all copies are counted in
\eqref{eq:opt-total-copies}. No physical measurement is chosen from earlier
outcomes. Lemma~\ref{lem:opt-lower} allows such adaptive choices and
still proves the matching lower bound.
The search uses at most $T$ pencil trials. At most $B$ insertions follow,
with three corrections each, before the $k$ ascent steps. Every guard has the same
fallback. So the procedure stops on every record, and its output is always
defined. The thresholds, $T$ and $k$ are explicit. The universal constant
$C$ in \eqref{eq:opt-bank-budgets} is not computed here.
\clearpage
\subsubsection{Reference-assisted reconstruction and verification}
\label{gen:sec:learner}
Write $\mathsf{Fit}(g,\zeta)$ for the five-stage reconstruction below.
Within each call the schedule of Appendix~\ref{gen:app:schedule} is evaluated
with $(f,\delta)=(g,\zeta)$. The test of a completed preparation is
separate from the frame and leakage tests inside $\mathsf{Fit}$.

\begin{algorithm}[H]
\caption{Reference-assisted structured-state learner (Theorem~\ref{thm:active-main})}
\label{alg:general-reference}
\begin{algorithmic}[1]
\Require The three requests of Section~\ref{sec:access}, $n\ge6$,
$0<f\le1/2$ and $0<\delta<1$.
\Ensure A Gaussian circuit and a structured input with infidelity at most
$f$ with probability at least $1-\delta$.
\AlgState{Set $a=f/5$. Compute $(m,k)=\mathsf{BinTest}(a/2,a;1/3,\delta)$
and the minimizing horizon $j$ of~\eqref{gen:eq:finite-verified}.}
\For{$i=1,\ldots,j$}
 \AlgState{Using fresh data, set $(\hat R,\ket{\hat\Xi})=\mathsf{Fit}(f/10,1/4)$.
 An inner fallback is tested like any other proposal.}
 \AlgState{Undo $U_{\hat R}$ on $m$ fresh targets and count rejections $Z$
 using independent settings of Lemma~\ref{gen:lem:whole-test}.
 The test may stop as soon as its pass/fail decision is determined.}
 \If{the count test accepts}
  \State \Return $(\hat R,\ket{\hat\Xi})$.
 \EndIf
\EndFor
\State \Return a fresh $\mathsf{Fit}(f,z_j)$.
\end{algorithmic}
\end{algorithm}

\noindent\emph{Reconstruction $\mathsf{Fit}(g,\zeta)$.}
Use the fixed calibration, the uniform active-subset probe and the
batches and caps from Appendix~\ref{gen:app:schedule}.
In steps 2--4 a failed test discards only the candidate; the next attempt
uses fresh records and the same incoming frame. A malformed structure
or exhaustion of any inner cap returns $(I,\ket0)$ from this call.
\begin{enumerate}[label=\arabic*.,itemsep=4pt,topsep=4pt]
 \item \emph{Learn the occupations.} Apply Lemma~\ref{gen:lem:structure}
 with failure parameter $\zeta/2$ to $N_{\rm s}$ input records. Keep the
 core, group frequent atoms by disjointness, and assign each other eligible
 atom using the fixed frequent unions. Keep the reference $G_b$ of every block.
 Reserve $N_{\rm w}$ independent input records, uninspected until step~5.
 A malformed structure ends the call before these records are requested.
 \item \emph{Initialize the frame.} In each of at most $A_0$ attempts
 use $N_a$ matching records to estimate the mode planes and $N_b$ fresh
 matching records to correct them (Lemmas~\ref{gen:lem:planes}
 and~\ref{gen:lem:mcorr}). Keep the first frame $\hat R_0$ that passes
 the reference click test $\mathsf T(u_0,\delta_0)$.
 \item \emph{Refine the frame.} At levels $k=1,\ldots,J-1$, make at most
 $A_k$ attempts. Average $N_k$ pairing records of strength
 $\tau_k=\sqrt{u_{k-1}/2}$ at $\hat R_{k-1}$ and set
 $\hat R_{\rm cand}=\hat R_{k-1}\polar(I+A_{\rm av})$.
 Keep the first $\hat R_k=\hat R_{\rm cand}$ passing $\mathsf T(u_k,\delta_k)$.
 \item \emph{Control leakage.} In each of at most $A_J$ attempts
 repeat the correction at $\hat R_{J-1}$ with $N_J$ records of strength
 $\sqrt{u_{J-1}/2}$. Keep the first $\hat R$ passing both
 $\mathsf T(u_{\rm c},\delta_J)$ on references and $\mathsf T_\ell$
 on fresh targets. Use the thresholds and batch sizes
 of~\eqref{gen:eq:click-schedule} and~\eqref{gen:eq:leak-size}.
 \item \emph{Recover the coherences.} For each of the two quadratures
 per partner in Lemma~\ref{gen:lem:readout}, undo $\hat R$ and measure
 $N_\phi$ fresh targets. One setting serves a partner in every block
 simultaneously. Count the occupation frequency $\hat w_{G_b}$ of $G_b$ in
 the reserved records and form the column $\hat v_b$ with entry $\hat w_{G_b}$
 at $G_b$ and $\hat\kappa_{ba}$ at every partner $a$. The block state is
 $\sum_a\hat v_{ba}\ket{F_a}/\|\hat v_b\|$ for a nonzero column and
 $\ket{F_{G_b}}$ otherwise. Move modes of zero-amplitude branches to the
 vacant core. If one branch remains in a block, fold its modes into the
 occupied core and absorb its phase into the global phase. Return $\hat R$
 and the product of the resulting core and block states.
\end{enumerate}
Empty retained families require no coherence
measurements. Ties use a fixed ordering and singular plane rounding
returns $I$. Matching and pairing records use uniformly random Fock
references. A matching record uses the random perfect matching of
Lemma~\ref{gen:lem:matching}; it involves no probe. A pairing record uses
the measurement of Lemma~\ref{gen:lem:record} with the flat probe of
Appendix~\ref{gen:app:schedule} on $n'=Q+1$ modes, at the strength of its
stage, with a fresh active subset and fresh phases in every record. The
click test counts the runs whose output string differs from the chosen
input string. The leakage test counts outputs outside the learned
patterns of Section~\ref{gen:app:tail}. Both are measured after undoing
the candidate frame.
The whole-preparation test uses only target copies. Its randomized
settings test pattern membership and the relative amplitudes within the
proposed blocks.
\end{document}